\documentclass[12pt,fbe_tez]{report}
\usepackage{fbe_tez}
\usepackage{amssymb}
\usepackage{amsmath}
\usepackage{amsthm}
\usepackage{textcomp}
\usepackage{array}
\usepackage[latin5]{inputenc}
\usepackage{verbatim} 
\usepackage{enumerate}
\usepackage{enumitem}
\usepackage[T1]{fontenc}
\usepackage{cite}

\newtheorem{theorem}{Theorem}[chapter]
\newtheorem{lemma}[theorem]{Lemma}
\newtheorem{corollary}[theorem]{Corollary}

\newtheorem{fact}[theorem]{Fact}
\newtheorem{definition}[theorem]{Definition}
\newtheorem{remark}[theorem]{Remark}

\theoremstyle{remark}
\newtheorem{openproblem}{Open Problem}

\newcommand{\bra}[1]{\langle #1|}
\newcommand{\ket}[1]{|#1\rangle}
\newcommand{\braket}[2]{\langle #1|#2\rangle}
\newcommand{\cent}[0]{\mbox{\textcent}}
\newcommand{\dollar}[0]{\$}

\newcommand{\leftstate}[1]{\overleftarrow{#1}}
\newcommand{\rightstate}[1]{\overrightarrow{#1}}
\newcommand{\stopstate}[1]{\mspace{-5mu}\downarrow\mspace{-5mu} #1}

\newcommand{\rightconf}[1]{|\overrightarrow{#1}\rangle}
\newcommand{\stopconf}[1]{|\mspace{-5mu}\downarrow\mspace{-5mu} #1\rangle}

\newcommand{\LofC}[0]{\left\langle\right.}
\newcommand{\RofC}[0]{\left.\right\rangle}
\newcommand{\LRofC}[1]{\left\langle #1 \right\rangle }

\newcommand{\setD}{<\mspace{-4mu}>}

\newcommand{\rev}[0]{\mathrm{r}}
\newcommand{\trans}[0]{\mathrm{T}}

\title{CLASSICAL AND QUANTUM COMPUTATION WITH SMALL SPACE BOUNDS}
\turkcebaslik{AZ BELLEĞE SAHİP KLASİK VE KUANTUM HESAPLAMA}
\degree{B.S., Computer Education and Educational Technology,
Bo\u{g}azi\c{c}i University, 2004 \\
       M.S., Computer Engineering, Bo\u{g}azi\c{c}i University, 2007}
\author{Abuzer Yakary{\i}lmaz}
\program{Computer Engineering}
\supervisor{Prof. Ahmet Celal Cem Say}
\examinerii{Assoc. Prof. Can Özturan}
\examineri{Assist. Prof. Ali Taylan Cemgil}
\examineriii{Assoc. Prof. Ferit Öztürk}
\examineriv{Assist. Prof. Hüsnü Yenigün}
\dateofapproval{25.11.2010}

\begin{document}

\pagenumbering{roman}
%
\makephdtitle      

\makeapprovalpage

\newpage

~~\newline~~
~~\newline~~
~~\newline~~
~~\newline~~
~~\newline~~
~~\newline~~
~~\newline~~
~~\newline~~
~~\newline~~
~~\newline~~
~~\newline~~
~~\newline~~
~~\newline~~
~~\newline~~
~~\newline~~
~~\newline~~
~~\newline~~
~~\newline~~
~~\newline~~
~~\newline~~
~~\newline~~
~~\newline~~
~~\newline~~
~~\newline~~
\begin{flushright}
\textit{To My Sister Fatma...}
\end{flushright}

\newpage

\begin{acknowledgements}

First of all, I would like to thank my PhD advisor A. C. Cem Say.
It has been a great experience for me to work with Cem Say
and it is an honour for me to be his first PhD graduate.
With his guidance, I have learned many things not only about ``theory'' 
but also about how to be a researcher.
His research excitement has always motivated me and 
they were very special moments to see his eyes' shining 
whenever we came to a new finding.

I also would like to express my gratitude to Ali Taylan Cemgil, Kıvanç Mıhçak, Can Özturan, Ferit Öztürk, 
and Hüsnü Yenigün, my thesis committee members, 
who devoted their time and energy to this research.

I would like to thank Dilek Çankaya for her great support at first period of this research.
During that time, I had completed the main framework of the thesis.

I would like to thank R\={u}si\c{n}\v{s} Freivalds for his collaborative works and 
his efforts to our joint works.
I would like to thank Andris Ambainis for offering to write a chapter on ``quantum automata'' together,
for his great advice and for many helpful discussions and for kindly accepting my short visit to Riga.
I would like to thank John Watrous for his great advice and for many helpful discussions and
for kindly accepting my short visit to Waterloo.
I would like to thank Farid Ablayev, Ali Taylan Cemgil, Flavio D'Alessandro,
Juraj Hromkovi\v{c}, Maris Ozols, Ferit Öztürk, Pascal Tesson, Mikahil Volkov for their helpful 
answers to my questions.

I would like to thank the organizing committee of 
SATA 2008 (Lisbon), TQC 2009 (Waterloo), AutoMathA 2009 (Li\`{e}ge), CSR 2009 (Novosibirsk),
Randomized and Quantum Computation 2010 (Brno),
QIP 2010 (Zurich), and QIP 2011 (Singapore).

I am grateful to my colleagues from the department of Computer Engineering at Boğaziçi University,
with whom we have not only shared many academic experiences but also participated in many social activities.
I would like especially to thank all participants of our regularly scheduled basketball games.

This research has been partially supported by Boğaziçi University Scientific Research Projects 
with grants 07A103 and 08A102 and the Scientific and Technological Research
Council of Turkey (TÜBİTAK) with grant 108E142.

The most beautiful finding during this research is Gönül. The proofs are written in my heart.

I would like to thank my family for their continuous encouragement and 
motivation not only with this work but also with much else.

\end{acknowledgements}

\begin{abstract}

In this thesis, we introduce a new quantum Turing machine model that supports 
general quantum operators, together with its pushdown, counter, and finite automaton variants,
and examine the computational power of classical and quantum machines using small space bounds
in many different cases. The main contributions are summarized below.

Firstly, we consider quantum Turing machines in the unbounded error setting:
(i) in some cases of sublogarithmic space bounds, the class of languages recognized by quantum Turing machines
is shown to be strictly larger than that of classical ones;
(ii) in constant space bounds, the same result can still be obtained for restricted quantum Turing machines;
(iii) the complete characterization of the class of languages recognized by 
realtime constant space nondeterministic quantum Turing machines is given.

Secondly, we consider constant space-bounded quantum Turing machines in the bounded error setting:
(i) we introduce a new type of quantum and probabilistic finite automata 
with a special two-way input head which is not allowed 
to be stationary or move to the left but has the capability to reset itself
to its starting position;
(ii) the computational power of this type of quantum machine 
is shown to be superior to that of the probabilistic machine;
(iii) based on these models, two-way probabilistic and two-way classical-head quantum 
finite automata are shown to be more succinct than two-way nondeterministic finite automata 
and their one-way variants;
(iv) we also introduce probabilistic and quantum finite automata with postselection with their
bounded error language classes, and give many characterizations of them.

Thirdly, the computational power of realtime quantum finite automata augmented with a 
write-only memory is investigated by showing many simulation results for different kinds of counter automata.
Parallelly, some results on counter and pushdown automata are obtained.

Finally, some lower bounds of realtime classical Turing machines in order to recognize a nonregular language
are shown to be tight. Moreover, the same question is investigated for some other kinds of realtime machines and
several nonregular languages recognized by them in small space bounds are presented.

\end{abstract}
%
%
\begin{ozet}

Bu tezde genel kuantum operatörlerini destekleyen yeni bir kuantum Turing makine modeli ile birlikte 
onun yığıt-bellekli, sayaçlı ve sonlu bellekli modelleri tanımlandı ve az belleğe sahip 
klasik ve kuantum makinelerin hesaplama güçleri bir çok durum için incelendi.
Temel katkılarımız aşağıda özetlenmiştir.

İlk olarak kuantum Turing makineleri sınırlı olmayan hata açısından ele alındı:
(i) logaritma-altı bellek kullanılan bazı durumlarda hesaplama gücü açısından
kuantum Turing makinelerinin klasik muadillerinden üstün olduğu gösterildi;
(ii) aynı sonuç sonlu belleğe sahip kısıtlı kuantum Turing makineleri için de elde edildi;
(iii) gerçek-zamanlı sonlu belleğe sahip belirlenimci olmayan kuantum Turing makinelerinin
tanıdığı dil ailesi belirlendi.

İkinci olarak, sonlu belleğe sahip kuantum Turing makineleri sınırlı hata açısından ele alındı:
(i) sola gitme veya durma hakkı yasaklanmış fakat kendisini başlangıç noktasına taşıyabilen
özel kafaya sahip yeni çift-yönlü kuantum ve olasılıksal sonlu durumlu makineler tanımlandı;
(ii) hesaplama gücü açısından bu tür kuantum makinelerin olasılıksal olanlardan daha güçlü olduğu gösterildi;
(iii) bu modeller temelinde, çift-yönlü olasıksal ve klasik kafaya sahip kuantum sonlu durum makinelerin,
çift-yönlü belirlenimci olmayan sonlu durum makineler ile kendi tek-yönlü muadillerinden 
daha az sonlu bellek kullandıkları gösterildi;
(iv) ayrıca olasılıksal ve kuantum sonseçimli sonlu durumlu makineler ile sınırlı hata payı
ile tanıdıkları dil sınıfları tanımlandı ve bir çok özellikleri gösterildi.

Üçüncü olarak, sadece yazma hakkı olan bir hafıza eklenen gerçek-zamanlı kuantum sonlu durumlu
makenelerin hesaplama gücü, farklı türdeki sayaçlı makineler üzerinden yapılan bir çok benzetim ile incelendi.
Paralel olarak, sayaçlı ve yığıt-bellekli makinelere dair bazı sonuçlar elde edildi.

Son olarak, literatürde geçen bazı alt sınırların, 
düzenli olmayan bir dili tanıyan gerçek-zamanlı klasik Turing makineler için mümkün olan en iyi sınırlar oldukları
gösterildi.
Ek olarak, benzer soru diğer tür gerçek-zamanlı makineler için araştırıldı ve 
onlar tarafından az bellek ile tanınan birçok düzenli olmayan dilin varlığı gösterildi.

\end{ozet}

\tableofcontents
\listoffigures
\listoftables



\chapter{INTRODUCTION} \label{int}
\pagenumbering{arabic}

Computation with small space-bounds has always had a special emphasis 
in the field of theoretical computer science.
In this thesis, we examine both classical and quantum small space complexity classes.
Our main tools are Turing machines (TM) and some restricted variants of them,
i.e. one-way or realtime TMs, finite automata, and counter and pushdown automata.
Since quantum computation is relatively a young research field 
(and so it has been less investigated than classical computation), 
our main focus is on quantum machines and their space-bounded classes.

Moreover, we think that examining small space complexity classes or 
the models working in small space-bounds is useful for understanding and comparing
different kinds of computational resources, such as \textit{nondeterminism}, \textit{randomization},
\textit{quantumness}, etc.
Therefore, we also try to investigate the cases in which quantumness (or randomization) 
has advantages and the cases having no advantage.

Our interest of ``space'' can be classified into four general cases,
which can also be seen as the understanding of what we mean by small space-bounds:
\begin{enumerate}
	\item constant space,
	\item sublogarithmic space,
	\item sublinear space, and
	\item linear space.
\end{enumerate}
Note that, such a classification may not be meaningful for some models since
their computational power increases only when using logarithmic or linear space \cite{Ga84,Sz94}.
On the other hand, even constant space can be sub-classified for some models
in terms of language recognition \cite{Fr81,KW97,AW02}.

In the following part, based on the outline of the thesis, 
we give a short description of each chapter by highlighting the main results.

Chapter \ref{sbc} begins with the preliminaries in which 
the basic notations, terminologies, and language recognition settings used throughout the thesis are presented.
After that, the conventional TM model for space bounded computations and its classical variants, i.e.
finite automata (FA), counter automata (CA), and pushdown automata (PDA), 
are given with their formal definitions and computational specifications.
The chapter ends with the definitions of many space classes.

In Chapter \ref{qtm}, we introduce a new kind of quantum Turing machine (QTM)
allowing to implement general quantum operators and its FA, PDA, and CA variants. 
Since the previously defined QTM models for space-bounded computations \cite{Wa98,Wa99,Wa03}
did not reflect the full generality of quantum mechanics,
our new QTM model is one of the contributions of the thesis.
We also present several well-formedness conditions, i.e. 
a list of constraints obeyed by the transition function
of the machine, for the quantum models\footnote{Some of them are given in Appendix \ref{well}.}. 
Moreover, we present the standard technique for quantum machines in order to
simulate their classical counterparts exactly, such that both machines agree on the accepting value
for any given input string.
Lastly, the ``linearization'' technique of a quantum finite automaton (QFA) operating 
in realtime\footnote{The input head of the machine is allowed to move only to the right in each step.} is given.

In Chapter \ref{uerr}, we focus on our results related with unbounded error, both in the general case and
in one-sided (equivalently, nondeterministic) case.
Firstly, we show that one-way\footnote{The input head of the machine is allowed either to be stationary or 
to move to the right in each step.} QFAs (1QFAs), that is, one-way QTMs (1QTMs) with constant space, 
can recognize some nonstochastic languages that were shown to be not recognizable 
by a probabilistic Turing machine (PTM)
(resp., one-way PTMs (1PTMs)) in space $ o(\log\log n) $ (resp., $ o(\log n) $).
This is one of our superiority results of quantum computation over classical computation in $ o(\log\log n) $
and $ o(\log n) $ spaces.
Thus, we partially solve an open problem addressed in \cite{Wa98}, i.e.
the relationship between QTMs and PTMs for sublogarithmic space bounds.
In the next part, we solve another open problem introduced in \cite{BP02}, i.e.
the complete characterization of a restricted type realtime QFA (RT-QFA), namely 
realtime Kondacs-Watrous QFA \cite{KW97} (RT-KWQFA), in the unbounded error setting.
We show that the language class recognized by RT-KWQFAs  
is equal to that of realtime probabilistic FAs (RT-PFAs) in this setting,
using a simulation technique.
Moreover, this technique is also used in order to demonstrate some other results.
Thirdly, we present the complete characterization of the class of languages recognized by 
realtime nondeterministic QFAs (RT-NQFAs) and then show some non-trivial properties of this class.

In Chapter \ref{berr}, we shift our focus on the results related to 
constant space-bounds and the bounded-error cases. 
Firstly, we define a new type of QFAs and PFAs with a special two-way input head which is not allowed 
to be stationary or move to the left, but has the capability to reset the input head
to its initial position and to change the internal state to a specified one. 
If the specified internal state is set to be the initial one,
these machines are called realtime QFAs and PFAs with restart, 
denoted by RT-QFA$ ^{\circlearrowleft} $s and RT-PFA$ ^{\circlearrowleft} $s.
Secondly, we show that the class of the languages recognized by this kind of QFAs is a proper superset of
that of PFAs in bounded error setting. 
Moreover, based on these models, we show that two-way QFAs and two-way PFAs (2QFAs and 2PFAs) can be more concise than
two-way nondeterministic FAs (2NFAs) and their one-way variants.
Additionally, in a special setting, we also show that 2QFAs can be more concise than 2PFAs.
Thirdly, we present more space- and time-efficient algorithms for 2QFAs in order to recognize 
some nonregular languages.
In the last part of this chapter, we introduce the FAs endowed with 
an unrealistic theoretical capability, i.e. \textit{postselection}.
By postselection, it is meant that the final decision on the input is given by a specified subset of 
the computation outcomes and so the rest of the computation outcomes are discarded \cite{Aa05}.
We introduce both RT-QFAs and RT-PFAs with postselection (RT-PostQFA and RT-PostPFA) 
and then give some characterizations of the classes of the languages recognized by them.
As an interesting observation, the class of languages recognized by RT-PostQFA and RT-QFA$ ^{\circlearrowleft} $
(resp., RT-PostPFA and RT-PFA$ ^{\circlearrowleft} $) are the same in both bounded and unbounded error settings.
Additionally, 
in \cite{LSF09,SLF10}, a somewhat different QFA model with postselection was defined based on RT-KWQFA.
We also examine probabilistic and quantum versions of these models, and
present some characterizations of the classes of languages recognized by them.

In Chapter \ref{wom}, we present our results that are extensions of the ones we gave in 
\cite{SY10A} and \cite{YFSA10}, that is, a RT-QFA augmented with a 
\textit{write-only memory} (WOM) (RT-QFA-WOM). 
Contrary to classical computation, due to the interference of the configurations,
the computational power of a quantum machine can be increased by using a WOM.
Throughout the chapter, we investigate the computational power of RT-QFA with 
IOC (\textit{increment-only counter}), POS (\textit{push-only stack}), and WOM 
by showing some simulations of classical machines, such as blind and reversal counter automata, 
and giving some example languages.

The minimum space usage required by a TM in order to recognize a nonregular language has been
determined for many cases \cite{Sz94}, but realtime TMs have not been considered.
In Chapter \ref{rtm}, we show that RT-TMs can share the same lower bounds with one-way TMs (1TMs)
by using a general simulation technique.
Moreover, the same question is investigated for some other kinds of realtime machines and
several nonregular languages recognized by them in small space bounds are presented.

In Chapter \ref{rwork}, we present the remaining results of our work that we chose not to 
organize as a separate chapter. 

Section \ref{conc} is the conclusion of the thesis. 
We present a technical summary of all results. 
A short review of the mathematical formalism used for quantum computation in this thesis
is given in Appendix \ref{qc}.

\chapter{SPACE-BOUNDED COMPUTATION} \label{sbc}

This chapter forms the foundation of the thesis.
In Section \ref{sbc:pre}, we list basic notations and terminologies
and language recognition settings used throughout the thesis.
The details of generic TMs convenient for space bounded computations are presented in Section \ref{sbc:TMs}.
Consequently, the details of generic PDAs and CAs are given in Sections \ref{sbc:PDAs} and \ref{sbc:CAs},
respectively.
Then, the formal definitions of classical (deterministic, nondeterministic, alternating, and probabilistic)
 machines are presented in 
Section \ref{sbc:formal-definition-classical-machine}.
The computational specifications of all types of machines that we cover (classical and quantum) are described in 
Section \ref{sbc:computation-of-machines}.
In the last section (\ref{sbc:space-classes}), 
the generic space complexity classes and some specific space classes introduced in this thesis are presented.

\section{Preliminaries} \label{sbc:pre}

\subsection{Basic Notation} \label{sbc:basic-notation}

\begin{enumerate}
	\item For a given set $ S $, $ |S| $ is the size of the set.
	\item For a given alphabet $ A $, $ A^{*} $ is the set of 
		\textit{the empty string} (or \textit{the empty symbol}), denoted as $ \varepsilon $,
		and all strings obtained by finitely concatenating the symbols of $ A $.
	\item For a given string $ w $, $ |w| $ is the length of $ w $, 
		$ w_{i} $ is the $ i^{th} $ symbol of $ w $.
	\item For a given vector (row or column) $ v $, $ v[i] $ is the $ i^{th} $ entry of $ v $.
	\item For a given matrix $ A $, $ A[i,j] $ is the $ (i,j)^{th} $ entry of $ A $.
	\item Let $ \{ A_{i} \mid 1 \le i \le n \} $  be a set.
		Then, $ \bar{A} $ denotes $ A_{1} \times \cdots \times A_{n} $ and any instance of $ \bar{A} $
		is denoted by small letter, i.e. $ \bar{a} $ or $ \bar{a}_{i \in \{1,\ldots,|\bar{A}| \} } $.
	\item Some fundamental conventions in Hilbert space are as follows:
		\begin{itemize}
			\item $ v $ and its conjugate transpose are denoted as $ \ket{v} $ and $ \bra{v} $, respectively;
			\item the multiplication of $ \bra{v_{1}} $  and $ \ket{v_{2}} $ is shortly written as
				$ \braket{v_{1}}{v_{2}} $;
			\item the tensor product of $ \ket{v_{1}} $  and $ \ket{v_{2}} $ can also be written as
				$ \ket{v_{1}} \ket{v_{2}} $ instead of $ \ket{v_{1}} \otimes \ket{v_{2}} $,
		\end{itemize}
		where $ v $, $ v_{1} $, and $ v_{2} $ are vectors.
\end{enumerate}

\subsection{Basic Terminology} \label{sbc:basic-terminology}

$ \Sigma $ (\textit{input alphabet}): $ \Sigma $ is a finite set of symbols,
	i.e. $ \Sigma = \{ \sigma_{1}, \ldots, \sigma_{|\Sigma|} \} $.
	As a convention, $ \Sigma $ never contains the symbol $ \# $ (\textit{the blank symbol}),
	$ \mathbf{\cent} $ (\textit{the left end-marker of the input}), and 
	$ \dollar $ (\textit{the right end-marker of the input}).
	$ \tilde{\Sigma} $ denotes the set $ \Sigma \cup \{\cent,\dollar\} $.
	Additionally, $ \tilde{w} $ denotes the string $ \cent w \dollar $, 
	for any given input string $ w \in \Sigma^{*} $.
	
$ \Gamma $ (\textit{work tape, stack, or memory alphabet}): $ \Gamma $ is a finite set of symbols,
	i.e. $ \Gamma = \{\gamma_{1},\ldots,\gamma_{|\Gamma|}\} $.
	In case of work tape, $ \Gamma $ contains $ \# $  and in case of memory, $ \Gamma $
	additionally contains $ \varepsilon $.
	On the other hand, in case of stack, $ \Gamma $ contains $ \vdash $ (\textit{the left end-marker of stack}) 
	but not $ \# $ and $ \tilde{\Gamma} = \Gamma \cup \{\#\} $.

$ \Delta $ (\textit{the set of outcomes}) and $ \Omega $ (\textit{the finite register alphabet}):
	In quantum computation, in order to implement general quantum operators, the main system
	is subjected to interact with a writeable finite register (See Appendix \ref{qc} for details).
	$ \Omega $ represents the alphabet of this finite register, 
	i.e. $ \Omega = \{\omega_{1},\ldots,\omega_{|\Omega|}\} $.
	In some cases, a selective (projective) measurement is done on the finite register. 
	In those contexts, $ \Delta $ represents the finite set of outcome results of the measurement,
	i.e. $ \Delta=\{\tau_{1},\ldots,\tau_{|\Delta|}\} $;
	$ \Omega $ is partitioned into $ |\Delta| $ pairwise disjoint parts,
	i.e. 
	\begin{equation}
		\Omega = \bigcup\limits_{\tau \in \Delta} \Omega_{\tau};
	\end{equation}
	and the projective measurement $ P $ is specified as
	\begin{equation}
		P = \{ P_{\tau \in \Delta} \mid P_{\tau} = \sum_{\omega \in \Omega_{\tau}} \ket{\omega}\bra{\omega} \}.
	\end{equation}
	In its \textit{standard usage}, $ \Delta $ contains three elements, i.e. $ \Delta = \{n,a,r\} $, and
	the following actions are associated with the measurement outcomes:
	\begin{itemize}
		\item ``$n$'': the computation continuous;
		\item ``$a$'': the computation halts, and the input is accepted;
		\item ``$r$'': the computation halts, and the input is rejected.
	\end{itemize}

$ Q $ (\textit{the set of internal states}): $ Q $ is a finite set of internal states,
	i.e. $ Q = \{ q_{1},\ldots,q_{|Q|} \} $. 
	Unless otherwise specified, we assume the following:
	\begin{itemize}
		\item $ q_{1} $ is \textit{the initial state};
		\item for classical computational models 
			except realtime finite automata (see Section \ref{sbc:formal-definition-classical-machine}),
			$ Q $ is partitioned into three disjoint subsets, i.e.
			$ Q_{a} $ (the set of \textit{the accepting states}),
			$ Q_{r} $ (the set of \textit{the rejection states}), and 
			$ Q_{n} = Q \setminus \{ Q_{a} \cup Q_{r} \} $ (the set of \textit{the nonhalting states});
		\item for realtime finite automata including quantum ones, 
			only one subset of $ Q $ is defined, $ Q_{a} $.
	\end{itemize}

$ \setD $ (\textit{the set of head directions}): 
	$ <\mspace{-4mu}> $ is set $ \{ \leftarrow, \downarrow , \rightarrow \} $, where 
	``$ \leftarrow $'' means that the (corresponding) head moves one square left,
	``$ \downarrow $'' means that the head stays on the same square, and
	``$ \rightarrow $'' means that the head moves one square right.
	When the tape is used as a stack,	the arrows are also associated with the following stack operations:
	``$ \leftarrow $'' means that the top symbol is popped from the stack,
	``$ \downarrow $'' means that the top symbol is popped from the stack and 
	a new symbol is pushed into the stack, and
	``$ \rightarrow $'' means that a symbol is pushed into the stack.
	As a special case, $ \rhd $ is set $ \{ \downarrow,\rightarrow \} $.
	Additionally, $ D_{i} $, $ D_{w} $, and $ D_{s} $ are some functions from $ Q $ to $ <\mspace{-4mu}> $ 
	(or to $ \rhd $).
	As will be seen later, the subscripts ``i'', ``w'', and ``s'' correspond to input tape, work tape, and
	stack, respectively. 

$ \lozenge $ (\textit{the counter operations}): $ \lozenge $ is set $ \{ -1, 0 , +1 \} $, 
	where ``$ -1 $'' means that the value of the counter is decreased by 1, ``$ 0 $'' means that the value 
	of the counter is not changed, and ``$ +1 $'' means that the value 
	of the counter is increased by 1. As a special case, $ \vartriangle $ is set $ \{0,+1\} $.
	Additionally, $ D_{c} $ is a function from $ Q $ to $ \lozenge^{k} $ (or to $ \vartriangle^{k} $),
	where $ k $, a nonnegative integer, denotes the number of the counter(s).

$ \Theta $ (\textit{the status of a counter}): $ \Theta $ is set $ \{ 0,1 \} $, where $ 1 $ means that the counter 		
	value is nonzero and $ 0 $ means that the counter value is zero.
	
$ \delta $ (\textit{the transition function}): The behavior of a machine is specified by
	its transition function. The domain and the range of a transition function may vary with respect to
	the capabilities of the model. For a general and formal definition, $ \delta $,
	which is called ``having $ m $ domain components and  $ n $ range components''
	throughout the thesis, is a function
	\begin{equation}
		\delta : X_{1} \times X_{2} \times \cdots \times X_{m} \rightarrow 
		Z ^{Y_{1} \times Y_{2} \times \cdots \times Y_{n} },
	\end{equation}
	where $ m $ and $ n $ depend on the model; $ X_{1 \le i \le m} $'s 
	(\textit{the domain components}), which decide the updates in a single step, 
	can be 
	\begin{itemize}
		\item the current internal state,
		\item the symbols read by the tape heads, or
		\item the status of a counter;
	\end{itemize}
	$ Y_{1 \le j \le n} $'s (\textit{the range components}), which are updates done in a single step, 
	can be
	\begin{itemize}
		\item the next internal state,
		\item the symbols written on the tapes, memories, etc.,
		\item the movements of the tape heads, or
		\item an update operation on a stack or a counter;
	\end{itemize}
	and $ Z $ (\textit{the set of the transition values}) is a set of either discrete or continuous numbers.
	More specifically, for each $ \bar{x} \in \bar{X} $,
	\begin{equation}
		\delta(\bar{x}) = \sum_{\bar{y}_{i} \in \bar{Y} } z_{i} \bar{y}_{i},
	\end{equation}
	and $ z_{i} \in Z $ is called \textit{the transition value} and by overloading notation $ \delta $,
	$ z_{i} $ is also represented as
	\begin{equation}
		z_{i} = \delta(\bar{x};\bar{y}_{i}) = 
		\delta(\bar{x}[1],\ldots,\bar{x}[m],\bar{y}_{i}[1],\ldots,\bar{y}_{i}[n] ),
	\end{equation}
	where
	\begin{equation}
		\delta : X_{1} \times \cdots \times X_{m} \times Y_{1} \times \cdots \times Y_{n} \rightarrow Z.
\end{equation}
	Note that, independent of the context, $ z_{i} = 0 $ always corresponds to the case where 
	there is no update defined from $ \bar{x} $ to $ \bar{y}_{i} $.	
	For \textit{nondeterministic and alternating} machines, $ Z $ is set to $ \{0,1\} $,
	where the transitions are defined only for values 1. 
	If there is exactly one transition defined for each domain elements, 
	we obtain a \textit{deterministic} machine.
	The transition function of a deterministic machine can be written in the simplified form as follows:
	\begin{equation}
		\delta: \bar{X} \rightarrow \bar{Y}.
	\end{equation}
	For \textit{probabilistic} machines, $ Z $ is set to $ [0,1] $ (or a subset of a $ [0,1] $),
	where the transitions are defined for nonzero values and each transition value is called 
	as \textit{the transition probability}.
	For each domain value, all related transition probabilities must always be summed up to 1,
	i.e. for each $ \bar{x} \in \bar{X} $,
	\begin{equation}
		\sum_{\bar{y}_{i} \in \bar{Y}} \delta(\bar{x};\bar{y}_{i}) = 1.
	\end{equation}
	This condition is known as \textit{the local condition for PTM wellformedness}.
	For \textit{quantum} machines, $ Z $ is set to a subset of complex numbers whose
	euclidean norm is at most 1. (The remaining details are given in Section \ref{qtm}.)

$ f_{\mathcal{M}}^{a}(w) $ (\textit{the accepting probability}): For a given machine $ \mathcal{M} $
	and an input string $ w \in \Sigma^{*} $, $ f_{\mathcal{M}}^{a}(w) $, or shortly $ f_{\mathcal{M}}(w) $, 
	is the accepting probability of $ w $ associated by $ \mathcal{M} $.
	If the range of $ f_{\mathcal{M}}(w) $ is extended to real number domain, 
	it is called \textit{the accepting value}. 
	Moreover, $ f_{\mathcal{M}}^{r}(w) $ is used 
	in order to represent \textit{the rejecting probability} of $ w $ associated by $ \mathcal{M} $.
	Note that,	for deterministic, nondeterministic, and alternating machines, 
	the value of these functions is either 0 or 1.

\subsection{Language Recognition} \label{sbc:language-recognition}

\subsubsection{Cutpoint} \label{sbc:cutpoint}

The language $ L \subseteq \Sigma^{*} $ recognized by machine $ \mathcal{M} $ with (strict) cutpoint 
$ \lambda \in \mathbb{R} $ is defined \cite{Ra63,BJKP05} as 
\begin{equation}
	L = \{ w \in \Sigma^{*} \mid f_{\mathcal{M}}(w) > \lambda \}.
\end{equation}

The language $ L \subseteq \Sigma^{*} $ recognized by machine $ \mathcal{M} $ with nonstrict cutpoint 
$ \lambda \in \mathbb{R} $ is defined \cite{BJKP05} as 
\begin{equation}
	L = \{ w \in \Sigma^{*} \mid f_{\mathcal{M}}(w) \geq \lambda \}.
\end{equation}

The language $ L \subseteq \Sigma^{*} $ is recognized by machine $ \mathcal{M} $ 
with strict (resp., nonstrict) cutpoint
if there exists a cutpoint $ \lambda \in \mathbb{R} $ such that
$ L  $ is recognized by machine $ \mathcal{M} $ with strict (resp., nonstrict) cutpoint $ \lambda $.

\subsubsection{One-Sided Cutpoint} \label{sbc:one-sided-cutpoint}

The language $ L \subseteq \Sigma^{*} $ is recognized by machine $ \mathcal{M} $ with (positive) one-sided cutpoint 
$ \lambda \in \mathbb{R} $ if $ L $ is recognized with cutpoint $ \lambda $ and 
$ f_{\mathcal{M}}(w) = \lambda $ for all $ w \notin L $.

The language $ L \subseteq \Sigma^{*} $ is recognized by machine $ \mathcal{M} $ with negative one-sided cutpoint 
$ \lambda \in \mathbb{R} $ if $ \overline{L} $ is recognized by machine $ \mathcal{M} $ 
with positive one-sided cutpoint $ \lambda $

The language $ L \subseteq \Sigma^{*} $ is recognized by machine $ \mathcal{M} $ 
with positive (resp., negative) one-sided cutpoint 
if there exists a cutpoint $ \lambda \in \mathbb{R} $ such that
$ L $ (resp., $ \overline{L} $) is recognized 
by machine $ \mathcal{M} $ with positive one-sided cutpoint $ \lambda \in \mathbb{R} $

\subsubsection{Exclusive and Inclusive Cutpoint} \label{sbc:exclusive-inclusive-cutpoint}

The language $ L \subseteq \Sigma^{*} $ is recognized by machine $ \mathcal{M} $ with exclusive cutpoint 
$ \lambda \in \mathbb{R} $ if $ L $ is defined as 
\begin{equation}
	L = \{ w \in \Sigma^{*} \mid f_{\mathcal{M}}(w) \neq \lambda \}.
\end{equation}

The language $ L \subseteq \Sigma^{*} $ is recognized by machine $ \mathcal{M} $ with inclusive cutpoint 
$ \lambda \in \mathbb{R} $ if $ \overline{L} $ is recognized by machine $ \mathcal{M} $ 
with exclusive cutpoint $ \lambda $.

\subsubsection{Unbounded Error} \label{sbc:unbounded-error}

The language $ L \subseteq \Sigma^{*} $ is recognized by machine $ \mathcal{M} $ with unbounded error
if $ L $ is recognized by $ \mathcal{M} $ with either strict or nonstrict cutpoint \cite{YS10C}.

\subsubsection{One-Sided Unbounded Error} \label{sbc:one-sided-unbounded-error}

Let $ f_{\mathcal{M}}(w) \in [0,1] $.

The language $ L \subseteq \Sigma^{*} $ recognized by machine $ \mathcal{M} $ with \textit{(positive) one-sided
unbounded error} is defined as
\begin{equation}
      L = \{ w \in \Sigma^{*} \mid f_{\mathcal{M}}(w) > 0 \}.
\end{equation}

The language $ L \subseteq \Sigma^{*} $ recognized by machine $ \mathcal{M} $ with \textit{negative one-sided
unbounded error} is defined as
\begin{equation}
      L = \{ w \in \Sigma^{*} \mid f_{\mathcal{M}}(w) = 1 \}.
\end{equation}

\subsubsection{Isolated Cutpoint or Bounded Error} \label{sbc:isolated-cutpoint}

Let $ f_{\mathcal{M}}(w) \in [0,1] $.

The language $ L \subseteq \Sigma^{*} $ is said to be recognized by machine $ \mathcal{M} $ 
with isolated cutpoint \cite{Ra63} $ \lambda \in (0,1) $
if there exists a $ \delta > 0 $ such that
\begin{itemize}
	\item $ f_{\mathcal{M}}(w) \geq \lambda+\delta $ when $ w \in L $ and
	\item $ f_{\mathcal{M}}(w) \leq \lambda-\delta $ when $ w \notin L $,
\end{itemize}
where $ \delta $ is known as the isolation gap.

The language $ L \subseteq \Sigma^{*} $ is said to be recognized by machine $ \mathcal{M} $ 
with bounded error if there exists a $ p \in (\frac{1}{2},1] $ such that
\begin{itemize}
	\item $ f_{\mathcal{M}}(w) \geq p $ when $ w \in L $ and
	\item $ f_{\mathcal{M}}(w) \leq 1-p $ when $ w \notin L $.
\end{itemize}
Equivalently, it can be said that $ L \subseteq \Sigma^{*} $ is recognized by machine $ \mathcal{M} $ 
with error bound $ \epsilon $, where $ \epsilon = 1-p $ (and so $ \epsilon \in [0,\frac{1}{2}) $).

\subsubsection{One-Sided Bounded Error} \label{ssbc:one-sided-bounded-error}

The language $ L \subseteq \Sigma^{*} $ is said to be recognized by machine $ \mathcal{M} $ 
with (positive) one-sided bounded error if there exists a $ p \in (0,1] $ such that
\begin{itemize}
	\item $ f_{\mathcal{M}}(w) \geq p $ when $ w \in L $ and
	\item $ f_{\mathcal{M}}(w) = 0 $ when $ w \notin L $.
\end{itemize}
Equivalently, it can be said that $ L \subseteq \Sigma^{*} $ is recognized by machine $ \mathcal{M} $ 
with (positive) one-sided error bound $ \epsilon $, where $ \epsilon = 1-p $ (and so $ \epsilon \in [0,1) $).

The language $ L \subseteq \Sigma^{*} $ is said to be recognized by machine $ \mathcal{M} $ 
with negative one-sided bounded error if there exists a $ p \in (0,1] $ such that
\begin{itemize}
	\item $ f_{\mathcal{M}}(w) = 1 $ when $ w \in L $ and
	\item $ f_{\mathcal{M}}(w) \le 1-p $ when $ w \notin L $.
\end{itemize}
Equivalently, it can be said that $ L \subseteq \Sigma^{*} $ is recognized by machine $ \mathcal{M} $ 
with negative one-sided error bound $ \epsilon $, where $ \epsilon = 1-p $ (and so $ \epsilon \in [0,1) $).

\section{Turing Machines} \label{sbc:TMs}

The Turing machine (TM) models used throughout the thesis consist of a
read-only input tape with a two-way tape head, a read/write work tape
with a two-way tape head, and a finite state control. (The quantum
versions also have a finite register that plays a part in the
implementation of general quantum operations, and is used to determine
whether the computation has halted, and if so, with which decision.
For reasons of simplicity, this register is not included in the
definition of the classical machines, since its functionality can
be emulated by a suitable partition of the set of internal states
without any loss of computational power.)
Unless otherwise specified, both tapes are assumed to be two-way infinite and indexed by $ \mathbb{Z} $.

Let $ w $ be an input string.
On the input tape, $ \tilde{w} $ is placed in the squares indexed by $ 1, \ldots, | \tilde{w} | $,
and all remaining squares contain $ \# $.
When the computation starts, the internal state is $ q_{1} $,
the input tape head and the work tape head are placed on the squares indexed by 1 and 0, respectively.
Additionally, we assume that the input tape head never visits the squares indexed by $ 0 $ or $ |\tilde{w}|+1 $.

The transition function of a classical TM, i.e.
\begin{equation}
	\delta : Q \times \tilde{\Sigma} \times \Gamma \rightarrow 
	Z^{Q \times \Gamma \times <\mspace{-4mu}> \times <\mspace{-4mu}> },
\end{equation}
or
\begin{equation}
	z = \delta(q,\sigma,\gamma,q',\gamma',d_{i},d_{w})  \in Z,
\end{equation} 
has 3 domain components and 4 range components such that
whenever $ z \neq 0 $,
\begin{center}
\begin{minipage}{0.85\textwidth}
	the TM -- that is in state $ q \in Q $ and reads symbols $ \sigma \in \tilde{\Sigma} $ 
	and $ \gamma \in \Gamma $ on the input tape and 
	the work tape, respectively -- changes its state to $ q' \in Q $, writes $ \gamma' \in \Gamma $ 
	on the work tape,
	and then updates the position of the input and work tapes with respect to 
	$ d_{i} \in <\mspace{-4mu}> $ and $ d_{w} \in <\mspace{-4mu}> $,
	respectively, with transition value $ z $.
\end{minipage}
\end{center}
In the quantum case, there is also a fifth range component $ \Omega $ in order to implement 
general quantum operations, i.e. a symbol is written on the register ($ \omega \in \Omega $).
(The details are given in Section \ref{qtm}.).

By restricting the movement of the input tape head to $ \{ \downarrow, \rightarrow \} $
(and so by replacing range component corresponding to the input tape head  $ \setD $ with $ \rhd $),
we obtain a one-way TM, denoted as 1TM.

If the input tape head is restricted to move only to the right (``$ \rightarrow $'')
(and so range component $ \setD $ corresponding to the input tape head 
is completely removed from the transition function and 
the input tape head is automatically moved one square to the right after each transition),
we obtain a realtime TM, denoted as RT-TM.
Note that, after reading $ \dollar $, the computation of a RT-TM must be terminated.

If a 1TM is allowed to be stationary on a symbol for only at most some fixed steps, say $ d > 0 $,
it is called RT-TM with delay $ d $, denoted as $ d $RT-TM.
However, note that, by 
allowing it to be stationary, the input tape head of a quantum
Turing machine (QTMs) may be put in superposition (see \cite{BV97,Wa98,Wa99} and also Chapter \ref{qtm}),
i.e. each head in the superposition may be
positioned on a different symbol of the input. 
This violates the idea behind the realtime computation.
Therefore, for realtime QTMs with delay $ d $, the input head is required to be classical, where $ d > 0 $.

In space-bounded computation, the models having more than one work tape can generally be simulated by the ones 
having one work tape with the same amount of space usage.
On the other hand, this is not the case for RT-TMs \cite{Ra63B,Aa74}.
Therefore, the number of the work tapes is a parameter for RT-TMs.
However, all RT-TMs presented in this thesis have only one work tape.
Additionally, we assume that the work tape of RT-TMs are one-way infinite and their 
left most square is indexed by 0.

\begin{table}[h!]
	\vskip\baselineskip
 \caption{The list of the abbreviations of Turing machines}
 \footnotesize
 \begin{center}
 \begin{tabular}{|l|r|r|r|}
 	\hline
 	\multicolumn{1}{|c|}{ \textbf{Types of TMs} } & \textbf{TM} & \textbf{1TM} & \textbf{(d-)RT-TM} 
 	\\ \hline
 	deterministic & DTM & 1DTM & (d-)RT-DTM 
 	\\ \hline
 	{nondeterministic} & NTM & 1NTM & (d-)RT-NTM 
 	\\ \hline
 	{probabilistic} & PTM & 1PTM & (d-)RT-PTM 
 	\\ \hline
 	{alternating} & ATM & 1ATM & (d-)RT-ATM 
 	\\ \hline 
 	{quantum} & QTM & 1QTM & (d-)RT-QTM
 	\\ \hline 
 	{nondeterministic quantum} & NQTM & 1NQTM & (d-)RT-NQTM
 	\\ \hline 
 	{classical head quantum} & CQTM & 1CQTM & (d-)RT-CQTM 
 	\\ \hline
 	{classical head nondeterministic quantum} & CNQTM & 1CNQTM & (d-)RT-CNQTM 
 	\\ \hline
 \end{tabular}
 \end{center}
\end{table}

If we remove the work tape of a TM (and so domain component $ \Gamma $ and range component $ \Gamma $ with 
its related range component $ \setD $ are completely removed from the transition function), 
we obtain a finite automaton, denoted as FA.

In classical and quantum FA domains,
prefix ``1'' has been sometimes used instead of prefix ``realtime'' \cite{RS59,Ra63,KW97,BP02}.
This is generally not a problem for classical FAs due to their equivalence. 
However, this is not the case for quantum FAs (QFAs) since 
allowing the head to be stationary increases the power of QFAs, 
i.e. in bounded and unbounded error language recognition \cite{KW97,NIH01,YS09D}.
In the current literature, QFAs are commonly denoted using the prefix ``1'' and ``1.5''
in order to present ``realtime'' and ``one-way'' input heads, respectively 
\cite{KW97,AI99,NIH01,YS09D}.
In this thesis, we adopt the prefix notation of TMs for one-way and realtime models:
\begin{enumerate}
	\item two-way finite automaton (2FA), TM with no work tape,
	\item one-way finite automaton (1FA), 1TM with no work tape,
	\item realtime finite automaton with delay $ d $ ($ d $RT-FA), $ d $RT-TM with no work tape, 
		where $ d>0 $, and,
	\item realtime finite automaton (RT-FA), RT-TM with no work tape.
\end{enumerate}
Note that, in FA domain, two-wayness of the input head has been clearly indicated.

\begin{table}[h!]
	\vskip\baselineskip
 \caption{The list of the abbreviations of finite automata}
 \footnotesize
 \begin{center}
 \begin{tabular}{|l|r|r|r|}
 	\hline
 	\multicolumn{1}{|c|}{ \textbf{Types of FAs} } & \textbf{2FA} & \textbf{1FA} & \textbf{(d-)RT-FA} 
 	\\ \hline
 	deterministic & 2DFA & 1DFA & (d-)RT-DFA 
 	\\ \hline
 	{nondeterministic} & 2NFA & 1NFA & (d-)RT-NFA 
 	\\ \hline
 	{probabilistic} & 2PFA & 1PFA & (d-)RT-PFA 
 	\\ \hline
 	{alternating} & 2AFA & 1AFA & (d-)RT-AFA 
 	\\ \hline 
 	{quantum} & 2QFA & 1QFA & (d-)RT-QFA
 	\\ \hline 
 	{nondeterministic quantum} & 2NQFA & 1NQFA & (d-)RT-NQFA
 	\\ \hline 
 	{classical head quantum} & 2CQFA & 1CQFA & (d-)RT-CQFA 
 	\\ \hline
 	{classical head nondeterministic quantum} & 2CNQFA & 1CNQFA & (d-)RT-CNQFA 
 	\\ \hline
 \end{tabular}
 \end{center}
\end{table}

A TM is said to be \textit{unidirectional} (or \textit{simple}) if the movements of input
and work tape heads are fixed for each internal state to be entered in any transition.
That is, for a unidirectional TM, denoted as uni-TM, $ D_{i} $ and $ D_{w} $
determine respectively the movements of the input and work tape heads in any transition.
Thus, the corresponding range component(s) of $ \delta $'s are dropped.
For simplicity, each internal state of a unidirectional FA, say $ q $, can be represented as follows: 
\begin{itemize}
	\item $ \leftstate{q} $ if $ D_{i}(q) = ``\leftarrow''  $,
	\item $ \stopstate{q} $ if $ D_{i}(q) = ``\downarrow''  $, and
	\item $ \rightstate{q} $ if $ D_{i}(q) = ``\rightarrow''  $.
\end{itemize}
All classical computational models presented throughout thesis
are equivalent to their unidirectional counterparts. 
Therefore, each of those machines is considered with default unidirectional.

A configuration of a TM is the collection of
\begin{itemize}
       \item the internal state of the machine,
       \item the position of the input tape head,
       \item the contents of the work tape, and the position of the work tape head.
\end{itemize}
$ \mathcal{C}^{w} $, or shortly $ \mathcal{C} $,
denotes the set of all configurations, which is a finite set in our
case of space bounded computations.
Let $ c_{i} $ and $ c_{j} $ be two configurations.
The value of the transition from $ c_{i} $ to $ c_{j} $ is 
given by the transition function $ \delta $ if
$ c_{i} $ is reachable from $ c_{j} $ in one step, and is zero otherwise.
A \textit{configuration matrix} is a square matrix whose rows and columns are
indexed by the configurations.
The $ (j,i)^{th} $ entry of the matrix denotes the value of the
transition from $ c_{i} $ to $ c_{j} $.
We assume $ c_{1} $ as \textit{the initial configuration}, in which, as described previously,
the heads are placed on the square indexed by 1 and the internal state is $ q_{1} $
(in quantum case, the finite register is prepared by the initial symbol as well).

For classical TMs, a configuration is called an \textit{accepting configuration} 
(resp., a \textit{rejecting configuration}) if its internal state belongs to the set of
the accepting states (resp., the set of the rejecting states).
Moreover, a configuration is called a \textit{halting configuration} 
if it is an accepting or rejecting configuration or
(in nondeterministic and alternating cases) there is no defined transition from it.
The computation does not continue from a halting configuration.
However, note that, for the models making their decisions after reading the whole input, such as RT-FAs,
the configurations can be associated with adjectives ``accepting'' or ``rejecting'' 
only after reading symbol $ \dollar $.

\section{Pushdown Automata} \label{sbc:PDAs}

A pushdown automaton (PDA) is a TM using its work tape as a stack,
on which three operations are applied: 
\textit{pop}, \textit{push}, and \textit{pop-and-push}.
As a special note, the push operation on a stack can be implemented by a TM in at least two steps,
however, it is a single step for a PDA.
We assume that the stacks are one-way infinite and bounded from the left.
The squares of a stack are indexed by nonnegative integers, where the left-most square
is indexed by 0 on which $ \vdash $ is placed throughout the computation.
At the beginning of the computation, all squares indexed by positive integers contain symbol $ \# $.
The stack head is assumed not to drop out from left.
For PDAs, $ \Gamma $ contains $ \vdash $ and not $ \# $ and $ \tilde{\Gamma} = \Gamma \cup \{\#\} $.

For a given input string  $ w \in \Sigma $,
a configuration of a PDA is composed of the following elements:
\begin{itemize}
	\item the current internal state,
	\item the position of the input head, and
	\item the contents of the stacks.
\end{itemize}

Throughout the thesis, it is assumed that the PDAs have one stack.
A two-way pushdown automaton (2PDA) has a two-way input tape head, i.e. in each transition, 
the position of the input tape head can be updated with respect to $ <\mspace{-4mu}> $, and 
a one-way pushdown automaton (1PDA) is a restricted 2PDA in which the input tape head cannot move to the left, 
i.e. in each transition, the position of the input tape head can be updated with respect to $ \rhd $.
By not allowing the input tape head of a 1PDA to be stationary, we obtain a realtime pushdown automaton (RT-PDA).

\begin{table}[b!]
	\vskip\baselineskip
 \caption{The list of the abbreviations of pushdown automata}
 \footnotesize
 \begin{center}
 \begin{tabular}{|l|r|r|r|}
 	\hline
 	\multicolumn{1}{|c|}{ \textbf{Types of PDAs} } & \textbf{2PDA} & \textbf{1PDA} & \textbf{(d-)RT-PDA} 
 	\\ \hline
 	deterministic & 2DPDA & 1DPDA & (d-)RT-DPDA 
 	\\ \hline
 	{nondeterministic} & 2NPDA & 1NPDA & (d-)RT-NPDA 
 	\\ \hline
 	{probabilistic} & 2PPDA & 1PPDA & (d-)RT-PPDA 
 	\\ \hline
 	{alternating} & 2APDA & 1APDA & (d-)RT-APDA 
 	\\ \hline 
 	{quantum} & 2QPDA & 1QPDA & (d-)RT-QPDA
 	\\ \hline 
 	{nondeterministic quantum} & 2NQPDA & 1NQPDA & (d-)RT-NQPDA
 	\\ \hline 
 	{classical head quantum} & 2CQPDA & 1CQPDA & (d-)RT-CQPDA 
 	\\ \hline
 	{classical head nondeterministic quantum} & 2CNQPDA & 1CNQPDA & (d-)RT-CNQPDA 
 	\\ \hline
 \end{tabular}
 \end{center}
\end{table}

The transition function of a PDA can be defined as a special case of a TM after making some 
small modifications for the push operation.
That is, syntactically, there is no difference between the transition function of a TM and a PDA.
However, a PDA has some restrictions on its transition function and 
the push operation has capability to change the next square (on the right side) on the stack.
In the following, the details of the operations implemented on the stack are presented.
Note that, all the remaining part of the transition function is exactly the same as a TM.

Let $ d_{s} \in \setD $ be the direction of the stack head
and $ \gamma \in \Gamma $ and $ \gamma' \in \tilde{\Gamma} $ be the symbols to be read and
to be written on the stack, respectively.
\begin{itemize}
	\item In a pop operation, $ \gamma \neq ``\vdash" $ is overwritten by ``$ \# $'' and the head is
		moved one square to the left. (For symbol $ \vdash $, the pop operation is forbidden.)
		Therefore, the pop operation is associated with setting of $ d_{s} $ to $ ``\leftarrow" $ and
		the values of all transitions having the setting of $ \gamma = ``\vdash" $ or 
		the setting of $ \gamma' \neq ``\#" $ are always zero.
	\item In a pop-and-push operation, $ \gamma $ is overwritten by $ \gamma' \neq ``\#" $ and the head
		stays in the same square. (It is not allowed to write symbol $ \# $.)
		Therefore, the pop-and-push operation is associated with setting of $ d_{s} $ to $ ``\downarrow" $ and
		the values of all transitions having the setting of $ \gamma' = ``\#" $ 
		or the setting of $ \gamma = ``\vdash" $  and $ \gamma' \neq ``\vdash" $ are always zero.
	\item In a push operation, the following two consecutive operations of a TM are combined: 
		(i) $ \gamma $ is overwritten by $ \gamma $ and the head is moved one square to the right,
		(ii) symbol $ \# $ is overwritten by $ \gamma' \neq ``\#" $ and the head stays in the same square.
		(It is not allowed to write symbol $ ``\#" $.)
		Therefore, the push operation is associated with setting of $ d_{s} $ to $ ``\rightarrow" $ and
		the values of all transitions having the setting of $ \gamma' = ``\#" $ are always zero.
\end{itemize}

As a last remark, \textit{unidirectional} PDAs, denoted as uni-PDAs, 
are defined exactly in the same way as TMs.
Thus, $ D_{s}(q') $ determines the stack operation of the transition in which
$ q' $ is the state to be entered.

\section{Counter Automata} \label{sbc:CAs}

A counter automaton (CA) can be seen as a PDA with multiple stacks 
whose alphabets are restricted to contain two symbols, i.e. $ \vdash $ and a \textit{counting symbol}.
That is, a CA can count by using its stacks and can check whether their values are zero or not.
In order to make a counting with negative numbers, we also assume that the stack tapes are two-way infinite
by requiring that the squares indexed by 0 always contains $ \vdash $ and the remaining squares 
contain either the blank symbol or counting symbol throughout the computation.
Note that, there is no blank symbol between the counting symbols.

Formally, we follow the domain specific conventions of counter automata.
A $ k \in \mathbb{Z}^{+} $ counter automaton  ($ k $CA) is a 
finite state automaton augmented with $ k $ counters which can be modified by some amount
from $ \lozenge $, and where the status of these counters is also taken into account during transitions.

For a given input string  $ w \in \Sigma $,
a configuration of a $ k $CA is composed of the following elements:
\begin{itemize}
	\item the current internal state,
	\item the position of the input head, and
	\item the contents of the counters.
\end{itemize}

Similar to PDAs, we define two-way $ k $-counter automaton (2$ k $CA), one-way $ k $-counter automaton (1$ k $CA), 
and realtime $ k $-counter automaton (RT-$ k $CA). 

\begin{table}[h!]
	\vskip\baselineskip
 \caption{The list of the abbreviations of counter automata}
 \footnotesize
 \begin{center}
 \begin{tabular}{|l|r|r|r|}
 	\hline
 	\multicolumn{1}{|c|}{ \textbf{Types of $ k $CAs} } & \textbf{2$ k $CA} & \textbf{1$ k $CA} & \textbf{(d-)RT-$ k $CA} 
 	\\ \hline
 	deterministic & 2D$ k $CA & 1D$ k $CA & (d-)RT-D$ k $CA 
 	\\ \hline
 	{nondeterministic} & 2N$ k $CA & 1N$ k $CA & (d-)RT-N$ k $CA 
 	\\ \hline
 	{probabilistic} & 2P$ k $CA & 1P$ k $CA & (d-)RT-P$ k $CA 
 	\\ \hline
 	{alternating} & 2A$ k $CA & 1A$ k $CA & (d-)RT-A$ k $CA 
 	\\ \hline 
 	{quantum} & 2Q$ k $CA & 1Q$ k $CA & (d-)RT-Q$ k $CA
 	\\ \hline 
 	{nondeterministic quantum} & 2NQ$ k $CA & 1NQ$ k $CA & (d-)RT-NQ$ k $CA
 	\\ \hline 
 	{classical head quantum} & 2CQ$ k $CA & 1CQ$ k $CA & (d-)RT-CQ$ k $CA 
 	\\ \hline
 	{classical head nondeterministic quantum} & 2CNQ$ k $CA & 1CNQ$ k $CA & (d-)RT-CNQ$ k $CA 
 	\\ \hline
 \end{tabular}
 \end{center}
\end{table}

Let $ k $ be a nonnegative integer.
The transition function of a classical 2$ k $CA, i.e.
\begin{equation}
	\delta : Q \times \tilde{\Sigma} \times \{ \Theta^{k} \} \rightarrow 
	Z^{Q \times \setD \times \{ \lozenge^{k} \} },
\end{equation}
or
\begin{equation}
	z = \delta(q,\sigma,\bar{\theta},q',d_{i},\bar{c})  \in Z,
\end{equation} 
has 3 domain components and 3 range components such that whenever $ z \neq 0 $,
\begin{center}
\begin{minipage}{0.85\textwidth}
	the 2$ k $CA -- that is in state $ q \in Q $, reads symbols $ \sigma \in \tilde{\Sigma} $ on the input tape,
	and has $ \bar{\theta} \in \Theta^{k} $ on the counter(s), i.e.
	$ \bar{\theta}[i] $ represents the status of the $ i^{th} $ counter, -- changes its state to $ q' \in Q $, 
	updates the position of the input head with respect to $ d_{i} $, and
	updates the values of the counter(s) with respect to $ \bar{c} $, i.e.
	the value of $ j^{th} $ counter is updated by $ \bar{c}[j] $, with transition value $ z $,
	where $ 1 \le i,j \le k $.
\end{minipage}
\end{center}
For 1$ k $CA, the range component $ \setD $ is replaced by $ \rhd $ and
for RT-$ k $CA, the range component $ \setD $ is completely removed.
In the quantum case, there is also an additional range component $ \Omega $ in order to implement 
general quantum operations, i.e. a symbol is written on the register ($ \omega \in \Omega $).

Similar to TMs, a CA is said to be \textit{unidirectional} or \textit{simple}, denoted as uni-CA, 
if the movement of the input tape head and the updates of the counters are fixed 
for each internal state to be entered in any transition.
(The function determining the updates of the counters is $ D_{c} $.)
Thus, the corresponding range component(s) of $ \delta $'s are dropped.

An $ r $-reversal $ k $CA \cite{Ch81}, denoted as $ r $-rev-$ k $CA, is a $ k $CA such that 
the number of alternations from increasing to decreasing 
and vice versa on each counter is restricted by $ r $, where $ r $ is a nonnegative integer.

A blind $ k $CA, denoted as a $ k $BCA, is a $ k $CA never knowing the status of its counter(s)
(and so the domain component $ \Theta^{k} $ is completely removed) and
it requires that the value of the each counter must be zero in order to accept a given input string.

Lastly, a $ k $CA$ (m) $ is a $ k $CA with the capability of updating its counter by a value
from the set $ \{-m,\ldots,0,\ldots,m\} $ in a single step, where $ m>1 $.
Indeed, for any given $ k $CA$ (m) $, an isomorphic $ k $CA can easily be built and so
such a capability does not increase the computational power of CAs.
We present this fact explicitly for realtime quantum CAs in Lemma \ref{qtm:lem:CA-m-isomorphic-CA}.

\section{Formal Definition of Classical Machines} \label{sbc:formal-definition-classical-machine}

Formally, a classical TM  or a classical PDA with any type of the input tape head has is a 7-tuple
\begin{equation}
	\mathcal{P} = (Q,\Sigma,\Gamma,\delta,q_{1},Q_{a},Q_{r}).
\end{equation}
For two-way or one-way classical FAs or $ k $CAs, we use a 6-tuple
\begin{equation}
	\mathcal{P} = (Q,\Sigma,\delta,q_{1},Q_{a},Q_{r}).
\end{equation}
For the above machines, $ Q_{n} = Q \setminus \{ Q_{a} \cup Q_{r} \} $ and 
additionally in case of alternating machines,
$ Q_{n} $ is composed of two disjoint subsets: 
$ Q_{e} $, the set of \textit{existential} states;
$ Q_{u} $, the set of \textit{universal} states.

A realtime classical FA or $ k $CA is a 5-tuple
\begin{equation}
	\mathcal{P} = (Q,\Sigma,\delta,q_{1},Q_{a}).
\end{equation}
When it describes an alternating machine, 
$ Q $ is composed of two disjoint subsets: $ Q_{e} $ and $ Q_{u} $.
Moreover, for probabilistic FAs, the transition function can also be defined as a finite set of 
(left) stochastic matrices, i.e. $ A_{\sigma \in \tilde{\Sigma}} $ and $ A_{\sigma} [j,i] $ represents
the transition probability from $ q_{i} $ to $ q_{j} $ when reading $ \sigma $.
Throughout the thesis, we follow this convention for those machines:
a two-way or one-way PFA is a 6-tuple
\begin{equation}
	\mathcal{P} = (Q,\Sigma,\{A_{\sigma \in \tilde{\Sigma}}\},q_{1},Q_{a},Q_{r})
\end{equation}
and a RT-PFA is a 5-tuple
\begin{equation}
	\mathcal{P} = (Q,\Sigma,\{A_{\sigma \in \tilde{\Sigma}}\},q_{1},Q_{a}).
\end{equation}
Similarly, the transition function of a probabilistic $ k $CA can be represented by a collection
of (left) stochastic matrices defined for each $ \sigma \in \tilde{\Sigma} $ 
and each $ \bar{\theta} \in  \Theta_{k} $.

\section{Computation of Machines} \label{sbc:computation-of-machines}

Let $ w \in \Sigma^{*} $ be a given input string and $ \mathcal{M} $ be a machine.
For all computational models, the computation begins with the initial configuration.
We use mainly \textit{the tree structure} in order to represent the computation of classical models
with the following specifications:
\begin{itemize}
	\item the root(s) of the tree(s) is (are) always the initial configuration;
	\item the nodes are the configurations and the halting configurations can only be placed
		on the leafs;
	\item the edges represent the one-step transitions between the configurations;
	\item any path from the root to a leaf is called a \textit{halting path}
\end{itemize}

\subsection{Deterministic Computation} 

The computation is traced by a finite path, 
i.e. the leaf of the path is either an accepting or rejecting configuration.
$ w $ is accepted if and only if the path of $ \mathcal{M} $ on $ w $ ends with an accepting configuration.

\subsection{Nondeterministic Computation} 

The computation is traced by a finite or infinite tree.
Contrary to the deterministic case, more than one outgoing edge may be defined for a configuration.
A terminating path of the tree is called \textit{an accepting path} (resp., \textit{rejecting path})
if its leaf is an accepting (resp., rejecting configuration).
$ w $ is accepted if and only if there exists an accepting path in the tree of $ \mathcal{M} $ on $ w $.

\subsection{Alternating Computation} 

This is a generalization of nondeterministic computation. 
The computation is traced by a forest, a set of trees.
The tree structure is similar to that of nondeterministic computation.
However, each tree is allowed to have only one outgoing edge from any node 
corresponding to a (nonhalting) configuration having an existential state.
Therefore, different nondeterministic choices lead to different computation trees.
A tree is called \textit{an accepting tree} if each leaf of the tree correspond to an accepting configuration,
that is, any terminating path of the tree is an accepting path.
$ w $ is accepted if and only if there exists an accepting tree in the forest of $ \mathcal{M} $ on $ w $.

\subsection{Probabilistic Computation} 

We can either use tree structure or \textit{vectors}.
The tree structure of a probabilistic computation is similar to that of nondeterministic computation,
in which each edge has a weight, i.e. the transition probability.
Therefore, each accepting or rejecting path can also be associated with a probability,
calculated by multiplying all weights of the path.
(Note that, the number of accepting or rejecting paths can be infinite.)
The overall accepting (resp., rejecting) probability of $ \mathcal{M} $ on $ w $, 
$ f_{\mathcal{M}}^{a}(w) $ (resp., $ f_{\mathcal{M}}^{r}(w) $),
is the summation of the probabilities associated to the all accepting (resp., rejecting) paths.

In the vector representation, 
a column vector, called \textit{configuration vector} or \textit{state vector},
whose $ i^{th} $ entry corresponds to the $ i^{th} $ configuration or state,
represent the probability distribution of the configuration in any step of the computation,
i.e. in each step of the computation, the current vector is multiplied by the configuration matrix from the left.
Note that, in the configuration matrix, the $ (i,i)^{th} $ entry is always assumed to be 1 if
the $ i^{th} $ column (or vector) corresponds to a halting configuration.
Thus, the probability of reaching a halting configuration can be cumulatively represented 
by the corresponding entry in any step of the computation.

As a special case, the computation of a RT-PFA can be traced by a stochastic state vector,
say $ v $, such that $ v(i) $ corresponds to state $ q_{i} \in Q $.
\begin{equation}
	v_{i} = A_{\tilde{w}_{i}} v_{i-1},
\end{equation}
where $ 1 \le i \le | \tilde{w} | $ and
$ v_{0} $ is the initial state vector whose first entry is equal to 1.
The transition matrices of a RT-PFA can be extended for any string as 
\begin{equation}
	A_{w\sigma} = A_{\sigma} A_{w},
\end{equation}
where $ \sigma \in \tilde{\Sigma} $, $ w \in (\tilde{\Sigma})^{*} $, and $ A_{\varepsilon} = I $.
The probability that $ w $ is accepted by 1PFA $ \mathcal{P} $ is
\begin{equation}
	f_{\mathcal{P}}(w) = \sum_{q_{i} \in Q_{a}} (A_{\tilde{w}}v_{0})(i) = 
	\sum_{q_{i} \in Q_{a}} v_{|\tilde{w}|}(i).
\end{equation}

By generalizing the linearization of RT-PFA, Turakainen \cite{Tu69} defined a more general computational model,
called \textit{generalized finite automaton} (GFA). 
The details of a GFA are given in Figure \ref{sbc:fig:GFA}.

\begin{figure}[h!]	
	\begin{center}
	\fbox{
	\begin{minipage}{0.85\textwidth}
		A GFA is formally a 5-tuple
		\begin{equation}
		      \mathcal{G}=(Q,\Sigma,\{A_{\sigma \in \Sigma} \},v_{0},f),
		\end{equation}
		where
		\begin{enumerate}
		      \item $ A_{\sigma \in \Sigma} $ are $ |Q| \times |Q| $-dimensional
		real valued transition matrices;
		      \item $ v_{0} $ and $ f $ are real valued  \textit{initial} (column)
		and \textit{final} (row) vectors,
		              respectively.
		\end{enumerate}
		Similar to what we have for RT-PFAs, the transition matrices of a GFA can be extended for any string.
		For a given input string, $ w \in \Sigma $, the acceptance value associated by GFA 
		$ \mathcal{G} $ to string $ w $ is
		\begin{equation}
		      f_{\mathcal{G}}(w)=f A_{w_{|w|}} \cdots A_{w_{1}} v_{0} = f A_{w} v_{0}.
		\end{equation}
	\end{minipage}
	}
	\end{center}
	\caption{Generalized finite automaton}
	\vskip\baselineskip
	\label{sbc:fig:GFA}
\end{figure}

\subsection{Quantum Computation} 

In this part, we roughly explain quantum computation\footnote{The formal definitions and 
detailed explanations are in Chapter \ref{qtm} and Appendix \ref{qc}.} 
by using the classical representation tools, tree structure and vectors together.

A \textit{pure} quantum state, or shortly, a quantum state, 
is represented by a column vector, where each entry represents
the amplitude of being in the classical state that is assumed as
the configurations of the quantum machines throughout the thesis.

In quantum computation, as a part of the one-step transition, 
some symbols are written on the finite register 
(and then the register is discarded either after a measurement or without any measurement).
For each written symbol, the current quantum state is transformed into a new quantum state
with the probability of the square of the euclidean norm of the new quantum state.
(Once the finite register is discarded, the new quantum states are normalized\footnote{
Although only one of them survives and so is normalized according to
any particular observer, we use a ``God's eye view'' and assume that
all of them survive with related probabilities.}.)

Similar to the probabilistic machines, the computation of a quantum machine can be represented by a tree
with some exceptions:
the nodes of the tree are the quantum states instead of the configurations 
(and so the root of the tree is the initial quantum state, 
where the entry of the initial configuration is 1 and the remaining entries are zeros);
the edges are again associated with the nonzero probabilities and a 
finite register symbol;
the leafs of the tree are the halting (accepting or rejecting) quantum states,
i.e. if the incoming edge of a quantum state is associated with an accepting (resp., a rejecting) 
finite register symbol, then it is called an accepting (resp., a rejecting) quantum state.
As a special remark, a symbol written on the finite register may not always been 
associated with an edge in the tree due to destructive interference interference.

The overall accepting and rejecting probabilities of a quantum machine on an input 
can be calculated in the same way as that of a probabilistic machine.

\section{Space Classes} \label{sbc:space-classes}

Let $ s $ be a function of the form $ s: \mathbb{N} \rightarrow \mathbb{R} $.

The space used by a path (a tree or a forest) is the difference 
between the leftmost visited square and the right most visited square 
when considering all the configurations in the structure.
Note that, in the quantum case, we restrict ourselves only to the configurations having nonzero amplitude
in the quantum states.

A machine is said to be strongly $ s(|w|) $ space-bounded if the space used by
the related computational structure (tree or forest) of $ w $ is no more than $ s(|w|) $,
for any input string $ w \in \Sigma^{*} $.

A machine is said to be middle $ s(|w|) $ space-bounded if the space used by
the related computational structure (tree or forest) of $ w $ is no more than $ s(|w|) $,
for any accepted input string $ w \in \Sigma^{*} $.

A nondeterministic or alternating machine is said to be weak $ s(|w|) $ space-bounded if the space used by
an accepting path or tree, respectively, is no more than $ s(|w|) $,
for any accepted input string $ w \in \Sigma^{*} $.

The generic name of space classes is \textbf{X}SPACE, where \textbf{X} is replaced by suitable letters
depending on the TM and/or error types.
In probabilistic and quantum computation,
the subscript of ``SPACE'', i.e. SPACE$ _{\mathbb{Y}} $, if it exists, specifically denotes the class of the numbers
$ \mathbb{Y} \in \{ \mathbb{Q}, \mathbb{A}, \tilde{\mathbb{R}}, \tilde{\mathbb{C}}, \mathbb{R}, \mathbb{C} \} $, 
to which the transition values of the corresponding TMs are restricted,
where $ \tilde{\mathbb{R}} $ and $ \tilde{\mathbb{C}} $ 
denote the computable real and complex numbers, respectively.
(Note that, a complex number (and so a real number) is computable 
if its real and imaginary parts are computed within the precision of
$2^{-n}$ by a deterministic algorithm in time polynomial in $ n $ \cite{BV97}.)
By default, probabilistic (resp., quantum) space classes are defined for  $ \mathbb{R} $ (resp., $ \mathbb{C} $).

There may be two more prefixes before the class names as well 
(when used together, the order below is followed): 
\begin{enumerate}
	\item the kind of the space usage, i.e. ``weak-'' or ``middle-'', and
	\item the type of the input tape head, i.e. ``one-way-'' or ``realtime-''.
\end{enumerate}

The list of the class names used throughout the thesis are as follows:
\begin{itemize}
	\item DSPACE(s), NSPACE(s), and ASPACE(s) denote the classes of the languages recognized by 
		DTMs, NTMs, and ATMs, respectively, in space $ s $.
	\item PrSPACE(s) and PrQSPACE(s) denote the classes of the languages recognized by PTMs and QTMs, respectively, 
		in space $ s $ with unbounded error.
	\item BPSPACE(s) and BQSPACE(s) denote the classes of the languages recognized by PTMs and QTMs, respectively, 
		in space $ s $ with bounded error.
	\item NQSPACE(s) denote the classes of the languages recognized by QTMs 
		in space $ s $ with one-sided unbounded error setting.
	\item C$ _{=} $SPACE(s) and C$ _{=} $QSPACE(s) denote the classes of the languages recognized by PTMs and QTMs, 
		respectively, in space $ s $, with inclusive cutpoint $ \frac{1}{2} $ providing that
		each computation must be halted in finite step.
\end{itemize}

For the pushdown and the counter machines, we use ``STACK'' and ``COUNTER'' complexity classes,
respectively, in order to represent the space usage of the machines on their storage devices
with respect to the length of the input string, i.e.
they are the counterparts of ``SPACE'' complexity class.
Moreover, $ k $STACK and $ k $COUNTER denote the classes defined by the machines having
$ k $ stacks and $ k $ counters, respectively.
As a generic class name, CFL is the class of languages recognized by 1NPDAs.

In the finite automata domain, the classes also have domain specific names.

The class of the languages recognized by RT-DFAs (and 1DFAs) are regular languages, denoted as REG.
In the finite automata domain, neither two-wayness nor nondeterminism 
increases the computational power of the deterministic machines \cite{RS59,Sh59}. 
That is, 
\begin{equation}
	\mbox{DSPACE(1)=NSPACE(1)=REG}.
\end{equation}

RT-PFAs, GFAs \cite{Tu69}, and 2PFAs \cite{Ka89} recognize the same
class of languages with (strict) cutpoint. This is the class of
\textit{stochastic languages}, denoted by S. The class of languages
recognized by these machines with nonstrict cutpoint is denoted by
coS. The class of languages recognized by
RT-PFAs, GFAs, and 2PFAs with unbounded error is therefore
S $ \cup $ coS, and is denoted by uS.
\begin{equation}
	\mbox{PrSPACE(1)=uS}.
\end{equation}

One-way-C$ _{=} $SPACE(1) is also denoted as S$ ^{=} $.
The complementary class of S$ ^{=} $, denoted as S$ ^{\neq} $, is 
known as exclusive stochastic languages \cite{Pa71}.

The class of languages recognized by RT-QFAs (Section \ref{qtm:RT-QFAs}) with cutpoint is denoted by QAL. 
The class of languages recognized by these machines with nonstrict cutpoint is denoted by coQAL.
QAL $ \cup $ coQAL is denoted by uQAL.
The class of the languages recognized by MCQFA (Section \ref{qtm:restricted-QTMs}) 
with cutpoint is Moore-Crotchfield languages, denoted as MCL.
coMCL and uMCL denote respectively complementary class of MCL and the class formed by MCL $ \cup $ coMCL. 

Brodsky and Pippenger \cite{BP02} defined the class of
languages recognized by RT-KWQFAs (Section \ref{qtm:restricted-QTMs}) with unbounded error, denoted as UMM,
in a way that is slightly different than the standard approach:
$ L \in $ UMM if and only if there exists a RT-KWQFA $ \mathcal{M} $ such that
\begin{itemize}
	\item $ f_{\mathcal{M}}(w) > \lambda $ when $ w \in L $ and
	\item $ f_{\mathcal{M}}(w) < \lambda $ when $ w \notin L $,
\end{itemize}
for some $ \lambda \in [0,1] $.

The languages recognized by RT-NQFAs are nondeterministic quantum automaton languages, denoted as NQAL.
NMCL denotes the class of the languages recognized by nondeterministic MCQFAs (MCNQFAs).

In the bounded error setting, the class of the languages recognized by RT-PFAs (and 1PFAs) 
are regular languages \cite{Ra63}.
On the other hand, 2PFAs can recognize some nonregular languages with bounded error \cite{Fr81}.
\begin{equation}
	\mbox{REG = one-way-BPSPACE(1)} \subsetneq \mbox{BPSPACE(1)}.
\end{equation}

The class of languages recognized by RT-PostPFAs and RT-PostQFAs (Section \ref{berr:postsel-definition})
with bounded error are PostS (\textit{post-stochastic languages}) and 
PostQAL (\textit{post-quantum automaton languages}), respectively. 
Additionally, the class of languages recognized by RT-LPostPFAs and RT-LPostQFAs 
(Section \ref{berr:latvian-postsel}) with bounded error are LPostS and LPostQAL, respectively. 

\chapter{A NEW KIND OF QUANTUM MACHINE} \label{qtm}

After a general framework for the quantum machines with some
introductory concepts (Section \ref{qtm:general-framework}), 
we present firstly a new kind of quantum Turing machine\footnote{
For descriptions of several other QTM variants, we refer the reader to \cite{BV97,Wa98,Wa99,Wa03,MW08,Ga09}.
} with its FA variants (Section \ref{qtm:QTMs}), 
and then, based on it, quantum pushdown and stack machines (Sections \ref{qtm:QPDAs} and \ref{qtm:QCAs}).
We also give some basic facts in this chapter.

\section{The General Framework for Quantum Machines} \label{qtm:general-framework}

In accordance with quantum theory, a quantum machine can be in a \textit{superposition} of
more than one configuration at the same time. The ``weight''
of each configuration in such a superposition is called its \textit{amplitude}.
Unlike the case with classical machines, these amplitudes are not restricted to
being positive real numbers, and that is what gives quantum computers
their interesting features. A superposition of configurations
\begin{equation}
       \ket{\psi} = \alpha_{1} \ket{c_{1}} + \alpha_{2} \ket{c_{2}} + \cdots
+  \alpha_{n} \ket{c_{n}}
\end{equation}
can be represented by a column vector $ \ket{\psi} $ with a row for each possible configuration,
where the $ i^{th} $ row contains the amplitude of the corresponding configuration in $ \ket{\psi} $.

If our knowledge that the quantum system under consideration is in
superposition $ \ket{\psi} $ is certain, then $ \ket{\psi} $ is called a \textit{pure
state}, and the vector notation described above is a suitable way of
manipulating this information. However, in some cases (e.g. during
classical probabilistic computation), we only know that the system is
in state $ \ket{\psi_{l}} $ with probability $ p_{l} $ for an ensemble of pure
states $ \{ (p_{l},\ket{\psi_{l}}) \} $,where
$ \sum_{l} p_{l}=1 $. A convenient representation tool for describing
quantum systems in such \textit{mixed states}
is the density matrix.
The \textit{density matrix}\footnote{The trace of a density matrix is 1,
and each density matrix is positive semidefinite.} representation of
$ \{ (p_{l},\ket{\psi_{l}}) \mid 1 \le l \le M < \infty \} $ is
\begin{equation}
      \rho = \sum_{l} p_{l} \ket{\psi_{l}} \bra{\psi_{l}}.
\end{equation}
We use both these representations for quantum states in the thesis. 
We refer the reader to Appendix \ref{qc} for further details.

A quantum machine is distinguished from a classical machine
by the presence of the items $\Omega$ (the finite register alphabet) 
and $ \Delta $ (the set of possible outcomes associated with the measurements of the finite register). 
Note that, $ \Omega $ is partitioned into $ |\Delta| = k $ 
subsets $ \Omega_{\tau_{1}}, \ldots , \Omega_{\tau_{k}} $.
As a part of each transition, a quantum machine has the following phases:
\begin{enumerate}
	\item \textit{pre-transition phase}: 
		initialize the finite register, i.e. the register is set to ``$ \omega_{1} $'';
	\item \textit{transition phase}:
		in addition to the transition of the classical machines, 
		update the content of the register;
	\item \textit{post-transition phase}: 
		make a selective measurement on the finite register with the outcome set $ \Delta $ 
		and then discard it\footnote{In some realtime quantum machines, 
		such a selective measurement can be done only at the end of the computation.}.
\end{enumerate}

Since we do not consider the register content as part of the
configuration, the register can be seen as the
``environment'' interacting with the ``principal system'' that is the rest of the quantum machine \cite{NC00}.
$ \delta $ therefore induces a set of configuration transition
matrices, $ \{ E_{\omega \in \Omega} \} $, where the $ (i,j)^{th} $ entry of $ E_{\omega}$,
the amplitude of the transition  from $ c_{j} $ to $ c_{i} $ by
writing $ \omega \in \Omega $ on the register,
is defined by $ \delta $ whenever $ c_{j} $ is reachable from $ c_{i}
$  in one step, and is zero otherwise. 
The $ \{ E_{\omega \in \Omega} \} $ form an operator $ \mathcal{E} $, with operation elements
$  \mathcal{E}_{\tau_{1}} \cup \mathcal{E}_{\tau_{2}} \cup \cdots \cup
\mathcal{E}_{\tau_{k}}  $, where
$ \mathcal{E}_{\tau \in \Delta} = \{ E_{\omega \in \Omega_{\tau}} \} $.

According to the modern understanding of quantum computation \cite{AKN98}, 
a quantum machine is said to be \textit{well-formed}\footnote{We also refer the reader to \cite{BV97}
for a detailed discussion of the well-formedness of QTMs that evolve unitarily.}
if $ \mathcal{E} $ is a superoperator (selective quantum operator), i.e.
\begin{equation}
      \sum_{\omega \in \Omega} E_{\omega}^{\dagger}E_{\omega} = I.
\end{equation}
$ \mathcal{E} $ can be represented by a $ | \mathcal{C} | |
\Omega | \times | \mathcal{C} | $-dimensional
matrix $ \mathsf{E} $ (Figure \ref{qtm:fiq:superoperators})
by concatenating each $ E_{\omega \in \Omega} $ one under the
other. It can be verified that $ \mathcal{E} $ is a superoperator if and only
if the columns of $ \mathsf{E} $
form an orthonormal set.

\begin{figure}[h]
	\begin{center}	
	\begin{minipage}{0.4\textwidth}
		\begin{equation}
		\begin{array}{ccccc}
			\multicolumn{1}{c|}{} & c_{1} & c_{2} & \ldots & \multicolumn{1}{c|}{ c_{|\mathcal{C}|} } \\
			\hline
			\multicolumn{1}{c|}{c_{1}} & & & & \multicolumn{1}{c|}{} \\
			\multicolumn{1}{c|}{c_{2}} & & & & \multicolumn{1}{c|}{} \\
			\multicolumn{1}{c|}{\vdots} & \multicolumn{4}{c|}{ E_{\omega_{1}} } \\
			\multicolumn{1}{c|}{c_{|\mathcal{C}|}} & & & & \multicolumn{1}{c|}{} \\
			\hline	
			\multicolumn{1}{c|}{c_{1}} & & & & \multicolumn{1}{c|}{} \\
			\multicolumn{1}{c|}{c_{2}} & & & & \multicolumn{1}{c|}{} \\
			\multicolumn{1}{c|}{\vdots} & \multicolumn{4}{c|}{ E_{\omega_{2}} } \\
			\multicolumn{1}{c|}{c_{|\mathcal{C}|}} & & & & \multicolumn{1}{c|}{} \\
			\hline\multicolumn{1}{c|}{c_{1}} & & & & \multicolumn{1}{c|}{} \\
			\multicolumn{1}{c|}{c_{2}} & & & & \multicolumn{1}{c|}{} \\
			\multicolumn{1}{c|}{\vdots} & \multicolumn{4}{c|}{\vdots } \\
			\multicolumn{1}{c|}{c_{|\mathcal{C}|}} & & & & \multicolumn{1}{c|}{} \\
			\hline	
			\multicolumn{1}{c|}{c_{1}} & & & & \multicolumn{1}{c|}{} \\
			\multicolumn{1}{c|}{c_{2}} & & & & \multicolumn{1}{c|}{} \\
			\multicolumn{1}{c|}{\vdots} & \multicolumn{4}{c|}{ E_{\omega_{|\Omega|}} } \\
			\multicolumn{1}{c|}{c_{|\mathcal{C}|}} & & & & \multicolumn{1}{c|}{} \\
			\hline	
		\end{array}
		\end{equation}
	\end{minipage}
	\end{center}
	\caption{The matrix representation of superoperators ($ \mathsf{E} $)}
	\vskip\baselineskip
	\label{qtm:fiq:superoperators}
\end{figure}

Let $ c_{j_{1}} $ and $ c_{j_{2}} $ be two configurations with corresponding columns 
$ v_{j_{1}} $ and $ v_{j_{2}} $ in $ \mathsf{E} $.
For an orthonormal set to be formed, we must have
\begin{equation}
      v_{j_{1}}^{\dagger} \cdot  v_{j_{2}} = \left \lbrace \begin{array}{ll}
              1 & j_{1}=j_{2} \\
              0 & j_{1} \neq j_{2}
      \end{array}
      \right.
\end{equation}
for all such pairs.
This constraint imposes some easily checkable restrictions on the transition function ($ \delta $).

The initial density matrix of a quantum machine is represented by  $ \rho_{0} =
\ket{c_{1}} \bra{c_{1}} $, where $ c_{1} $ is the initial configuration corresponding to the given input string.
Note that, unless otherwise specified, $ \Delta $ is set to $ \{n,a,r\} $.

\section{Quantum Turing Machines} \label{qtm:QTMs}

We define a quantum Turing machine (QTM) $ \mathcal{M} $ to be a 7-tuple
\begin{equation}
      M=(Q,\Sigma,\Gamma,\Omega,\delta,q_{1},\Delta).
\end{equation}
We refer the reader to Appendix \ref{well:QTMs} for the list of the local conditions for QTM wellformedness.

\subsection{Two-Way Quantum Finite Automata} \label{qtm:2QFAs}

The two-way quantum finite automaton (2QFA) is obtained by removing the work tape of the QTM:
\begin{equation}
      \mathcal{M}=(Q,\Sigma,\Omega,\delta,q_{1},\Delta).
\end{equation}
 
See below for a list of easily checkable local conditions for wellformedness of 2QFAs.
Let $ x_{1} $ and $ x_{2} $ denote the positions of the input tape heads.
In order to evolve to the same configuration in one step, the difference between $ x_{1} $ and $ x_{2} $
must be at most 2.
Therefore, we obtain a total of 13 different cases, listed below, that
completely define the restrictions on the transition function. Note that, by taking the
conjugates of each summation, we handle the symmetric cases that are shown in the parentheses.
\noindent
For $ q_{1},q_{2} \in Q; \sigma \in \tilde{\Sigma} $,
\\
\footnotesize
1. $ x_{1} = x_{2} $:
\begin{equation}
	 \sum\limits_{q' \in Q, d \in <\mspace{-4mu}>,\omega \in \Omega}
      \overline{ \delta(q_{1},\sigma,q',d,\omega) }
      \delta(q_{2},\sigma,q',d,\omega) =
      \left\lbrace
              \begin{array}{ll}
                      1 & q_{1} = q_{2} \\
                      0 & \mbox{otherwise}
              \end{array}
      \right.
\end{equation}
\\
2. $ x_{1} = x_{2} - 1 $ ($ x_{1} = x_{2} + 1 $):
\begin{equation}
	\sum_{q' \in Q,\omega \in \Omega} 
      \overline{\delta(q_{1},\sigma,q',\rightarrow,\omega)} 
      \delta(q_{2},\sigma,q',\downarrow,\omega) + 
      \overline{\delta(q_{1},\sigma,q',\downarrow,\omega)} 
      \delta(q_{2},\sigma,q',\leftarrow,\omega) = 0.
\end{equation}
\\
3. $ x_{1} = x_{2} - 2 $ ($ x_{1} = x_{2} + 2 $):
\begin{equation}
	\sum_{q' \in Q,\omega \in \Omega}
	\overline{\delta(q_{1},\sigma,q',\rightarrow,\omega)}
     \delta(q_{2},\sigma,q',\leftarrow,\omega) = 0.
\end{equation}
\normalsize

\subsection{Unidirectional Quantum Turing Machines} \label{qtm:uniQTMs}

If the QTM is unidirectional, wellformedness can be checked using the
simpler conditions in Figure \ref{figure:uniQTM-wellformedness}.
Removing the reference to worktape symbols, we obtain the analogous
constraints for unidirectional 2QFAs as shown in Figure \ref{figure:uni2QFA-wellformedness}.

\begin{figure}[h!]       
       \begin{center}
       \fbox{
       \begin{minipage}{\textwidth}
               \footnotesize{
               For $ q_{1},q_{2} \in Q; \sigma \in \tilde{\Sigma}; \gamma_{1},\gamma_{2} \in \Gamma $,
               \begin{equation}
                       \sum\limits_{q' \in Q,\gamma' \in \Gamma,\omega \in \Omega}
                       \overline{ \delta(q_{1},\sigma,\gamma_{1},q',\gamma',\omega)
}
                       \delta(q_{2},\sigma,\gamma_{2},q',\gamma',\omega) =
                       \left\lbrace
                               \begin{array}{ll}
                                       1 & q_{1} = q_{2} \mbox{ and } \gamma_{1} = \gamma_{2} \\
                                       0 & \mbox{otherwise}
                               \end{array}
                       \right..
               \end{equation}
               }
       \end{minipage}
       } 
       \end{center}
       \caption{The local conditions for unidirectional QTM wellformedness}
       \vskip\baselineskip
       \label{figure:uniQTM-wellformedness}
\end{figure}

\begin{figure}[h!]       
       \begin{center}
       \fbox{
       \begin{minipage}{\textwidth}
               \footnotesize{
               For $ q_{1},q_{2} \in Q; \sigma \in \tilde{\Sigma} $,
               \begin{equation}
                       \sum\limits_{q' \in Q,\omega \in \Omega}
                       \overline{ \delta(q_{1},\sigma,q',\omega) }
                       \delta(q_{2},\sigma,q',\omega) =
                       \left\lbrace
                               \begin{array}{ll}
                                       1 & q_{1} = q_{2} \\
                                       0 & \mbox{otherwise}
                               \end{array}
                       \right..
               \end{equation}
               }
       \end{minipage}
       }
       \end{center}
       \caption{The local conditions for unidirectional 2QFA wellformedness}
       \vskip\baselineskip
       \label{figure:uni2QFA-wellformedness}       
\end{figure}

As is the case with PTMs, the transition
function of a uni-QTM can be specified easily by transition matrices of the form
$ \{ E_{\sigma,\omega} \} $, whose rows and columns are indexed by
(internal state, work tape symbol) pairs
for each $ \sigma \in \tilde{\Sigma} $ and $ \omega \in \Omega $. 
It can be verified that the wellformedness condition is then equivalent to the requirement that, 
for each  $ \sigma \in \tilde{\Sigma} $,
\begin{equation}
       \sum_{\omega \in \Omega} E_{\sigma,\omega}^{\dagger} E_{\sigma,\omega} = I.
\end{equation}

Similarly, for each $ \sigma \in \tilde{\Sigma} $ and $ \omega \in \Omega $,
well-formed unidirectional 2QFAs can be described by transition matrices of the form
$ \{ E_{\sigma,\omega} \} $, whose rows and columns are indexed by
internal states,
such that for each  $ \sigma \in \tilde{\Sigma} $,
\begin{equation}
       \sum_{\omega \in \Omega} E_{\sigma,\omega}^{\dagger} E_{\sigma,\omega} = I.
\end{equation}

\begin{openproblem}
	Given a QTM (resp., 2QFA) $ \mathcal{M} $, 
	does there always exist a uni-QTM (resp., uni-2QFA) $ \mathcal{M'} $ such that
	\begin{equation}
		f_{\mathcal{M}}(w) = f_{\mathcal{M'}}(w)
	\end{equation}
	for all $ w \in \Sigma^{*} $?
\end{openproblem}

\subsection{Quantum Turing Machines with Classical Heads} \label{qtm:CQTMs}

Although our definition of space usage as the number of work tape
squares used during the computation is standard in the study of small
as well as large space bounds \cite{Si06,Sz94,FK94},
some researchers prefer to utilize QTM models where the tape head
locations are classical (i.e. the heads do not enter quantum
superpositions) to avoid the possibility of using quantum resources
that increase with input size for the implementation of the heads \cite{AW02,Wa03}.
For details of this specialization of our model, which we call the QTM 
with classical heads (CQTM)

To specialize our general QTM model in order to ensure that the head positions are classical,
we associate combinations of head movements with measurement outcomes.
There are 9 different pairs of possible movement directions, i.e.
\begin{equation}
	<\mspace{-4mu}>^{2} = \{ \leftarrow, \downarrow, \rightarrow \} \times 
	\{ \leftarrow, \downarrow, \rightarrow \},
\end{equation}
for the input and work tape heads, and so we can classify register symbols with the function
\begin{equation}
       D_{r} : \Omega \rightarrow {<\mspace{-4mu}>}^{2}.
\end{equation}
We have $ D_{r}( \omega ) = ( \downarrow, \downarrow ) $
if $ \omega \in \Omega_{a} \cup \Omega_{r} $.
We split $ \Omega_{n} $ into $ 9 $ parts, i.e.
\begin{equation}
       \Omega_{n} = \bigcup\limits_{d_{i},d_{w} \in <\mspace{-4mu}> }
\Omega_{n,d_{i},d_{w}},
\end{equation}
where
\begin{equation}
       \Omega_{n,d_{i},d_{w}} = \{ \omega \in \Omega_{n} \mid D_{r}(\omega)=(d_{i},d_{_{w}})  \}.
\end{equation}
Therefore, the outcome set have 11 elements, represented as triples, specified as follows:
\begin{enumerate}
       \item ``$ (n,d_{i},d_{w}) $'': the computation continuous and the
positions of the input and work tape heads
               are updated     with respect to $ d_{i} $ and $ d_{w} $, respectively;
       \item ``$ (a,\downarrow,\downarrow) $'': the computation halts and the
input is accepted with no head movement;
       \item ``$ (r,\downarrow,\downarrow) $'': the computation halts and the
input is rejected with no head movement.
\end{enumerate}

The transition function of CQTMs are specified so that
when the CQTM is in state $ q $ and reads $ \sigma $ and $ \gamma $
respectively on the input and work tapes,
it enters state $ q' $, and writes $ \gamma' $ and
$ \omega $ respectively
on the work tape and the finite register with the amplitude
\begin{equation}
       \delta(q,\sigma,\gamma,q',\gamma',\omega).
\end{equation}
Since the update of the positions of the input and work tape heads is
performed classically, it is no longer a part of the transitions.
Note that the transition function of 2QFAs with classical head
(2CQFAs) \cite{AW02}
is obtained by removing the mention of the work tape from the above description.

Moreover, as with unidirectional QTMs (resp. unidirectional 2QFAs),
for each $ \sigma \in \tilde{\Sigma} $ and $ \omega \in \Omega $,
CQTMs (2CQFAs) can be described by transition matrices
$ \{ E_{\sigma,\omega} \} $ satisfying the same properties.

As also argued in \cite{Wa03}, CQTMs are sufficiently general for simulating any classical TM.
We present a trivial simulation.

\begin{lemma}
	\label{qtm:lem:pm-simulated-by-qm}
       CQTMs can simulate any PTM exactly.
\end{lemma}
\begin{proof}
Let $ \mathcal{P} = (Q, \Sigma, \Gamma, \delta_{\mathcal{P}}, Q_{a},
Q_{r}) $ be a PTM and 
$ \mathcal{M} = (Q, \Sigma, \Gamma, \Omega, \delta_{\mathcal{M}}, \Delta ) $
be the CQTM simulating $ \mathcal{P} $.
For each $ (q,\gamma,q',\gamma') \in Q \times \Gamma
\times Q \times \Gamma $,
we define a register symbol $ \omega $ such that
\begin{enumerate}
       \item if $ q' \in Q_{a} $: $ \omega \in
\Omega_{(a,\downarrow,\downarrow)} $;
       \item if $ q' \in Q_{r} $: $ \omega \in
\Omega_{(r,\downarrow,\downarrow)} $;
       \item if $ q' \in Q_{n} $: $ \omega \in
\Omega_{(n,D_{i}(q'), D_{w}(q'))} $.
\end{enumerate}
We conclude with setting
\begin{equation}
       \delta_{\mathcal{M}}(q,\sigma,\gamma,q',\gamma',\omega) =
       \sqrt{ \delta_{\mathcal{P}}(q,\sigma,\gamma,q',\gamma')},
\end{equation}
i.e. the quantum transitions behave exactly as if they are probabilistic.
\end{proof}
In fact, this result can be extended for other kind of machines, such as, PDAs, CAs, and FAs.

Watrous' QTM model in \cite{Wa03}, which we call Wa03-QTM for ease of
reference, is a CQTM variant
that has an additional classical work tape and classical internal
states. Every Wa03-QTM can be simulated exactly (i.e. preserving the
same acceptance probability for every input) by CQTMs with only some
time overhead. Note that
Wa03-QTMs allow only algebraic transition amplitudes by definition.

\subsection{Realtime Quantum Finite Automata} \label{qtm:RT-QFAs}

Let us consider realtime versions of 2QFAs, 
whose tape heads are forced by definition to have classical locations. 
If the quantum machine model used is sufficiently general, 
then the intermediate measurements can be postponed easily to the end of
the algorithm in realtime computation. That final measurement can be performed on the set of
internal states, rather than the finite register.
Therefore, as with RT-FAs, we specify a subset of the internal
states of the machine as the collection of accepting states, denoted as $ Q_{a} $.

A realtime quantum finite automaton (RT-QFA) \cite{Hi08} is a 5-tuple
\begin{equation}
	\mathcal{M}=(Q,\Sigma,\{\mathcal{E}_{\sigma \in \tilde{\Sigma}}\},q_{1},Q_{a}),
\end{equation}
where each $ \mathcal{E}_{\sigma } $ is an operator having elements
$ \{ E_{\sigma,1},\ldots,E_{\sigma,k} \} $ for some $ k \in \mathbb{Z}^{+} $
satisfying
\begin{equation}
	\sum_{i=1}^{k} E_{\sigma,i}^{\dagger} E_{\sigma,i} = I.
\end{equation}
Additionally, we define the projector
\begin{equation}
	P_{a} = \sum_{q \in Q_{a}} \ket{q}\bra{q}
\end{equation}
in order to check for acceptance.
For any given input string $ w \in \Sigma^{*} $, $ \tilde{w} $ is
placed on the tape, and the computation can be traced by density matrices
\begin{equation}
	\rho_{j} = \mathcal{E}_{\tilde{w}_{j}} (\rho_{j-1}) =
	\sum_{i=1}^{k} E_{\tilde{w}_{j},i} \rho_{j-1} E_{\tilde{w}_{j},i}^{\dagger},
\end{equation}
where $ 1 \le j \le | \tilde{w} |  $ and $ \rho_{0} = \ket{q_{1}}
\bra{q_{1}} $ is the initial density matrix. 
The transition operators can be extended easily for any string as
\begin{equation}
	\mathcal{E}_{w \sigma} = \mathcal{E}_{w} \circ \mathcal{E}_{\sigma},
\end{equation}
where $ \sigma \in \tilde{\Sigma} $, $ w \in (\tilde{\Sigma})^{*} $, and $ \mathcal{E}_{\varepsilon} = I $.
Note that, if $ \mathcal{E}=\{E_{i} \mid 1 \le i \le k \} $ and 
$ \mathcal{E}'=\{ E'_{i} \mid 1 \le j \le k' \} $, then
\begin{equation}
	\mathcal{E'} \circ \mathcal{E} = \{ E_{j}'E_{i} \mid 1 \le i \le k, 1 \le j \le k' \}.
\end{equation}
The probability that RT-QFA $ \mathcal{M} $ accepts $ w $ is
\begin{equation}
	f_{\mathcal{M}}(w) = tr( P_{a} \mathcal{E}_{\tilde{w}}(\rho_{0})) = tr(P_{a} \rho_{| \tilde{w} |} ).
\end{equation}

\begin{lemma}
	\label{qtm:lem:RT-QFA-to-RT-GFA}
	For a given RT-QFA $ \mathcal{M} $ with $ n $ internal states,
	there exists a GFA  $ \mathcal{G} $ with $ n^{2} $ internal states such that
	$ f_{\mathcal{M}}(w) = f_{\mathcal{G}}(w) $ for all $ w \in \Sigma^{*} $.
\end{lemma}
\begin{proof}
	Let $ \mathcal{M} = (Q_{1},\Sigma,\mathcal{E}_{\sigma \in \tilde{\Sigma}},q_{1},Q_{a} ) $ be the RT-QFA with
	$ n $ internal states, and let each $ \mathcal{E}_{\sigma \in \tilde{\Sigma}} $ have $ k $ elements, 
	without loss of generality.
	We construct GFA $ \mathcal{G}=(Q_{2},\Sigma,\{A_{\sigma \in \Sigma}\},v_{0},f) $ with 
	$ 2n^{2} $ internal states.
	We use mapping $ vec $ described in Figure \ref{qtm:fig:vec} 
	in order to linearize the computation of $ \mathcal{M} $ so that it can be traced by GFA $ \mathcal{G} $.
	
	\begin{figure}[h!]		
		\begin{center}
		\fbox{
		\begin{minipage}{0.90\textwidth}
		\footnotesize{
			(Page 73 in \cite{Wa03})
			Let $ A $, $ B $, and $ C $      be $ n \times n $ dimensional matrices.
			$ vec $ is a linear mapping from $ n \times n $ 
			matrices to $ n^{2} $ dimensional (column) vectors defined as
			\begin{equation}
				vec(A)[(i-1)n+j] = A[i,j],
			\end{equation}
			where $ 1 \le i,j \le N $. 
			One can verify the following properties:
			\begin{equation}
				\label{equation:vec-ABC}
				vec(ABC) = (A \otimes C^{\trans})vec(B)
			\end{equation}
			and
			\begin{equation}
				\label{equation:vec-AtB}
				tr(A^{\trans}B)=vec(A)^{\trans}vec(B).
			\end{equation}
		}
		\end{minipage}
		}
		\end{center}
		\caption{The definition and properties of $ vec $}
		\vskip\baselineskip
		\label{qtm:fig:vec}		
	\end{figure}

	We define
	\begin{equation}
		v_{0}' = vec(\rho_{1}),
	\end{equation}
	where
	\begin{equation}
		\rho_{1} = \mathcal{E}_{\cent}(\rho_{0}) = \sum_{i=1}^{k} E_{\cent,i} \rho_{0} E_{\cent,i}^{\dagger}.
	\end{equation}
	For each $ \sigma \in \Sigma $, we define
	\begin{equation}
		A_{\sigma}' = \sum_{i=1}^{k} E_{\sigma,i} \otimes E_{\sigma,i}^{*}
	\end{equation}
	and so we obtain (by Equation \ref{equation:vec-ABC})
	\begin{equation}
		\rho' = \mathcal{E_{\sigma}}(\rho) = \sum_{i=1}^{k} E_{\sigma,i} \rho E_{\sigma,i}^{\dagger}
		\rightarrow vec(\rho') = A_{\sigma}' vec(\rho),
	\end{equation}
	for any density matrix $ \rho $.
	Finally, we define
	\begin{equation}
		f' = vec(P_{a})^{T} \sum_{i=1}^{k} E_{\dollar,i} \otimes E_{\dollar,i}^{*}.
	\end{equation}
	It can be verified by using Equation \ref{equation:vec-AtB} that for any input string $ w \in \Sigma^{*} $,
	\begin{equation}
		f' A'_{w_{|w|}} \cdots A'_{w_{1}} v_{0}' =
		tr(P_{a} \mathcal{E}_{\dollar} \circ \mathcal{E}_{w} \circ \mathcal{E}_{\cent} (\rho_{0})) =
		f_{\mathcal{M}}(w).
	\end{equation}
	The complex entries of $ v_{0}' $, $ \{ A_{\sigma \in \Sigma}' \} $, and $ f' $
	can be replaced \cite{MC00} with $ 2 \times 2 $ 
	dimensional real matrices\footnote{$ a+bi $ is replaced with 
	$ \left( \begin{array}{cc} a & b \\ -b & a \end{array} \right) $.},
	and so we obtain the equations
	\begin{equation}
		\left(
		\begin{array}{cc}
			f_{\mathcal{M}}(w) & 0 \\
			0 & f_{\mathcal{M}}(w)
		\end{array}
		\right) = f''  A''_{w_{|w|}} \cdots 
		A''_{w_{1}} v_{0}'',
	\end{equation}
	where the terms with double primes are obtained from the corresponding terms with single primes.
	We finish the construction of $ \mathcal{G} $ by stating that
	\begin{enumerate}
		\item $ v_{0} $ is the first column of $ v_{0}'' $,
		\item $ A_{\sigma} $ is equal to $ A_{\sigma}'' $, for each $ \sigma \in \Sigma $, and
		\item $ f $ is the first row of $ f'' $.
	\end{enumerate}
    We also refer the reader to \cite{MC00,LQ08} presenting similar constructions for other
	types of realtime QFAs.
	An alternative demonstration, not using the density matrix formalism, can be found in \cite{YS09A}.	
\end{proof}

\begin{corollary}
	QAL = S.
\end{corollary}
We therefore have that realtime unbounded-error probabilistic and quantum finite automata are equivalent in power.
We show in Section \ref{uerr:general} that 
this equivalence does not carry over to the one-way and two-way cases.

\subsection{Quantum Turing Machines with Restricted Measurements} \label{qtm:restricted-QTMs}
                      
In another specialization of the QTM model, the \textit{QTM with
restricted measurements}, the machine is unidirectional, the
heads can enter quantum superpositions, $ \Delta = \{n,a,r\} $, and
$ | \Omega_{n} | = | \Omega_{a} | = | \Omega_{r} | = 1  $. The first
family of QTMs that was formulated for the analysis of space
complexity issues \cite{Wa98,Wa99}, which we call the Wa98-QTM,
corresponds to such a model, with the added restriction that the transition
amplitudes are only allowed to be rational numbers.
The finite automaton versions of QTMs with restricted measurements\footnote{These models,
which also allow unrestricted transition amplitudes by the convention
in automata theory, are introduced in the paper written by Kondacs and Watrous \cite{KW97}.}
are known as Kondacs-Watrous quantum finite automata, and abbreviated
as 2KWQFAs, 1KWQFA, or RT-KWQFAs, depending on the set of allowed
directions of movement for the input head.
These are pure state models, 
since the non-halting part of the computation is always represented by
a single quantum state.
Therefore, configuration or state vectors, rather than the density
matrix formalism, can be used in order to
trace the computation easily.
To be consistent with the literature on 2KWQFAs, we specialize the
2QFA model by the following process:
\begin{enumerate}
	\item The finite register does not need to be refreshed, since the 
		computation continuous if and only if the initial symbol is observed.
	\item In fact, 2KWQFAs do not need to have the finite register at all, instead, similarly to 2PFAs, the set of
		internal states of the 2KWQFA is partitioned to sets of nonhalting, accepting, and rejecting states, 
		denoted as $ Q_{n} $, $ Q_{a} $, and $ Q_{r} $, respectively, which can be obtained easily by taking 
		the tensor product of the internal states of the 2QFA and the set $ \{n, a, r\} $.
	\item A configuration is designated as nonhalting (resp., accepting or rejecting),
		if its internal state is a member of $ Q_{n} $ (resp., $ Q_{a} $ or $ Q_{r} $). 
		Nonhalting (resp., accepting or rejecting) configurations form the set $ \mathcal{C}^{w}_{n} $
		(resp., $ \mathcal{C}^{w}_{a} $ or $ \mathcal{C}^{w}_{r} $) 
		(for a given input string $ w \in \Sigma^{*} $).
	\item The evolution of the configuration sets can be represented by a unitary matrix.
	\item The measurement is done on the configuration set with projectors $ P_{n} $, $ P_{a} $, and $ P_{r}$,
		defined as
        \begin{equation}
			P_{\tau \in \{n,a,r\}}  =\sum_{c \in \mathcal{C}^{w}_{\tau}} \ket{c}\bra{c}
		\end{equation}
        for a given input string $ w \in \Sigma^{*} $, where 
        the standard actions are associated with the outcomes ``$n$'', ``$a$'', and ``$r$''.
\end{enumerate}

Formally, a 2KWQFA is a 6-tuple
\begin{equation}
	\mathcal{M} = \{Q,\Sigma,\delta,q_{1},Q_{a},Q_{r}\},
\end{equation}
where $ Q_{n} = Q \setminus \{ Q_{a} \cup Q_{r} \} $ and $ q_{1} \in Q_{n} $.
$ \delta $ induces a unitary matrix $ U_{\sigma}$, whose rows and columns are
indexed by internal states for each input symbol $ \sigma $.
Note that, since all 2KWQFAs are unidirectional, 
we use the notations $ \leftstate{q} $, $ \stopstate{q} $, and $ \rightstate{q} $
for internal state $ q $ in order to represent the value of $ D_{i}(q) $ as
$ \leftarrow $, $ \downarrow $, and $ \rightarrow $, respectively.

A RT-KWQFA is a 6-tuple
\begin{equation}
	\mathcal{M} = \{Q,\Sigma,\{U_{\sigma \in \tilde{\Sigma}}\},q_{1},Q_{a},Q_{r}\},
\end{equation}
where $ \{U_{\sigma \in \tilde{\Sigma}}\} $ are unitary transition matrices.
In contrast to the other kinds of realtime finite automata, a RT-KWQFA is
measured at each step during computation
after the unitary transformation is applied.
The projectors are defined as
\begin{equation}
	 P_{\tau \in \Delta}     =\sum_{q \in Q_{\tau}} \ket{q}\bra{q}.
\end{equation}
The nonhalting portion of the computation of a RT-KWQFA can be traced by
a state vector, say $ \ket{u} $, such that
$ \ket{u}[i] $ corresponds to state $ q_{i} $.
The computation begins with $ \ket{u_{0}} = \ket{q_{1}} $. For a given
input string $ w \in \Sigma^{*} $,
at step $ j $ $ (1 \le j \le |\tilde{w}|) $:
\begin{equation}
	\ket{u_{j}} = P_{n}U_{\tilde{w}_{j}} \ket{u_{j-1}},
\end{equation}
the input is accepted with probability
\begin{equation}
	|| P_{a}U_{\tilde{w}_{j}} \ket{u_{j-1}} ||^{2},
\end{equation}
and rejected with probability
\begin{equation}
	|| P_{r}U_{\tilde{w}_{j}} \ket{u_{j-1}} ||^{2}.
\end{equation}
The overall acceptance and rejection probabilities are accumulated by
summing up these values at each step. Note that, the state vector
representing the nonhalting portion is not normalized in the
description given above.

The most restricted QFA model is the Moore-Crutchfield QFA (MCQFA) \cite{MC00}, 
which can be seen as a special case of RT-KWQFA
such that only a unique measurement is done after reading symbol $ \dollar $.

\section{Quantum Pushdown Automata} \label{qtm:QPDAs}

Two-way quantum pushdown automaton (2QPDA),
one-way quantum pushdown automaton (1QPDA), and
realtime quantum pushdown automaton (RT-QPDA) can all be represented to be a 7-tuple
\begin{equation}
	\mathcal{P} = ( \Sigma,\Gamma, \Omega, Q, \delta, q_{1}, \Delta ).
\end{equation}
We refer the reader to Appendix \ref{well:QPDAs} for the list of the local conditions for 2QPDAs, 1QPDAs,
and RT-QPDAs wellformedness.

If a 2QPDA (resp., 1QPDA or RT-QPDA) is unidirectional, 
then there is only one local condition for well-formedness:
for any choice of $ q_{1},q_{2} \in Q $, $ \sigma \in \tilde{\Sigma} $, and
$ \gamma_{1} \in \Gamma,\gamma_{2} \in \tilde{\Gamma} $,
\small
\begin{equation}
	\sum\limits_{q' \in Q,\gamma' \in \Upsilon,\omega \in \Omega} \mspace{-15mu}
	\overline{\delta( q_{1},\sigma, \gamma_{1},q',\gamma',\omega )}
	\delta( q_{2},\sigma, \gamma_{2},q',\gamma',\omega )
	= \left\lbrace
	\begin{array}{ll}
		1 ~~ & q_{1} = q_{2} \mbox{ and } \gamma_{1} = \gamma_{2} \\
		0 & \mbox{otherwise}
	\end{array}
	\right. .
\end{equation}
\normalsize
In the unidirectional case, we can define an admissible operator for each $ \sigma \in \tilde{\Sigma} $,
i.e. $ \mathcal{E}_{\sigma} =\{ E_{\sigma,\omega}\} $,
where $\omega \in \Omega  $ and  $ E_{\sigma,\omega}[j,i] $ represents
the amplitude of the transition $ \delta(q_{i},\sigma,\gamma_{i},q_{j},\gamma_{j},\omega) $,
where $ (q_{i},\gamma_{i}) $ and $ (q_{j},\gamma_{j}) $ are the pairs corresponding to
the $ i^{th} $ and $ j^{th} $ columns (rows), respectively, and
$ q_{i},q_{j} \in Q; \gamma_{i},\gamma_{j} \in \Upsilon $.

It is an open problem whether the computation powers of 2QPDAs (resp., 1QPDAs or RT-QPDAs) and uni-2QPDAs 
(resp., uni-1QPDAs or uni-RT-QPDAs) are the same or not.

\section{Quantum Counter Automata} \label{qtm:QCAs}

Since quantum CAs are a special case of quantum PDAs, we do not give the details for all
kind of quantum CAs. Instead, we focus on the realtime counterpart of quantum counter automata.

A realtime quantum $ k $-counter automaton (RT-Q$ k $CA) is a 6-tuple
\begin{equation}
	\mathcal{M} = (Q,\Sigma,\Omega,\delta,q_{1},Q_{a}).
\end{equation}
We give the local conditions of well-formedness for RT-Q1CAs:
For any choice of $ q_{1},q_{2} \in Q $, $ \sigma \in \tilde{\Sigma} $, and
$ \theta_{1},\theta_{2} \in \Theta $,
\small
\begin{equation}
       \sum\limits_{q^{\prime} \in Q, c \in \lozenge, \omega \in \Omega }
\mspace{-27mu}
               \overline{
               \delta(q_{1},\sigma,\theta_{1},q',c,\omega ) }
               \delta(q_{2},\sigma,\theta_{2},q',c,\omega )
               = \left\lbrace
               \begin{array}{lll}
                       1 ~~ &  q_{1} = q_{2} \mbox{ and } \theta_{1} = \theta_{2} \\
                       0 &  \mbox{otherwise}
               \end{array}
               \right. ,
\end{equation}
\normalsize
\begin{equation}
       \begin{array}{ll}
               \displaystyle \sum\limits_{q^{\prime} \in Q, \omega \in \Omega }
               &
               \overline{\delta(q_{1},\sigma,\theta_{1},q',+1,\omega)}
               \delta(q_{2},\sigma,\theta_{2},q', 0 ,\omega)
               \\
               & + ~
               \overline{\delta(q_{1},\sigma,\theta_{1},q',0,\omega)}
               \delta(q_{2},\sigma,\theta_{2},q', -1 ,\omega) = 0
       \end{array}     ,
\end{equation}
and
\begin{equation}
               \displaystyle \sum\limits_{q^{\prime} \in Q, \omega \in \Omega }
               \overline{\delta(q_{1},\sigma,\theta_{1},q',+1,\omega)}
               \delta(q_{2},\sigma,\theta_{2},q', -1 ,\omega) = 0.
\end{equation}

If the RT-Q1CA (and RT-Q$ k $CA) is unidirectional, then there is only one local condition for well-formedness:
for any choice of $ q_{1},q_{2} \in Q $, $ \sigma \in \tilde{\Sigma} $, and
$ \bar{\theta_{1}},\bar{\theta_{2}} \in \Theta^{k} $,
\begin{equation}
	\mspace{-1mu}
	\sum\limits_{q^{\prime} \in Q,  \omega \in \Omega } \mspace{-15mu}
	\overline{
	\delta(q_{1},\sigma,\bar{\theta_{1}},q',\omega ) }
	\delta(q_{2},\sigma,\bar{\theta_{2}},q',\omega )
	= \left\lbrace
	\begin{array}{lll}
		1 ~~ &  q_{1} = q_{2} \mbox{ and } \bar{\theta_{1}} = \bar{\theta_{2}} \\
		0 &  \mbox{otherwise}
	\end{array}
	\right. .
\end{equation}
In the unidirectional case, we can define an admissible operator $ \mathcal{E}_{\sigma,\bar{\theta}}$ 
for each $ \sigma \in \tilde{\Sigma} $ and $ \bar{\theta} \in \Theta^{k} $, 
which is described by a collection $\{ E_{\sigma,\bar{\theta},\omega}\} $,
where $\omega \in \Omega  $ and  $ E_{\sigma,\bar{\theta},\omega}[j,i] $ represents
the amplitude of the transition  from state $ q_{i} $ to $ q_{j} $ 
when reading symbol $ \sigma $, having counter signs $ \bar{\theta} $, 
and writing $ \omega $ on the finite register 
(and so the value of the counter is updated by $ D_{c}(q') $).
Note that, $ \mathcal{E}_{\sigma,\bar{\theta}} $ is admissible if and only if
\begin{equation}
	\sum_{\omega \in \Omega} E_{\sigma,\bar{\theta},\omega}^{\dagger} E_{\sigma,\bar{\theta},\omega} = I.
\end{equation}
It is an open problem whether the computation powers of RT-Q$ k $CA and uni-RT-Q$ k $CA are the same or not.

We close the section by showing the isomorphism of RT-Q$ k $CA($ m $) and RT-Q$ k $CA.

\begin{lemma}
	\label{qtm:lem:CA-m-isomorphic-CA}
	Any RT-Q$ k $CA($ m $) can be exactly simulated by a RT-Q$ k $CA.
\end{lemma}
\begin{proof}
	Let $ \mathcal{M} = (Q,\Sigma,\Omega,\delta,q_{1},Q_{a}) $ be the RT-Q$ k $CA($ m $) and 
	$ \mathcal{M}' = (Q',\Sigma,\Omega,\delta',q_{1}',Q'_{a}) $ be
	the RT-Q$ k $CA simulating $ \mathcal{M} $ exactly.
	For each internal state of $ \mathcal{M} $, say $ q \in Q $, 
	we define $ m^{k} $ internal states, i.e. $ \LRofC{q,i_{1},\ldots,i_{k}} \in Q' $ 
	$ (0 \le i_{j} \le m-1, 1 \le j \leq k) $, for $ \mathcal{M}' $.
	Moreover, $ q_{1}' = \LRofC{q_{1},0,\ldots,0} $
	and $ Q'_{a} = Q_{a} \times \{0,\ldots,m-1\}^{k} $.

	Let 
	\begin{equation}
		\varphi : \mathbb{Z}^{k} \rightarrow \mathbb{Z}^{k} \times \{ 0,\ldots,m-1\}^{k}
	\end{equation} 
	be a bijection such that 
	\begin{equation}
		\varphi(x_{1},\ldots,x_{k})= 
		\left( \left\lfloor \frac{x_{1}}{m} \right\rfloor, \ldots, 
		\left\lfloor \frac{x_{k}}{m} \right\rfloor, (x_{1} \mod m),\ldots,(x_{k} \mod m) \right).
	\end{equation}	
	Hence, we can say that the counter values of $ \mathcal{M} $, say $ \bar{x} \in \mathbb{Z}^{k} $, 
	can be equivalently represented by $ \varphi(\bar{x}) $, based on which we construct $ \mathcal{M}' $,
	where $ \varphi(\bar{x})[i] $ is stored by the $ i^{th} $ counter and $ \varphi(\bar{x})[k+i] $
	is stored by the internal state.	
	That is, for any configuration of $ \mathcal{M} $, say ($ q,\bar{x} $),
	we have an equivalent configuration of $ \mathcal{M}' $ as 
	\begin{equation}
		( \LRofC{ q , \varphi(\bar{x})[k+1],\ldots, \varphi(\bar{x})[2k] },
			\varphi(\bar{x})[1],\ldots,\varphi(\bar{x})[k]).
	\end{equation}
	Moreover, the transitions of $ \mathcal{M}' $ can be obtained from those of $ \mathcal{M} $
	in the following way:
	for any $ (i_{1},\ldots,i_{k}) \in \{-m,\ldots,m\}^{k} $ 
	and $ (j_{1},\ldots,j_{k}) \in \{0,\ldots,m-1\}^{k} $, 
	the part of the transition 
	\begin{equation}
		(q,\sigma) \overset{\delta}{\longrightarrow} \alpha (q',i_{1},\ldots,i_{k},\omega)
	\end{equation} 
	of $ \mathcal{M} $ is replaced by transition
	\begin{equation}
		\small
		\begin{array}{l}
		( \LRofC{q,j_{1},\ldots,j_{k}},\sigma) \overset{\delta'}{\longrightarrow} 
		\\
		~~~~~
		\alpha \left( \LRofC{ q', (j_{1}+i_{1} \mod m),\ldots,(j_{k}+i_{k} \mod m)}, 
		\left\lfloor \frac{j_{1}+i_{1}}{m} \right\rfloor, \ldots,
		\left\lfloor \frac{j_{k}+i_{k}}{m} \right\rfloor,
		\omega \right)
		\end{array}
	\end{equation}
	in $ \mathcal{M}' $, where $ q \in Q $, $ \sigma \in \tilde{\Sigma} $, $\omega \in \Omega $, and
	$ \alpha \in \mathbb{C} $ is the amplitude of the transition.
	Since $ \varphi $ is a bijection, the configuration matrix of $ \mathcal{M} $ is isomorphic to
	the one of $ \mathcal{M}' $ for any input string $ w \in \Sigma^{*} $.
	(Similarly, $ \mathcal{M} $ is well-formed if and only if $ \mathcal{M}' $ is well-formed.) 
	Therefore, they process exactly the same computation on a given input string, say $ w \in \Sigma^{*} $,
	and so
	\begin{equation}
		f_{\mathcal{M}}(w) = f_{\mathcal{M}'}(w).
	\end{equation}
\end{proof}

\section{Nondeterministic Quantum Machines} \label{qtm:nondeterministic}

A nondeterministic quantum machine is a quantum machine with the error setting of 
positive one-sided unbounded error.
``N'' is used before ``Q'' in the abbreviations of quantum machines.

\chapter{SUBLOGARITHMIC-SPACE UNBOUNDED-ERROR COMPUTATION} \label{uerr}

In this chapter, we focus on unbounded error quantum computation in sublogatihmic space 
(Section \ref{uerr:general})
and nondeterminitic quantum computation in constant space (Section \ref{uerr:nondeterministic}).

\section{Unbounded Error Results} \label{uerr:general}

Watrous compared the unbounded-error probabilistic space complexity
classes (PrSPACE$ _{\mathbb{Q}} $($s$) and PrSPACE$ _{\mathbb{A}} $($s$)) 
with the corresponding classes for both Wa98-QTMs \cite{Wa98,Wa99} and 
Wa03-QTMs \cite{Wa03}, respectively, for space bounds $ s=\Omega(\log n) $, 
establishing the identity of the associated
quantum space complexity classes with each other, and also with the
corresponding probabilistic ones.
The case of $s=o(\log n)$ was left as an open question \cite{Wa99}. 
In this section, we provide an answer to that question.

\subsection{Probabilistic versus Quantum Computation with Sublogarithmic Space} 
\label{uerr:probabilistic-versus-quantum}

We already know that QTMs allowing superoperators are at least as powerful as PTMs for any common space bound. 
We now exhibit a 1KWQFA which performs a task that is impossible for PTMs with small space bounds.

Consider the nonstochastic and context-free language \cite{NH71}
	\begin{equation}
		\small
		L_{NH} = \{a^{x}ba^{y_{1}}ba^{y_{2}}b \cdots a^{y_{t}}b \mid x,t,y_{1}, \cdots, y_{t}
       \in \mathbb{Z}^{+} \mbox{ and } \exists k ~ (1 \le k \le t), x=\sum_{i=1}^{k}y_{i} \}
	\end{equation}
over the alphabet $ \Sigma = \{a,b \} $.
Freivalds and Karpinski \cite{FK94} have
proven the following facts about $ L_{NH} $:

\begin{fact}
	No PTM using space $ o(\log\log n) $ can recognize $ L_{NH} $ with unbounded error.
\end{fact}

\begin{fact}
	No 1PTM using space $ o(\log n) $ can recognize $ L_{NH} $ with unbounded error.
\end{fact}

There exists a one-way \textit{deterministic} TM that recognizes $ L_{NH} $ within the optimal space bound 
$O(\log n)$ \cite{FK94}. 
No (two-way) PTM which recognizes $ L_{NH} $ using $ o(\log n) $ space is known as of the time of writing.

\begin{theorem}
       \label{theorem:1KWQFA}
       There exists a 1KWQFA that recognizes $ L_{NH} $ with unbounded error.
\end{theorem}
\begin{proof}
       Consider the 1KWQFA $ \mathcal{M}=(Q, \Sigma, \delta, q_{0}, Q_{a},Q_{r}) $, 
       where $ \Sigma=\{a,b\} $ and the state sets are as follows:
       \begin{equation}
               \begin{array}{lcl}
                       Q_{n} & = &
                               \{ \rightstate{q_{0}} \} \cup \{ \rightstate{q_{i}} \mid 1 \le i \le 6 \} \cup \{
\rightstate{p_{i}} \mid 1 \le i \le 6 \}
                               \cup \{ \rightstate{a_{i}} \mid 1 \le i \le 4 \} \\
                               & & \cup~\{ \rightstate{r_{i}} \mid 1 \le i \le 4 \}
                                       \cup \{ \stopstate{w_{i}} \mid 1 \le i \le 6 \}, \\
                       Q_{a} &  =  & \{ \stopstate{A_{i}} \mid 1 \le i \le 18  \},
                               ~Q_{r} = \{ \stopstate{R_{i}} \mid 1 \le i \le 18  \}. \\
               \end{array}
        \end{equation}
       Let each $ U_{\sigma} $ induced by $ \delta $ act as indicated in Figures
\ref{figure:1.5KWQFA-1} and \ref{figure:1.5KWQFA-2}, and extend each to be unitary.
\begin{figure}[!h]
	\begin{center}
\footnotesize
\begin{tabular}{|c|l|l|}
\hline
       Stages & \multicolumn{1}{c|}{$ U_{\cent}, U_{a} $} &
                \multicolumn{1}{c|}{$ U_{\dollar} $} \\
\hline
       & $ U_{\cent} \rightconf{q_{0}} = \frac{1}{\sqrt{2}}\rightconf{q_{1}} +
\frac{1}{\sqrt{2}}\rightconf{p_{1}} $  & \\
\hline
       $ \begin{array}{@{}c@{}} \mbox{I} \\ ( \mathsf{path_{1}} ) \end{array} $
       &
       $ \begin{array}{@{}l@{}}
               U_{a} \rightconf{q_{1}} = \frac{1}{\sqrt{2}}\rightconf{q_{2}} + \frac{1}{2}
\stopconf{A_{1}} + \frac{1}{2} \stopconf{R_{1}} \\
               U_{a} \rightconf{q_{2}} = \frac{1}{\sqrt{2}}\rightconf{q_{2}} - \frac{1}{2}
\stopconf{A_{1}} - \frac{1}{2} \stopconf{R_{1}}
       \end{array} $
       &
       $ \begin{array}{@{}l@{}}
               U_{\dollar} \rightconf{q_{1}}= \frac{1}{\sqrt{2}} \stopconf{A_{1}} + \frac{1}{\sqrt{2}}\stopconf{R_{1}} \\
               U_{\dollar} \rightconf{q_{2}}= \frac{1}{\sqrt{2}} \stopconf{A_{2}} + \frac{1}{\sqrt{2}}\stopconf{R_{2}} \\
               U_{\dollar} \rightconf{q_{3}}= \frac{1}{\sqrt{2}} \stopconf{A_{3}} + \frac{1}{\sqrt{2}}\stopconf{R_{3}}
       \end{array}$
       \\
\hline
       $ \begin{array}{@{}c@{}} \mbox{I} \\ ( \mathsf{path_{2}} ) \end{array} $
       &
       $ \begin{array}{@{}l@{}}
               U_{a} \rightconf{p_{1}} = \stopconf{w_{1}} \\
               U_{a} \stopconf{w_{1}} = \frac{1}{\sqrt{2}}\rightconf{p_{2}} + \frac{1}{2}
\stopconf{A_{2}} + \frac{1}{2} \stopconf{R_{2}} \\
               U_{a} \rightconf{p_{2}} = \stopconf{w_{2}} \\
               U_{a} \stopconf{w_{2}} = \frac{1}{\sqrt{2}}\rightconf{p_{2}} - \frac{1}{2}
\stopconf{A_{2}} - \frac{1}{2} \stopconf{R_{2}}
       \end{array} $
       &
       $ \begin{array}{@{}l@{}}
               U_{\dollar} \rightconf{p_{1}}= \frac{1}{\sqrt{2}} \stopconf{A_{4}} + \frac{1}{\sqrt{2}}\stopconf{R_{4}} \\
               U_{\dollar} \rightconf{p_{2}}= \frac{1}{\sqrt{2}} \stopconf{A_{5}} + \frac{1}{\sqrt{2}}\stopconf{R_{5}} \\
               U_{\dollar} \rightconf{p_{3}}= \frac{1}{\sqrt{2}} \stopconf{A_{6}} + \frac{1}{\sqrt{2}}\stopconf{R_{6}}
       \end{array}$
       \\
\hline
       $ \begin{array}{@{}c@{}} \mbox{II} \\ ( \mathsf{path_{1}} ) \end{array} $
       &
       $ \begin{array}{@{}l@{}}
               U_{a} \rightconf{q_{3}}=\stopconf{w_{3}} \\
               U_{a} \stopconf{w_{3}} = \frac{1}{\sqrt{2}}\rightconf{q_{4}} + \frac{1}{2}
\stopconf{A_{3}} + \frac{1}{2} \stopconf{R_{3}} \\
               U_{a} \rightconf{q_{4}}=\stopconf{w_{4}} \\
               U_{a} \stopconf{w_{4}} = \frac{1}{\sqrt{2}}\rightconf{q_{4}} - \frac{1}{2}
\stopconf{A_{3}} - \frac{1}{2} \stopconf{R_{3}}
       \end{array} $
       &
       $ \begin{array}{@{}l@{}}
               U_{\dollar} \rightconf{q_{4}} = \frac{1}{\sqrt{2}} \stopconf{A_{7}} + \frac{1}{\sqrt{2}}\stopconf{R_{7}} \\
               U_{\dollar} \rightconf{q_{5}}=  \frac{1}{\sqrt{2}} \stopconf{A_{8}} + \frac{1}{\sqrt{2}}\stopconf{R_{8}}
       \end{array} $
       \\
\hline
       $ \begin{array}{@{}c@{}} \mbox{II} \\ ( \mathsf{path_{2}} ) \end{array} $
       &
       $ \begin{array}{@{}l@{}}
               U_{a} \rightconf{p_{3}}=\frac{1}{\sqrt{2}} \rightconf{p_{4}} + \frac{1}{2}
\stopconf{A_{4}} + \frac{1}{2} \stopconf{R_{4}} \\
               U_{a} \rightconf{p_{4}}=\frac{1}{\sqrt{2}} \rightconf{p_{4}} - \frac{1}{2}
\stopconf{A_{4}} - \frac{1}{2} \stopconf{R_{4}}
       \end{array} $
       &
       $ \begin{array}{@{}l@{}}
               U_{\dollar} \rightconf{p_{4}} = \frac{1}{\sqrt{2}} \stopconf{A_{9}} + \frac{1}{\sqrt{2}}\stopconf{R_{9}} \\
               U_{\dollar} \rightconf{p_{5}}= \frac{1}{\sqrt{2}} \stopconf{A_{10}} + \frac{1}{\sqrt{2}}\stopconf{R_{10}}
       \end{array} $
       \\
\hline
       $ \begin{array}{@{}c@{}} \mbox{III} \\ ( \mathsf{path_{1}} ) \end{array} $
       &
       $ \begin{array}{@{}l@{}}
               U_{a} \rightconf{q_{5}}=\stopconf{w_{5}} \\
               U_{a} \stopconf{w_{5}} = \frac{1}{\sqrt{2}}\rightconf{q_{6}} + \frac{1}{2}
\stopconf{A_{5}} + \frac{1}{2} \stopconf{R_{5}} \\
               U_{a} \rightconf{q_{6}}=\stopconf{w_{6}} \\
               U_{a} \stopconf{w_{6}} = \frac{1}{\sqrt{2}}\rightconf{q_{6}} - \frac{1}{2}
\stopconf{A_{5}} - \frac{1}{2} \stopconf{R_{5}}
       \end{array} $
       &
       $ \begin{array}{@{}l@{}}
               U_{\dollar} \rightconf{q_{6}} = \frac{1}{\sqrt{2}} \stopconf{A_{11}} + \frac{1}{\sqrt{2}}\stopconf{R_{11}}
       \end{array} $
       \\
\hline
       $ \begin{array}{@{}c@{}} \mbox{III} \\ ( \mathsf{path_{2}} ) \end{array} $
       &
       $ \begin{array}{@{}l@{}}
               U_{a} \rightconf{p_{5}}=\frac{1}{\sqrt{2}} \rightconf{p_{6}} + \frac{1}{2}
\stopconf{A_{6}} + \frac{1}{2} \stopconf{R_{6}} \\
               U_{a} \rightconf{p_{6}}=\frac{1}{\sqrt{2}} \rightconf{p_{6}} - \frac{1}{2}
\stopconf{A_{6}} - \frac{1}{2} \stopconf{R_{6}}
       \end{array} $
       &
       $ \begin{array}{@{}l@{}}
               U_{\dollar} \rightconf{p_{6}} = \frac{1}{\sqrt{2}} \stopconf{A_{12}} + \frac{1}{\sqrt{2}}\stopconf{R_{12}}
       \end{array} $
       \\
\hline
       $ \begin{array}{@{}c@{}} \mbox{III} \\ ( \mathsf{path_{accept}} )
\end{array} $
       &
       $ \begin{array}{@{}l@{}}
               U_{a} \rightconf{a_{1}}=\frac{1}{\sqrt{2}} \rightconf{a_{2}} + \frac{1}{2}
\stopconf{A_{7}} + \frac{1}{2} \stopconf{R_{7}} \\
               U_{a} \rightconf{a_{2}}=\frac{1}{\sqrt{2}} \rightconf{a_{2}} - \frac{1}{2}
\stopconf{A_{7}} - \frac{1}{2} \stopconf{R_{7}} \\
               U_{a} \rightconf{a_{3}}=\frac{1}{\sqrt{2}} \rightconf{a_{4}} + \frac{1}{2}
\stopconf{A_{8}} + \frac{1}{2} \stopconf{R_{8}} \\
               U_{a} \rightconf{a_{4}}=\frac{1}{\sqrt{2}} \rightconf{a_{4}} - \frac{1}{2}
\stopconf{A_{8}} - \frac{1}{2} \stopconf{R_{8}}
       \end{array} $
       &
       $ \begin{array}{@{}l@{}}
               U_{\dollar} \rightconf{a_{1}} = \stopconf{A_{17}} \\
               U_{\dollar} \rightconf{a_{3}} = \stopconf{A_{18}} \\
               U_{\dollar} \rightconf{a_{2}} = \frac{1}{\sqrt{2}} \stopconf{A_{13}} + \frac{1}{\sqrt{2}}\stopconf{R_{13}} \\
               U_{\dollar} \rightconf{a_{4}} = \frac{1}{\sqrt{2}} \stopconf{A_{14}} + \frac{1}{\sqrt{2}}\stopconf{R_{14}}
       \end{array} $
       \\
\hline
       $ \begin{array}{@{}c@{}} \mbox{III} \\ ( \mathsf{path_{reject}} )
\end{array} $
       &
       $ \begin{array}{@{}l@{}}
               U_{a} \rightconf{r_{1}}=\frac{1}{\sqrt{2}} \rightconf{r_{2}} + \frac{1}{2}
\stopconf{A_{9}} + \frac{1}{2} \stopconf{R_{9}} \\
               U_{a} \rightconf{r_{2}}=\frac{1}{\sqrt{2}} \rightconf{r_{2}} - \frac{1}{2}
\stopconf{A_{9}} - \frac{1}{2} \stopconf{R_{9}} \\
               U_{a} \rightconf{r_{3}}=\frac{1}{\sqrt{2}} \rightconf{r_{4}} + \frac{1}{2}
\stopconf{A_{10}} + \frac{1}{2} \stopconf{R_{10}} \\
               U_{a} \rightconf{r_{4}}=\frac{1}{\sqrt{2}} \rightconf{r_{4}} - \frac{1}{2}
\stopconf{A_{10}} - \frac{1}{2} \stopconf{R_{10}}
       \end{array} $
       &
       $ \begin{array}{@{}l@{}}
               U_{\dollar} \rightconf{r_{1}} = \stopconf{R_{17}} \\
               U_{\dollar} \rightconf{r_{3}} = \stopconf{R_{18}} \\
               U_{\dollar} \rightconf{r_{2}} = \frac{1}{\sqrt{2}} \stopconf{A_{15}} + \frac{1}{\sqrt{2}}\stopconf{R_{15}} \\
               U_{\dollar} \rightconf{r_{4}} = \frac{1}{\sqrt{2}} \stopconf{A_{16}} + \frac{1}{\sqrt{2}}\stopconf{R_{16}}
       \end{array} $
       \\
       \hline
\end{tabular}
\end{center}	
	\caption{Specification of the transition function of the 1KWQFA presented in
		the proof of Theorem \ref{theorem:1KWQFA} (I)}
	\vskip\baselineskip
	\label{figure:1.5KWQFA-1}
\end{figure}
\begin{figure}[!h]
	\begin{center}
       \centering
\footnotesize
\begin{tabular}{|c|l|l|}
       \hline
       Stages & \multicolumn{2}{c|}{$ U_{b} $} \\
       \hline
       $ \begin{array}{@{}c@{}} \mbox{I} \\ ( \mathsf{path_{1}} ) \end{array} $
       &
       \multicolumn{2}{l|}{
               $ \begin{array}{@{}l@{}}
                       U_{b} \rightconf{q_{1}} = \frac{1}{\sqrt{2}} \stopconf{A_{1}} + \frac{1}{\sqrt{2}}\stopconf{R_{1}} \\
                       U_{b} \rightconf{q_{2}} = \rightconf{q_{3}} \\
                       U_{b} \rightconf{q_{3}} = \frac{1}{\sqrt{2}} \stopconf{A_{2}} + \frac{1}{\sqrt{2}}\stopconf{R_{2}}
               \end{array} $ }
       \\
       \hline
       $ \begin{array}{@{}c@{}} \mbox{I} \\ ( \mathsf{path_{2}} ) \end{array} $
       &
       \multicolumn{2}{l|}{
               $ \begin{array}{@{}l@{}}
                       U_{b} \rightconf{p_{1}} = \frac{1}{\sqrt{2}} \stopconf{A_{3}} + \frac{1}{\sqrt{2}}\stopconf{R_{3}} \\
                       U_{b} \rightconf{p_{2}} = \rightconf{p_{3}} \\
                       U_{b} \rightconf{p_{3}} = \frac{1}{\sqrt{2}} \stopconf{A_{4}} + \frac{1}{\sqrt{2}}\stopconf{R_{4}}
               \end{array} $
       }
       \\
       \hline
       $ \begin{array}{@{}c@{}} \mbox{II} \\ ( \mathsf{path_{1}} ) \end{array} $
       &
       \multicolumn{2}{l|}{
               $ \begin{array}{@{}l@{}}
                       U_{b} \rightconf{q_{4}} = \frac{1}{2}
\rightconf{q_{5}}+\frac{1}{2\sqrt{2}}\rightconf{a_{1}}+\frac{1}{2\sqrt{2}}\rightconf{r_{1}}
                               + \frac{1}{2} \stopconf{A_{11}} + \frac{1}{2} \stopconf{R_{11}}\\
                       U_{b} \rightconf{q_{5}} = \frac{1}{\sqrt{2}} \stopconf{A_{5}} + \frac{1}{\sqrt{2}}\stopconf{R_{5}}
               \end{array} $
       }
       \\
       \hline
       $ \begin{array}{@{}c@{}} \mbox{II} \\ ( \mathsf{path_{2}} ) \end{array} $
       &
       \multicolumn{2}{l|}{
               $ \begin{array}{@{}l@{}}
                       U_{b} \rightconf{p_{4}} = \frac{1}{2} \rightconf{p_{5}} + \frac{1}{2\sqrt{2}}\rightconf{a_{1}}
                       - \frac{1}{2\sqrt{2}}\rightconf{r_{1}} + \frac{1}{2} \stopconf{A_{12}} +
\frac{1}{2} \stopconf{R_{12}}\\
                       U_{b} \rightconf{p_{5}} = \frac{1}{\sqrt{2}} \stopconf{A_{6}} + \frac{1}{\sqrt{2}}\stopconf{R_{6}}
               \end{array} $
       }
       \\
       \hline
       $ \begin{array}{@{}c@{}} \mbox{III} \\ ( \mathsf{path_{1}} ) \end{array} $
       &
       \multicolumn{2}{l|}{
               $ \begin{array}{@{}l@{}}
                       U_{b} \rightconf{q_{6}} = \frac{1}{2}
\rightconf{q_{5}}+\frac{1}{2\sqrt{2}}\rightconf{a_{1}}+\frac{1}{2\sqrt{2}}\rightconf{r_{1}}
                               - \frac{1}{2} \stopconf{A_{11}} - \frac{1}{2} \stopconf{R_{11}}
               \end{array} $
       }
       \\
       \hline
       $ \begin{array}{@{}c@{}} \mbox{III} \\ ( \mathsf{path_{2}} ) \end{array} $
       &
       \multicolumn{2}{l|}{
               $ \begin{array}{@{}l@{}}
                       U_{b} \rightconf{p_{6}} = \frac{1}{2} \rightconf{p_{5}} + \frac{1}{2\sqrt{2}}\rightconf{a_{1}}
                               - \frac{1}{2\sqrt{2}}\rightconf{r_{1}} - \frac{1}{2} \stopconf{A_{12}} -
\frac{1}{2} \stopconf{R_{12}}
               \end{array} $
       }
       \\
       \hline
       $ \begin{array}{@{}c@{}} \mbox{III} \\ ( \mathsf{path_{accept}} )
\end{array} $
       &
       \multicolumn{2}{l|}{
               $ \begin{array}{@{}l@{}}
                       U_{b} \rightconf{a_{2}}=\frac{1}{\sqrt{2}} \rightconf{a_{3}} + \frac{1}{2}
\stopconf{A_{13}} + \frac{1}{2} \stopconf{R_{13}} \\
                       U_{b} \rightconf{a_{1}} = \frac{1}{\sqrt{2}} \stopconf{A_{7}} + \frac{1}{\sqrt{2}}\stopconf{R_{7}} \\
                       U_{b} \rightconf{a_{4}}=\frac{1}{\sqrt{2}} \rightconf{a_{3}} - \frac{1}{2}
\stopconf{A_{13}} - \frac{1}{2} \stopconf{R_{13}} \\
                       U_{b} \rightconf{a_{3}} = \frac{1}{\sqrt{2}} \stopconf{A_{8}} + \frac{1}{\sqrt{2}}\stopconf{R_{8}}
               \end{array} $
       }
       \\
       \hline
       $ \begin{array}{@{}c@{}} \mbox{III} \\ ( \mathsf{path_{reject}} )
\end{array} $
       &
       \multicolumn{2}{l|}{
               $ \begin{array}{@{}l@{}}
                       U_{b} \rightconf{r_{2}}=\frac{1}{\sqrt{2}} \rightconf{r_{3}} + \frac{1}{2}
\stopconf{A_{14}} + \frac{1}{2} \stopconf{R_{14}} \\
                       U_{b} \rightconf{r_{1}} = \frac{1}{\sqrt{2}} \stopconf{A_{9}} + \frac{1}{\sqrt{2}}\stopconf{R_{9}} \\
                       U_{b} \rightconf{r_{4}}=\frac{1}{\sqrt{2}} \rightconf{r_{3}} - \frac{1}{2}
\stopconf{A_{14}} - \frac{1}{2} \stopconf{R_{14}} \\
                       U_{b} \rightconf{r_{3}} = \frac{1}{\sqrt{2}} \stopconf{A_{10}} + \frac{1}{\sqrt{2}}\stopconf{R_{10}}
               \end{array} $
       }
       \\
       \hline
\end{tabular}
\end{center}
       \caption{Specification of the transition function of the 1KWQFA presented in
		the proof of Theorem \ref{theorem:1KWQFA} (II)}
		\vskip\baselineskip
	\label{figure:1.5KWQFA-2}
\end{figure}

       Machine $ \mathcal{M} $ starts computation on symbol $ \cent $ by
branching into two paths,
       $ \mathsf{path_{1}} $ and $ \mathsf{path_{2}} $, with equal amplitude.
       Each path and their subpaths, to be described later, check whether
the input is of the form
       $ (aa^{*}b)(aa^{*}b)(aa^{*}b)^{*} $.
       The different stages of the program indicated in Figures
\ref{figure:1.5KWQFA-1} and \ref{figure:1.5KWQFA-2}
       correspond to the
       subtasks of this regular expression check. Stage I ends successfully
if the input begins with $ (aa^{*}b) $.
       Stage II checks the second $ (aa^{*}b)$. Finally, Stage III controls
whether the input ends with $ (aa^{*}b)^{*} $.

       The reader note that many transitions in the machine are of the form
       \begin{equation} U_{\sigma} \ket{q_{i}} = \ket{\psi} + \alpha \ket{A_{k}} + \alpha
\ket{R_{k}}, \end{equation}
       where $ \ket{\psi} $ is a superposition of configurations such that $
\braket{\psi}{\psi}=1-2\alpha^{2} $,
       $ A_{k} \in Q_{a} $, $ R_{k} \in Q_{r} $.
       The equal-probability transitions to the ``twin halting states'' $
A_{k} $ and $ R_{k} $ are included to ensure that
       the matrices are unitary, without upsetting the ``accept/reject
balance'' until a final decision about
       the membership of the input in $ L_{NH} $ is reached. If the regular
expression check mentioned
       above fails, each path in question splits equiprobably to one
rejecting and one accepting configuration, and the overall probability
of
       acceptance of the machine
       turns out to be precisely $ \frac{1}{2} $.      If the input is indeed of
the form $ (aa^{*}b)(aa^{*}b)(aa^{*}b)^{*} $, whether
       the acceptance probability exceeds $ \frac{1}{2} $ or not
       depends on the following additional tasks performed by the
computation paths in order to test for the equality
       mentioned in the definition of $ L_{NH} $:
       \begin{enumerate}
               \item $ \mathsf{path_{1}} $ walks over the $ a $'s at the speed of
one tape square per step until
               reading the first $ b $. After that point, $ \mathsf{path_{1}} $
pauses for one step over each $ a $
               before moving on to the next symbol.
               \item $ \mathsf{path_{2}} $ pauses for one step over each $ a $
until reading the first $ b $.
               After that point, $ \mathsf{path_{2}} $ walks over each $ a $ at the
speed of one square per step.
               \item On each $ b $ except the first one, $ \mathsf{path_{1}} $ and
$ \mathsf{path_{2}} $ split to take the
               following two courses of action with equal probability:
               \begin{enumerate}
                       \item In the first alternative, $ \mathsf{path_{1}} $ and $
\mathsf{path_{2}} $
                       perform a 2-way quantum Fourier transform (QFT) \cite{KW97}:
                       \begin{enumerate}
                               \item The targets of the QFT are two new computational paths,
i.e., $ \mathsf{path_{accept}} $
                               and $ \mathsf{path_{reject}} $. Disregarding the equal-probability
transitions to
                               the twin halting states mentioned above, the QFT is realized as:
                               \begin{equation} \mathsf{path_{1}} \rightarrow \frac{1}{\sqrt{2}}
\mathsf{path_{accept}} + \frac{1}{\sqrt{2}}
                                       \mathsf{path_{reject}}
                                \end{equation}
                               \begin{equation} \mathsf{path_{2}} \rightarrow \frac{1}{\sqrt{2}}
\mathsf{path_{accept}} - \frac{1}{\sqrt{2}}
                                       \mathsf{path_{reject}}
                                \end{equation}
                               \item $ \mathsf{path_{accept}} $ and $ \mathsf{path_{reject}} $
continue computation at the
                               speed of  $ \mathsf{path_{2}} $, walking over the $b$'s without
performing the QFT any more.
                       \end{enumerate}
                       \item In the second alternative, $ \mathsf{path_{1}} $ and $
\mathsf{path_{2}} $
                       continue computation without performing the QFT.
               \end{enumerate}
               \item On symbol $ \dollar $,
                       $ \mathsf{path_{accept}} $ enters an accepting state,
                       $ \mathsf{path_{reject}} $ enters a rejecting state,
                       $ \mathsf{path_{1}} $ and $ \mathsf{path_{2}} $ enter accepting and
rejecting states with
                       equal probability.
       \end{enumerate}

       Suppose that the input is of the form
       \begin{equation} w=a^{x}ba^{y_{1}}ba^{y_{2}}b \cdots a^{y_{t}}b, \end{equation}
       where $ x,t,y_{1}, \cdots, y_{t}  \in \mathbb{Z}^{+} $.

       $ \mathsf{path_{1}} $ reaches the first $ b $ earlier than $
\mathsf{path_{2}} $.
       Once it has passed the first $ b $, $ \mathsf{path_{2}} $ becomes
faster, and may or may not catch up with
       $ \mathsf{path_{1}} $, depending on the number of $ a $'s in the
input after the first $ b $.
       The two paths can meet on the symbol following the $x$'th $a$ after
the first $ b $, since at that point
       $ \mathsf{path_{1}} $ has paused for the same number of steps
as $ \mathsf{path_{2}} $.
       Only if that symbol is a $ b $, the two paths perform a QFT in
the same place and at the same time.
       To paraphrase, if there exists a $ k $ $ (1 \le k \le t) $ such that
$ x=\sum_{i=1}^{k}y_{i}\ $,
       $ \mathsf{path_{1}} $ and $ \mathsf{path_{2}} $ meet over the $
(k+1)^{th} $ $ b $ and perform the QFT
       at the same step. If there is no such $ k $, the paths either never
meet, or meet over an $ a $ without a QFT.

       The $ \mathsf{path_{accept}} $ and $ \mathsf{path_{reject}} $s that
are offshoots of $ \mathsf{path_{1}} $
       continue their traversal of the string faster than $
\mathsf{path_{1}} $. On the other hand,
       the offshoots of $ \mathsf{path_{2}} $ continue their traversal at
the same speed as $ \mathsf{path_{2}} $.

       By definition, the twin halting states reached during the computation
contribute equal amounts to the
       acceptance and rejection probabilities. $ \mathsf{path_{1}} $ and $
\mathsf{path_{2}} $ accept and reject
       equiprobably when they reach the end of the string. If  $
\mathsf{path_{1}} $ and $ \mathsf{path_{2}} $ never
       perform the QFT at the same time and in the same position, every QFT
produces two equal-probability paths which
       perform identical tasks, except that one accepts and the other one
rejects at the end.

       The overall acceptance and rejection probabilities are equal, $
\frac{1}{2} $, unless a $ \mathsf{path_{reject}} $
       with positive amplitude and a $ \mathsf{path_{reject}} $ with
negative amplitude can meet and therefore cancel
       each other. In such a case, the surviving $ \mathsf{path_{accept}}
$'s contributes the additional acceptance
       probability that tips the balance. As described above, such a
cancellation is only possible when
       $ \mathsf{path_{1}} $ and $ \mathsf{path_{2}} $ perform the QFT together.

       Therefore, if $ w \in L_{NH} $,
       the overall acceptance probability is greater than $ \frac{1}{2} $.
If $ w \notin L_{NH} $,
       the overall acceptance probability equals $ \frac{1}{2} $.
\end{proof}

\begin{corollary}
	For any space bound $s$ satisfying $ s(n)=o(\log \log n) $, 
	\begin{equation}
		\mbox{PrSPACE}(s) \subsetneq \mbox{PrQSPACE}(s). 
	\end{equation}
\end{corollary}

\begin{openproblem}
	Is PrSPACE($ s $) $ \subsetneq $ PrQSPACE($ s $), for $ s \in \Omega(\log \log n) \cap o(\log n) $?
\end{openproblem}

\begin{corollary}
	For any space bound $s$ satisfying $ s(n)=o(\log n) $,
	\begin{equation}
		\mbox{one-way-PrSPACE}(s) \subsetneq \mbox{one-way-PrQSPACE}(s).
	\end{equation}
\end{corollary}

\begin{corollary}
	\label{corollary:coCSPACE-coCQSPACE}
	$ \mbox{coC}_{=}\mbox{SPACE}(1) \subsetneq \mbox{coC}_{=}\mbox{QSPACE}(1). $
\end{corollary}
\begin{proof}
	Since $ \mbox{coC}_{=}\mbox{SPACE}(1) $ is a proper subset of $ S $ \cite{Pa71}, 
	$ L_{NH} $ is not a member of $ \mbox{coC}_{=}\mbox{SPACE}(1) $.
	On the other hand, as shown in Theorem \ref{theorem:1KWQFA}, $ L_{NH} $
	is also a member of $ \mbox{one-way-coC}_{=}\mbox{QSPACE}(1) $.
\end{proof}

In the next section, we prove a fact which allows us to state
a similar inclusion relationship between the classes of languages
recognized by QTMs with restricted measurements and PTMs using constant space.

As noted before, Watrous proved the equality
PrQSPACE($s$)=PrSPACE($s$) ($ s \in \Omega(\log n)$) for the cases where
PrQSPACE is defined in terms of Wa98-QTMs \cite{Wa98,Wa99}, and Wa03-QTMs \cite{Wa03}. 
However, we do not know how to prove these results for our more general QTMs.

\begin{openproblem}
	Is PrQSPACE($ s $) $ \subseteq $ PrSPACE($ s $), for $ s \in \Omega(\log n) $?
\end{openproblem}

We continue with presenting some other nonstochastic languages\footnote{Note that, 
their complementary languages are also nonstochastic.} in one-way-PrQSPACE(1)
by extending the 1KWQFA algorithm for $ L_{NH} $ as follows:
\begin{itemize}
	\item On symbol $ \cent $, the computation splits into two paths, 
		$ \mathsf{path}_{1} $ and $ \mathsf{path}_{2} $,
		with equal amplitude. $ \mathsf{path}_{1} $ (resp., or $ \mathsf{path}_{2} $) immediately 
		moves to the right 
		and $ \mathsf{path}_{2} $ (resp., or $ \mathsf{path}_{1} $) 
		stays $ c $ steps on $ \cent $ before moving to the right.
	\item While scanning the input, $ \mathsf{path}_{1} $ (resp., $ \mathsf{path}_{2} $) stays 
		on each symbol $ c_{1} $ (resp., $ c_{2} $) step(s); and,	
		just before moving to the right, it produces a subpath, 
		say $ \mathsf{subpath}_{1} $ (resp., $ \mathsf{subpath}_{2} $),
		with some amplitude when reading $ \sigma_{1} \in \Sigma $ (resp.,  $ \sigma_{2} \in \Sigma $).
	\item $ \mathsf{subpath}_{1} $ (resp., $ \mathsf{subpath}_{2} $) travels the remaining part of the input 
		by staying $ c_{1}' $ (resp., $ c_{2}' $) step(s) on each symbol.
	\item When $ \mathsf{path}_{1} $ and $ \mathsf{path}_{2} $ read $ \dollar $,
		the input is equiprobably accepted and rejected by each of them.
	\item The subpaths, on the other hand, make the following 2-QFT
		\begin{equation}
			\begin{array}{lcl}
				\mathsf{subpath}_{1}: \ket{S_{1}} & \rightarrow & 
					\frac{1}{\sqrt{2}} \ket{A} + \frac{1}{\sqrt{2}} \ket{R} \\
				\mathsf{subpath}_{2}: \ket{S_{2}} & \rightarrow &
					\frac{1}{\sqrt{2}} \ket{A} - \frac{1}{\sqrt{2}} \ket{R}			
			\end{array},
		\end{equation}
	where $ \ket{S_{1}} $ (resp., $ \ket{S_{2}} $) is the configuration that 
	$ \mathsf{subpath}_{1} $ (resp., $ \mathsf{subpath}_{2} $) is in before reading $ \dollar $;
	$ \ket{A} $ (resp., $ \ket{R} $) is a specified accepting (resp., rejecting) configuration.
\end{itemize}

After reading $ \dollar $, the overall accepting probability exceeds $ \frac{1}{2} $ only if 
any pair of $ \mathsf{subpath}_{1} $ and $ \mathsf{subpath}_{2} $
reach to $ \dollar $ at the same time.
Otherwise, the overall accepting and rejecting probabilities become equal.

Let $ L \subseteq \Sigma $ be the language recognized by the above 
1KWQFA algorithm with one-sided cutpoint $ \frac{1}{2} $.
One can easily verifies that $ w \in \Sigma^{*} $ is a member of $ L $
if and only if 
there are two indexes, say $ i $ and $ j $ ($ 1 \le i,j \le |w| $), satisfying
$ w_{i} = \sigma_{1} $ and $ w_{j}=\sigma_{2} $  and
\begin{equation}
	c_{1}i+c_{1}'(|w|-i) = c_{2}j+c_{2}'(|w|-j)+c',
\end{equation}
where $ c' = - c $ if $ \mathsf{path}_{1} $ reads $ \cent $ at least twice and $ c' = c $ otherwise.
By simplifying the equation, we obtain
\begin{equation}
	(c_{1}-c_{1}')i = (c_{2}-c_{2}')j + (c_{2}'-c_{1}')|w| + c',
\end{equation}
or 
\begin{equation}
	d_{1}i = d_{2}j + d_{3}|w| + d,
\end{equation}
where $ d_{1} = c_{1} - c_{1}' $, $ d_{2}=c_{2} - c_{2}' $, $ d_{3} = c_{2}'-c_{1}' $, and $ d=c' $.
Note that, $ c_{1} $, $ c_{1}' $, $ c_{2} $, and $ c_{2}' $ are nonnegative integers,
but $ d_{1} $, $ d_{2} $, and $ d_{3} $ can take negative values.
By setting $ d $, $ d_{1}' $, $ d_{2}' $, $ d_{2} $, $ d_{1} $, 
$ \sigma_{1} $, and $ \sigma_{2} $ appropriately,
many different languages in one-way-PrQSPACE(1) are obtained.
The languages given in Figure \ref{uerr:nonstochastic-languages} are examples of such languages.
(Except the special ones, they are nonstochatic).
Another (nonstochastic) example is 
$ L_{center} = \{ w \in \{a,b\}^{*} \mid |w| = 2k-1, w_{k}=b, k > 0 \} $
due to the fact that it can be characterized 
by the equation $ 2i=|w|+1 $ and $ \sigma_{1} = \sigma_{2} = b $.

\begin{figure}[h!]
	\begin{center}
	\fbox{
	\small
	\begin{minipage}{0.9\textwidth}
		In \cite{FYS10A}, a joint work with R. Freivalds, 
		we presented a new family of languages:
		\begin{equation}
			\footnotesize
			L_{d_{1},\sigma_{1},d_{2},\sigma_{2},d,\Sigma} =
			\left\{w \in \Sigma^{*} \mid \exists i,j \left(w_{i}=\sigma_{1},w_{j}=\sigma_{2}, 
			d_{1}i=d_{2}(|w|+1-j)+d\right) \right\},
		\end{equation}
		where $ d,d_{1},d_{2} \in \mathbb{Z} $, $ \sigma_{1},\sigma_{2} \in \Sigma $.
		We can classify those languages as follows:
		\begin{itemize}
			\item They are in REG, if
			\begin{enumerate}
				\item $ \Sigma $ is unary,
				\item $ d_{1} $ and $ d_{2} $ have different signs, or
				\item $ \gcd(d_{1},d_{2}) $ does not divide $ d $.
			\end{enumerate}
			\item They are in S, if they are members of $ \{ L_{1,a,1,b,d,\{a,b\}} \mid d \in \mathbb{Z} \} $.
			\item They are not in uS, otherwise.
		\end{itemize}
		One of the simplest languages that are not in uS is $ L_{1,b,1,b,0,\{a,b\}} $, that is,
		\begin{equation}
			L_{say} = \{ w  \mid \exists u_{1},u_{2},v_{1},v_{2} \in \{a,b\}^{*},
			 w = u_{1}bu_{2} = v_{1}bv_{2}, |u_{1}| = |v_{2}| \},
		\end{equation}
		where the name comes from the surname of A. C. Cem Say, who considered this language
		for the first time.
	\end{minipage}
	}
	\end{center}
	\caption{A new family of nonstochastic languages}
	\vskip\baselineskip
	\label{uerr:nonstochastic-languages}
\end{figure}

\begin{theorem}
 The nonstochastic language\footnote{See page 88 on \cite{SS78}.}
	\begin{equation} L_{div}=\{ a^{n}ba^{kn} \mid k,n \in \mathbb{Z}^{+} \} 
	\end{equation} can be recognized by a 2QFA with unbounded error.
\end{theorem}
\begin{proof}
	We give a sketch of the algorithm:
	\begin{enumerate}
		\item If the input string is not of the form $ a^{+} b a^{+} $, 
			it is accepted with probability $ \frac{1}{2} $.
		\item Otherwise, the computation splits into two paths, 
			say $ \mathsf{path}_{1} $ and $ \mathsf{path}_{2} $, on the symbol $ b $.
		\item In an infinite loop, $ \mathsf{path}_{1} $ 
			travels to the left end-marker and comes back to the $b$ with the speed of one symbol per step.
			On the $ b $, it
			\begin{enumerate}
							\item performs the following transition with probability $ \frac{1}{2} $ 
				\begin{equation}
					\mathsf{path}_{1} \rightarrow \frac{1}{\sqrt{2}} (Accept) + \frac{1}{\sqrt{2}} (Reject),
				\end{equation}
				\item and goes on with probability $ \frac{1}{2} $.
			\end{enumerate}
		\item $ \mathsf{path}_{2} $ travels to the right end-marker and comes back to the $b$
			with the speed of one symbol per step, 
			on which it performs the following transition exactly:
			\begin{equation}
				\mathsf{path}_{2} \rightarrow \frac{1}{\sqrt{2}} (Accept) - \frac{1}{\sqrt{2}} (Reject).
			\end{equation}
	\end{enumerate}
	The two paths meet only if the input string is a member of $L_{div}$, in which case the acceptance 		
	probability would exceed $ \frac{1}{2} $ due to interference. 
	Otherwise, the acceptance probability becomes equal to $ \frac{1}{2} $.
\end{proof}

\begin{openproblem}
	Is $ L_{div} $ in one-way-PrQSPACE(1)?
\end{openproblem}

\begin{openproblem}
	Is one-way-PrQSPACE(1) $ \subsetneq $ PrQSPACE(1)?  
\end{openproblem}

\subsection{Languages Recognized by Kondacs-Watrous Quantum Finite Automata} \label{uerr:KWQFA-languages}

In this section, we settle an open problem of Brodsky and Pippenger
\cite{BP02}, giving a complete characterization of the class of
languages recognized with unbounded error by RT-KWQFAs. It turns out
that these restricted RT-QFAs, which are known to be inferior to RT-PFAs
in the bounded error case, are equivalent to them in the unbounded
error setting.

\begin{figure}[h!]
	\begin{center}
	\fbox{
	\begin{minipage}{0.9\textwidth}
		\footnotesize
		Let $ S $ be a finite set and $ \{ A_{s} \mid s \in S \} $ 
		be a set of $ m \times m $-dimensional matrices such that
		the norm of each column belonging to $ A_{s \in S} $ does not exceed 1.
		We present a method in order to find a set of $ m \times m $-dimensional matrices, 
		$ \{ B_{s} \mid s \in S\} $, with a generic constant $ l $ such that
		the columns of the matrix 
		\begin{equation}
			\frac{1}{l}			
			\left( \begin{array}{c}  A_{s} \\ \hline  B_{s}  \end{array} \right)
		\end{equation}
		form an orthonormal set
		for each $ s \in S $. The details of the method is given below.
		\begin{itemize}
			\item The entries of $ B_{s \in S} $ are set to 0.
			\item $ l $ is set to $ 2m+1 $.
			\item For each $ s \in S $, by executing the following loop, the entries of $ B_{s} $ 
				are updated to make the length of each column of 
				$ \left( \begin{array}{c} A_{s} \\ \hline B_{s} \end{array} \right) $ equal to $ l $,		
				and also to make the columns of 
				$ \left( \begin{array}{c} A_{s} \\ \hline B_{s} \end{array} \right) $
				pairwise orthogonal,
				where $ l $ must have been set to a value so that the loop works properly\footnote{
					\scriptsize
					The unique constraint for the loop to work properly is that 
					the value of $ l_{i} $, calculated at the (ii)$ ^{nd} $ step, must be at most $ l $.
					By setting $ l $ to $ 2m+1 $, the following bounds can be easily verified
					for each iteration of the loop:
					\begin{itemize}
						\item $ l_{i} < 2 $ at the (ii)$ ^{nd} $ step;
						\item $ 2m < |b_{i,i}| < 2m+1 $ at the (iii)$ ^{rd} $ step;
						\item $ |b_{j,i}| < \frac{1}{m} $ at the (v)$ ^{th} $ step.
					\end{itemize}
				}.\\			
				\begin{tabular}{ll}
					i.   & for $ i=1 $ to $ n+2 $ \\
					ii.  & ~~~set $ l_{i} $ to the current length of the $ i^{th} $ column \\
					iii. & ~~~set $ b_{i,i} $ to $ \sqrt{l^{2}-l_{i}^{2}} $ \\
					iv.  & ~~~for $ j=i+1 $ to $ n+2 $ \\
					v.   & ~~~~~~set $ b_{i,j} $ to some value so that $ i^{th} $ and $ j^{th} $ 
						 				columns can become orthogonal
				\end{tabular}
		\end{itemize}
	\end{minipage}
	}
	\end{center}
	\caption{General template to build a unitary matrix (I)}
	\vskip\baselineskip
	\label{uerr:fig:general-template-1}
\end{figure}
	
\begin{lemma}
       \label{uerr:lem:rt-kwqfa}
       Any language recognized with cutpoint (or nonstrict cutpoint) 
       $ \frac{1}{2} $ by a RT-PFA with $ n $ internal states 
       can be recognized with cutpoint (or nonstrict cutpoint) $ \frac{1}{2} $ 
       by a RT-KWQFA with $ O(n) $ internal states.
\end{lemma}
\begin{proof}
	Let $ L $ be a language recognized by a RT-PFA with $ n $ internal states
	\begin{equation} \mathcal{P} =(Q,\Sigma,\{A_{\sigma \in \tilde{\Sigma}} \},q_{1},Q_{a}) \end{equation}
	with (nonstrict) cutpoint $ \frac{1}{2} $.
	We construct a RT-KWQFA 
	\begin{equation} 
		\mathcal{M}=(R, \Sigma, \{U_{\sigma \in \tilde{\Sigma}} \}, r_{1},R_{a}, R_{r}) 
	\end{equation} 
	with $ (3n+6)$ internal states recognizing $ L $ with (nonstrict) cutpoint $ \frac{1}{2} $. 
	The idea is to ``embed'' the (not necessarily
	unitary) matrices $A_{\sigma}$ of the RT-PFA within the larger unitary matrices $U_{\sigma}$ of the RT-KWQFA.
	
	We define $ Q' $, $ v_{0}' $, and $ \{ A_{\sigma \in \tilde{\Sigma}}' \} $ as follows:
	\begin{enumerate}
		\item $ Q' = Q \cup \{ q_{n+1}, q_{n+2} \} $;
		\item $ v_{0}' = (1,0,\ldots,0)^{T} $  is an $ (n+2) $-dimensional column vector;
		\item Each $ A_{\sigma}' $ is a $ (n+2) \times (n+2) $-dimensional matrix:
			\begin{equation}			
             A_{\sigma \in \Sigma \cup \{\cent\}}'= \left(
				\begin{array}{c|c}
					A_{\sigma} &  0_{n \times 2}  \\
					\hline
					0_{2 \times n} & I_{2 \times 2} \\
                     \end{array}
			\right),
             A_{\dollar}'=\left(
				\begin{array}{c|c}
					0_{n \times n} & 0_{2 \times n}  \\
					\hline
					T_{2 \times n} & I_{2 \times 2} \\
				\end{array}
			\right)
			\left(
				\begin{array}{c|c}
					A_{\dollar} &  0_{n \times 2}  \\
					\hline
					0_{2 \times n} & I_{2 \times 2} \\
                     \end{array}
			\right),
			\end{equation}
     where $ T(1,i)=1 $ and $ T(2,i)=0 $ when $ q_{i} \in Q_{a} $, and
           $ T(1,i)=0 $ and $ T(2,i)=1 $ when $ q_{i} \notin Q_{a} $
           for $ 1 \le i \le n $.
     \end{enumerate}
     For a given input $ w \in \Sigma^{*} $, 
     \begin{equation}
		v_{| \tilde{w} |}' = A_{\dollar}' A_{w_{|w|}}' \cdots
			A_{w_{1}}' A_{\cent}' v_{0}.
	\end{equation}
     It can easily be verified that
	\begin{equation} 
		v_{| \tilde{w}|}' = (0_{1 \times n} \mid f_{\mathcal{P}}(w), 1 - f_{\mathcal{P}}(w))^{T}. 
	\end{equation}	
	
	For each $ \sigma \in \tilde{\Sigma} $, we obtain constant $ l $ and the set 
	$ \{B_{\sigma \in \tilde{\Sigma}} \} $ for the set $ \{A'_{\sigma \in \tilde{\Sigma}} \} $ 
	according to template described in Figure \ref{uerr:fig:general-template-1} 
	such that the columns of the matrix 
	\begin{equation}
		\frac{1}{l} 
		\left( \begin{array}{c} A'_{\sigma} \\ \hline B_{\sigma} \end{array} \right)
	\end{equation}	
	form an orthonormal set.
	Then, we obtain $ U_{\sigma \in \tilde{\Sigma}} $ as
	\begin{equation}
		U_{\sigma}=\left( 
			\begin{array}{c|c}				
				\begin{array}{c}
					A_{\sigma}''
					\\ \hline 
					B_{\sigma}'
					\\ \hline
					B_{\sigma}''
				\end{array}
				 &
				D_{\sigma}
			\end{array}
		\right),
	\end{equation}
	where $ A_{\sigma}'' = \frac{1}{l} A_{\sigma}' $, 
	$ B_{\sigma}' = B_{\sigma}'' = \frac{1}{\sqrt{2}l} B_{\sigma} $, and 
	the entries of $ D_{\sigma} $ are selected to make $ U_{\sigma} $ a unitary matrix.
	
	The state set $ R = R_{n} \cup R_{a} \cup R_{r}  $ is specified as:
	\begin{enumerate}
		\item $ r_{n+1} \in R_{a} $ corresponds to state $ q_{n+1} $;
		\item $ r_{n+2} \in R_{r} $ corresponds to state $ q_{n+2} $;
		\item $ \{r_{1},\ldots,r_{n}\} \in R_{n} $ correspond to the states of $ Q $, 
			where $ r_{1} $ is the start state;
		\item All the states defined for the rows of $ B_{\sigma}' $ 
			and $ B_{\sigma}'' $ are respectively accepting and rejecting states.
	\end{enumerate}
	
	$ \mathcal{M} $ simulates the computation of $ \mathcal{P} $
	for the input string $ w $ by multiplying the amplitude of each non-halting state with
	$ \frac{1}{l} $ in each step. 
	Hence, the top $ n+2 $ entries of the state vector of $ \mathcal{M} $ equal
	\begin{equation}	\left( \frac{1}{l} \right)^{|\tilde{w}|}
		\left(0_{1 \times n} , f_{\mathcal{P}}(w), 1- f_{\mathcal{P}}(w) \right)^{\trans}
	\end{equation}
	just before the last measurement on the right end-marker. 
	Note that, the halting states, except $ q_{n+1} $ and $ q_{n+2} $, come in accept/reject pairs, so
	that transitions to them during the computation add equal amounts to the overall 
	acceptance and rejection probabilities, and therefore not affect the decision on the membership of the
	input in $ L $. We conclude that 
	\begin{equation}
		f_{\mathcal{M}}(w) > \frac{1}{2} \mbox{ if and only if } f_{\mathcal{P}}(w) > \frac{1}{2},
	\end{equation}
	and
	\begin{equation}
		f_{\mathcal{M}}(w) \geq \frac{1}{2} \mbox{ if and only if } f_{\mathcal{P}}(w) \geq \frac{1}{2}.
	\end{equation}
\end{proof}

\begin{theorem}
       \label{theorem:SLUMM}
       The class of languages recognized by RT-KWQFAs with unbounded error is uS (uQAL).
\end{theorem}
\begin{proof}
       Follows from Lemma~\ref{uerr:lem:rt-kwqfa}, Lemma \ref{qtm:lem:RT-QFA-to-RT-GFA} and \cite{Tu69}.
\end{proof}

\begin{corollary}
	UMM = QAL $ \cap $ coQAL = S $ \cap $ coS.
\end{corollary}
\begin{proof}
	It is obvious that UMM $ \subseteq $ QAL $ \cap $ coQAL. Let $ L \in $ QAL $ \cap $ coQAL.
	Then, there exist two RT-KWQFAs $ \mathcal{M}_{1} $ and $ \mathcal{M}_{2} $ such that
	for all $ w \in L $, $ f_{\mathcal{M}_{1}} (w) > \frac{1}{2} $
	and $ f_{\mathcal{M}_{2}} (w) \geq \frac{1}{2} $ and
	for all $ w \notin L $, $ f_{\mathcal{M}_{1}} (w) \leq \frac{1}{2} $
	and $ f_{\mathcal{M}_{2}} (w) < \frac{1}{2} $.
	Let $ \mathcal{M}_{3} $ be a RT-KWQFA running $ \mathcal{M}_{1} $ and $ \mathcal{M}_{2} $
	with equal probability. Thus, we obtain that for all $ w \in L $,
	$ f_{\mathcal{M}_{3}} (w) > \frac{1}{2}$, and for all $ w \notin L $,
	$ f_{\mathcal{M}_{3}} (w) < \frac{1}{2} $. 
	Therefore, $ L \in  $ UMM.
\end{proof}

Considering this result together with Theorem \ref{theorem:1KWQFA}, we
conclude that, unlike classical deterministic and probabilistic finite automata,
allowing the tape head to ``stay put'' for
some steps during its left-to-right traversal of the input increases
the language recognition power of
quantum finite automata in the unbounded error case.

Since unbounded-error RT-PFAs and 2PFAs are equivalent in
computational power \cite{Ka89}, we are now able to state the
following corollary to Theorem \ref{theorem:1KWQFA}:
\begin{corollary}
	The class of languages recognized with unbounded error by
	constant-space PTMs is a proper subclass of the respective class for
	QTMs with restricted measurements.
\end{corollary}

Also note that, since the algorithm described in the proof of Theorem
\ref{theorem:1KWQFA} is
presented for a 1KWQFA, Corollary \ref{corollary:coCSPACE-coCQSPACE} is still valid when $
\mbox{coC}_{=}\mbox{QSPACE}(1) $
is defined for QTMs with restricted measurements.

\section{Nondeterministic Quantum Computation} \label{uerr:nondeterministic}

We begin with some basic facts. 

\subsection{Basic Facts} \label{uerr:non-basic-fact}

For a fixed $ \Sigma $, the pair $ (\mathcal{A},\lambda) $ is said to be 
\textit{equivalent under cutpoint separation to} the pair
$ (\mathcal{A'},\lambda') $, denoted as $ (\mathcal{A},\lambda) \equiv 
(\mathcal{A'},\lambda') $,      
if the equalities 
\begin{eqnarray}
	 \{w \in \Sigma^{*} \mid f_{\mathcal{A}}(w)<\lambda \} 
	 	& = & \{w \in \Sigma^{*} \mid f_{\mathcal{A'}}(w)<\lambda' \}, \\
	 \{w \in \Sigma^{*} \mid f_{\mathcal{A}}(w)=\lambda \} 
	 	& = & \{w \in \Sigma^{*} \mid f_{\mathcal{A'}}(w)=\lambda' \}, \mbox{ and} \\
	 \{w \in \Sigma^{*} \mid f_{\mathcal{A}}(w) > \lambda \} 
	 	& = & \{w \in \Sigma^{*} \mid f_{\mathcal{A'}}(w) > \lambda' \}
\end{eqnarray}
hold, where $ \mathcal{A} $, $ \mathcal{A'} $ are machines and 
$ \lambda, \lambda' \in \mathbb{R} $ are cutpoints.

\begin{fact}
	\label{fact:GFA-GFA}
	\cite{Tu69} Let $ \mathcal{G}_{1} $ be a GFA  and $\lambda_{1} \in \mathbb{R} $ be a cutpoint.
	For any cutpoint $ \lambda_{2} \in \mathbb{R} $, there exists a GFA $ \mathcal{G}_{2} $  such that 
	$ (\mathcal{G}_{1},\lambda_{1}) \equiv (\mathcal{G}_{2},\lambda_{2}) $.
\end{fact}

\begin{fact}
	\label{fact:PFA-PFA}
	\cite{Pa71} Let $ \mathcal{P}_{1} $ be a RT-PFA  and $ \lambda_{1} \in [0,1) $ be a cutpoint.
	For any cutpoint $ \lambda_{2} \in (0,1) $, there exists a RT-PFA $ \mathcal{P}_{2} $  such that
	$ (\mathcal{P}_{1},\lambda_{1}) \equiv (\mathcal{P}_{2},\lambda_{2}) $.
\end{fact}

\begin{fact}
	\label{fact:PFA-KWQFA} (see also Lemma \ref{uerr:lem:rt-kwqfa})
	\cite{YS09D,YS10C} For any RT-PFA $ \mathcal{P} $, there exists a RT-KWQFA $\mathcal{M} $ such that
	$ (\mathcal{P},\frac{1}{2}) \equiv (\mathcal{M},\frac{1}{2}) $.
\end{fact}

\begin{fact}
	\label{fact:KWQFA-GFA} (see also Lemma \ref{qtm:lem:RT-QFA-to-RT-GFA})
	\cite{LQ08,YS09D} For any RT-KWQFA $ \mathcal{M} $ and cutpoint $ \lambda \in [0,1) $, there exists a
	GFA $ \mathcal{G} $ such that $ (\mathcal{M},\lambda) \equiv (\mathcal{G},\lambda) $.
\end{fact}

\begin{fact}
	\label{fact:GPFA-PFA}
	\cite{Tu69} For any GFA $ \mathcal{G} $ and cutpoint $ \lambda_{1} \in \mathbb{R} $,
	there exist a RT-PFA $ \mathcal{P} $ and a cutpoint $ \lambda_{2} \in (0,1) $ such that
	$ (\mathcal{G},\lambda_{1}) \equiv (\mathcal{P},\lambda_{2}) $.
\end{fact}

\begin{fact}
     \label{fact:MCLpropersubsetS}
     \cite{BC01B} MCL $ \subsetneq $ S$ ^{>} $.
\end{fact}

\begin{fact}
	\label{fact:MCLpropersubsetQAL}
    Since MCQFA is a special case of RT-KWQFA, MCL $ \subseteq $ QAL and NMCL $ \subseteq $ NQAL.
\end{fact}

\begin{fact}
     \label{fact:RLpropersubsetSLneqSLeq}
     \cite{Pa71} REG is a proper subset of both S$ ^{\neq} $ and S$ ^{=} $.
\end{fact}

\begin{fact}
	\label{fact:SLneq-propersubsetSL}
	\cite{Pa71} S$ ^{\neq} $ $ \subsetneq $ S and
	S$ \setminus $S$ ^{=} \neq \emptyset $.
\end{fact}

\begin{fact}
	\label{fact:RLpropersubsetNQAL}
	\cite{BP02,NIHK02}
	REG is a proper subset of NQAL.
\end{fact}

\subsection{Languages Recognized with One-Sided Error} \label{uerr:non-languages}

RT-PFAs (and 1PFAs) can recognize all and only the regular languages with cutpoint $ 0 $ \cite{Ma93}.
RT-KWQFAs (and so 1QFAs) can do more than that, as be characterized in this section.

We start the presentation of our main result by stating a fact which
is useful in several proofs in this section.

\begin{lemma}
	\label{lemma:Sneq-GPFAonesidedcutpoint0}
	For any language $ L $, $ L \in $ S$ ^{\neq} $ if and only if there exists a GFA that recognizes $ L $
	with one-sided cutpoint $ 0 $.
\end{lemma}
\begin{proof}
	The forward direction is proven on page 171 of \cite{Pa71}. In the reverse direction, 
	if a GFA recognizes $ L $ with one-sided cutpoint $ 0 $, 
	then  $ L \in $ S$^{\neq} $ by Fact \ref{fact:GPFA-PFA}.
\end{proof}

\begin{lemma}
   \label{lemma:S-QL_0}
   S$ ^{\neq} \subseteq $ NQAL.
\end{lemma}
\begin{proof}
	If $ L \in S^{\neq}$, then there exists an $ n $-state  RT-PFA
	$ \mathcal{P} = (Q,\Sigma,\{A_{\sigma \in \tilde{\Sigma}} \},q_{1},Q_{a}) $ such that
	$ w \in L \leftrightarrow f_{\mathcal{P}}(w) \neq \frac{1}{2} $.
	We define $ Q' $, $ v_{0}' $, and 
	$ \{ A_{\sigma \in \Gamma}' \} $ as follows:
	\begin{enumerate}
		\item $ Q' = Q \cup \{ q_{n+1}, q_{n+2}, q_{n+3} \} $;
		\item $ v_{0}'=(1,0,\ldots,0) $ is a $ (n+3) \times 1 $-dimensional column vector;
		\item Each $ A_{\sigma}' $ is a $ (n+3) \times (n+3) $-dimensional matrix:
			\begin{equation*}
				A_{\cent}'= \left(
				\begin{array}{c|ccc}
					& 1 & \cdots & 1 \\
					\frac{1}{2}A_{\cent}[c_{1}] & 0 & \cdots & 0  \\
					& \vdots & \ddots & \vdots \\
					& 0 & \cdots & 0 \\
					\hline
					0 & 0 & \cdots & 0 \\
					0 & 0 & \cdots & 0 \\
					\frac{1}{2} & 0 & \cdots & 0 \\
				\end{array}
				\right), ~~~~
				A_{\sigma \in \Sigma}'= \left(
				\begin{array}{c|ccc}
				\\
					A_{\sigma} &  & 0_{n \times 3} &  \\\\
					\hline
					& 1 & 0 & 0 \\
					0_{3 \times n}   & 0 & 1 & 0 \\
					& 0 & 0 & 1 \\
				\end{array}
				\right),				
			\end{equation*}
			\begin{equation*}
				A_{\dollar}'=				
				\left(
				\begin{array}{ccc|ccr}
					&&&&& \\
					& 0_{n \times n} & & & 0_{n \times 3} \\
					&&&&& \\
					\hline
					t_{1,1} & \cdots & t_{n,1} & ~1 & 0 & \mbox{-}\frac{1}{2} \\
					t_{1,2} & \cdots & t_{n,2} & ~0 & 1 & \frac{1}{2} \\
					0 & \cdots & 0 & 0 & 0 & 0 \\
				\end{array}
				\right)
				\left(
				\begin{array}{c|ccc}
					\\
					A_{\dollar} &  & 0_{n \times 3} &  \\\\
					\hline
					& 1 & 0 & 0 \\
					0_{3 \times n}   & 0 & 1 & 0 \\
					& 0 & 0 & 1 \\
				\end{array}
				\right),
			\end{equation*}
		where $ A_{\cent}[c_{1}] $ is the first column of $ A_{\cent} $;
		$ t_{i,1}=1 $ and $ t_{i,2}=0 $ when $ q_{i} \in Q_{a} $, and
		$ t_{i,1}=0 $ and $ t_{i,2}=1 $ when $ q_{i} \notin Q_{a} $
		for $ 1 \le i \le n $.
	\end{enumerate}
	For a given input $ w \in \Sigma^{*} $, let	
	$ v_{|\tilde{w}|}' = A_{\dollar}' 	A_{w_{|w|}}' \cdots
	A_{w_{|1|}}' A_{\cent}'  v_{0}' $.
	It is easily verified that this computation ``imitates''
	the processing of $w$ by $\mathcal{P}$; the first $n$ entries of
	the manipulated vector $v'$  contain exactly the state vector of $\mathcal{P}$
	(multiplied by $ \frac{1}{2} $) in the corresponding steps of its execution. 
	The last matrix multiplication results in
	\begin{equation}
		v_{|\tilde{w}|}'=\left(0_{n \times 1} ,
		\frac{2f_{\mathcal{P}}(w)-1}{4},\frac{3-2f_{\mathcal{P}}(w)}{4},0 \right).
	\end{equation}
	The $ (n+1)^{st} $ entry of $ v_{|\tilde{w}|}' $ equals $ 0 $ if and only if $ w \notin L $.

	Using a modified version of the RT-PFA simulation method described in Section 
	\ref{uerr:KWQFA-languages},
	we can construct a RT-KWQFA $ \mathcal{M}=(R, \Sigma, \{U_{\sigma \in \Gamma} \},r_{1},
	R_{a}, R_{r} ) $ recognizing $ L $ with cutpoint 0.
	For each $ \sigma \in \tilde{\Sigma} $, $ U_{\sigma} $ is built as follows:
	\begin{equation}
		U_{\sigma}=\left( 
			\begin{array}{c|c}				
				\begin{array}{c}
					\frac{1}{l}A_{\sigma}'
					\\ \hline 
					\frac{1}{l}B_{\sigma}
				\end{array}
				 &
				D_{\sigma}
			\end{array}
		\right),
	\end{equation}
	where the constant $ l $ and the set $ \{B_{\sigma} \mid \sigma \in \tilde{\Sigma} \} $ 
	are obtained from $ \{A'_{\sigma} \mid \sigma \in \tilde{\Sigma} \} $ 
	according to the template described in Figure \ref{uerr:fig:general-template-1}.
	
	The state set $ R = R_{n} \cup R_{a} \cup R_{r}  $ is specified as:
	\begin{enumerate}
		\item $ r_{n+1} \in R_{a} $ corresponds to state $ q_{n+1} $;
		\item $ r_{n+2} \in R_{r} $ corresponds to state $ q_{n+2} $;
		\item $ \{r_{1},\ldots,r_{n},r_{n+3}\} \in R_{n} $ correspond to the remaining states of
			$ Q' $, where $ r_{1} $ is the start state;
		\item All the new states that are defined during the construction of
			$ \{U_{\sigma \in \Gamma}\} $ are rejecting ones.
	\end{enumerate}
	$ \mathcal{M} $ simulates the computation of $ \mathcal{P} $ for a given input string $ w \in \Sigma^{*} $ by
	representing the probability of each state $ q_{j}$ by the amplitude of the corresponding state $ r_{j}$.
	The transitions from the $ 2n+6 $ states added
	during the construction of $ U_{\sigma \in \Gamma } $ for ensuring unitarity do not interfere
	with this simulation, since the computation halts immediately on the
	``branches'' where these states are entered.
	Therefore, the top $ n+3$ entries of the state vector of $ \mathcal{M} $ equal
	\begin{equation}
		\left( \frac{1}{l} \right)^{|\tilde{w}|}
		\left(0_{n \times 1} , \frac{2f_{\mathcal{P}}(w)-1}{4},
		\frac{3-2f_{\mathcal{P}}(w)}{4},0 \right)^{\mathtt{T}}
	\end{equation}
just before the last measurement on $ \dollar $.
Since the amplitude of the only accepting state is nonzero if and only
if $ w \in L $, $ L $ is recognized by $ \mathcal{M} $ with cutpoint $ 0 $.
\end{proof}

\begin{lemma}
	NQAL$ \subseteq  $ S$ ^{\neq} $.
\end{lemma}
\begin{proof}
	By Fact \ref{fact:KWQFA-GFA}, there exists a GFA with one-sided cutpoint 0
	for any member $L$ of NQAL. By Lemma \ref{lemma:Sneq-GPFAonesidedcutpoint0},
	$L$ is an exclusive stochastic language.
\end{proof}

\begin{theorem}
   \label{theorem:maintheorem}
	The class of the languages recognized by RT-KWQFA 
	with one-sided unbounded error is precisely equal to S$ ^{\neq} $.
\end{theorem}

\begin{corollary}
	\label{corollary:SLneq-equal-NQAL}
	S$ ^{\neq} $ = NQAL.
\end{corollary}

\begin{corollary}
	\label{corollary:SLeq-equal-coNQAL}
	S$ ^{=} $ = coNQAL.
\end{corollary}

By Fact \ref{fact:SLneq-propersubsetSL}, there exist languages that 
RT-KWQFAs can recognize with two-sided (unbounded), but not one-sided error.
The class of these languages is precisely 
$ \mbox{uS}  \setminus ( \mbox{S} ^{\neq} \cup \mbox{S} ^{=} ) $.
Note that the above results also establish that the class of languages
recognized by RT-NQFAs is not closed under complementation (Fact \ref{fact:closure-properties-Seq-Sneq}).

\begin{openproblem}
      Does S$ ^{\neq} \cap $ S$ ^{=} $ contain a nonregular language?
\end{openproblem}

\begin{openproblem}
      Is S$ ^{\neq} $ countable or uncountable?
\end{openproblem}

Watrous \cite{Wa99} has shown that
NQSPACE$ _{\mathbb{Q}} $($ s $) = coC$ _{=} $SPACE$ _{\mathbb{Q}} $($ s $) for $ s = \Omega(\log(n)) $.
Due to the fact \cite{Ka89} that S$ ^{\neq} $ = coC$ _{\mbox{=}} $SPACE(1),
we have proven that $ \mbox{co-C} _{=} $SPACE(1) $
\subseteq $ NQSPACE(1),
and whether the inclusion is strict or not depends on
whether a two-way (or one-way) head would increase the computational power of a
NQFA\footnote{For any 2PFA $ \mathcal{M} $ and cutpoint $ \lambda_{1} \in [0,1) $,
there exist a one-way PFA $ \mathcal{P} $ and a cutpoint $ \lambda_{2} \in [0,1) $ such that
$ (\mathcal{M},\lambda_{1}) \equiv (\mathcal{P},\lambda_{2}) $ \cite{Ka89}, 
whereas one-way QFAs are more powerful than RT-QFAs in the general unbounded error setting.}.

\begin{openproblem}
	Can NQFAs with a two-way (or one-way) tape head recognize more languages than 
	the realtime model discussed here?
\end{openproblem}

For sublogarithmic space, we can conclude the following corollary, firstly stated in \cite{YS10A},
 by combining some previously known facts:

\begin{corollary}
	NSPACE($ s $) $ \subsetneq $ NQSPACE($ s $) for $ s=o(\log(n)) $.
\end{corollary}
\begin{proof} (Sketch.)
	QTMs can simulate PTMs easily for any common space bound.
	$ L_{neq}=\{w\in\{a,b\}^{*} \mid |w|_{a} \neq |w|_{b}\} \in $ S$ ^{\neq} $ \cite{BC01B}.
	It is easily seen that $ L_{neq} $ is a nonregular deterministic context-free language (DCFL).
	It is known that no nonregular DCFL is in NSPACE($ s $) for $ s=o(\log(n)) $ \cite{AGM92}.
\end{proof}

For space bounds $ s \in \Omega(\log(n)) $,
all we know in this regard is the trivial fact that NSPACE($ s $) $ \subseteq $ NQSPACE($ s $).

\begin{openproblem}
	Is NSPACE($ s $) $ \subsetneq $ NQSPACE($ s $), for $ s \in \Omega(\log n) $?
\end{openproblem}

\subsection{Space Efficiency of Nondeterministic Quantum Finite Automata} \label{uerr:non-space-efficiency}

It is well known \cite{AF98,MPP01} that some infinite families of languages can
be recognized with one-sided bounded error by just tuning the
transition amplitudes of a RT-QFA with a constant number of states,
whereas the sizes of the corresponding RT-PFAs grow without bound.
After a simple example, we argue that this advantage is also valid in the unbounded error case.

For $ m \in \mathbb{Z}^{+} $, $ L_{m} \subseteq \{a\}^{*} $ is defined as
\begin{equation}
L_{m}=\{a^{i} \mid i \mod(m) \neq 0 \}.
\end{equation}

\begin{theorem}
	\label{theorem:L_m}
	For any $ m>1 $, $ L_{m} $ can be recognized by a 2-state 
	MCQFA\footnote{There is an equivalent 4-state RT-KWQFA.} with cutpoint $ 0 $.
\end{theorem}
\begin{proof}
    $ \mathcal{M} $ begins the computation at state $ q_{0} $, and
each transition with the symbol $a$
    corresponds to a rotation\footnote{For details of a similar
construction for a nonregular language,
    see \cite{BC01B}.} by angle $ \frac{\pi}{m} $ in the $
\ket{q_{0}} $-$ \ket{q_{1}} $ plane,
    where $ q_{1} $ is the accepting state.
\end{proof}

For any positive $n$, it is known \cite{Ma93} that every $n$-state RT-PFA with cutpoint $ 0 $ has an equivalent
nondeterministic finite automaton with the same number of states.
Therefore, only finitely many distinct languages can be recognized
with one-sided unbounded error by RT-PFAs with at most $n$ states.

Combining this with the fact that any $n$-state RT-PFA with cutpoint $ 0 $ can be simulated by a
RT-KWQFA with $ 2n+6 $ states (see Section \ref{uerr:KWQFA-languages}),
the superiority of RT-KWQFAs (and so RT-QFAs) over RT-PFAs in this regard is established.

\subsection{Languages Recognized with Two-Sided Error} \label{uerr:two-sided}

To gain a better understanding of the classes of languages
recognizable by positive one-sided, negative one-sided, and
necessarily two-sided error by QFAs, we examine some examples from
each of those families. Bertoni and Carpentieri \cite{BC01B} showed
that $ L_{neq} $ is in NMCL, and that its complement, 
say, $ L_{eq} $,
is not in MCL. Now that we have Theorem \ref{theorem:maintheorem}, we
can use the well-known
results \cite{Pa71, Ma93} from the RT-PFA literature that state that
$ L_{eq} \in $  S$ ^{=} $, $ L_{neq} \in $ S$ ^{\neq} $, but not
vice versa, to conclude that stronger RT-QFA variants also can not recognize $ L_{eq} $ with positive
one-sided error, and neither can they recognize $ L_{neq} $ with
negative one-sided error. Similarly, L\={a}ce \textit{et al.} \cite{LSF09}
proved recently that the complement of the palindrome language $
L_{pal}=\{w \in \{a,b\}^{*} \mid w=w^{r} \} $ is in NQAL. We can show the corresponding
result for $ L_{pal} $ using the following fact:

\begin{fact} \cite{Di71}
	\label{fact:dieu}
	Let $ L \in $ S$ ^{=} $. Then there exists a natural number\footnote{This number can be chosen
	as the number of states of a PFA $ \mathcal{P} $ such that $ L=\mathbb{L}(\mathcal{P},=\lambda) $ for some 
	$ \lambda \in [0,1] $.} $ ~n \ge 1 $ such that for any strings $ u, v, y \in \Sigma^{*} $,
	\begin{equation}
		\mbox{if } uv, uyv, \ldots, uy^{n-1}v \in L, \mbox{then } uy^{*}v \subseteq L.
	\end{equation}
\end{fact}

\begin{theorem}
	\label{theorem:L_pal}
	$ L_{pal} \notin $  S$ ^{\neq} $.
\end{theorem}
\begin{proof}
	Suppose that $ L_{pal} \in $ S$ ^{\neq} $. Then $ \overline{L_{pal}} \in $ S$ ^{=} $.
	Let $ u=a^{n}b $, $ y=a $, and $ v=\varepsilon $.
	\begin{equation}
		a^{n}b, a^{n}ba,\ldots,a^{n}ba^{n-1} \in \overline{L_{pal}}
	\end{equation}
	imply that $ a^{n}ba^{n} \in $ $ \overline{L_{pal}} $ by Fact \ref{fact:dieu}.
    Since this string is actually a member of $ L_{pal} $,  we have a contradiction.
\end{proof}
We now exhibit some languages which can only be recognized by two-sided error by a QFA.

\begin{theorem}
	\label{theorem:L_twosided-1}
	$ L=\{aw_{1} \cup bw_{2} \mid w_{1} \in L_{eq}, w_{2} \in L_{neq} \}  
	\in $ S$ \setminus ($S$ ^{=} \cup $ S$ ^{\neq} )$.
\end{theorem}
\begin{proof}
    Suppose that $ L \in $ S$ ^{\neq} $, then there exists a GFA
    \begin{equation} \mathcal{G}=(S,\Sigma,\{A_{\sigma \in \{a,b\}}
\},v_{0},f) \end{equation}
    recognizing $ L$ with one-sided cutpoint $ 0 $. The GFA
    \begin{equation} \mathcal{G}'=(S,\Sigma,\{A_{\sigma \in \{a,b\}} \},A_{a}v_{0},f) \end{equation}
    recognizes $ L_{eq} $ with one-sided cutpoint 0, meaning that $
L_{eq} \in $ S$ ^{\neq} $.
    This contradicts the well-known fact mentioned in the first
paragraph of this section.
    Suppose now that $ L \in $ S$ ^{=} $,  then
    \begin{equation} \overline{L}=\{\varepsilon \cup aw_{2} \cup bw_{1} \mid w_{1}
\in L_{eq}, w_{2} \in L_{neq} \} \end{equation}
    is in S$ ^{\neq} $, which also results in a contradiction for
the same reason.
    Since both $ L_{eq} $ and its complement are stochastic, it is
not difficult to show that $ L $ is stochastic.
\end{proof}

\begin{lemma}
      \label{lemma:L_twosided-2}
      $ L_{lt}= \{w \in \{a,b\}^{*} \mid |w|_{a} < |w|_{b} \} $ $ \notin ($S$ ^{=} \cup $  S$ ^{\neq} )$.
\end{lemma}
\begin{proof}
      Suppose that $ L_{lt} \in $ S$ ^{=} $. Let $ u=\varepsilon  $, $ y=a
$, and $ v=b^{n} $.
      \begin{equation} b^{n}, ab^{n},\ldots,a^{n-1}b^{n} \in L_{lt} \end{equation}
      imply that $ a^{n}b^{n} \in $ $ L_{lt} $ by Fact \ref{fact:dieu}.
      Since this string is actually a member of $ \overline{L_{lt}} $,  we
have a contradiction.

      Similarly, suppose that $ L_{lt} \in $ S$ ^{\neq} $, or $
\overline{L_{lt}} \in $ S$ ^{=} $.
      Let $ u=a^{n} $, $ y=b $, and $ v=b $.
      \begin{equation} a^{n}b, a^{n}b^{2},\ldots,a^{n}b^{n} \in \overline{L_{lt}} \end{equation}
      imply that $ a^{n}b^{n+1} \in \overline{L_{lt}} $ by Fact
\ref{fact:dieu}.
      Since this string is actually a member of $ L_{lt} $,  we have a contradiction.
\end{proof}

\begin{corollary}
      \label{corollary:L_twosided-2}
      $ L_{lt} \in $ S$  \setminus ($S$ ^{=} \cup $  S$ ^{\neq} )$.
\end{corollary}
\begin{proof}
      This follows from Lemma \ref{lemma:L_twosided-2} and the fact that $ L_{lt} \in  $ S \cite{Ra92,Ka89}.
\end{proof}

\begin{theorem}
      \label{theorem:L_twosided-3}
      $ L_{eq \cdot b}= L_{eq} \cdot b^{+} $ $ \in $
      S$  \setminus ($S$ ^{=} \cup $  S$ ^{\neq} )$.
\end{theorem}
\begin{proof}
      The proof of $ L_{eq \cdot b} \notin ($S$ ^{=} \cup $  S$ ^{\neq} ) $ uses the
      setup presented in Lemma \ref{lemma:L_twosided-2}, i.e.,
      \begin{enumerate}
              \item select $ u=\varepsilon  $, $ y=a $, and $ v=b^{n} $ to 
              contradict with $ L_{eq \cdot b} \in $ S$ ^{=} $,
              \item select $ u=a^{n} $, $ y=b $, and $ v=b $ to contradict with $ L_{eq \cdot b} \in $ S$ ^{\neq} $.
      \end{enumerate}
      Any string $ w $ is a member of $ L_{eq \cdot b} $ if and only if it
      has the following three properties:
      \begin{itemize}
              \item $ w $ ends with $ b $.
              \item $ w \in L_{lt} $.
              \item Let $ u $ be the longest prefix of $ w $ ending with $ a $
              ($ u = \varepsilon $ if $ w \in \{b^{*}\} $). Then, $ u \in \overline{L_{lt}} $.
      \end{itemize}
      Since these properties can be checked easily by a 2PFA with bounded error,
      $ L_{eq \cdot b} \in $ S \cite{Ra92,Ka89}.
\end{proof}

Now, we show the stochasticity of an important family of languages.

The \textit{word problem} for a group is the problem of deciding whether or not a product of group elements 
is equal to the identity element \cite{LZ77}.
Let $ \mathcal{G}^{k}=(G,\circ) $ be a finitely generated free group with a basis
\begin{equation} 
	\Sigma =\{\sigma_{1},\ldots,\sigma_{k},\sigma_{1}^{-1},\ldots,\sigma_{k}^{-1}\},
\end{equation}
where $ k \in \mathbb{Z}^{+} $ is the rank of $ \mathcal{G}^{k} $.
$ L_{wp}(\mathcal{G}^{k}) \subseteq \Sigma^{*} $ is the language defined as
\begin{equation} L_{wp}(\mathcal{G}^{k})=\{ w \mid w_{i} \in \Sigma, 1 \le i \le |w|,
w_{1} \circ \cdots \circ w_{|w|} = \imath \}, \end{equation}
where $ \imath \in G $ is the identity element of $ \mathcal{G}^{k} $.

\begin{fact}
      \label{fact:isomorphic-free-group}
      (Page 1 of \cite{LS77})
      Let $ \mathcal{G}_{1}^{k_{1}} $ and $ \mathcal{G}_{2}^{k_{2}} $ be finitely generated free groups.
      Then $ \mathcal{G}_{1}^{k_{1}} $ and $ \mathcal{G}_{2}^{k_{2}} $ are 
      isomorphic if and only if $ k_{1} = k_{2} $.
\end{fact}

\begin{corollary}
      \label{corollary:isomorphic-free-group}
      $ L_{wp}(\mathcal{G}_{1}^{k_{1}}) $ and $ L_{wp}(\mathcal{G}_{2}^{k_{2}}) $ are isomorphic
      if and only if $ k_{1} = k_{2} $,
      where $ \mathcal{G}_{1}^{k_{1}} $ and $ \mathcal{G}_{2}^{k_{2}} $ are finitely generated free groups.
\end{corollary}

As a generic name, $ L_{wp}(k) $ can be used instead of $ L_{wp}(\mathcal{G}^{k}) $
due to Corollary \ref{corollary:isomorphic-free-group}, where $ k \in \mathbb{Z}^{+} $.

\begin{fact}
      \label{fact:word-problem-one-SL}
      \cite{Tu81}
      $ L_{wp}(1) \in $ S.
\end{fact}

\begin{fact}
      \label{fact:free-group-co-NMCL}
      \cite{BP02}
      $ L_{wp}(k) \in $ coNMCL, the class of languages whose complements are in NMCL,
       for any $ k \in \mathbb{Z}^{+} $.
\end{fact}

\begin{corollary}
      \label{corollary:free-group-SLeq}
      $ L_{wp}(k) \in $ S$ ^{=} $ for any $ k \in \mathbb{Z}^{+} $.
\end{corollary}

We now provide a proof of the following theorem.

\begin{theorem}
      \label{theorem:word-problem-is-stochastic}
      $ L_{wp}(k) \in $ S for any $ k \ge 2 $.
\end{theorem}

In fact, Theorem \ref{theorem:word-problem-is-stochastic} was stated
as a corollary on page 1463 of \cite{BP02}, but the purported proof there
was based on the claim that coNMCL$ \subseteq $MCL$ \subseteq $ S.
It is however known \cite{BC01B}, as we mentioned above, that a member
of coNMCL ($ L_{eq} $) lies outside MCL.
Furthermore, the same demonstration can be easily extended to $
L_{wp}(k) $, where $ k \in \mathbb{Z}^{+} $.

\begin{corollary}
      \label{corollary:word-problem-is-not-in-MCL}
      $ L_{wp}(k) \notin $ MCL for any $ k \in \mathbb{Z}^{+} $.
\end{corollary}

Since it is still an open problem whether S$ ^{=} \subseteq $ S or not,
we cannot use Corollary \ref{corollary:free-group-SLeq} directly to
prove Theorem \ref{theorem:word-problem-is-stochastic}.
Instead, we focus on a subclass of S$ ^{=} $ that is known to be 
a subset of S, S$ ^{=}_{\mathbb{Q}} $ \cite{Tu69B}.

\begin{fact}
      \label{fact:SLeq-rat-subset-of-SL}
      \cite{Tu69B}
      S$ ^{=}_{\mathbb{Q}} $ $ \subsetneq $ S.
\end{fact}

SO$ _{3}(\mathbb{Q}) $ is the group of rotations on $ \mathbb{R}^{3} $ that are
$ 3 \times 3 $ dimensional orthogonal matrices having only rational entries with determinant $ +1 $.

\begin{fact}
       \label{fact:SQ3-rational-free-group}
       \cite{Sw94,DD06} For any $ k \ge 2 $, SO$ _{3}(\mathbb{Q}) $
       contains a free subgroup with rank $ k $, namely $ \mathcal{S}^{k} $.
\end{fact}
\begin{proof}[Proof of Theorem \ref{theorem:word-problem-is-stochastic}]
	For $ L_{wp}(k) $, we define a rational GFA
	\begin{equation}
		\mathcal{G}_{k}=(\{s_{1},s_{2},s_{3}\}, \Sigma, \{A_{\sigma \in \Sigma}\},v_{0},f),
	\end{equation}
	where
	\begin{enumerate}
		\item $ \Sigma=\{R_{1},\ldots,R_{k},R_{1}^{-1},\ldots,R_{k}^{-1}\} $ is a basis of $ \mathcal{S}^{k} $;
		\item $ A_{\sigma}=\sigma $ for each $ \sigma \in \Sigma $;
		\item $ v_{0}=(1 ~~ 0 ~~ 0)^{\trans} $;
		\item $ f=(1 ~~ 0 ~~ 0) $.
	\end{enumerate}
	It is obvious that $ w \in L_{wp}(k) $ if and only if 
	$ A_{w} \cdots A_{1} = I_{3 \times 3} $ if and only if
	$ f_{\mathcal{G}_{k}}(w)= f A_{w} \cdots A_{1} v_{0} = 1 $,
	where $ w \in \Sigma^{*} $.
	Thus, $ L_{wp}(k) \in $ S$ ^{=}_{\mathbb{Q}} $ by selecting the cutpoint as $ 1 $.
	We can conclude with Fact \ref{fact:SLeq-rat-subset-of-SL}.
\end{proof}

\subsection{Closure Properties} \label{uerr:closure}

The previously discovered closure properties of S, S$ ^{\neq} $ and S$ ^{=} $ are listed below.

\begin{fact}
	\label{fact:closure-properties-SL}~
	\begin{enumerate}
		\item S is not closed under union and intersection \cite{Fl72,Fl74,La74}.
		\item S is closed under union and intersection with a regular language \cite{Bu68,Tu68}.
		\item S is closed under reversal \cite{Tu69}.
		\item S is not closed under concatenation, Kleene closure, and homomorphism \cite{Co69,Tu71}.
		\item S is closed under complementation over unary alphabets \cite{Fl74}.
	\end{enumerate}
\end{fact}

\begin{fact}
	\label{fact:closure-properties-Seq-Sneq}~
	\begin{enumerate}
		\item Both S$ ^{\neq} $ and S$ ^{=} $ are closed under union and intersection \cite{Pa71}.
		\item Neither S$ ^{\neq} $ nor S$ ^{=} $ is closed under complementation \cite{Di71}.
		\item S  is closed under intersection with a member of S$ ^{\neq} $ \cite{Pa71}.
	\end{enumerate}
\end{fact}

In \cite{YS10A}, we proved several new nontrivial closure properties of the 
``one-sided'' classes S$ ^{\neq} $ and S$ ^{=} $.
We refer the reader to \cite{YS10A} for the proofs of the following theorems.

\begin{theorem}
	\label{theorem:one-sided-closure-concatenation} S$ ^{\neq} $ is closed under concatenation.
\end{theorem}

\begin{theorem}
      \label{theorem:negative-one-sided-not-close-concatenation}
      S$ ^{=} $ is not closed under concatenation.
\end{theorem}
	
\begin{theorem}
     \label{theorem-one-sided-closure-star}
     S$ ^{\neq} $ is closed under Kleene closure.
\end{theorem}

\begin{theorem}
      \label{theorem:negative-one-sided-not-close-Kleene}
      S$ ^{=} $ is not closed under Kleene closure.
\end{theorem}

\begin{lemma}
     \label{lemma:homomorphism1}
     Let $ h: \Sigma \rightarrow \Sigma \setminus \{\kappa\} $ be a homomorphism such that
     \begin{equation}
             h(\sigma)=\left\lbrace
                     \begin{array}{ll}
                             \sigma ,&  \sigma \neq \kappa \\
                             \varepsilon, & \sigma=\kappa
                     \end{array}
             \right.,
     \end{equation}
     where $ \kappa $ is a specific symbol in $ \Sigma $.
     If $ L \subseteq \Sigma^{*} $ is in S$ ^{\neq} $, then so is $ h(L) $.
\end{lemma}

\begin{lemma}
     \label{lemma:homomorphism2}
     Let $ h : \Sigma \rightarrow \Upsilon^{*} $ be a homomorphism such
that $ |h(\sigma)|>0 $
     for all $ \sigma \in \Sigma $. If $ L \subseteq \Sigma^{*} $ is in S$
^{\neq} $, then so is $ h(L) $.
\end{lemma}

\begin{theorem}
     S$ ^{\neq} $ is closed under homomorphism.
\end{theorem}

\begin{theorem}
      \label{theorem:close-inverse-homomorphism}
     S$ ^{\neq} $ and S$ ^{=} $ are closed under inverse homomorphism.
\end{theorem}

\begin{theorem}
      \label{theorem:close-reversal}
    S$ ^{\neq} $ and S$ ^{=} $ are closed under reversal.
\end{theorem}

\begin{theorem}
      \label{theorem:close-word-quotient}
     S$ ^{\neq} $ and S$ ^{=} $  are closed under word quotient.
\end{theorem}

\begin{theorem}
      \label{theorem:close-difference}
     S$ ^{\neq} $ and S$ ^{=} $ are not closed under difference.
\end{theorem}

\begin{theorem}
      \label{theorem:close-difference-with-regular-language}
    S$ ^{\neq} $ and S$ ^{=} $ are closed under difference with a
regular language.
\end{theorem}

For completeness, we list below the following easy facts about MCL and NMCL:
\begin{enumerate}
      \item NMCL is closed under both union and intersection.
      \item Neither MCL nor NMCL is closed under
complementation \cite{BC01B}.
      \item Both MCL and NMCL are closed under inverse homomorphism
\cite{MC00}.
      \item Both MCL and NMCL are closed under word quotient
\cite{BP02}.
\end{enumerate}

Some related open problems are as follows:
\begin{openproblem}
      Is MCL closed under union? Intersection?
\end{openproblem}
\begin{openproblem}
      Do NMCL and MCL coincide?
\end{openproblem}
\begin{openproblem}
      Is S closed under complementation? (page 158 of \cite{Pa71})
\end{openproblem}
\begin{openproblem}
      Is S$ ^{=} $ a subset of S? (page 173 of \cite{Pa71})
\end{openproblem}

\chapter{CONSTANT-SPACE BOUNDED-ERROR COMPUTATION} \label{berr}

In this chapter, we focus on bounded-error computation in constant space.
The results presented in Sections 
\ref{berr:reset-probabilistic}, \ref{berr:reset-quantum}, \ref{berr:succinctness},
and \ref{berr:amplification} are obtained by our new two-way PFAs and QFAs 
(defined in Section \ref{berr:reset-definition}).
In the second part of the chapter (Section \ref{berr:postsel}),
we present our results related with realtime PFAs and QFAs with postselection.

\section{Probabilistic and Quantum Automata with Resetting} \label{berr:resetting-models}

In this section, we keep the original definitions of the quantum machines 
having capability of resetting as stated in \cite{YS10B}
(the underlying model is RT-KWQFA)
since the computational power of them does not increase 
(as shown in Theorem \ref{berr:thm:RT-GFA-restart-simulated-by-RT-QFA-restart}) 
when the underlying model is selected as RT-QFA.

\subsection{Definitions} \label{berr:reset-definition}
A \textit{two-way quantum finite automaton with reset} (2QFA$^{\curvearrowleft} $) is a 7-tuple
\begin{equation}
	\label{equation:2qfa-with-reset-tuple} 
	\mathcal{M}=(Q,\Sigma,\delta,q_{1},Q_{a},Q_{r},Q_{reset}=\cup_{ q \in Q_{n}} Q^{\curvearrowleft}_{q} ).
\end{equation}
In contrast to the previous models,
\begin{enumerate}
	\item $ Q_{n} = Q \setminus (Q_{a} \cup Q_{r} \cup Q_{reset}) $ is the set of nonhalting 
		and nonresetting states;
	\item $ Q_{reset} $ is the union of disjoint reset sets, i.e., each $ Q_{q \in Q_{n}}^{\curvearrowleft} $ 
		contains reset states that cause the computation to restart with state $ q $.
\end{enumerate}
We assume that the states in $ Q_{n} $ have smaller indices than
other members of $Q$; $ q_{i} \in Q_{n} $ for $ 1 \le i < | Q_{n} | $.

Apart from the left reset capability, 2QFA$ ^{\curvearrowleft} $s are identical to 2KWQFAs. 
(We refer the reader to Section \ref{qtm:restricted-QTMs} 
and \cite{KW97} for detailed coverage of the technical properties of 2KWQFAs.)
In each step of its execution, a 2QFA$ ^{\curvearrowleft} $ undergoes two linear
operations: The first one is a unitary transformation of the current
superposition according to $\delta$, and the second one is the measurement,
done on the configuration set, $ \mathcal{C}^{w} $, with set of projectors, $ P_{\tau \in \Delta} $,
for a given input string $ w \in \Sigma^{*} $, where 
\begin{itemize}
	\item $ \Delta=\{a,n,r\} \cup \{\mbox{reset-}i \mid 1 \le i \le |Q_{n}| \} $,
	\item $ \mathcal{C}^{w}_{\tau \in \{a,n,r\}} = \{ c \mid c \in Q_{\tau} \times \{1,\ldots,|\tilde{w}|\}  \} $,
	\item $ \mathcal{C}^{w}_{\scriptsize \mbox{reset-}i} = 
		\{ c \mid c \in Q_{q_{i} \in Q_{n}}^{\curvearrowleft} \times \{1,\ldots,|\tilde{w}| \}  \} $ 
		for $ 1 \le i \le |Q_{n}| $, and
	\item $ P_{\tau \in \Delta} = \sum\limits_{c \in \mathcal{C}_{\tau}^{w}} \ket{c} \bra{c} $.
\end{itemize}
Thus, the outcome of the measurement is one of ``accept'', ``reject'', ``continue without resetting'', 
or ``reset with state $ q $'', for any $ q \in Q_{n} $.
Note that, if ``reset with state $ q $'' is measured, the tape head is reset to point
to the left end-marker, and the machine continuous from the superposition $ \ket { q,1 }$ in the next step. 

A \textit{two-way quantum finite automaton with restart} (2QFA$ ^{\circlearrowleft} $) is a
restricted 2QFA$ ^{\curvearrowleft} $ in which the ``reset moves'' can
target only the original start state of
the machine, that is, in terms of Equation \ref{equation:2qfa-with-reset-tuple},
all the $ Q^{\curvearrowleft}_{q} $ of a 2QFA$ ^{\circlearrowleft} $ are
empty, with the exception of $ Q^{\curvearrowleft}_{q_{1}} $,
represented as $ Q^{\circlearrowleft} $.

Other variants of two-way automata with reset that are examined in this thesis are
\begin{enumerate}
	\item A \textit{realtime (Kondacs-Watrous) quantum finite automaton with reset} (RT-QFA$ ^{\curvearrowleft} $) 
		is a restricted 2QFA$ ^{\curvearrowleft} $ which uses neither ``move one square to the left'' 
		nor ``stay put'' transitions, and whose tape head is therefore classical,
	\item A \textit{realtime (Kondacs-Watrous) quantum finite automaton with restart} 
		(RT-QFA$ ^{\circlearrowleft} $) is a RT-QFA$ ^{\curvearrowleft} $ where the reset moves can target 
		only the original start state, and,
	\item A \textit{realtime probabilistic finite automaton with restart} (RT-PFA$ ^{\circlearrowleft} $) 
		is a RT-PFA which has been enhanced with the capability of resetting the tape head to the left 
		end-marker and swapping to the original start state.
\end{enumerate}

\subsection{Basic Facts} \label{berr:reset-basic-fact}

We start by stating some basic facts concerning automata with restart, which are used in later sections.

A segment of computation which begins with a (re)start, and ends with
a halting or restarting configuration is called a \textit{round}.
Clearly, every automaton with restart which makes nontrivial use of
its restarting capability runs for infinitely many rounds on some
input strings. Throughout this chapter, we make the assumption that our
two-way automata do not contain infinite loops within a round, that
is, the computation restarts or halts with probability 1
in a finite number steps for each round.

Everywhere in this section, $ \mathcal{R} $ stands for a finite
state automaton with restart, and $ w \in \Sigma^{*} $ represents an input string using the alphabet $ \Sigma $.
$ p_{\mathcal{R}}^{a}(w) $ (resp., $ p_{\mathcal{R}}^{r}(w) $) denote the probability that 
$ \mathcal{R} $ accepts (resp., rejects) $ w $, in the first round.
Moreover, $ p_{\mathcal{R}}^{h}(w) = p_{\mathcal{R}}^{a}(w) + p_{\mathcal{R}}^{r}(w) $.

\begin{lemma}
	\label{lemma:acc-rej}
	\begin{equation}
		\label{equation:Paccw-and-Prejw}
		f_{\mathcal{R}}^{a}(w)=\dfrac{1}{1+\dfrac{p_{\mathcal{R}}^{r}(w)}{p_{\mathcal{R}}^{a}(w)}};~~
		f_{\mathcal{R}}^{r}(w)=\dfrac{1}{1+\dfrac{p_{\mathcal{R}}^{a}(w)}{p_{\mathcal{R}}^{r}(w)}}.
	\end{equation}
\end{lemma}
\begin{proof}
	\begin{eqnarray*}
		f_{\mathcal{R}}^{a}(w) & = &
        \sum_{i=0}^{\infty}\left(1-p_{\mathcal{R}}^{a}(w)-p_{\mathcal{R}}^{r}(w) \right)^{i}
        p_{\mathcal{R}}^{a}(w)\\
		& = & p_{\mathcal{R}}^{a}(w) \left(
		\dfrac{1}{1-(1-p_{\mathcal{R}}^{a}(w)-p_{\mathcal{R}}^{r}(w))} \right) \\
		& = &
		\dfrac{p_{\mathcal{R}}^{a}(w)}{p_{\mathcal{R}}^{a}(w)+p_{\mathcal{R}}^{r}(w)} \\
		& = &
		\dfrac{1}{1+\frac{p_{\mathcal{R}}^{r}(w)}{p_{\mathcal{R}}^{a}(w)}}.
	\end{eqnarray*}
	$ f_{\mathcal{R}}^{r}(w) $ is calculated in the same way.
\end{proof}

\begin{lemma}
	\label{lemma:prej-over-pacc}
	The language $ L \subseteq \Sigma^{*} $ is recognized by $ \mathcal{R} $ with error bound $ \epsilon >0 $  
	if and only if $ \frac{p_{\mathcal{R}}^{r}(w)}{p_{\mathcal{R}}^{a}(w)} \le
	\frac{\epsilon}{1-\epsilon} $  when $ w \in L $, and $ \frac{p_{\mathcal{R}}^{a}(w)}{p_{\mathcal{R}}^{r}(w)} 
	\le \frac{\epsilon}{1-\epsilon} $ when $ w \notin L $.
	Furthermore, if $ \frac{p_{\mathcal{R}}^{r}(w)}{p_{\mathcal{R}}^{a}(w)} $ 
	(resp., $ \frac{p_{\mathcal{R}}^{a}(w)}{p_{\mathcal{R}}^{r}(w)}  $) is at most $ \epsilon $, then 
	$ f_{\mathcal{R}}^{a}(w) $ (resp, $ f_{\mathcal{R}}^{r}(w) $) is at least $ 1-\epsilon $.
\end{lemma}
\begin{proof}
	This follows from Lemma \ref{lemma:acc-rej}, since, for all $ p \ge 0 $, $ \epsilon \in [0,\frac{1}{2}) $,
	\begin{equation}
		\label{equation:1overp-epsilon}
		\frac{1}{1+p} \ge 1-\epsilon \Leftrightarrow p \le \frac{\epsilon}{1-\epsilon}, \mbox{ and }
	\end{equation}
	\begin{equation}
		\label{equation:1overp-epsilon-simplified}
		p \le \epsilon \Rightarrow \frac{1}{1+p} \ge 1-\epsilon.
	\end{equation}
\end{proof}

\begin{lemma}
	\label{lemma:expected-runtime}
	Let $ p=p_{\mathcal{R}}^{h}(w) $, and let $ s(w) $ be the maximum number of steps in any branch of a
	round of $ \mathcal{R} $ on $ w $.
	The worst-case expected runtime of $ \mathcal{R} $ on $ w $ is
	\begin{equation}
		\label{equation:expected-runtime}
		\frac{1}{p}(s(w)).	
	\end{equation}
\end{lemma}
\begin{proof}
	The worst-case expected running time of $ \mathcal{R} $ on $ w $ is
	\begin{equation}
		\label{equation:expected-runtime-proof}
		\sum_{i=0}^{\infty} (i+1)(1-p)^{i} (p)(s(w)) = (p)(s(w))\frac{1}{p^{2}}=\frac{1}{p}(s(w)).
	\end{equation}
\end{proof}

\begin{lemma}
	\label{lemma:restart-time}
	Any one-way automaton with restart with expected runtime 
	$ t $ can be simulated by a corresponding two-way automaton without restart 
	in expected time no more than $ 2t $.
\end{lemma}
\begin{proof}
The program of the two-way machine ($ \mathcal{R}_{2} $) is identical
to that of the one-way machine with restart ($ \mathcal{R}_{1} $),
except for the fact that each restart move of $ \mathcal{R}_{1} $ is
imitated by $ \mathcal{R}_{2} $ by moving the head one square per step
all the way to the left end-marker. This causes the runtimes of the
$i$ nonhalting rounds in the summation in Equation
(\ref{equation:expected-runtime-proof})
in Lemma \ref{lemma:expected-runtime} to increase by a factor of 2.
\end{proof}

We now give a quick review of the technique of probability
amplification. Suppose that we are given a machine (with or without reset) $
\mathcal{A} $,
which recognizes a language $L$ with error bounded by $\epsilon$, and
we wish to construct another machine which recognizes $L$ with a much
smaller, but still positive, probability of error, say, $ \epsilon'
$. It is well known\footnote{See, for instance, pages 369-370 of
\cite{Si06}.} that one can achieve this by running $ \mathcal{A} $  $
 O(\log(\frac{1}{\epsilon'})) $ times on the same input, and then
giving the majority answer as our verdict about the membership of the
input string in $L$.

Suppose that the original machine $ \mathcal{A} $ needs
to be run $ 2k+1 $ times for the overall procedure to work with the
desired correctness probability. Two counters can be used to count the
acceptance and rejection responses, and
the overall computation
accepts (resp., rejects) when the number of recorded acceptances (resp., rejections)
reaches $ k+1 $.
To implement these counters in the finite automaton setting, we need
to ``connect'' $ (k+1)^{2} $ copies of $ \mathcal{A} $, $ \{
\mathcal{A}_{i,j} \mid 0 \le i,j \le k \} $,
where the subscripts indicate the values of the two counters,
i.e., the states of $ \mathcal{A}_{i,j} $ encode the information that
$ \mathcal{A} $
has accepted $ i $ times and rejected $ j $ times in its previous
runs. The new machine $ \mathcal{M} $ is constructed
from the $ \mathcal{A}_{i,j} $'s as follows:
\begin{itemize}
 \item The start state of $ \mathcal{M} $ is the start state of $
\mathcal{A}_{0,0} $;
 \item Upon reaching any accept state of $ \mathcal{A}_{i,j} $ ($
0 \le i,j < k $), $ \mathcal{M} $ moves the head back to the left
end-marker and then switches to the start state of $
\mathcal{A}_{i+1,j} $;
    \item Upon reaching any reject states of $ \mathcal{A}_{i,j} $ ($ 0
\le i,j < k $), $ \mathcal{M} $ moves the head back to the left
end-marker and then switches to the start state of $
\mathcal{A}_{i,j+1} $;
 \item The accept states of $ \mathcal{M} $ are the accept states
of $ \mathcal{A}_{k,j} $ ($ 0 \le j < k $);
 \item The reject states of $ \mathcal{M} $ are the reject states
of $ \mathcal{A}_{i,k} $ ($ 0 \le i < k $).
\end{itemize}
\begin{lemma}
 \label{lemma:restart-to-reset}
 If language $ L \subseteq \Sigma^{*} $ is recognized by
$\mathcal{R} $ with a fixed error
 bound $ \epsilon > 0 $, then for any positive error bound $
\epsilon' < \epsilon$, there exists a finite automaton with reset,
 $ \mathcal{R}' $, recognizing $ L $.
 Moreover, if $ \mathcal{R} $ has $ n $ states and its (expected)
runtime is $ O(s(|w|)) $,
 then $ \mathcal{R}' $ has $
O(\log^{2}(\frac{1}{\epsilon'})n) $ states, and
 its (expected) runtime is $ O(\log(\frac{1}{\epsilon'})s(|w|)) $,
 where $ w $ is the input string.
\end{lemma}
\begin{proof}
 Follows easily from the above description.
\end{proof}

Finally, we note the following relationship between the
computational powers of the 2CQFA and the RT-QFA$^{\curvearrowleft}$.

\begin{lemma}
	\label{lemma:1qfareset-simulatedby-2qcfa}
	For any RT-QFA$ ^{\curvearrowleft} $ $ \mathcal{M}_{1} $  with $ n $ states and expected runtime
	$ t(|w|) $, there exists a 2CQFA $ \mathcal{M}_{2} $ with $ n $ states and 
	expected runtime $ O(t(|w|)) $, such that $ \mathcal{M}_{2} $ accepts every input string $w$ with the same 		
	probability that $ \mathcal{M}_{1} $ accepts $w$.
\end{lemma}

\subsection{Computational Powers of Realtime Probabilistic Finite Automata with Restart} \label{berr:reset-probabilistic}

It is interesting to examine the power of the restart move in
classical computation. Any RT-PFA$ ^{\circlearrowleft} $ which
runs in expected $ t $ steps
can be simulated by a 2PFA which runs in expected $ 2t $ steps (see
Lemma \ref{lemma:restart-time}).
We ask in this section whether the restart move can substitute the ``left'' and
``stationary'' moves of a 2PFA without loss of computational power.
Since every polynomial-time 2PFA recognizes a regular language, which
can of course be recognized by using only ``right'' moves, we focus on
the best-known example of a nonregular language that can be recognized
by an exponential-time 2PFA.
\begin{theorem}
	There exists a natural number $k$, such that for any $ \epsilon>0 $,
	there exists a $k$-state RT-PFA$ ^{\circlearrowleft} $ $ \mathcal{P}_{\epsilon} $
	recognizing language $ L_{eq} $ with error bound $ \epsilon $ and expected runtime 
	$ O( (\frac{2}{\epsilon^{2}})^{|w|}|w|) $, where $ w $ is the input string.
\end{theorem}
\begin{proof}
       We construct the RT-PFA$ ^{\circlearrowleft} $ $
\mathcal{P}_{\epsilon} $, shortly $ \mathcal{P} $, as follows:
       Let $ x = \frac{\epsilon^{2}}{2} $.
       The computation splits into three paths called $ \mathsf{path_{1}} $,
$ \mathsf{path_{2}} $, and
       $ \mathsf{path_{3}} $ with equal probabilities on symbol $ \cent $.
       All three paths, while performing their main tasks, parallelly check
whether the input is of the form
       $ a^{*}b^{*} $, if not, all paths simply reject.
       The main tasks of the  paths are as follows:
       \begin{list}{$ \bullet $}{}
               \item $ \mathsf{path_{1}} $ moves on with probability $ x $ and
restarts with probability $ 1-x $
               when reading symbols $ a $ and $ b $. After reading the right
end-marker $ \dollar $,
               it accepts with probability with $ 1 $.
               \item $ \mathsf{path_{2}} $ moves on with probability $ x^{2} $ and
restarts with probability $ 1-x^{2} $
               when reading symbol $ a $. On $ b $'s, it continuous with the ``syntax'' check.
               After reading the $ \dollar $, it rejects with probability $
\frac{\epsilon}{2} $ and
               restarts with probability $ 1-\frac{\epsilon}{2} $.
               \item $ \mathsf{path_{3}} $ is similar to $ \mathsf{path_{2}} $,
except that the transitions
               of symbols $ a $ and $ b $ are interchanged.
       \end{list}

       If the input is of the form $ a^{m}b^{n} $, then the accept and
reject probabilities
       of the first round are calculated as
   \begin{equation}
       p_{\mathcal{P}}^{a}(w)=\frac{1}{3} x^{m+n},~ \mbox{ and }~
           p_{\mathcal{P}}^{r}(w)= \frac{\epsilon}{6} \left(
x^{2m} + x^{2n} \right).
    \end{equation}

       If $ m = n $, then
       \begin{equation} \frac{p_{\mathcal{P}}^{r}(w)}{p_{\mathcal{P}}^{a}(w)}
= \epsilon. \end{equation}

       If $ m \neq n $ (assume without loss of generality that $ m = n+d $
for some $ d \in \mathbb{Z}^{+} $) ,
       then
    \begin{equation}
            \frac{p_{\mathcal{P}}^{a}(w)}{p_{\mathcal{P}}^{r}(w)}
=
            \frac{2}{\epsilon} \frac{x^{2n+d}}{x^{2n+2d}+x^{2n}} =
            \frac{2}{\epsilon} \frac{x^{d}}{x^{2d}+1} <
            \frac{2}{\epsilon} x^{d} \le \frac{2}{\epsilon} x
       \end{equation}
       By replacing $ x=\dfrac{\epsilon^{2}}{2} $, we can get
       \begin{equation} \frac{p_{\mathcal{P}}^{a}(w)}{p_{\mathcal{P}}^{r}(w)}
< \epsilon. \end{equation}

       By using Lemma \ref{lemma:prej-over-pacc}, we can conclude that $ \mathcal{P} $
       recognizes $ L_{eq} $ with error bound $ \epsilon $.

       Since $ p_{\mathcal{P}}^{h}(w) $ is always greater than $ \frac{1}{3} x^{|w|}  $,
       the expected runtime of the algorithm is $ O((\frac{2}{\epsilon^{2}})^{|w|}|w|) $,
       where $ w $ is the input string.
\end{proof}

\subsection{Computational Powers of Realtime Quantum Finite Automata with Restart} \label{berr:reset-quantum}

In this section, we focus on the RT-QFA$ ^{\circlearrowleft} $, which turns out to be the simplest and 
most restricted known model of 
quantum computation that is strictly superior in terms of bounded-error language recognition to its classical 
counterpart.

Our first result shows that RT-QFA$ ^{\circlearrowleft} $s can
simulate any RT-PFA$ ^{\circlearrowleft} $ with small state cost, albeit with great slowdown.
Note that no such relation is known between the 2KWQFA and its classical counterpart, the 2PFA, 
in the bounded error case.

\begin{theorem}
	\label{berr:thm:pfa-restart-simulated-by-qfa-restart}
	Any language $ L \subseteq \Sigma^{*} $ recognized by an $ n $-state RT-PFA$ ^{\circlearrowleft} $
	with error bound $ \epsilon $ can be recognized by a $ (2n+4) $-state RT-QFA$ ^{\circlearrowleft} $
	with the same error bound. 
	Moreover, if the expected runtime of the RT-PFA$ ^{\circlearrowleft} $ is $ O(s(|w|)) $,
	then the expected runtime of the RT-QFA$  ^{\circlearrowleft} $ is $ O(l^{2|w|}s^{2}(|w|)) $
	for a constant $ l>1 $ depending on $n$, where $ w $ is the input string.
\end{theorem}
\begin{proof}
	Let $ \mathcal{P} $ be an $ n $-state RT-PFA$ ^{\circlearrowleft} $ recognizing $ L $ 
	with error bound $ \epsilon $. We construct a $ (2n+4) $-state RT-QFA$ ^{\circlearrowleft} $ 
	$ \mathcal{M} $ recognizing the same language with error bound $ \epsilon' \le \epsilon $.

	By adding two more states, $ q_{a} $ and $ q_{r} $, to $ \mathcal{P} $, 
	we obtain a new RT-PFA$ ^{\circlearrowleft} $, $ \mathcal{P}' $, 
	where the halting of the computation in each round is postponed to the last symbol, $ \dollar $, 
	on which the overall accepting and rejecting probabilities are summed up into 
	$ q_{a} $ and $ q_{r} $, respectively.
	Therefore, for any given input string $ w \in \Sigma^{*} $, 
	the value of $ q_{a} $ and $ q_{r} $ are $ p_{\mathcal{P}}^{a}(w) $ and $ p_{\mathcal{P}}^{r}(w) $,
	respectively, at the end of the first round.

	By using the method described in Section \ref{uerr:KWQFA-languages}, 
	each stochastic matrix can be converted to a unitary one
	with twice the size, i.e. each transition matrix of $ \mathcal{P}' $ 
	can be converted to a $ (2n+4) \times (2n+4) $-dimensional
	unitary matrix. These are the transition matrices of $ \mathcal{M} $. 
	The state set of $ \mathcal{M} $ can be specified as follows:
	\begin{enumerate}
		\item The initial state of $ \mathcal{M} $ is the state corresponding 
			to the initial state of $ \mathcal{P} $;
		\item The states corresponding to $ q_{a} $ and $ q_{r} $ are
		the accepting and rejecting states, $ q_{a}' $ and $ q_{r}' $, respectively;
		\item the states corresponding to the non-halting and non-restarting states of 
		$ \mathcal{P}' $ are non-halting and non-restarting states, respectively; and,
		\item all remaining states are restarting states.
	\end{enumerate}	

	When $ \mathcal{M} $ runs on input string $ w $, the amplitudes of $ q_{a}' $ and $ q_{r}' $, 
	the only halting states of $ \mathcal{M} $, at the end of the first round are
	$ \left( \frac{1}{l} \right)^{|\tilde{w}|}p_{\mathcal{P}}^{a}(w)  $ and
	$ \left( \frac{1}{l} \right)^{|\tilde{w}|}p_{\mathcal{P}}^{r}(w)  $, respectively,
	where $ l $ is set to $ 2n+5 $ with respect to the template described in 
	Figure \ref{uerr:fig:general-template-1}.
	Therefore, when $ w \in L $,
	\begin{equation}
		\frac{p_{\mathcal{M}}^{r}(w)}{p_{\mathcal{M}}^{a}(w)} =
		\frac{(p_{\mathcal{P}}^{r}(w))^{2}}{(p_{\mathcal{P}}^{a}(w))^{2}}
		\leq
		\frac{\epsilon^{2}}{(1-\epsilon)^{2}},
	\end{equation}
	and similarly, when $ w \notin L $,
	\begin{equation}
		\frac{p_{\mathcal{M}}^{a}(w)}{p_{\mathcal{M}}^{r}(w)} =
		\frac{(p_{\mathcal{P}}^{a}(w))^{2}}{(p_{\mathcal{P}}^{r}(w))^{2}}
		\leq
		\frac{\epsilon^{2}}{(1-\epsilon)^{2}}.
	\end{equation}
	By solving the equation
	\begin{equation} \frac{\epsilon'}{1-\epsilon'} = \frac{\epsilon^{2}}{(1-\epsilon)^{2}}, \end{equation}
	we obtain
	\begin{equation} \epsilon'=\frac{\epsilon^{2}}{1 - 2\epsilon + 2\epsilon^{2}} \leq \epsilon. \end{equation}

	The expected runtime of $ \mathcal{P} $ is 
	\begin{equation} 
		\frac{| \tilde{w} |}{p_{\mathcal{P}}^{a}(w)+p_{\mathcal{P}}^{r}(w)} \in O(s(|w|)), 
	\end{equation}
	and so the expected runtime of $ \mathcal{M} $ is
	\begin{equation} 
		\left( l \right)^{2|\tilde{w}|} 
			\frac{| \tilde{w} |}{(p_{\mathcal{P}}^{a}(w))^{2}+(p_{\mathcal{P}}^{r}(w))^{2}} < 3 
		\left( l \right)^{2|\tilde{w}|} 
			\left( \frac{| \tilde{w} |}{p_{\mathcal{P}}^{a}(w)+p_{\mathcal{P}}^{r}(w)} \right)^{2}
		\in O(l^{2|w|}s^{2}(|w|)).
	\end{equation}
\end{proof}

\begin{corollary}
     RT-QFA$  ^{\circlearrowleft} $s can recognize all regular languages with zero error.
\end{corollary}

If the underlying QFA model of realtime QFA with restart is chosen as RT-QFA instead of RT-KWQFA,
we obtain the general RT-QFA$  ^{\circlearrowleft} $, denoted shortly as RT-GQFA$ ^{\circlearrowleft}$.
Thus, the accepting, rejecting, and restarting parts can easily be postponed to the end of the
computation, that is, only one observation is implemented after the whole input is read.
A \textit{general} realtime quantum finite automaton with restart is a 6-tuple
\begin{equation}
	\mathcal{M} = (Q,\Sigma,\{\mathcal{E}_{\sigma \in \tilde{\Sigma}}\},q_{1},Q_{a},Q_{r}),
\end{equation} 
where all specifications are the same as RT-QFA (see Section \ref{qtm:RT-QFAs}) except:
\begin{itemize}
	\item $ Q_{r} $ is the set of rejecting states;
	\item $ Q^{\circlearrowleft} = Q \setminus ( Q_{a} \cup Q_{r} ) $ is the set of restarting states;
	\item $ \Delta = \{a,r,\circlearrowleft\} $ with the following specifications:
		\begin{enumerate}
			\item ``a'': the computation is halted and the input is accepted,
			\item ``r'': the computation is halted and the input is rejected, and
			\item ``$ \circlearrowleft $'': the computation is restarted.
		\end{enumerate}
		The corresponding projectors, $ P_{a},~ P_{r}, \mbox{ and } P^{\circlearrowleft} $, 
		are defined in a standard way, based on the related set of states, 
		$ Q_{a},~ Q_{r},\mbox{ and }  Q^{\circlearrowleft} $, respectively.
\end{itemize}

\begin{theorem}
	\label{berr:thm:RT-GFA-restart-simulated-by-RT-QFA-restart}
	Any language $ L \subseteq \Sigma^{*} $ recognized by an $ n $-state RT-GQFA$ ^{\circlearrowleft} $
	with error bound $ \epsilon $ can be recognized by a $ O(n) $-state RT-QFA$ ^{\circlearrowleft} $
	with the same error bound. 
	Moreover, if the expected runtime of the RT-GQFA$ ^{\circlearrowleft} $ is $ O(s(|w|)) $,
	then the expected runtime of the RT-QFA$  ^{\circlearrowleft} $ is $ O(l^{2|w|}s^{2}(|w|)) $
	for a constant $ l>1 $, where $ w $ is the input string.
\end{theorem}
\begin{proof}
	We use almost the same idea presented in the proof of 
	Theorem \ref{berr:thm:pfa-restart-simulated-by-qfa-restart} after linearizing the computation of
	the given RT-GQFA$ ^{\circlearrowleft} $.
	Let $ \mathcal{G} = (Q,\Sigma,\{ \mathcal{E}_{\sigma \in \tilde{\Sigma}} \},q_{1},Q_{a},Q_{r}) $ 
	be an $ n $-state RT-GQFA$ ^{\circlearrowleft} $ recognizing $ L $ 	
	with error bound $ \epsilon $. We construct a $ 3n^{2}+6 $-state RT-QFA$ ^{\circlearrowleft} $ 
	$ \mathcal{M} $ recognizing the same language with error bound $ \epsilon' \le \epsilon $.
	
	In order to linearize $ \mathcal{G} $, we use the technique described in 
	Lemma \ref{qtm:lem:RT-QFA-to-RT-GFA} and so we obtain
	$ n^{2} \times n^{2} $-dimensional matrices for each $ \sigma \in \tilde{\Sigma} $, i.e.
	$ A_{\sigma} = \sum\limits_{i=1}^{|\mathcal{E}_{\sigma}|} E_{\sigma,i} \otimes E_{\sigma,i}^{*} $.
	By adding two more states, $ q_{n^{2}+1} $ and $ q_{n^{2}+2} $, the overall accepting and rejecting 
	probabilities are respectively summed up on them, i.e. 
	\begin{equation}
		A_{\sigma \in \Sigma \cup \{\cent\}}'= \left(
		\begin{array}{c|c}
			A_{\sigma} &  0_{n \times 2}  \\
			\hline
			0_{2 \times n} & I_{2 \times 2} \\
		\end{array}
		\right),
		A_{\dollar}'=\left(
			\begin{array}{c|c}
			0_{n \times n} & 0_{2 \times n}  \\
			\hline
			T_{2 \times n} & I_{2 \times 2} \\
		\end{array}
		\right)
		\left(
		\begin{array}{c|c}
			A_{\dollar} &  0_{n \times 2}  \\
			\hline
			0_{2 \times n} & I_{2 \times 2} \\
		\end{array}
		\right),
	\end{equation}
	where all the entries of $ T $ are zeros except that 
	$ T[1,(i-1)n^{2}+i] = 1  $ when $ q_{i} \in Q_{a} $ and
	$ T[2,(i-1)n^{2}+i] = 1  $ when $ q_{i} \in Q_{r} $.
	Let $ v_{0} = (1,0,\ldots,0) $ be a $ (n^{2}+2) \times 1 $-dimensional column vector.
	It can be easily verified that, for any $ w \in \Sigma^{*} $,
	\begin{equation}
		v_{|\tilde{w}|}'=A_{\dollar}' A_{w_{|w|}}' \cdots A_{w_{1}}' A_{\cent}' v_{0}
		= (0_{n^{2} \times 1},p_{\mathcal{G}}^{a}(w),p_{\mathcal{G}}^{r}(w))
	\end{equation}
	Based on the template given on Figure \ref{berr:fig:general-template-2}, we obtain constant $ l $ and
	the sets $ B_{\sigma \in \tilde{\Sigma}} $ and $ C_{\sigma \in \tilde{\Sigma}} $
	such that the columns of the following matrix form an orthonormal set,
	\begin{equation}
		\frac{1}{l} \left( \begin{array}{c} A'_{\sigma} \\ \hline B_{\sigma} 
		\\ \hline C_{\sigma} \end{array} \right).
	\end{equation}
	The remaining part is as in the proof of Theorem \ref{berr:thm:pfa-restart-simulated-by-qfa-restart}.
\end{proof}

	\begin{figure}[h!]
	\begin{center}
		\fbox{
		\begin{minipage}{0.9\textwidth}
			\footnotesize
			Let $ S $ be a finite set and $ \{ A_{s} \mid s \in S \} $ 
			be a set of $ m \times m $-dimensional matrices.
			We present a method in order to find two sets of $ m \times m $-dimensional matrices, 
			$ \{ B_{s} \mid s \in S\} $ and $ \{ C_{s} \mid s \in S\} $,  with a generic constant $ l $ such that
			the columns of the following matrix form an orthonormal set
			\begin{equation}
				\frac{1}{l} 
				\left( \begin{array}{c} A_{s} \\ \hline B_{s} \\ \hline C_{s} \end{array} \right),
			\end{equation}
			for each $ s \in S $. The details of the method is given below.
			\begin{enumerate}
				\item The entries of $ B_{s \in S} $ and $ C_{s \in S} $ are set to 0.
				\item For each $ s \in S $, the entries of $ B_{s} $ are updated to make the columns of
					$ \left( \begin{array}{c} A_{s} \\ \hline B_{s} \end{array} \right) $  pairwise orthogonal.
					Specifically, \\
					\begin{tabular}{ll}
						& for $ i=1, \ldots, m-1 $ \\
						& ~~~~set $ b_{i,i}=1 $ \\
						& ~~~~for $ j=i+1, \ldots, m $ \\
						& ~~~~~~~~set $ b_{i,j} $ to some value so that the $ i^{th} $ and $ j^{th} $ columns 
						become orthogonal \\
						& set $ l_{s} $ to the maximum of the lengths (norms) of the columns of
						$ \left( \begin{array}{c} A_{s} \\ \hline B_{s} \end{array} \right) $
					\end{tabular}
				\item  $ l = \max( \{ l_{s} \mid s \in S \} ) $.
				\item For each $ s \in S $, the diagonal entries of $ C_{s} $ are 
					updated to make the length of each column of
					$ \left( \begin{array}{c} A_{s} \\ \hline B_{s} \\ \hline C_{s} 
					\end{array} \right) $ equal to $ l $.
			\end{enumerate}
		\end{minipage}
		}
		\end{center}
		\caption{General template to build a unitary matrix (II)}
		\vskip\baselineskip
		\label{berr:fig:general-template-2}
	\end{figure}
	
For an automaton $ \mathcal{M} $ recognizing a language $ L $, we define
the \textit{gap function}, $ g_{\mathcal{M}} : N \rightarrow [0, 1] $,
such that $ g_{\mathcal{M}}(n)$ is the
difference between the minimum
acceptance probability of a member of $ L $ with length at most $ n $ and
the maximum acceptance probability of a non-member of $ L $ with length
at most $ n $\footnote{The definition of $ g_{\mathcal{M}} $ is
due to Bertoni and Carpentieri \cite{BC01A}, who call it the ``error function.''}.

\begin{lemma}
	\label{lemma:kwqfa-to-1qfarestart-exponential}
	If a language $ L $ is recognized by a RT-KWQFA $ \mathcal{M} $ with positive (negative)
	one-sided unbounded error such that $ g_{\mathcal{M}}(n) \geq c^{-n} $ for some $ c>1 $,
	then for all $ \epsilon \in (0,\frac{1}{2}) $, $ L $ is recognized by some 
	RT-QFA$ ^{\circlearrowleft} $ having three more states than $ \mathcal{M} $ with positive (negative) 
	one-sided error $ \epsilon $ in expected time $ O(\frac{1}{\epsilon}c^{|w|}|w|) $,
	where $ w $ is an input string.
\end{lemma}
\begin{proof}
	We consider the case of positive one-sided error. The adaptation to the other case is trivial.
	$ \mathcal{M} $ is converted into a RT-QFA$ ^{\circlearrowleft} $ $\mathcal{M}'_{\epsilon} $,
	shortly $ \mathcal{M}' $, as follows.
	$ \mathcal{M}' $ starts by branching to two equiprobable paths,
	$ \mathsf{path_{1}} $ and $ \mathsf{path_{2}} $, at the beginning of the computation.
	$ \mathsf{path_{1}} $ imitates the computation of $ \mathcal{M} $, 
	except that all reject states that appear in its subpaths are replaced
	by restart states. Regardless of the form of the input, $ \mathsf{path_{2}} $ moves right with amplitude $
	\frac{1}{\sqrt{c}} $, (and so restarts the computation with the remaining probability,) on every input symbol. 
	When it arrives at the right end-marker, $ \mathsf{path_{2}} $ rejects with amplitude $ \sqrt{\epsilon} $, and
	restarts the computation with amplitude $ \sqrt{1-\epsilon} $. 	

	When $ w \notin L $,
	\begin{equation}
		p_{\mathcal{M}'}^{a}(w) = 0,
		\mbox{ and }
		p_{\mathcal{M}'}^{r}(w) = \frac{\epsilon}{2c^{|w|}},
    \end{equation}
	and so the input is rejected with probability 1.
	When $ w \in L $,
	\begin{equation}
		p_{\mathcal{M}'}^{a}(w) \ge \frac{1}{2c^{|w|}},
		\mbox{ and }
		p_{\mathcal{M}'}^{r}(w) = \frac{\epsilon}{2c^{|w|}},
	\end{equation}
	and so the input is accepted with error bound $ \epsilon > 0 $ due to Lemma \ref{lemma:prej-over-pacc}, since
    \begin{equation} \frac{p_{\mathcal{M}'}^{r}(w)}{p_{\mathcal{M}'}^{a}(w)} \le \epsilon. \end{equation}
    Since $ p_{\mathcal{M}'}^{h}(w) $ is always greater than $ \frac{\epsilon}{2c^{|w|}} $,
    the expected runtime of $ \mathcal{M}'_{\epsilon} $ is $ O(\frac{1}{\epsilon}c^{|w|}|w|) $.
\end{proof}
\begin{lemma}
	\label{lemma:kwqfa-to-1qfarestart-constant}
	If a language $ L $ is recognized by a RT-KWQFA $ \mathcal{M} $ with positive (negative)
	one-sided bounded error such that $ g_{\mathcal{M}}(n) \geq c^{-1} $ for some $ c>1 $,
	then for all $ \epsilon \in (0,\frac{1}{2}) $, $ L $ is recognized by some 
	RT-QFA$ ^{\circlearrowleft} $ having three more states than $ \mathcal{M} $ with positive (negative) 
	one-sided error $ \epsilon $ in expected time $ O(\frac{1}{\epsilon}c|w|) $,
	where $ w $ is an input string.
\end{lemma}
\begin{proof}
	The construction is almost identical to that in Lemma \ref{lemma:kwqfa-to-1qfarestart-exponential}, 
	except that $ \mathsf{path_{2}} $ rejects with amplitude $ \sqrt{\epsilon} $, and
	restarts the computation with amplitude $ \sqrt{1-\epsilon} $ immediately on the left end-marker,
	 thereby causing every input to be rejected with the constant probability $ \frac{\epsilon}{2c} $.
	Hence, the expected runtime of 
	the RT-QFA$ ^{\circlearrowleft} $s turns out to be $ O(\frac{1}{\epsilon}c|w|) $.
\end{proof}
Lemma \ref{lemma:kwqfa-to-1qfarestart-exponential} is a useful step towards an eventual
characterization of the class of languages that are recognized with one-sided bounded error by 
RT-QFA$ ^{\circlearrowleft} $s, since full classical characterizations are known 
(see Section \ref{uerr:non-languages}) for the
classes of languages recognized by one-sided unbounded error by several RT-QFA models, including the RT-KWQFA.

\begin{theorem}
	\label{theorem:1qfa-restart-s=rat}
	For every language $L$ $ \in $ S$ ^{=}_{\mathbb{Q}} $, there exists a number $n$ such that for all error bounds 
	$ \epsilon > 0 $, there exist $n$-state RT-QFA$ ^{\circlearrowleft} $s that recognize  
	$L$ and $\overline{L} $ with one-sided error bounded by $ \epsilon $.
\end{theorem}
\begin{proof}
	For a language  $ L $ in S$ ^{=}_{\mathbb{Q}} $, let $ \mathcal{P} $ 
	be the rational RT-PFA (i.e. each transition probability must be a rational number) associated by $ L $.
	Turakainen \cite{Tu69B} showed that there exists a constant $ b > 1 $ such that 
	for any string $ w \notin L $,
	the probability that $ \mathcal{P} $ accepts $w$
	cannot be in the interval $ (\frac{1}{2}-b^{-|w|},\frac{1}{2}+b^{-|w|}) $.
	By using the method described in Section \ref{uerr:non-languages}, 
	we can convert $ \mathcal{P} $ to a RT-KWQFA $ \mathcal{M} $
	recognizing $ \overline{L} $ with one-sided unbounded error, so that $ \mathcal{M} $
	accepts any $ w \in \overline{L} $ with probability greater than $ c^{-|w|} $,  for a constant $ c > b $.
	We can conclude with Lemma \ref{lemma:kwqfa-to-1qfarestart-exponential}.
\end{proof}
S$ ^{=}_{\mathbb{Q}} $ contains many well-known languages, such as 
$ L_{eq} $, 
$ L_{pal} $, 
$ L_{twin} =\{ wcw \mid w \in \{a,b\}^{*} \} $,
$ L_{mult}=\{x \# y \# z \mid x,y,z \mbox{ are natural numbers in binary notation and } x \times y = z \} $, 
$ L_{square}= \{a^{n}b^{n^{2}} \mid n > 0 \} $,
$ L_{power} = \{ a^{n}b^{2^{n}} \mid n > 0 \} $, 
the word problem for finitely generated free groups, and
all \textit{polynomial languages},  \cite{Tu82} defined as
\begin{equation} \{a_{1}^{n_{1}} \cdots a_{k}^{n_{k}} b_{1}^{p_{1}(n_{1},\ldots,n_{k})} \cdots 
	b_{r}^{p_{r}(n_{1},\ldots,n_{k})} \mid p_{i}(n_{1},\ldots,n_{k}) \ge 0 \}, \end{equation}
where $ a_{1}, \ldots, a_{k},b_{1}, \ldots, b_{r} $ are distinct symbols, and 
each $ p_{i} $ is a polynomial with integer coefficients.
Note that Theorem \ref{theorem:1qfa-restart-s=rat} and 
Lemma \ref{lemma:1qfareset-simulatedby-2qcfa} answer a question posed by Ambainis and
Watrous \cite{AW02} about whether $ L_{square} $ and $ L_{power} $ can be recognized
with bounded error by 2CQFAs affirmatively.
\begin{corollary}
      The class of languages recognized by RT-QFA$  ^{\circlearrowleft} $s with bounded error
      properly contains the class of languages recognized by RT-PFA$ ^{\circlearrowleft} $s.
\end{corollary}
\begin{proof}
	This follows from Theorems \ref{berr:thm:pfa-restart-simulated-by-qfa-restart} and 
	\ref{theorem:1qfa-restart-s=rat}, Lemma
	\ref{lemma:restart-time}, and the fact \cite{DS92,FK94} that $ L_{pal} $ cannot be recognized with
	bounded error by 2PFAs.
\end{proof}
Since RT-QFAs \cite{Pa00,Hi08,YS10C} are known to be equivalent in language
recognition power to RT-PFAs, one has to consider a two-way model to
demonstrate the superiority of quantum computers over classical ones.
The 2CQFA is known \cite{AW02} to be superior to its
classical counterpart, the 2PFA, also by virtue of $ L_{pal} $.
Recall that, by Lemma \ref{lemma:1qfareset-simulatedby-2qcfa}, 2CQFAs can simulate
RT-QFA$  ^{\circlearrowleft} $s easily, and we do not know of a simulation in the other direction.

\section{Succinctness of Two-Way Models} \label{berr:succinctness}

In this section, we demonstrate several infinite families of regular
languages which can be recognized with some fixed probability greater
than $ \frac{1}{2} $ by
just tuning the transition amplitudes of a RT-QFA$ ^{\circlearrowleft} $
with a constant number of states, whereas the sizes of the
corresponding RT-QFAs, RT-PFAs, and 2NFAs grow without bound. One of
our constructions can be
adapted easily to
show that RT-PFA$ ^{\circlearrowleft} $s, (and, equivalently, 2PFAs),
also possess the same advantage over those machines.

\begin{definition} For an alphabet $ \Sigma $ containing symbols $ a $ and $ b $,
   and $ m \in \mathbb{Z}^{+} $, the family of languages $ A_{m} $
is defined as
   \begin{equation} A_{m}=\{ ua \mid u \in \Sigma^{*}, |u| \le m \}. \end{equation}
\end{definition}
Note that Ambainis et al. \cite{ANTV02} report that any Nayak
realtime quantum finite automaton\footnote{This is a realtime QFA model of
intermediate power, subsuming the RT-KWQFA, but strictly weaker
than RT-QFA in bounded error.}
that recognizes $ A_m $ with some fixed probability greater than $
\frac{1}{2} $ has $ 2^{\Omega(m)} $ states.

\begin{theorem}
	\label{theorem:A_m}
	$ A_{m} $  is recognized by a 6-state RT-QFA$ ^{\circlearrowleft} $ $ \mathcal{M}_{m,\epsilon} $
	for any error bound $ \epsilon >0  $. Moreover, the expected runtime of $ \mathcal{M}_{m,\epsilon} $ 
	on input $ w $ is $ O( \left( \frac{1}{\epsilon} \right)^{2m}|w|) $.	
\end{theorem}
\begin{proof}
	Let $ \mathcal{M}_{m,\epsilon} ~ (\mbox{shortly }\mathcal{M}) =\{Q,\Sigma,\delta,q_{0},Q_{a},Q_{r},Q^{\circlearrowleft}\} $ 
	be a RT-QFA$ ^{\circlearrowleft} $ with 
	$ Q_{n}=\{ q_{0}, q_{1} \} $,
	$ Q_{a}=\{A\} $, $ Q_{r}=\{R\} $, $ Q^{\circlearrowleft}=\{I_{1},I_{2}\} $.
	$ \mathcal{M}$ contains the transitions
	\begin{eqnarray*}
		U_{\cent} \ket{q_{0}} & = & \epsilon \ket{q_{1}} +  \epsilon^{\frac{2m+5}{2}} \ket{R} +
			\sqrt{1-\epsilon^{2}-\epsilon^{2m+5}} \ket{I_{1}} \\
		U_{a} \ket{q_{0}} & = & \epsilon \ket{q_{0}} + \sqrt{\frac{1}{2}-\epsilon^{2}} \ket{I_{1}} 
			+ \frac{1}{\sqrt{2}} \ket{I_{2}} \\
		U_{a} \ket{q_{1}} & = & \epsilon \ket{q_{0}} + \sqrt{\frac{1}{2}-\epsilon^{2}} \ket{I_{1}} 
			- \frac{1}{\sqrt{2}} \ket{I_{2}} \\
		U_{\Sigma \setminus \{a\}} \ket{q_{0}} & = & \epsilon \ket{q_{1}} 
			+ \sqrt{\frac{1}{2}-\epsilon^{2}} \ket{I_{1}} + \frac{1}{\sqrt{2}} \ket{I_{2}} \\
		U_{\Sigma \setminus \{a\}} \ket{q_{1}} & = & \epsilon \ket{q_{1}}
			+ \sqrt{\frac{1}{2}-\epsilon^{2}} \ket{I_{1}}- \frac{1}{\sqrt{2}} \ket{I_{2}} \\
        U_{\dollar} \ket{q_{0}} & = & \ket{A} \\
		U_{\dollar} \ket{q_{1}} & = & \ket{R} \\
	\end{eqnarray*}
	and the transitions not mentioned above can be completed easily, by
	extending each $ U_{\sigma}$ to be unitary.

	On the left end-marker, $\mathcal{M}$ rejects with probability $ \epsilon^{2m+5} $, goes on to 
	scan the input string with amplitude $ \epsilon $, and restarts immediately with the remaining probability.
	States $q_{0}$  and $q_{1}$ implement the check for the regular
	expression $ \Sigma^{*}a $, but the machine restarts with probability $ 1 - \epsilon^{2} $
	on all input symbols during this check.

	If $ w=u \sigma' $ for $ u \in \Sigma^{*} $, and $ \sigma' \neq a $, the input is rejected with 
	probability $1$, since $ p_{\mathcal{M}'}^{a}(w)=0 $.

	If $ w=ua $ for $ u \in \Sigma^{*} $,
	\begin{equation}
		p_{\mathcal{M}}^{a}(w)= \epsilon^{2|w|+2}, ~~
		p_{\mathcal{M}}^{r}(w)=\epsilon^{2m+5}.
	\end{equation}
	Hence, if $ w \in A_{m} $,
	\begin{equation} p_{\mathcal{M}}^{a}(w) \ge \epsilon^{2m+4}, \end{equation}
	and if $ w \notin A_{m} $,
	\begin{equation} p_{\mathcal{M}}^{a}(w) \le \epsilon^{2m+6}. \end{equation}
	In both cases, the corresponding ratio 
	$ \frac{p_{\mathcal{M}}^{r}(w)}{p_{\mathcal{M}}^{a}(w)} $ or 
	$ \frac{p_{\mathcal{M}}^{a}(w)}{p_{\mathcal{M}}^{r}(w)}$ is not greater than
	$ \epsilon $. 
	Thus, by Lemma \ref{lemma:prej-over-pacc}, 
	we conclude that $ \mathcal{M} $ recognizes $A_{m} $ with error bounded by $ \epsilon $.
	Since $ p_{\mathcal{M}}^{h}(w) $ is always greater than $ \epsilon^{2m+5} $,
	the expected runtime of $ \mathcal{M} $ is $ O( \left( \frac{1}{\epsilon} \right)^{2m}|w|) $.
\end{proof}

By a theorem of Rabin \cite{Ra63}, for any fixed error bound, if a
language $L$ is recognized with bounded error by
a RT-PFA with $n$ states, then there exists a RT-DFA that recognizes $L$ with
$ 2^{O(n)} $ states.
Parallelly, Freivalds et al. \cite{FOM09} note that one-way quantum
finite automata with mixed states are
no more than superexponentially more concise than RT-DFAs.
These facts can be used to conclude that a collection of RT-PFAs (or
RT-QFAs) with a fixed common number of
states that recognize an infinite family of languages with a fixed common
error bound less than $ \frac{1}{2} $, \textit{à la} the two-way
quantum automata of Theorem
\ref{theorem:A_m}, cannot exist, since that would imply the existence
of a similar family of RT-DFAs of fixed size.
By the same reasoning, the existence of such families of 2NFAs can
also be overruled.

The reader should note that there exists a bounded-error RT-PFA$ ^{\circlearrowleft} $
(and therefore, a 2PFA\footnote{See Section \ref{berr:reset-probabilistic}
for an examination of the relationship between the
computational powers of the RT-PFA$ ^{\circlearrowleft} $
and the 2PFA.},) for $ A_{m} $, which one can obtain simply by replacing
each transition amplitude of 1QFA$
^{\circlearrowleft} $ $ \mathcal{M}_{m,\epsilon} $
defined in Theorem \ref{theorem:A_m} by the square of its modulus.
This establishes the fact that 2PFAs also possess the succinctness
advantage discussed above over RT-PFAs, RT-QFAs and RT-NFAs.

We proceed to present two more examples.
\begin{definition}
	\label{definition:B_m}
	For $ m \in \mathbb{Z}^{+} $, the language family $ B_{m} \subseteq \{a\}^{*} $ is defined as
	\begin{equation} B_{m}=\{a^{i} \mid i \mod(m) \equiv 0 \}. \end{equation}
\end{definition}
\begin{theorem}
	\label{theorem:B_m}
	For any error bound $ \epsilon > 0 $, there exists a
	7-state RT-QFA$  ^{\circlearrowleft} $ $ \mathcal{M}_{m,\epsilon} $ which
	accepts any $ w \in B_{m} $ with certainty,
	and rejects any $ w \notin B_{m} $ with probability at least $ 1-\epsilon $.
	Moreover, the expected runtime of $ \mathcal{M}_{m,\epsilon} $ on $ w $ is
	$ O \left(\frac{1}{\epsilon}\sin^{-2}(\frac{\pi}{m})|w| \right) $.
\end{theorem}
\begin{proof}
	We construct a $ 4 $-state RT-KWQFA recognizing $ \overline{B_{m}} $ 
	with positive one-sided bounded error, as described in \cite{AF98}.
	Let $ \mathcal{M}_{m} = (Q,\Sigma,\delta,q_{0},Q_{a},Q_{r}) $ be RT-KWQFA with
	$ Q_{n}=\{q_{0},q_{1}\} $, $ Q_{a}=\{A\} $, and $ Q_{r}=\{R\} $.
	$ \mathcal{M}_{m} $ contains the transitions
	\begin{eqnarray*}
		U_{\cent} \ket{q_{0}} & = & \ket{q_{0}} \\
		U_{a} \ket{q_{0}} & = & \cos(\frac{\pi}{m})\ket{q_{0}}+\sin(\frac{\pi}{m})\ket{q_{1}} \\
		U_{a} \ket{q_{1}} & = & -\sin(\frac{\pi}{m})\ket{q_{0}}+\cos(\frac{\pi}{m})\ket{q_{1}} \\
		U_{\dollar}\ket{q_{0}} & = & \ket{R} \\
		U_{\dollar}\ket{q_{1}} & = & \ket{A}
	\end{eqnarray*}
	and the transition amplitudes not listed above are filled in to satisfy unitarity.
	$ \mathcal{M}_{m} $ begins computation at the
	$ \ket{q_{0}} $-axis, and performs a rotation by angle $ \frac{\pi}{m} $ 
	in the $ \ket{q_{0}} $-$ \ket{q_{1}} $ plane for each $ a $ it reads.
	Therefore, the value of the gap function, $ g_{\mathcal{M}_{m}} $, is not less than $ \sin^{2}(\frac{\pi}{m}) $
	for $ |w|>0 $.
	By Lemma \ref{lemma:kwqfa-to-1qfarestart-constant}, there exists a $ 7 $-state 
	RT-QFA$  ^{\circlearrowleft} $ $ \mathcal{M}_{m,\epsilon} $ recognizing $ \overline{B_{m}} $
	with positive one-sided bounded error and whose expected runtime is 
	$ O \left(\frac{1}{\epsilon}\sin^{-2}(\frac{\pi}{m})|w| \right) $.
	By swapping the accepting and rejecting states of $ \mathcal{M}_{m,\epsilon} $,
	we can get the desired machine.
\end{proof}

\begin{definition} 
	For an alphabet $ \Sigma $, and $ m \in \mathbb{Z}^{+} $, 
	the language family $ C_{m} $ is defined as
	 \begin{equation} C_{m}=\{ w \in \Sigma^{*} \mid |w| = m \}. \end{equation}
\end{definition}
\begin{theorem}
	\label{theorem:C_m}
	For any error bound $ \epsilon > 0 $, there exists a 7-state
	RT-QFA$ ^{\circlearrowleft} $ $ \mathcal{M}_{m,\epsilon} $ which accepts any $ w \in C_{m} $ with certainty, and
	rejects any $ w \notin C_{m} $ with probability at least $ 1-\epsilon $.
	Moreover, the expected runtime of $ \mathcal{M}_{m,\epsilon} $ on $ w $ is
	$ O(\frac{1}{\epsilon}2^{m}|w|) $.
\end{theorem}
\begin{proof}
	We construct a $ 4 $-state RT-KWQFA recognizing $ \overline{C_{m}} $ with 
	positive one-sided unbounded error.
	Let $ \mathcal{M}_{m} = (Q,\Sigma,\delta,q_{0},Q_{a},Q_{r}) $ be RT-KWQFA with
	$ Q_{n}=\{q_{0},q_{1}\} $, $ Q_{a}=\{A\} $, and $ Q_{r}=\{R\} $.
	$ \mathcal{M}_{m} $ contains the transitions
	\begin{eqnarray*}
		U_{\cent}\ket{q_{0}} & = &  \frac{1}{\sqrt{2}} \ket{q_{0}} + 
			\left( \frac{1}{\sqrt{2}} \right)^{m+1} \ket{q_{1}} + 
			\sqrt{\frac{1}{2}-\left( \frac{1}{2} \right)^{m+1} } \ket{R} \\
		U_{\sigma \in \Sigma} \ket{q_{0}} & = & \frac{1}{\sqrt{2}} \ket{q_{0}}
			+ \frac{1}{\sqrt{2}} \ket{R} \\
		U_{\sigma \in \Sigma} \ket{q_{1}} & = & \ket{q_{1}} \\
		U_{\dollar}\ket{q_{0}} & = & \frac{1}{\sqrt{2}} \ket{A} + \frac{1}{\sqrt{2}} \ket{R} \\
		U_{\dollar}\ket{q_{1}} & = & -\frac{1}{\sqrt{2}} \ket{A} + \frac{1}{\sqrt{2}} \ket{R} \\
	\end{eqnarray*}
	with the amplitudes of the transitions not mentioned above filled in to ensure unitarity. 

	$ \mathcal{M}_{m}$ encodes the length of the input string in the amplitude of state $ q_{0} $, 
	which equals $ \left( \frac{1}{\sqrt{2}} \right)^{|w|+1} $ just before the processing of the right end-marker.
	The desired length $m$ is ``hardwired'' into the amplitudes of $ q_{1} $.
	For a given input string $ w \in \Sigma^{*} $, if $ w \in C_{m} $, 
	then the amplitudes of states $ q_{0} $ and $ q_{1} $ are equal, and
	the QFT \cite{KW97} performed on 
	the right end-marker sets the amplitude of $ A $ to 0. 
	Therefore, $ w $ is rejected with certainty.
	If $ w \in \overline{C_{m}} $, then the accepting probability is equal to
	\begin{equation}
		\left( \left( \frac{1}{\sqrt{2}} \right)^{|w|+2} -
		 \left( \frac{1}{\sqrt{2}} \right)^{m+2} \right)^{2}
	\end{equation}
	and it is minimized when $ |w|=m+1 $, which gives us the inequality
	\begin{equation}  g_{\mathcal{M}_{m}}(w) > \left( \frac{1}{2} \right)^{m+6}. \end{equation}
	By Lemma \ref{lemma:kwqfa-to-1qfarestart-constant}, there exists a $ 7 $-state 
	RT-QFA$  ^{\circlearrowleft} $ $ \mathcal{M}_{m,\epsilon} $ recognizing $ \overline{C_{m}} $
	with positive one-sided bounded error and whose expected runtime is 
	$ O \left(\frac{1}{\epsilon}2^{m}|w| \right) $.
	By swapping the accepting and rejecting states of $ \mathcal{M}_{m,\epsilon} $,
	we can get the desired machine.
\end{proof}

Note that, unlike what we had with Theorem \ref{theorem:A_m}, the QFAs of Theorems \ref{theorem:B_m}
and \ref{theorem:C_m} cannot be converted so easily to 2PFAs.
In fact, we can prove that there exist no 2PFA families of fixed size which recognize $ B_m $
and $ C_m $ with fixed one-sided error less than $ \frac{1}{2} $, like those QFAs:
Assume that such a 2PFA family exists. Switch the accept and reject states to obtain a family for the complements
of the languages. The 2PFAs thus obtained operate with cutpoint 0.
Obtain an equivalent 2NFA with the same number of states by converting all 
transitions with nonzero weight to nondeterministic transitions.
But there are only finitely many 2NFAs of this size, meaning that 
they cannot recognize our infinite family of languages.

\section{Probability Amplification} \label{berr:amplification}

Many automaton descriptions in this thesis, and elsewhere in the theory
of probabilistic and quantum automata, describe not a single
algorithm, but a general template which one can use for building a
machine $ M_{\epsilon} $ that operates with a desired error bound $ \epsilon $.
The dependences of the runtime and number of states of $ M_{\epsilon} $ on $ \frac{1}{\epsilon} $
are measures of the complexity of the probability
amplification process involved in the construction method used.
Viewed as such, the constructions described in the theorems in
Section \ref{berr:succinctness} are maximally efficient in terms of the state cost,
with no dependence on the error bound. In this section, we present improvements over 
previous results about the efficiency of probability amplification in 2QFAs.

\subsection{Improved Algorithms for UPAL Language} \label{berr:L-upal-two-way}

In classical computation, one only needs to sequence $ O(\log(\frac{1}{\epsilon})) $
identical copies of a given probabilistic automaton with one sided
error $ p < 1 $ to run on the same
input in order to obtain a machine with error bound $ \epsilon $.
Yakary{\i}lmaz and Say \cite{YS09B} noted that this method of
probability amplification
does not yield efficient results for 2KWQFAs; the number of machine
copies required to reduce the error to $\epsilon$
can be as high as $ ( \frac{1}{\epsilon} )^2 $.
The most succinct 2KWQFAs for $ L_{upal} $, $ \{ a^{n}b^{n} \mid n > 0 \} $,
produced by alternative methods developed in \cite{YS09B}
have $ O(\log^{2}(\frac{1}{\epsilon})\log\log(\frac{1}{\epsilon})) $ states,
and runtime linear in the size of the input $ w $.
In Section \ref{berr:L-upal-RT}, we present a construction  which yields (exponential
time) RT-QFA$ ^{\circlearrowleft} $s that recognize $ L_{upal} $ within any desired error bound
$ \epsilon $, with no dependence of the state set size on $ \epsilon $.
Ambainis and Watrous \cite{AW02} present a method which can be used to build
2QCFAs that recognize $ L_{upal} $ also with constant state set size,
where the ``tuning'' of the automaton for a particular
error bound is achieved by setting some transition amplitudes
appropriately, and the expected runtime of those machines is $ O(|w|^4) $.
We now show that the 2QFA$ ^{\circlearrowleft} $ formalism allows more
efficient probability amplification.
\begin{theorem}
	\label{theorem:2qfarestart_anbn}
	There exists a constant $n$, such that, for any $ \epsilon>0 $, an $n$-state 2QFA$ ^{\circlearrowleft} $
	which recognizes $ L_{upal} $ with one-sided error bound $ \epsilon $ 
	within $ O(\frac{1}{\epsilon}|w|) $ expected runtime can be constructed, where $ w $ is the input string.
\end{theorem}
\begin{proof}
	We start with Kondacs and Watrous' original 2KWQFA \cite{KW97} $ M_N $, 
	which recognizes $ L_{upal} $ with one-sided error $ \frac{1}{N} $, for any integer $ N>1 $. 
	After a deterministic test for membership of $ a^{*}b^{*} $, 
	$ M_N $ branches to $ N $ computational paths, each of which perform a QFT at the end of the computation. 
	Set $ N=2 $. $ M_2 $ accepts all members of $ L_{eq} $ with probability 1.
	Non-members of $ L_{eq} $ are rejected with probability at least $ \frac{1}{2} $.
	We convert $ M_2 $ to a 2QFA$ ^{\circlearrowleft} $ $ \mathcal{M}_{\epsilon}' $ 
	by changing the target states of the QFT as follows:
	\begin{equation}
		\mathsf{path_{1}} \rightarrow \frac{1}{\sqrt{2}}\ket{\mbox{Reject}}
		+ \sqrt{\frac{\epsilon}{2}} \ket{\mbox{Accept}}+ \sqrt{\frac{1-\epsilon}{2}}\ket{\mbox{Restart}}
	\end{equation}
	\begin{equation}
		\mathsf{path_{2}} \rightarrow -\frac{1}{\sqrt{2}}\ket{\mbox{Reject}}
		+ \sqrt{\frac{\epsilon}{2}} \ket{\mbox{Accept}}+ \sqrt{\frac{1-\epsilon}{2}}\ket{\mbox{Restart}}
	\end{equation}
	where the amplitude of each path is $ \frac{1}{\sqrt{2}} $.
	For a given input $ w \in \Sigma^{*} $, 
	\begin{enumerate}
		\item if $ w $ is not of the form $ a^{*}b^{*} $, then $ p_{\mathcal{M}'}^{r}(w)=1 $;
		\item if $ w $ is of the form $ a^{*}b^{*} $ and $ w \notin L $, 
			then $ p_{\mathcal{M}'}^{r}(w)=\frac{1}{2} $, and 
			$ p_{\mathcal{M}'}^{a}(w)=\frac{\epsilon}{2} $;
		\item if $ w \in L $, then $ p_{\mathcal{M}'}^{r}(w)=0 $ and 
			$ p_{\mathcal{M}'}^{a}(w)=\epsilon $.
	\end{enumerate}
	It is easily seen that the error is one-sided. 
	Since $ \frac{p_{\mathcal{M}'}^{a}(w)}{p_{\mathcal{M}'}^{r}(w)} = \epsilon $,
	we can conclude with Lemma \ref{lemma:prej-over-pacc}.
	Moreover, the minimum halting probability occurs in the third case above, and so
	the expected runtime of $ \mathcal{M}_{\epsilon}' $ is $ O(\frac{1}{\epsilon}|w|) $.
\end{proof}
\begin{theorem}
	\label{theorem:2qfareset_anbn}
	For any $ \epsilon \in (0,\frac{1}{2}) $, there exists a 2QFA$ ^{\curvearrowleft} $ with
	$ O(\log(\frac{1}{\epsilon})) $ states that recognizes $ L_{upal} $ with 
	one-sided error bound $ \epsilon $ in $ O(\log(\frac{1}{\epsilon})|w|) $ steps, 
	where $ w $ is the input string.
\end{theorem}
\begin{proof}
	Let $ M_{2} $ be the 2KWQFA recognizing $ L_{upal} $ with one-sided error bound $ \frac{1}{2} $ 
	mentioned in the proof of Theorem \ref{theorem:2qfarestart_anbn}.
	Then, a 2QFA$ ^{\curvearrowleft} $ that is constructed by sequentially connecting
	$ O(\log(\frac{1}{\epsilon})) $ copies of $ M_{2} $, so that the input is accepted only if it is accepted by
	all the copies, and rejected otherwise, can recognize $ L_{upal} $ with one-sided error bound $ \epsilon $.
\end{proof}

\subsection{A Realtime Quantum Finite Automata with Restart Algorithm for UPAL Language} \label{berr:L-upal-RT}

\begin{theorem}
	\label{berr:L-upal-for-RT-QFA-restart}
	For any $ \epsilon>0 $, there exists a 15-state RT-QFA$ ^{\circlearrowleft} $ $ \mathcal{M}_{\epsilon} $,
	which accepts any $ w \in L_{upal} $ with certainty, and rejects any $ w \notin L_{upal} $
	with probability at least $ 1-\epsilon $.
	Moreover, the expected runtime of $ \mathcal{M}_{\epsilon} $ on $ w $ is 
	$ O(\frac{1}{\epsilon}(2\sqrt{2})^{|w|}|w|) $ .
\end{theorem}
\begin{proof}
	We construct a $ 12 $-state RT-KWQFA recognizing $ \overline{L_{eq}} $ 
	with positive one-sided unbounded error.
	Let $ \mathcal{M} = (Q,\Sigma,\delta,q_{0},Q_{a},Q_{r}) $ be RT-KWQFA with	
	$ Q_{a}=\{A_{1},A_{2},A_{3}\} $, $ Q_{rej}=\{R_{1},R_{2},R_{3}\} $, and
	$ Q_{n}=\{p_{0},p_{1},p_{2},q_{0},q_{1},q_{2}\} $.
	The transition function of $ \mathcal{M} $ is shown in Figure \ref{figure:anbn1QFARestart}.	
	As before, we assume that the transitions not specified in the figure are filled in to ensure that
	the $ U_{\sigma} $ are unitary.

	\begin{figure}[here]		
		\setlength{\extrarowheight}{2pt}
		\scriptsize{
		\begin{center}
		\begin{tabular}{|c|l|l|l|}
			\hline
			Paths & \multicolumn{1}{c|}{$ U_{\cent}, U_{a} $} &
			\multicolumn{1}{c|}{$ U_{b} $} & \multicolumn{1}{c|}{$ U_{\dollar} $} \\
			\hline
			& $ U_{\cent} \ket{q_{0}} = \frac{1}{\sqrt{2}} \ket{p_{0}} + \frac{1}{\sqrt{2}}\ket{q_{0}} $ 
			& & \\
			\hline
			$ \mathsf{path_{1}} $
			&
			$ \begin{array}{@{}l@{}}
				U_{a} \ket{p_{0}} = \frac{1}{2} \ket{p_1} + \frac{1}{2}\ket{R_{1}} 
				+ \frac{1}{\sqrt{2}}\ket{R_{2}} \\
				U_{a} \ket{p_{1}} = \frac{1}{2} \ket{p_1} + \frac{1}{2}\ket{R_{1}}
				- \frac{1}{\sqrt{2}}\ket{R_{2}} \\
				U_{a} \ket{p_{2}} = \ket{A_{1}}
			\end{array} $
			&
			$ \begin{array}{@{}l@{}}
				U_{b} \ket{p_{0}} = \ket{A_{1}} \\
				U_{b} \ket{p_{1}} = \frac{1}{\sqrt{2}} \ket{p_{2}}+\frac{1}{\sqrt{2}}\ket{R_{1}} \\
				U_{b} \ket{p_{2}} = \frac{1}{\sqrt{2}} \ket{p_{2}}-\frac{1}{\sqrt{2}}\ket{R_{1}}
			\end{array} $
			&
			$ \begin{array}{@{}l@{}}
				U_{\dollar}\ket{p_{0}}=\ket{R_{1}} \\
				U_{\dollar}\ket{p_{1}} = \ket{A_{1}} \\
				U_{\dollar}\ket{p_{2}} = \frac{1}{\sqrt{2}} \ket{R_{2}} + \frac{1}{\sqrt{2}} \ket{A_{2}} \\
			\end{array} $
			\\
			\hline
			$ \mathsf{path_{2}}$
			&
			$ \begin{array}{@{}l@{}}
				U_{a} \ket{q_{0}} = \frac{1}{\sqrt{2}} \ket{q_1} + \frac{1}{\sqrt{2}} \ket{R_{3}} \\
				U_{a} \ket{q_{1}} = \frac{1}{\sqrt{2}} \ket{q_1}-\frac{1}{\sqrt{2}} \ket{R_{3}} \\
				U_{a} \ket{q_{2}} = \ket{A_{2}}
			\end{array} $
			&
			$ \begin{array}{@{}l@{}}
				U_{b} \ket{q_{0}} = \ket{A_{2}} \\
				U_{b} \ket{q_{1}} = \frac{1}{2} \ket{q_2} + \frac{1}{2}\ket{R_{2}} 
				+ \frac{1}{\sqrt{2}}\ket{R_{3}} \\
				U_{b} \ket{q_{2}} = \frac{1}{2} \ket{q_2} +
				\frac{1}{2}\ket{R_{2}} - \frac{1}{\sqrt{2}}\ket{R_{3}}
			\end{array} $
			&
			$ \begin{array}{@{}l@{}}
				U_{\dollar}\ket{q_{0}} = \ket{R_{3}} \\
				U_{\dollar}\ket{q_{1}} = \ket{A_{3}} \\
				U_{\dollar}\ket{q_{2}} = \frac{1}{\sqrt{2}} \ket{R_{2}} - \frac{1}{\sqrt{2}} \ket{A_{2}} \\
			\end{array} $
			\\
			\hline       
		\end{tabular}
		\end{center}
		}
		\caption{Specification of the transition function of the RT-KWQFA presented in
		the proof of Theorem \ref{berr:L-upal-for-RT-QFA-restart}}
		\vskip\baselineskip
		\label{figure:anbn1QFARestart}
	\end{figure}

	As seen in the figure, $ \mathcal{M} $ branches to two paths on the left end-marker. 
	Both paths reject immediately if the input $ w \in \{a,b\}^{*} $ is the empty string, 
	and accept with nonzero probability, say $ \alpha $, if it is of the form 
	$ (\{a,b\}^{*} \setminus a^{*}b^{*}) \cup a^{+} \cup b^{+}  $.
	Otherwise, $ w = a^{m}b^{n}$ $ (m,n>0) $, and the amplitudes of the
	paths just before the transition associated with the right end-marker in the first round are as follows:
	\begin{itemize}
		\item State $ p_{2} $ has amplitude $ \frac{1}{\sqrt{2}} (\frac{1}{2})^{m} (\frac{1}{\sqrt{2}})^{n}$,
		\item state $ q_{2} $ has amplitude $ \frac{1}{\sqrt{2}} (\frac{1}{\sqrt{2}})^{m} (\frac{1}{2})^{n}$.
	\end{itemize}
	If $ m=n $,  then the accepting probability is zero. 
	If $ m \neq n $ (assume without loss of generality that $ m=n+d $ for some $ d \in \mathbb{Z}^{+} $), 
	then the accepting probability is equal to
	\begin{equation} \left( \frac{1}{2} \right)^{m+n+1} \left( \left( \frac{1}{\sqrt{2}} \right)^{m} - 
	\left( \frac{1}{\sqrt{2}} \right)^{n} \right)^{2}
	=
	\underbrace{\left( \frac{1}{2} \right)^{m+2n+1}}_{ > \left( \frac{1}{2} \right)^{\frac{3|w|}{2}+1} } 
	\underbrace{\left( 1 - \left( \frac{1}{\sqrt{2}} \right)^{d-2}
	+ \left( \frac{1}{2} \right)^{d} \right)}_{ > \frac{1}{16}} 
	\end{equation}
	Since $ \alpha $ is always greater than this value,
	\begin{equation} g_{\mathcal{M}}(|w|) > \left( \frac{1}{2} \right)^{\frac{3|w|}{2}+5}, \end{equation}
	for $ |w|>0 $.
	By Lemma \ref{lemma:kwqfa-to-1qfarestart-exponential}, there exists a $ 15 $-state 
	RT-QFA$  ^{\circlearrowleft} $ $ \mathcal{M}_{\epsilon} $ recognizing $ \overline{L_{upal}} $
	with positive one-sided bounded error and whose expected runtime is 
	$ O(\frac{1}{\epsilon}(2\sqrt{2})^{|w|}|w|) $.
	By swapping accepting and rejecting states of $ \mathcal{M}_{m} $,
	we can get the desired machine.
\end{proof}

\subsection{An Improved Algorithm for PAL Language} \label{berr:L_pal}

Ambainis and Watrous \cite{AW02} present a 2CQFA construction which
decides $ L_{pal} $ in expected time $ O( \left( \frac{1}{\epsilon} \right)^{|w|}|w|) $
with error bounded by $ \epsilon > 0 $, where $ w $ is the input string.
(Watrous \cite{Wa98} describes a 2KWQFA which accepts all members of the
complement of $ L_{pal} $ with probability 1, and fails to halt for all palindromes;
it is not known if 2KWQFAs can recognize this language by halting for all inputs.)
We now present a RT-QFA$ ^{\circlearrowleft} $ construction, which,
by Lemma \ref{lemma:1qfareset-simulatedby-2qcfa},
can be adapted to yield 2CQFAs with the same complexity, which reduces
the dependence of the Ambainis-Watrous method on the desired error bound considerably.

\begin{theorem}
	\label{berr:L-pal-for-RT-QFA-restart}
	For any $ \epsilon>0 $, there exists a 15-state RT-QFA$ ^{\circlearrowleft} $ $ \mathcal{M}_{\epsilon} $
	which accepts any $ w \in L_{pal} $ with certainty, and rejects any $ w \notin L_{pal} $
	with probability at least $ 1-\epsilon $.
	Moreover, the expected runtime of $ \mathcal{M}_{\epsilon} $ on $ w $ is $ O(\frac{1}{\epsilon}3^{|w|}|w|) $.
\end{theorem}
\begin{proof}
	We first construct a modified version of the RT-KWQFA algorithm of L\={a}ce et al. \cite{LSF09} 
	for recognizing the nonpalindrome language.
	The idea behind the construction is that we encode both the input 
	string and its reverse into the amplitudes of two of the states of the machine, 
	and then perform a substraction between these amplitudes using the QFT \cite{LSF09}.
	If the input is not a palindrome, the two amplitudes do not cancel each other completely, 
	and the nonzero difference is transferred to an accept state. 
	Otherwise, the accepting probability is zero.

	Let $ \mathcal{M} = (Q,\Sigma,\delta,q_{0},Q_{a},Q_{r}) $ be RT-KWQFA with
	$ Q_{n}=\{p_{1},p_{2},q_{0},q_{1},q_{2},q_{3}\} $ ,
	$ Q_{a}=\{A\} $,
	$ Q_{r}=\{R_{i} \mid 1 \le i \le 5\} $.
	The transition function of $ \mathcal{M} $ is shown in Figure \ref{figure:pal1QFARestart}.	
	As before, we assume that the transitions not specified in the figure are filled in to ensure that
	the $ U_{\sigma} $ are unitary.
	 \begin{figure}[here]		
		\setlength{\extrarowheight}{5pt}
		\scriptsize{
		\begin{center}
		\begin{tabular}{|c|l|l|}
			\hline
			Paths & \multicolumn{1}{c|}{$ U_{\cent}, U_{a} $} &
			\multicolumn{1}{c|}{$ U_{b} $} \\
			\hline
			& 
			$ U_{\cent} \ket{q_{0}} = \frac{1}{\sqrt{2}} \ket{p_{1}} +  \frac{1}{\sqrt{2}} \ket{q_{1}} $
			& \\
			\hline
			$ \mathsf{path_{1}} $ 
			&
			$ \begin{array}{@{}l@{}}
				U_{a} \ket{p_{1}} = \sqrt{\frac{2}{3}} \ket{p_{1}} - \frac{1}{\sqrt{3}} \ket{R_{1}} \\
				U_{a} \ket{p_{2}} = \frac{1}{\sqrt{6}} \ket{p_{1}} + \frac{1}{\sqrt{6}} \ket{p_{2}}
				+ \frac{1}{\sqrt{3}} \ket{R_{1}} + \frac{1}{\sqrt{3}} \ket{R_{2}}
			\end{array} $
			&
			$ \begin{array}{@{}l@{}}
				U_{b} \ket{p_{1}} = \frac{1}{\sqrt{6}} \ket{p_{1}} + \frac{1}{\sqrt{6}} \ket{p_{2}}
				+ \frac{1}{\sqrt{3}} \ket{R_{1}} + \frac{1}{\sqrt{3}} \ket{R_{2}} \\
				U_{b} \ket{p_{2}} = \sqrt{\frac{2}{3}} \ket{p_{2}} -  \frac{1}{\sqrt{3}} \ket{R_{1}} \\
			\end{array} $
			\\
			\hline
			$ \mathsf{path_{2}}$
             &
             $ \begin{array}{@{}l@{}}
				U_{a} \ket{q_{1}} = \frac{1}{\sqrt{6}} \ket{q_{1}} + \frac{1}{\sqrt{6}} \ket{q_{3}}
				-\frac{1}{\sqrt{3}} \ket{R_{3}} + \frac{1}{\sqrt{3}} \ket{R_{4}} \\
				U_{a} \ket{q_{2}} = \sqrt{\frac{2}{3}} \ket{q_{2}} + \frac{1}{\sqrt{3}} \ket{R_{5}} \\
				U_{a} \ket{q_{3}} = \sqrt{\frac{2}{3}} \ket{q_{3}} + \frac{1}{\sqrt{3}} \ket{R_{3}} \\
			\end{array} $
			&
			$ \begin{array}{@{}l@{}}
				U_{b} \ket{q_{1}} =\frac{1}{\sqrt{6}} \ket{q_{1}} + \frac{1}{\sqrt{6}} \ket{q_{2}}
				-\frac{1}{\sqrt{3}} \ket{R_{3}} + \frac{1}{\sqrt{3}} \ket{R_{4}} \\
				U_{b} \ket{q_{2}} = \sqrt{\frac{2}{3}} \ket{q_{2}} + \frac{1}{\sqrt{3}} \ket{R_{3}} \\
				U_{b} \ket{q_{3}} = \sqrt{\frac{2}{3}} \ket{q_{3}} + \frac{1}{\sqrt{3}} \ket{R_{5}} \\
			\end{array} $
			\\
			\hline
			&
			\multicolumn{2}{c|}{$ U_{\dollar} $}
			\\
			\hline
			$ \mathsf{path_{1}} $
			&
			\multicolumn{2}{l|}{
				$ \begin{array}{@{}l@{}}
					U_{\dollar}\ket{p_{1}} = \ket{R_{1}} \\
					U_{\dollar}\ket{p_{2}} = \frac{1}{\sqrt{2}} \ket{A} + \frac{1}{\sqrt{2}} \ket{R_{2}} \\
				\end{array} $
			}
			\\
			\hline
			$ \mathsf{path_{2}} $
			&
			\multicolumn{2}{l|}{
				$ \begin{array}{@{}l@{}}
					U_{\dollar}\ket{q_{1}} = \ket{R_{3}} \\
					U_{\dollar}\ket{q_{2}} = -\frac{1}{\sqrt{2}} \ket{A} + \frac{1}{\sqrt{2}} \ket{R_{2}} \\
					U_{\dollar}\ket{q_{3}} = \ket{R_{4}} \\
				\end{array} $
			}
			\\
             \hline
		\end{tabular}
		\end{center}
		}
		\caption{Specification of the transition function of the RT-KWQFA presented in
		the proof of Theorem \ref{berr:L-pal-for-RT-QFA-restart}}
		\vskip\baselineskip
		\label{figure:pal1QFARestart}
	\end{figure}

	$ \mathsf{path_{2}} $ and $ \mathsf{path_{1}} $ encode the input
	string and its reverse \cite{Ra63,Pa71} into the amplitudes of  states $ q_{2} $ and $ p_{2} $, respectively. 
	If the input is $ w=w_{1}w_{2}\cdots w_{l} $, then the values of
	these amplitudes just before the transition associated with the right end-marker in 
	the first round are as follows:
	\begin{itemize}
		\item State $ p_{2} $ has amplitude $ \frac{1}{\sqrt{2}} \left(\sqrt{\frac{2}{3}}\right)^{|w|}
		(0.w_{l}w_{l-1}\cdots w_{1})_{2} $, and 
		\item state $ q_{2} $ has amplitude $ \frac{1}{\sqrt{2}} \left(\sqrt{\frac{2}{3}}\right)^{|w|}
		(0.w_{1}w_{2} \cdots w_{l})_{2}. $
	\end{itemize}
	The factor of $ \sqrt{\frac{2}{3}} $ is due to the ``loss'' of amplitude necessitated by the fact that the
	originally non-unitary encoding matrices of \cite{Ra63,Pa71} have to be ``embedded'' 
	in a unitary matrix \cite{YS09D,YS10C}. 
	Note that the symbols $ a $ and $ b $ are encoded by 0 and 1, respectively.

	If $ w \in L_{pal} $, the acceptance probability is zero.
	If $ w \in \overline{L_{pal}} $, the acceptance probability is minimized by
	strings which are almost palindromes, except for a single defect in the middle, that is,
	when $ |w|=2k $ for $ k \in \mathbb{Z}^{+} $,  $ w_{i}=w_{2k-i+1} $, where $ 1 \le i \le k-1 $, and
	$ w_{k} \neq w_{k+1} $, so,
	\begin{equation} g_{\mathcal{M}}(w) \ge \frac{1}{8} \left( \frac{1}{3} \right)^{|w|}. \end{equation}
	By Lemma \ref{lemma:kwqfa-to-1qfarestart-exponential}, there exists a $ 15 $-state 
	RT-QFA$  ^{\circlearrowleft} $ $ \mathcal{M}_{\epsilon} $ recognizing $ \overline{L_{pal}} $
	with positive one-sided bounded error, whose expected runtime is 
	$ O(\frac{1}{\epsilon}3^{|w|}|w|) $.
	By swapping accepting and rejecting states of $ \mathcal{M}_{m} $,
	we can get the desired machine.
\end{proof}

Note that the technique used in the proof above can be extended easily
to handle bigger input alphabets by using the matrices defined on Page 169 of \cite{Pa71},
and the method of simulating stochastic matrices by unitary matrices described in 
Sections \ref{uerr:KWQFA-languages} and \ref{uerr:non-languages}.

\section{Probabilistic and Quantum Automata with Postselection} \label{berr:postsel}

In this section, we define the realtime finite automaton with
postselection (RT-PostFA) in the spirit of Aaronson \cite{Aa05}.
In fact, there is no difference between a standard finite automaton
model and its counterpart with postselection
in the processing of the input except for the final decision.
Instead of accepting states, RT-PostFAs have a set of \textit{postselection states},
denoted as $ Q_{p} $, which is the union of two disjoint subsets
$ Q_{pa} $ and $ Q_{pr} $, the \textit{accepting} and \textit{rejecting} \textit{postselection states},
and have the capability of discarding all computation branches except
the ones belonging to $ Q_{p} $ at the end,
on which a normalization is performed, and the output is given.
Therefore, the probability of being in a postselection state at the
end must be nonzero.

\subsection{Definitions} \label{berr:postsel-definition}

The acceptance and rejection probabilities of a machine, say $ \mathcal{M} $,
on a given input, say $ w \in \Sigma^{*} $, in postselection are denoted as 
$ p_{\mathcal{M}}^{a}(w) $ and $ p_{\mathcal{M}}^{r}(w) $, respectively.
Thus, by normalizing these probabilities, we obtain
\begin{equation}
       f_{\mathcal{P}}^{a}(w) =
\dfrac{p_{\mathcal{P}}^{a}(w)}{p_{\mathcal{P}}^{a}(w)+p_{\mathcal{P}}^{r}(w)},
\end{equation}
and
\begin{equation}
       f_{\mathcal{P}}^{r}(w) =
\dfrac{p_{\mathcal{P}}^{r}(w)}{p_{\mathcal{P}}^{a}(w)+p_{\mathcal{P}}^{r}(w)}.
\end{equation}

A RT-PFA with postselection (RT-PostPFA) is a 5-tuple
\begin{equation}
	\mathcal{P}=(Q,\Sigma,\{ A_{\sigma \in \tilde{\Sigma}} \},q_{1},Q_{p}),
\end{equation}
satisfying that for each input string $ w \in \Sigma^{*} $,
\begin{equation}
       \sum_{q_{i} \in Q_{p}} v_{|\tilde{w}|}[i] > 0.
\end{equation}
The acceptance and rejection probabilities of $ \mathcal{P} $ on $ w \in \Sigma^{*} $ in postselection
are defined respectively as 
\begin{equation}
       p_{\mathcal{P}}^{a}(w) = \sum_{q_{i} \in Q_{pa}} v_{|\tilde{w}|}[i]
\end{equation}
and
\begin{equation}
       p_{\mathcal{P}}^{r}(w) = \sum_{q_{i} \in Q_{pr}} v_{|\tilde{w}|}[i].
\end{equation}

We call the class of languages recognized by RT-PostPFAs with bounded error 
\textit{PostS} (post-stochastic languages).

A RT-QFA with postselection (RT-PostQFA) is a 5-tuple
\begin{equation}
       \mathcal{M}=(Q,\Sigma,\{\mathcal{E}_{\sigma \in \tilde{\Sigma}}\},q_{1},Q_{p}),
\end{equation}
satisfying that for each input string $ w \in \Sigma^{*} $,
\begin{equation}
       tr(P_{p} \rho_{| \tilde{w} |} ) > 0,
\end{equation}
where $ P_{p} $ is the projector defined as
\begin{equation}
       P_{p} = \sum_{q \in Q_{p}} \ket{q}\bra{q}.
\end{equation}
Additionally we define projectors $ P_{pa} $ and $ P_{pr} $ as follows:
\begin{equation}
       P_{pa} = \sum_{q \in Q_{pa}} \ket{q}\bra{q}
\end{equation}
and
\begin{equation}
       P_{pr} = \sum_{q \in Q_{pr}} \ket{q}\bra{q}.
\end{equation}
The acceptance and rejection probabilities of $ \mathcal{M} $ on $ w
\in \Sigma^{*} $ in postselection
are defined as
\begin{equation}
       p_{\mathcal{M}}^{a}(w) = tr(P_{pa} \rho_{| \tilde{w} |} )
\end{equation}
and
\begin{equation}
       p_{\mathcal{M}}^{r}(w) = tr(P_{pr} \rho_{| \tilde{w} |} ).
\end{equation}
We call the class of languages recognized by RT-PostQFAs with bounded
error \textit{PostQAL} (post-quantum automaton languages).

The error bound of a given postselection machine can be
improved by performing a tensor product of the machine with itself as
many times as required.
Specifically, if we combine $ k $ copies of a machine with postselection state set
$ Q_{pa} \cup Q_{pr} $, the new accepting and rejecting postselection
state sets can be chosen as
\begin{equation}
       Q_{pa}' = \underbrace{Q_{pa} \times \cdots \times Q_{pa}}_{k
\mbox{ times}}
\end{equation}
and
\begin{equation}
       Q_{pr}' = \underbrace{Q_{pr} \times \cdots \times Q_{pr}}_{k
\mbox{ times}},
\end{equation}
respectively. 
\begin{lemma}
	If $ L $ is recognized by RT-PostQFA (resp., RT-PostPFA) $ \mathcal{M} $ with error bound 
	$ \epsilon \in (0,\frac{1}{2})  $,
	then there exists a RT-PostQFA (resp., RT-PostPFA), say $ \mathcal{M}' $,
	recognizing $ L $ with error bound $  \epsilon^{2} $.
\end{lemma}
\begin{proof}
	We give a proof for RT-PostQFAs and it can easily be extended for RT-PostPFAs.
	$ M' $ can be obtained by tensoring $ k $ copies of $ \mathcal{M} $, where
	the new accepting (resp., rejecting) postselection states, $ Q_{pa}' $  (resp., $ Q_{pr}' $),
	are $ \otimes_{i=1}^{k} Q_{pa} $ (resp., $ \otimes_{i=1}^{k} Q_{pa} $),
	where $ Q_{pa} $ (resp., $ Q_{pr} $) are accepting postselection states of $ \mathcal{M} $.
	
	Let $ \rho_{\tilde{w}} $ and $ \rho_{\tilde{w}}' $ be the respectively density matrices of 
	$ \mathcal{M} $ and $ \mathcal{M}' $
	after reading $ \tilde{w} $ for a given input string $ w \in \Sigma^{*} $. 
	By definition, we have
	\begin{equation}
		p_{\mathcal{M}}^{a}(w) = \sum_{q_{i} \in Q_{pa} }\rho_{\tilde{w}}[i,i],
		~~~~
		p_{\mathcal{M}'}^{a}(w) = \sum_{q_{i'} \in Q_{pa}' }\rho_{\tilde{w}}[i',i']		
	\end{equation}
	and 
	 \begin{equation}
		p_{\mathcal{M}}^{r}(w) = \sum_{q_{i} \in Q_{pr} }\rho_{\tilde{w}}[i,i],
		~~~~
		p_{\mathcal{M}'}^{r}(w) = \sum_{q_{i'} \in Q_{pr}' }\rho_{\tilde{w}}[i',i'].
	\end{equation}
	By using the equality $ \rho_{\tilde{w}}' = \otimes_{i=1}^{k} \rho_{\tilde{w}}  $,
	the following can be obtained after a straightforward calculation:
	\begin{equation}
		p_{\mathcal{M}'}^{a}(w) = \left( p_{\mathcal{M}}^{a}(w) \right)^{k}
	\end{equation}
	and
	\begin{equation}
		p_{\mathcal{M}'}^{r}(w) = \left( p_{\mathcal{M}}^{r}(w) \right)^{k}.
	\end{equation}
	
	We examine the case of $ w \in L $ (the case $ w \notin L $ is symmetric).
	Since $ L $ is recognized by $ \mathcal{M} $ with error bound $ \epsilon $, 
	we have
	\begin{equation}
		 \frac{p_{\mathcal{M}}^{r}(w)}{p_{\mathcal{M}}^{a}(w)}
		 \leq 
		 \frac{\epsilon}{1-\epsilon}.
	\end{equation}
	If $ L $ is recognized by $ \mathcal{M}' $ with error bound $ \epsilon^{2} $,
	we must have
	\begin{equation}
		\frac{p_{\mathcal{M}'}^{r}(w)}{p_{\mathcal{M}'}^{a}(w)}
		\leq
		\frac{\epsilon^{2}}{1-\epsilon^{2}}.
	\end{equation}
	Thus, any $ k $ satisfying the following inequality provides the desired machine $ \mathcal{M}' $:
	\begin{equation}		
		\left( \frac{\epsilon}{1-\epsilon} \right)^{k}
		\leq		
		\frac{\epsilon^{2}}{1-\epsilon^{2}}
	\end{equation}
	due to the fact that 
	\begin{equation}
		\label{berr:eq:k}
		\frac{p_{\mathcal{M}'}^{r}(w)}{p_{\mathcal{M}'}^{a}(w)} =
		\left( \frac{p_{\mathcal{M}}^{r}(w)}{p_{\mathcal{M}}^{a}(w)} \right)^{k}.
	\end{equation}
	By solving Equation \ref{berr:eq:k}, we can get
	\begin{equation}
		k = 1 + \left\lceil \frac{ \log \left( \frac{1}{\epsilon} + 1 \right) }{ 
		\log \left( \frac{1}{\epsilon} - 1 \right) } \right\rceil.
	\end{equation}
	Therefore, for any $ 0 < \epsilon < \frac{1}{2} $, we can find a value for $ k $.
\end{proof}

\begin{corollary}
	If $ L $ is recognized by RT-PostQFA (resp., RT-PostPFA) $ \mathcal{M} $ with error bound 
	$ 0 < \epsilon < \frac{1}{2}  $,
	then there exists a RT-PostQFA (resp., RT-PostPFA), say $ \mathcal{M}' $, recognizing $ L $ with error bound 
	$ \epsilon' < \epsilon $ such that $ \epsilon' $ can be arbitrarily close to 0.
\end{corollary}

\subsection{Characterization of Realtime Postselection Automata} \label{berr:postsel-QFA-PFA}

\begin{theorem}
	The classes of languages recognized by RT-PFA$ ^{\circlearrowleft} $ and RT-QFA$ ^{\circlearrowleft} $
	with bounded error are identical to PostS and PostQAL, respectively.
\end{theorem}
\begin{proof}
	As shown in Theorem \ref{berr:thm:RT-GFA-restart-simulated-by-RT-QFA-restart}, 
	the computational power of RT-GQFA$ ^{\circlearrowleft} $ and RT-QFA$ ^{\circlearrowleft} $
	are identical. Therefore, we assume RT-GQFA$ ^{\circlearrowleft} $ as the restart machine in the 
	remaining part of the proof.
	For a given RT-PostFA, we obtain a machine with restart by converting $ Q_{pa} $ and $ Q_{pr} $ to
	respectively the accepting and rejecting states, and restarting computation at the end of the input 
	in the cases where  the original machine halts in a state not in $ Q_{p} $.
	For a given machine with restart, (we assume the computation is restarted and halted
	only at the end of the input,) we obtain a RT-PostFA by taking the accepting and rejecting states of
	the original machine as the members of $ Q_{pa} $ and $ Q_{pr} $, respectively, and converting the
	remaining states to nonpostselection states.
\end{proof}

\begin{corollary}
	PostQAL and PostS are subsets of the class of the languages recognized
	by 2QFAs and 2PFAs, respectively, with bounded error.
\end{corollary}

\begin{corollary}
	\label{corollary:L-pal}
	$ L_{pal} $ is a member of PostQAL but not PostS.
\end{corollary}
\begin{proof}
	$ L_{pal} \in $ PostQAL since there is a RT-QFA$ ^{\circlearrowleft} $ algorithm for $ L_{pal} $ 
	(See \ref{berr:reset-quantum} and \cite{YS10B}).
	However,  $ L_{pal} $ cannot be recognized with bounded error even by 2PFAs \cite{DS92}.
\end{proof}

\begin{theorem}
	\label{theorem:post-closure}
	PostQAL and PostS are closed under complementation, union, and intersection.
\end{theorem}
\begin{proof}
	If a language is recognized by a RT-PostFA with bounded error,
	by swapping the accepting and rejecting postselection states, we obtain a new RT-PostFA
	recognizing the complement of the language with bounded error.
	Therefore, both classes are closed under complementation.
	
	Let $ L_{1} $ and $ L_{2} $ be members of PostQAL (resp., PostS).
	Then, there are two RT-PostQFAs (resp., RT-PostPFAs) $ \mathcal{P}_{1} $ and $ \mathcal{P}_{2} $	
	recognizing $ L_{1} $ and $ L_{2} $ with error bound $ \epsilon < \frac{3}{4} $, respectively.
	Moreover, let $ Q_{pa_{1}} $ and $ Q_{pr_{1}} $ (resp., $ Q_{pa_{2}} $ and $ Q_{pr_{2}} $)
	represent the sets of the accepting and rejecting postselection states of $ \mathcal{P}_{1} $
	(resp., $ \mathcal{P}_{2} $), respectively, and let $ Q_{p_{1}} =  Q_{pa_{1}} \cup Q_{pr_{1}} $
	and $ Q_{p_{2}} =  Q_{pa_{2}} \cup Q_{pr_{2}} $.
	By tensoring  $ \mathcal{P}_{1} $ and $ \mathcal{P}_{2} $, we obtain two new machines, 
	say $ \mathcal{M}_{1} $ and $ \mathcal{M}_{2} $, such that
	\begin{itemize}
		\item the sets of the accepting and rejecting postselection states of $ \mathcal{M}_{1} $ is
			\begin{equation}
				Q_{p_{1}} \otimes Q_{p_{2}} \setminus Q_{pr_{1}} \otimes Q_{pr_{2}}
			\end{equation} 
			and
			\begin{equation}
				Q_{pr_{1}} \otimes Q_{pr_{2}},
			\end{equation}
			respectively, and
		\item the sets of the accepting and rejecting postselection states of $ \mathcal{M}_{2} $ is
			\begin{equation}
				Q_{pa_{1}} \otimes Q_{pa_{2}},
			\end{equation} 
			and
			\begin{equation}
				Q_{p_{1}} \otimes Q_{p_{2}} \setminus Q_{pa_{1}} \otimes Q_{pa_{2}},
			\end{equation}
			respectively.
	\end{itemize}
	Thus, the following inequalities can be verified for a given input string $ w \in \Sigma^{*} $:
	\begin{itemize}
		\item if $ w \in L_{1} \cup L_{2} $, $ f_{\mathcal{M}_{1}}^{a}(w) \ge \frac{15}{16} $;
		\item if $ w \notin L_{1} \cup L_{2} $, $ f_{\mathcal{M}_{1}}^{a}(w) \leq \frac{7}{16} $;
		\item if $ w \in L_{1} \cap L_{2} $, $ f_{\mathcal{M}_{2}}^{a}(w) \ge \frac{9}{16} $; 
		\item if $ w \notin L_{1} \cap L_{2} $, $ f_{\mathcal{M}_{2}}^{a}(w) \leq \frac{1}{16} $.
	\end{itemize}
	We can conclude that both classes are closed under union and intersection.
\end{proof}

\begin{theorem}
	\label{theorem:PostQ-subset-Q}
	PostQAL and PostS are subsets of S (QAL).
\end{theorem}
\begin{proof}
	A given RT-PostFA can be converted to its corresponding standard model (without postselection) as follows:
	\begin{itemize}
		\item All nonpostselection states of the RT-PostFA are made to transition to accepting states
			with probability $ \frac{1}{2} $ at the end of the computation.
		\item All members of $ Q_{pa} $ are accepting states in the new machine.
	\end{itemize}
	Therefore, for the members, the overall accepting probability of the new machine 
	exceeds $ \frac{1}{2} $, and for nonmembers, it can be at most $ \frac{1}{2} $.
\end{proof}

By using the fact that S is not closed under union and intersection
\cite{Fl72,Fl74,La74}, Corollary \ref{corollary:L-pal}, and Theorems
\ref{theorem:post-closure} and \ref{theorem:PostQ-subset-Q},
we obtain the following corollary.

\begin{corollary}
	PostS $ \subsetneq $ PostQAL $ \subsetneq $ S (QAL).
\end{corollary}

\subsection{Latvian Postselection Automata} \label{berr:latvian-postsel}

In \cite{LSF09,DF10,SLF10}, a somewhat different RT-PostQFA model,
that violates the assumption that the postselection is done on a set
of computation branches having nonzero probability, is presented.
Although the motivation for this feature is not clear in
\cite{LSF09,DF10,SLF10}, such a violation seems reasonable
due to some fundamental reasons related to the capabilities of finite automata.
For example, when we are given ``more'' resources, we can create some
computational paths with
sufficiently small probabilities as a part of the postselection set such that
they do not affect the overall computation but can help to accept or to reject
the input as desired whenever there is zero probability of observing the other postselection states.
However, we do not know how to implement such a solution
for quantum or probabilistic automata\footnote{As a similar issue,
we do not know how to increase or decrease an
acceptance probability that is exactly equal to the cutpoint,
for a given quantum or probabilistic automaton (in most cases).
A related open problem is whether $ coS = S $ or not \cite{Pa71,YS10A}
even when we restrict ourselves to computable transition probabilities \cite{Di77}.}.

We call the machines defined in \cite{LSF09,SLF10} Latvian RT-PostFAs (RT-LPostFAs). These
have an additional component $ \tau \in \{A,R\} $ such that 
whenever the postselection probability is zero for a given input string $ w \in \Sigma^{*} $,
\begin{itemize}
       \item $ w $ is accepted if $ \tau = A $,
       \item $ w $ is rejected if $ \tau = R $.
\end{itemize}
The bounded-error classes corresponding to the RT-LPostPFA and RT-LPostQFA
models are called LPostS and LPostQAL, respectively.

\begin{theorem}
	LPostS = PostS.
\end{theorem}
\begin{proof}
	We need to show that LPostS $ \subseteq $ PostS.
	Let $ L $ be in LPostS and $ \mathcal{P} $ with $ \tau \in \{A,R\} $ be the RT-LPostPFA
	recognizing $ L $ with error bound $ \epsilon < \frac{1}{2} $.
	Suppose that $ L' $ is the language such that for each member of $ L' $,
	the probability of postselection assigned by $ \mathcal{P} $ is zero.
	By designating all postselection states as accepting states and removing the probability
	values of transitions, we obtain a RT-NFA which recognizes $ \overline{L'} $.
	Thus, there exists a RT-DFA, say $ \mathcal{D} $, recognizing $ L' $.
	
	By combining (tensoring) $ \mathcal{P} $ and $ \mathcal{D} $, a RT-PostPFA, say $ \mathcal{P}' $,
	can be obtained such that
	if the input string is a member of $ L' $, 
	the decision is given deterministically with respect to $ \tau $, and
	if it is not a member of $ L' $ (the probability of the postselection is nonzero), 
	the decision is given by standard postselection procedure.
	Therefore, $ L $ is recognized by $ \mathcal{P}' $ with the same error bound
	and so $ L $ is in PostS, too.
\end{proof}

However, we cannot use the same idea in the quantum case due to the
fact that the class of the languages recognized
by NQFAs, is a proper superclass of the regular languages (see Section \ref{uerr:nondeterministic}).

\begin{theorem}
	NQAL $ \cup $ coNQAL $ \subseteq $ LPostQAL.
\end{theorem}
\begin{proof}
	For $ L \in $ NQAL, take the accepting states of the NQFA recognizing $ L $ as postselection accepting states
	with $ \tau = R $. (There are no postselection rejecting states.)
	For $ L \in $ coNQAL, take the accepting states of the NQFA recognizing $ \overline{L} $
	as postselection rejecting states with $ \tau = A $. (There are no postselection accepting states.)
\end{proof}

Interestingly, all languages in NQAL $ \cup $ coNQAL are recognized with zero error by 
RT-LPostQFAs\footnote{L\={a}ce et al. \cite{LSF09} describe a zero error machine for $ L_{pal} $.}.

\begin{theorem}
	LPostQAL is closed under complementation.
\end{theorem}
\begin{proof}
	If a language is recognized by a RT-LPostQFA with bounded error,
	by swapping the accepting and rejecting postselection states and
	by setting $ \tau $ to $ \{A,R\} \setminus \tau $, we obtain a new RT-LPostQFA
	recognizing the complement of the language with bounded error.
	Therefore, LPostQAL is closed under complementation.
\end{proof}

\begin{theorem}
	LPostQAL $ \subseteq $ uQAL (uS).
\end{theorem}
\begin{proof}
	The proof is similar to the proof of Theorem \ref{theorem:PostQ-subset-Q} with the exception that
	\begin{itemize}
		\item if $ \tau = A $, we have recognition with nonstrict cutpoint;
		\item if $ \tau = R $, we have recognition with strict cutpoint.
	\end{itemize}
\end{proof}

\chapter{WRITE-ONLY MEMORY} \label{wom}

In this chapter, we examine realtime quantum finite automaton (RT-QFA) 
models augmented with a ``write-only memory'' (WOM) under several types of restrictions related to WOM access.

\section{Definitions} \label{wom:definition}

A WOM is a two-way write-only work tape having alphabet $ \Gamma $ containing $ \# $ and 
$ \varepsilon $, where $ \varepsilon $ means that the square under the tape head is not changed.
If we restrict the tape head movement of WOM to $ \rhd $, i.e. one-way,
we obtain a ``push-only stack'' (POS).
For POSs, we assume that, if $ \varepsilon $ is written, the work tape head does not move and
if a symbol different than $ \varepsilon $ is written, the work tape head automatically 
moves one square to the right.
A special case of POS is obtained by restricting $ \Gamma $ with $ \varepsilon $ and a \textit{counting} symbol 
(different than $ \# $), this is called an ``increment-only counter'' (IOC).
As a further restriction, one symbol (except $ \varepsilon $) is required to be written on the work tape
at every step of the computation. We call this type of memory as ``trash tape'' (TT).

In this chapter, we examine the power of a RT-QFA augmented by WOM, POS, or IOC,
namely RT-QFA-WOM, RT-QFA-POS (0-rev-RT-QPDA), or RT-QFA-IOC (0-rev-RT-Q1CA), respectively.
Note that, RT-QFA-WOMs, RT-QFA-POSs, and RT-QFA-IOCs are special cases of
realtime quantum Turing machines, 
realtime quantum pushdown automata, and 
realtime quantum one-counter automata.

Formally, a RT-QFA-WOM $ \mathcal{M} $ is a 7-tuple
\begin{equation}
	\label{def:RT-QFA-WOM}
	(Q,\Sigma,\Gamma,\Omega,\delta,q_{1},Q_{a}).
\end{equation}
	When in state $ q \in Q $ and reading symbol $ \sigma \in \tilde{\Sigma} $ 
	on the input tape, $ \mathcal{M} $ changes its state to $ q' \in Q $, writes $ \gamma \in \Gamma $ 
	and $ \omega \in \Omega $ on the WOM tape and the finite register, respectively,
	and then updates the position of the WOM tape head with respect to $ d_{w} \in \setD $
	with transition amplitude $ \delta(q,\sigma,q',\gamma,d_{w},\omega) = \alpha $, where
	$ \alpha \in \mathbb{C} $ and $ |\alpha| \leq 1 $.
	
	In order to represent all transitions from the case where $ \mathcal{M} $ is in state $ q \in Q $ and reading
	symbol $ \sigma \in \tilde{\Sigma} $ together, we use the notation
	\begin{equation}
		\delta(q,\sigma) =  \sum_{(q',\gamma,d_{w},\omega) \in Q \times \Gamma \times \setD \times \Omega }
			\delta(q,\sigma,q',\gamma,d_{w},\omega) (q',\gamma,d_{w},\omega),
	\end{equation}
	where
	\begin{equation}
		\sum_{(q',\gamma,d_{w},\omega) \in Q \times \Gamma \times \setD \times \Omega }
			|\delta(q,\sigma,q',\gamma,d_{w},\omega)|^{2} = 1.
	\end{equation}

A configuration of a RT-QFA-WOM is the collection of
\begin{itemize}
	\item the internal state of the machine,
	\item the position of the input tape head,
	\item the contents of the WOM tape, and the position of the WOM tape head.
\end{itemize}

The formal definition of the RT-QFA-POS is similar to that of the RT-QFA-WOM,
except that the movement of the WOM tape head is restricted to $ \rhd $, and so
the position of that head does not need to be a part of a configuration.
On the other hand, the definition of the RT-QFA-IOC can be simplified by removing the $ \Gamma $
component from (\ref{def:RT-QFA-WOM}):

A RT-QFA-IOC $ \mathcal{M} $ is a 6-tuple
\begin{equation}
	(Q,\Sigma,\Omega,\delta,q_{1},Q_{a}).
\end{equation}
	When in state $ q \in Q $, and reading symbol $ \sigma \in \tilde{\Sigma} $ 
	on the input tape,  $ \mathcal{M} $ changes its state to $ q' \in Q $, writes $ \omega $ in the register, and updates the value of its counter
	by $ c \in \vartriangle = \{0,+1\} $
	with transition amplitude $ \delta(q,\sigma,q',c,\omega) = \alpha $, where
	$ \alpha \in \mathbb{C} $ and $ |\alpha| \leq 1 $.
	
	In order to show all transitions from the case where $ \mathcal{M} $ is in state $ q \in Q $ and reads
	symbol $ \sigma \in \tilde{\Sigma} $ together, we use the notation
	\begin{equation}
		\delta(q,\sigma) =  \sum_{(q',c,\omega) \in Q \times \vartriangle \times \Omega }
			\delta(q,\sigma,q',c,\omega) (q',c,\omega),
	\end{equation}
	where
	\begin{equation}
		\sum_{(q',c,\omega) \in Q \times \vartriangle \times \Omega }
			|\delta(q,\sigma,q',c,\omega)|^{2} = 1.
	\end{equation}

A configuration of a RT-QFA-IOC is the collection of
\begin{itemize}
       \item the internal state of the machine,
       \item the position of the input tape head, and
       \item the value of the counter.
\end{itemize}

RT-QFA-IOC($ m $) is a RT-QFA-IOC with capability of incrementing its counter by a 
value from the set $ \{0,\ldots,m\} $, where $ m>1 $.

\section{Basic Facts} \label{wom:basic-facts}

We present some facts that are useful in the next parts in this section.

\begin{lemma}
	\label{wom:lem:classical-wom}
	The computational power of any realtime classical finite automaton is unchanged
	when the model is augmented with a WOM.
\end{lemma}
\begin{proof}
	For a given machine $ \mathcal{M} $ and an input string $w$, consider the tree $ \mathcal{T} $ of states, where the root is the initial state, each subsequent level corresponds to the processing of the next input symbol, and the children of each node $N$ are the states that have nonzero-probability transitions from $S$ with the input symbol corresponding to that level. Each such edge in the tree is labeled with the corresponding transition probability. The probability of node $N$ is the product of the probabilities on the path to $N$ from the root. The acceptance probability is the sum of the probabilities of the accept states at the last level.

Now consider attaching a WOM to $ \mathcal{M} $, and augmenting its program so that every transition now also specifies the action to be taken on the WOM. Several new transitions of this new machine may correspond to a single transition of $ \mathcal{M} $, since, for example, a transition with probability $p$ can be divided into two transitions with probability $ p \over 2 $, whose effects on the internal state are identical, but which write different symbols on the WOM. It is clear that many different programs can be obtained by augmenting $ \mathcal{M} $ in this manner with different WOM actions. Visualize the configuration tree $ \mathcal{T}_{new} $ of any one of these new machines on input $w$. There exists a homomorphism $h$ from $ \mathcal{T}_{new} $ to $ \mathcal{T} $, where $h$ maps nodes in $ \mathcal{T}_{new} $ to nodes on the same level in $ \mathcal{T} $, the configurations in $h^{-1}(N)$ all have $N$ as their states, and the total probability of the members of $h^{-1}(N)$ equals the probability of $N$ in $ \mathcal{T} $, for any $N$. We conclude that all the machines with WOM accept $w$ with exactly the same probability as $w$, so the WOM does not make any difference.
\end{proof}
 
\begin{fact}
	\cite{Fr79}
	For every $k$, if $ L $ is recognized by a deterministic RT-$ k $BCA (RT-D$ k $BCA), then for every $ \epsilon \in (0, \frac{1}{2}) $, then
	there exists a probabilistic RT-1BCA (RT-P1BCA) recognizing $L$ with negative one-sided error bound $ \epsilon $.
\end{fact}
We can generalize this result to probabilistic RT-$ k $BCAs (RT-P$ k $BCAs), where $ k>1 $.
\begin{lemma}
	\label{wom:lem:RT-PkBCA-by-RT-P1BCA}
	Let $ \mathcal{P} $ be a given RT-P$ k $BCA and $ \epsilon \in (0,\frac{1}{2}) $ be a given error bound. 
	Then, there exists a RT-P1BCA($ R $) $ \mathcal{P}' $ such that for all $ w \in \Sigma^{*} $,
	\begin{equation}
		f_{\mathcal{P}}(w) \le f_{\mathcal{P}'}(w) \le f_{\mathcal{P}}(w)+\epsilon(1-f_{\mathcal{P}}(w)),
	\end{equation}
	where $ R = 2^{\left\lceil \frac{k}{\epsilon} \right\rceil} $.
\end{lemma}
\begin{proof} 
	Based on the probabilistic method described in Figure \ref{wom:fig:fr-79},
	we can obtain $ \mathcal{P}' $ by making the following modifications on $ \mathcal{P} $:
	\begin{figure}[h!]
		\begin{center}
       \fbox{
       \begin{minipage}{0.9\textwidth}
               \small
               In this figure, we review a method presented by Freivalds in \cite{Fr79}:
               Given a machine with $ k>1 $ counters, say $ C_{1},\ldots,C_{k}  $,
whose values
               can be updated using the increment set $ \{-1,0,1\} $, we can build
a machine with a single counter, say $ C $,  whose value can be
updated using the increment set $ \{ -R,\ldots,R \} $
               ($ R = 2^{\left\lceil \frac{k}{\epsilon} \right\rceil} $), such that
               all updates on $ C_{1},\ldots,C_{k} $ can be simulated on $ C $ in
the sense that
               (i) if all values of $ C_{1},\ldots,C_{k} $ are zeros, then the
value of $ C $ is zero; and
               (ii) if the value of at least one of $ C_{1},\ldots,C_{k} $ is nonzero, then
               the value of $ C $ is nonzero with probability $ 1-\epsilon $,
               where $ \epsilon \in (0,\frac{1}{2}) $.
               The  probabilistic method for this simulation is as follows:
               \begin{itemize}
                       \item Choose a number $ r $ equiprobably from the set $ \{1,\ldots,R\} $.
                       \item The value of $ C $ is increased (resp., decreased) by $ r^{i} $
                               if the value of $ C_{i} $ is increased (resp., decreased) by 1.
               \end{itemize}
       \end{minipage}
       }
       \end{center}
       \caption{Probabilistic zero-checking of multiple counters by one counter}
       \vskip\baselineskip
       \label{wom:fig:fr-79}
\end{figure}
	\begin{enumerate}
		\item At the beginning of the computation, $ \mathcal{P}' $ equiprobably 
			chooses a number $ r $ from the set $ \{1,\ldots,R\} $.
		\item For each transition of $ \mathcal{P} $, in which the values of counters are updated
		by $ (c_{1},\ldots,c_{k}) $ $ \in \{-1,0,1\}^{k} $, i.e., the value of the $ i^{th} $ counter
		is updated by $ c_{i} $ ($ 1 \le i \le k $), $ \mathcal{P} $ 
		makes the same transition by updating its counter values by $ \sum \limits_{i=1}^{k} r^{i}c_{i} $.
	\end{enumerate}
	Hence, (i) for each accepting path of $ \mathcal{P} $, the input is accepted by $ \mathcal{P}' $, too;
	(ii) for each rejecting path of $ \mathcal{P} $, the input may be accepted by $ \mathcal{P}' $
	with a probability at most $ \epsilon $.
	By combining these cases, we obtain the following inequality for any input string $ w \in \Sigma^{*} $: 
	\begin{equation}
		f_{\mathcal{P}}(w) \le f_{\mathcal{P}'}(w) \le f_{\mathcal{P}}(w)+\epsilon(1-f_{\mathcal{P}}(w))
	\end{equation}
 \end{proof}
\begin{theorem}
	\label{wom:thm:RT-PkBCA-by-RT-P1BCA}
	If $ L $ is recognized by a RT-P$ k $BCA
	with error bound $ \epsilon \in (0, \frac{1}{2}) $, then
	$ L $ is recognized by a RT-P1BCA with error bound 
	$ \epsilon' $ ($ 0 < \epsilon < \epsilon' < \frac{1}{2} $). 
	Moreover, $ \epsilon' $ can be tuned to be arbitrarily close to $ \epsilon $.
\end{theorem}
\begin{proof}
	Let $ \mathcal{P} $ be a RT-P$ k $BCA recognizing $ L $ with error bound $ \epsilon $.
	By using the previous lemma (Lemma \ref{wom:lem:RT-PkBCA-by-RT-P1BCA}), 
	for any $ \epsilon'' \in (0,\frac{1}{2}) $,
	we can construct a RT-P1BCA($ R $), say $ \mathcal{P}'' $, from $ \mathcal{P} $,
	where $ R=2^{\left\lceil \frac{k}{\epsilon''} \right\rceil} $.
	Hence, depending on the value of $ \epsilon $,
	we can select $ \epsilon'' $ to be sufficiently small such that
	$ L $ is recognized by $ \mathcal{P}'' $ with error bound 
	$ \epsilon'=\epsilon+\epsilon''(1-\epsilon) < \frac{1}{2} $.
	Since for each RT-P1BCA($ m $), there is an equivalent RT-P1BCA for any $ m>1 $,
	$ L $ is also recognized by	a RT-P1BCA 
	with bounded error $ \epsilon' $, which can be tuned to be arbitrarily close to $ \epsilon $.
 \end{proof}
\begin{corollary}
	If $ L $ is recognized by a RT-P$ k $BCA
	with negative one-sided error bound $\epsilon \in (0, 1) $, then
	$ L $ is recognized by a RT-P1BCA with negative one-sided error bound 
	$ \epsilon' $ ($ 0 < \epsilon  < \epsilon' $). 
	Moreover, $ \epsilon' $ can be tuned to be arbitrarily close to $ \epsilon  $.
\end{corollary}

\begin{lemma}
	\label{wom:lem:inc-m-equal-inc-1}
	For a given RT-QFA-IOC($ m $) $ \mathcal{M} $, these exists a RT-QFA-IOC $ \mathcal{M'} $
	such that
	\begin{equation}
		f_{\mathcal{M}}(w)=f_{\mathcal{M'}}(w),
	\end{equation}
	for all $ w \in \Sigma^{*} $, where $ m > 2 $.
\end{lemma}
\begin{proof}
	It can be easily followed from Lemma \ref{qtm:lem:CA-m-isomorphic-CA} 
	by considering only one counter\footnote{In fact, the proof is independent
	from the number of the counters.} and then 
	restricting its update values with $ \{0,\ldots,m\} $. 
\end{proof}

\section{Realtime Quantum Finite Automata with Incremental-Only Counters} \label{wom:IOCs}

We examine the capabilities of RT-QFA-IOCs in both the bounded and unbounded error settings, 
and show that they can simulate a family of conventional counter machines, which are themselves 
superior to RT-QFAs, in both these cases. 

\subsection{Bounded Error} \label{wom:ioc:bounded}

The main theorem to be proven in this subsection is

\begin{theorem}
	\label{wom:thm:RT-Q1BCA-by-RT-QFA-IOC}
	The class of languages recognized with bounded error by RT-QFA-IOCs contains all languages recognized with 
	bounded error by conventional realtime quantum automata with one blind counter (RT-Q1BCAs).
\end{theorem}

Before presenting our proof of Theorem \ref{wom:thm:RT-Q1BCA-by-RT-QFA-IOC}, 
let us demonstrate the underlying idea by showing how RT-QFA-IOCs can simulate a simpler 
family of machines, namely, deterministic automata with one blind counter. 

\begin{lemma}
	\label{wom:lem:RT-D1BCA-by-RT-QFA-IOC}
	If a language $ L $ is recognized by a RT-D1BCA,
	then $ L $ can also be recognized by a RT-QFA-IOC 
	with negative one-sided error bound $ \frac{1}{m} $, for any desired value of $m$.
\end{lemma}
\begin{proof}
We build a RT-QFA-IOC($ m $) that recognizes $L$, 
which is sufficient by Lemma \ref{wom:lem:inc-m-equal-inc-1}. 

	Throughout this proof, the symbol ``$ i $'' is reserved for the imaginary number $ \sqrt{-1} $.
	Let the given RT-D1BCA be $ \mathcal{D} = (Q,\Sigma,\delta,q_{1},Q_{a}) $,
	where $ Q = \{q_{1}, \ldots, q_{n} \} $. We build $ \mathcal{M} = (Q',\Sigma,\Omega,\delta',q_{1,1},Q_{a}') $,
	where
	\begin{itemize}
		\item $ Q'= \{q_{j,1},\ldots,q_{j,n} \mid 1 \le j \le m\} $,
		\item $ Q'_{a} = \{ q_{m,i} \mid q_{i} \in Q_{a} \} $, and 
		\item $ \Omega = \{ \omega_{1}, \ldots, \omega_{n} \} $.
	\end{itemize}
	$ \mathcal{M} $ splits the computation into $m$ paths, i.e. $ \mathsf{path}_{j} $ ($ 1 \le j \le m $),
	with equal amplitude on the left end-marker $ \cent $. That is,
	\begin{equation}
		\delta' (q_{1,1},\cent) =
			\underbrace{\frac{1}{\sqrt{m}}(q_{1,t},0,\omega_{1})}_{\mathsf{path}_{1}} + \cdots +
			\underbrace{\frac{1}{\sqrt{m}}(q_{m,t},0,\omega_{1})}_{\mathsf{path}_{m}},
	\end{equation}
	whenever $ \delta(q_{1},\cent,q_{t},0) = 1 $, where $ 1 \le t \le n $.
	Until reading the right end-marker $ \dollar $, $ \mathsf{path}_{j} $ proceeds in the following way:
	For each  $ \sigma \in \Sigma $ and $  s \in \{1,\ldots,n\} $,
	\begin{eqnarray}
		\label{eq:deterministic-transition}
		\mathsf{path}_{j}: \delta'(q_{j,s},\sigma) & = & (q_{j,t},c_{j},\omega_{s})	
	\end{eqnarray}
	whenever $ \delta(q_{s},\sigma,q_{t},c) = 1 $, where $ 1 \le t \le n $, and
	\begin{itemize}
		\item $ c_{j} = j $ if $ c = 1 $,
		\item $ c_{j} = m-j+1 $ if $ c = -1 $, and
		\item $ c_{j} = 0 $, otherwise.
	\end{itemize} 
	
To paraphrase, each path separately simulates\footnote{Note that each transition of $ \mathcal{M} $ 
in Equation \ref{eq:deterministic-transition} writes a symbol determined by the source state of 
the corresponding transition of  $ \mathcal{D} $ to the register. 
This ensures the orthonormality condition for quantum machines described earlier.}
the computation of $ \mathcal{D} $ on the input string, going through states that correspond 
to the states of $ \mathcal{D} $, and incrementing their counters whenever $ \mathcal{D} $ changes its counter, as follows:
	\begin{itemize}
		\item $ \mathsf{path}_{j} $ increments the counter by $ j $
			whenever $ \mathcal{D} $ increments the counter by 1,
		\item $ \mathsf{path}_{j} $ increments the counter by $ m-j+1 $
			whenever $ \mathcal{D} $ decrements the counter by 1, and
		\item $ \mathsf{path}_{j} $ does not make any incrementation, otherwise.
	\end{itemize}

	On symbol $ \dollar $, the following transitions are executed 
	(note that the counter updates in this last step are also made according to the setup described above):
	\newline
	If $ q_{t} \in Q_{a} $,
	\begin{equation}
		\label{eq:deterministic-accept}
		\mathsf{path}_{j}: \delta' (q_{j,s},\dollar) =
		\frac{1}{\sqrt{m}} \sum_{l=1}^{m} e^{\frac{2 \pi i}{m}jl} (q_{l,t},c_{j},\omega_{s})
	\end{equation}
	and if $ q_{t} \notin Q_{a} $,
	\begin{equation}
		\label{eq:deterministic-reject}
		\mathsf{path}_{j}: \delta' (q_{j,s},\dollar) = (q_{j,t},c_{j},\omega_{s}),
	\end{equation}
	whenever $ \delta(q_{s},\dollar,q_{t},c) = 1 $, where $ 1 \le t \le n $. 
	
The essential idea behind this setup, where different paths increment their counters with different values to represent increments and decrements performed by $ \mathcal{D} $ is that the increment values used by $ \mathcal{M} $ have been selected carefully to ensure that the counter has the same value in all of $ \mathcal{M} $'s paths at any time if $ \mathcal{D} $'s counter is zero at that time. Furthermore, all of $ \mathcal{M} $'s paths are guaranteed to have different counter values if $ \mathcal{D} $'s counter is nonzero\footnote{This idea has been adapted from an algorithm by Kondacs and Watrous for a different type of quantum automaton, whose analysis can be found in \cite{KW97}.}.

\begin{figure}[h]
	\begin{center}
       \fbox{
       \begin{minipage}{0.95\textwidth}
       \footnotesize
               Let $ N > 1 $ be a integer. The $ N $-way quantum Fourier transform (QFT) is the transformation
\begin{equation}
                       \delta(d_{j}) \rightarrow \alpha
                               \sum\limits_{l = 1}^{N} e^{\frac{2 \pi i }{N}jl} (r_{l}), ~~~~ 1 \le j \le N ,
\end{equation}
       from the
               \textit{domain} states  $ d_{1}, \ldots, d_{N} $ to the
\textit{range} states $ r_{1}, \ldots, r_{N} $.
               $ r_{N} $ is the \textit{distinguished} range element. 
               $ \alpha $ is a real number such that $ \alpha^{2}N \leq 1 $.
               The QFT  can be used to check whether separate computational paths
of a quantum program that are in superposition have converged to the
same configuration at a particular step. Assume that the program has
previously split to $N$ paths, each of which have the same amplitude,
and whose state components are the  $ d_{j} $. In all the uses of the
QFT in our algorithms, one of the following conditions is satisfied:
\begin{enumerate}
                       \item The WOM component of the configuration is different in each
of the $N$ paths: In this case, the QFT further divides each path
to $N$ subpaths, that differs from each other by the internal
state component. No interference takes place.
                       \item Each path has the same WOM content at the moment of the QFT:
In this case, the paths that have $ r_{1}, \ldots, r_{N-1} $ as their
state components destructively interfere with each
other \cite{YS09B}, and $ \alpha^{2} N $ of the probability of
the $N$ incoming paths are accumulated on a single resulting path
with that WOM content, and $ r_{N} $ as its state component.
               \end{enumerate}
       \end{minipage}}
       \end{center}
       \caption{The description of $ N $-way quantum Fourier transform used by the machines having a WOM}
       \vskip\baselineskip
       \label{fig:wom:N-way-QFT}
\end{figure}

	For a given input string $ w \in \Sigma^{*} $,
	\begin{enumerate}
		\item if $ \mathcal{D} $ ends up in a state not in $ Q_{a} $ (and so $ w \notin L $), 
			then $ \mathcal{M} $ rejects the input in each of its $m$ paths, and the overall
			rejection probability is 1;
		\item if $ \mathcal{D} $ ends up in a state in $ Q_{a} $, all paths make an
			$ m $-way QFT (see Figure \ref{fig:wom:N-way-QFT}) 
			whose distinguished target is an accepting state:
			\begin{enumerate}
				\item if the counter of $ \mathcal{D} $ is zero (and so $ w \in L $),
					all paths have the same counter value, that is, they interfere with each other, and so $ \mathcal{M} $ accepts with probability 1;
				\item if the counter of $ \mathcal{D} $ is not zero (and so $ w \notin L $), there is no interference, and each path ends by accepting $ w $ with probability $ \frac{1}{m^{2}} $, leading to a total acceptance probability of $ \frac{1}{m} $, and a rejection probability of $ 1 - \frac{1}{m} $.
			\end{enumerate}
	\end{enumerate}
 \end{proof}

\begin{proof}[Proof of Theorem \ref{wom:thm:RT-Q1BCA-by-RT-QFA-IOC}]
	Given	a RT-Q1BCA that recognizes a language $L$ $ \mathcal{M} $ with error bound $ \epsilon < \frac{1}{2} $, we build a RT-QFA-IOC($ m $) $ \mathcal{M'} $, using essentially the same construction as in 
	Lemma \ref{wom:lem:RT-D1BCA-by-RT-QFA-IOC}: $ \mathcal{M'} $ simulates $m$ copies of $ \mathcal{M} $, and these copies use the set of increment sizes described in the proof of Lemma \ref{wom:lem:RT-D1BCA-by-RT-QFA-IOC} to mimic the updates to $ \mathcal{M} $'s counter. Unlike the deterministic machine of that lemma, $ \mathcal{M} $ can fork to multiple computational paths, which is handled by modifying the transformation of Equation \ref{eq:deterministic-transition} as 
	\begin{equation}
		\label{eq:quantum-transition}
		\mathsf{path}_{j}: \delta'(q_{j,s},\sigma,q_{j,t},c_{j},\omega) = \alpha
	\end{equation}
	whenever $ \delta(q_{s},\sigma,q_{t},c,\omega) = \alpha $, where $ 1 \le t \le n $, and
	$ \omega \in \Omega $, and that of Equation \ref{eq:deterministic-accept} as
	\begin{equation}
		\label{eq:quantum-accept}
		\mathsf{path}_{j}: \delta' (q_{j,s},\dollar,q_{l,t},c_{j},\omega) =
		\frac{\alpha}{\sqrt{m}} e^{\frac{2 \pi i}{m}jl},
		\mbox{ for } l \in \{1,\ldots,m\}
	\end{equation}
	whenever $ \delta(q_{s},\dollar,q_{t},c,\omega) = \alpha $, where $ 1 \le t \le n $ and $ \omega \in \Omega $;
causing the corresponding paths of the $m$ copies of $ \mathcal{M} $ to undergo the $m$-way QFTs associated by each accept state as described above at the end of the input.

We therefore have that the paths of $ \mathcal{M} $ that end in non-accept states do the same thing with the same total probability in $ \mathcal{M'} $.  The paths of $ \mathcal{M} $ that end in accept states with the counter containing zero make $ \mathcal{M'} $ accept also with their original total probability, thanks to the QFT. The only mismatch between the machines is in the remaining case of the paths of $ \mathcal{M} $ that end in accept states with a nonzero counter value. As explained in the proof of Lemma \ref{wom:lem:RT-D1BCA-by-RT-QFA-IOC}, each such path contributes $ \frac{1}{m} $ of its probability to acceptance, and the rest to rejection.

	For any given input string $ w \in \Sigma^{*} $:
	\begin{itemize}
		\item If $ w \in L $, we have $ f_{\mathcal{M}}^{a}(w) \geq 1-\epsilon $ 
			and $ f_{\mathcal{M}}^{r}(w) \leq \epsilon $, then 
			\begin{equation}
				f_{\mathcal{M'}}^{a}(w) = f_{\mathcal{M}}^{a}(w) + \frac{1}{m} f_{\mathcal{M}}^{r}(w) 
				\geq 1 -\epsilon .
			\end{equation}
		\item If $ w \notin L $, we have $ f_{\mathcal{M}}^{a}(w) \leq \epsilon $ and 
			$ f_{\mathcal{M}}^{r}(w) \geq 1 - \epsilon $, then 
			\begin{equation}
				f_{\mathcal{M'}}^{a}(w) = f_{\mathcal{M}}^{a}(w) + \frac{1}{m} f_{\mathcal{M}}^{r}(w) 
				\leq  \epsilon + \frac{1}{m}(1-\epsilon).
			\end{equation}
	\end{itemize}
	Therefore, by setting $ m $ to a value greater than $ \frac{2-2\epsilon}{1-2\epsilon} $,
	$ L $ is recognized by $ \mathcal{M'} $ with error bound 
	$ \epsilon' = \epsilon + \frac{1}{m}(1-\epsilon) < \frac{1}{2} $.
	Moreover, by setting $ m $ to sufficiently large values, $ \epsilon' $ can be tuned to be arbitrarily close to $ \epsilon $.
 \end{proof}

\begin{corollary}
	If $ L $ is recognized by a RT-Q1BCA (or a RT-P1BCA) $ \mathcal{P} $ 
	with negative one-sided error bound $ \epsilon < 1 $,
	then $ L $ is recognized by a RT-QFA-IOC $ \mathcal{M} $ with 
	negative one-sided error bound $ \epsilon' $, i.e. $ \epsilon < \epsilon' < 1 $.
	Moreover, $ \epsilon' $ can be tuned to be arbitrarily close to $ \epsilon $.
\end{corollary}

For a given nonnegative integer $ k $, $ L_{eq-k} $ is the language defined over the alphabet 
$ \{a_{1},\ldots,a_{k},b_{1},\ldots,b_{k}\} $ as the set of all strings containing equal
numbers of $ a_{i} $'s and $ b_{i} $'s, for each $ i \in \{1,\ldots,k\} $.

\begin{fact}
	\cite{Fr79}
	For any nonnegative $ k $, 
	$ L_{eq-k} $ can be recognized by a RT-P$ 1 $BCA with negative one-sided bounded error $ \epsilon $,
	where $ \epsilon < \frac{1}{2} $.	
\end{fact}

\begin{corollary}
	RT-QFA-IOCs can recognize some non-context-free languages with bounded error.
\end{corollary}

We have therefore established that realtime quantum finite automata equipped with a WOM tape are more powerful than plain RT-QFAs, even when the WOM in question is restricted to be just a counter.

$ L_{eq-1} $'s complement, which can of course be recognized with positive one-sided bounded error by a 
RT-QFA-IOC by the results above, is a deterministic context-free language (DCFL). 
Using the fact \cite{AGM92} that no nonregular DCFL can be recognized by a nondeterministic 
TM using $ o(\log(n)) $ space, together with 
Lemma \ref{qtm:lem:pm-simulated-by-qm}, we are able to conclude the following.

\begin{corollary}
	\label{cor:QTM-WOMS-superior-PTM-WOMs}
       QTM-WOMs are strictly superior to PTM-WOMs for any space bound $ o(\log(n)) $
       in terms of language recognition with positive one-sided bounded error.
\end{corollary}

\subsection{Unbounded Error} \label{wom:ioc:unbounded}

The simulation method introduced in Lemma \ref{wom:lem:RT-D1BCA-by-RT-QFA-IOC} turns out to be useful in the analysis of the power of increment-only counter machines in the unbounded error mode as well:

\begin{theorem}
	\label{wom:thm:RT-NQ1BCA-by-RT-QFA-IOC-one-sided}
	Any language recognized by a nondeterministic realtime automaton with one blind counter (RT-NQ$ 1 $BCA) is
	recognized by a RT-QFA-IOC with cutpoint $ \frac{1}{2} $.
\end{theorem}
\begin{proof}
	Given a RT-NQ$ 1 $BCA $ \mathcal{N} $, we note that it is just a RT-Q$ 1 $BCA recognizing a language $ L $ with positive one-sided unbounded error \cite{YS10A}, and we can simulate it using the technique described in the proof of Theorem \ref{wom:thm:RT-Q1BCA-by-RT-QFA-IOC}. We set $m$, the number of copies of the RT-Q$ 1 $BCA to be parallelly simulated, to 2. We obtain a RT-QFA-IOC(2) $ \mathcal{M} $ such that
	\begin{enumerate}
		\item paths of $ \mathcal{N} $ that end in an accepting state with the counter equaling zero
			lead  $ \mathcal{M} $ to accept with the same total probability;
		\item paths of $ \mathcal{N} $ that end in an accepting state with a nonzero counter value
			contribute half of their probability to $ \mathcal{M} $'s acceptance probability, with the other half contributing to rejection; and
		\item paths of $ \mathcal{N} $ that end in a reject state
			cause $ \mathcal{M} $ to reject with the same total probability.
	\end{enumerate}
Finally, we modify the transitions on the right end-marker that enter the reject states mentioned in the third case above, so that they are replaced by equiprobable transitions to an (accept,reject) pair of states.
	The resulting machine recognizes $ L $ with ``one-sided'' cutpoint $ \frac{1}{2} $, that is, 
	the overall acceptance probability 
	exceeds $ \frac{1}{2} $ for the members of the language, and equals $ \frac{1}{2} $ for the nonmembers.
 \end{proof}

We now present a simulation of a classical model with non-blind counter.

\begin{theorem}
	\label{wom:1-rev-RT-D1CA-by-RT-QFA-IOC}
	If $ L $ is recognized by a realtime deterministic one-reversal one-counter automaton (1-rev-RT-D1CA), then it is recognized by a RT-QFA-IOC with cutpoint $ \frac{1}{2} $.
\end{theorem}
\begin{proof}
	We assume that the 1-rev-RT-D1CA $ \mathcal{D} = (Q,\Sigma,\delta,q_{1},Q_{a}) $ recognizing $ L $
	is in the following canonical form:
	\begin{itemize}		
		\item the counter value of $ \mathcal{D} $ never becomes nonnegative;
		\item the transition on $ \cent $ does not make any change 
			($ \delta(q_{1},\cent,0,q_{1})=1 $, and $ D_{c}(q_{1})=0 $);
		\item $ Q $ is the union of two disjoint subsets $ Q_{1} $ and $ Q_{2} $, i.e.
			\begin{enumerate}
				\item until the first decrement, the status of the counter is never checked
					-- this part is implemented by the members of $ Q_{1} $,
				\item during the first decrement, the internal state of $ \mathcal{D} $ switches to one of 
					the members of $ Q_{2} $, and
				\item the computation after the first decrement is implemented by the members of $ Q_{2} $;
			\end{enumerate}
		\item once the counter value is detected as zero, the status of the counter is not checked again.
	\end{itemize}	
	We construct a RT-QFA-IOC $ \mathcal{M} = (Q',\Sigma,\Omega,\delta',q_{1},Q'_{a}) $, 
	to recognize $ L $ with cutpoint $ \frac{1}{2} $,
	where 
	\begin{itemize}
		\item $ Q' = \{q_{1}\} \cup \{q_{j,i} \mid j \in \{1,\ldots,4\}, i \in \{1,\ldots,|Q|\}  \}, $
		\item $ Q'_{a} = \{ q_{j,i} \mid j \in \{1,2,3\}, q_{i} \in Q_{a} \} \cup
			\{ q_{4,i} \mid q_{i} \in Q_{r} \} $, and
		\item $ \Omega =  \{  \omega_{i} \cup \omega'_{i} \mid i \in \{1,\ldots,|Q|\} \} $.
	\end{itemize}
and the details of $ \delta' $ are given in Figures \ref{wom:fig:transition-1-rev-RT-D1CA-by-RT-QFA-IOC-1}
	and \ref{wom:fig:transition-1-rev-RT-D1CA-by-RT-QFA-IOC-2}.
	
	\begin{figure}[h!]
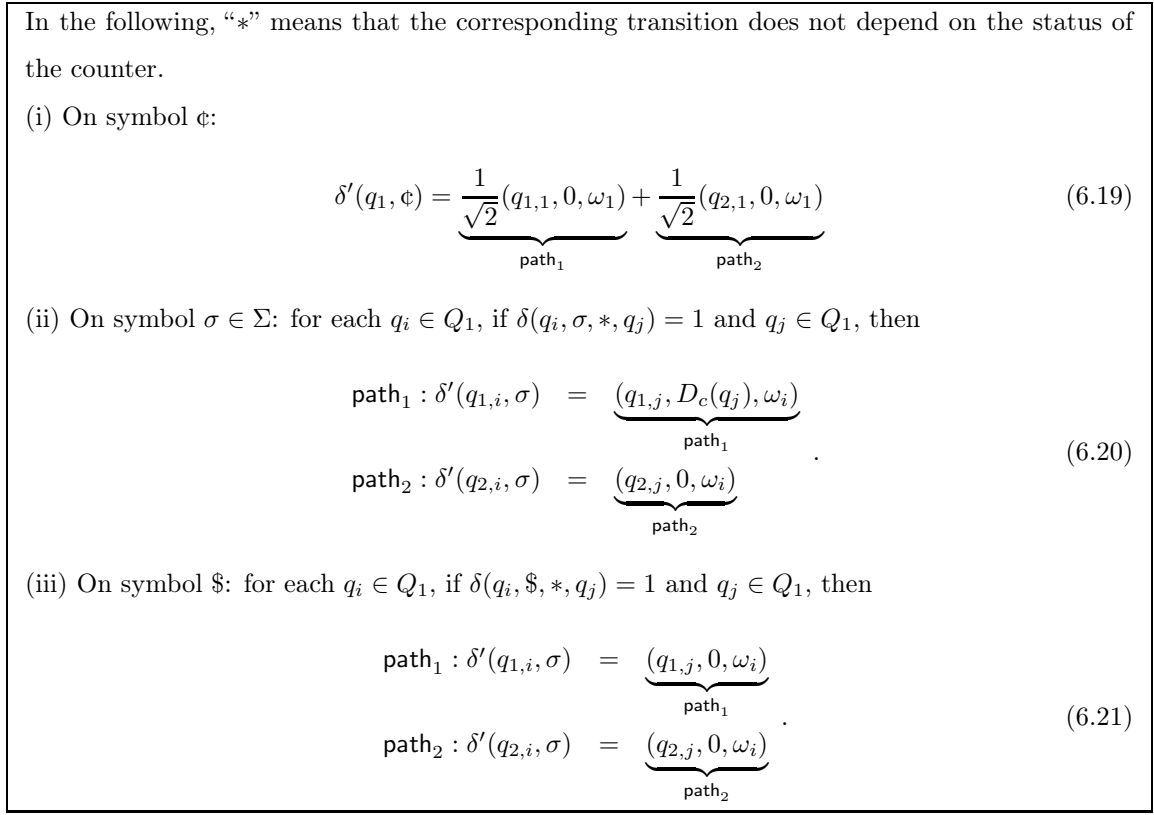

	\begin{center}
	\fbox{
	\begin{minipage}{0.95\textwidth}
		\footnotesize
		In the following, ``$ * $'' means that the corresponding transition does not depend on 
		the status of the counter.	
		\\
		(i) On symbol $ \cent $:
		\begin{equation}
			\delta'(q_{1},\cent) = \underbrace{ \frac{1}{\sqrt{2}} (q_{1,1},0,\omega_{1}) }_{\mathsf{path}_{1}} + 
				\underbrace{ \frac{1}{\sqrt{2}} (q_{2,1},0,\omega_{1}) }_{\mathsf{path}_{2}} 
		\end{equation}
		(ii) On symbol $ \sigma \in \Sigma $: for each $ q_{i} \in Q_{1} $, 
		if $ \delta(q_{i},\sigma,*,q_{j}) = 1 $ and $ q_{j} \in Q_{1} $, then 
		\begin{equation}
			\begin{array}{lcl}
				\mathsf{path}_{1}: \delta'(q_{1,i},\sigma) & = & 
					\underbrace{(q_{1,j},D_{c}(q_{j}),\omega_{i})}_{\mathsf{path}_{1}}
					\\
				\mathsf{path}_{2}: \delta'(q_{2,i},\sigma) & = & 
					\underbrace{(q_{2,j},0,\omega_{i})}_{\mathsf{path}_{2}}		
			\end{array}.
		\end{equation}
		(iii) On symbol $ \dollar $: for each $ q_{i} \in Q_{1} $, 
			if $ \delta(q_{i},\dollar,*,q_{j}) = 1 $ and $ q_{j} \in Q_{1} $, then
		\begin{equation}
			\begin{array}{lcl}
				\mathsf{path}_{1}: \delta'(q_{1,i},\sigma) & = & 
					\underbrace{(q_{1,j},0,\omega_{i})}_{\mathsf{path}_{1}}  \\
				\mathsf{path}_{2}: \delta'(q_{2,i},\sigma) & = & 
					\underbrace{(q_{2,j},0,\omega_{i})}_{\mathsf{path}_{2}}
			\end{array}.
		\end{equation}
		\end{minipage}
		}
		\end{center}
	\caption{The details of the transition function of the RT-QFA-IOC presented in the proof of 
		Theorem \ref{wom:1-rev-RT-D1CA-by-RT-QFA-IOC} (I)}
		\vskip\baselineskip
		\label{wom:fig:transition-1-rev-RT-D1CA-by-RT-QFA-IOC-1}
	\end{figure}
	
	\begin{figure}[h!]
	\begin{center}
	\fbox{
	\begin{minipage}{0.95\textwidth}
		\footnotesize		
		(iv) On symbol $ \sigma \in \Sigma $: for each $ q_{i} \in Q $, 
		if $ \delta(q_{i},\sigma,1,q_{j}) = 1 $ and $ q_{j} \in Q_{2} $, then 
		\begin{equation}
			\begin{array}{lcl}
				\mathsf{path}_{1}: \delta'(q_{1,i},\sigma) & = & 
					\underbrace{\frac{1}{\sqrt{3}}(q_{1,j},0,\omega_{i})}_{\mathsf{path}_{1}} +
					\underbrace{\frac{1}{\sqrt{3}}(q_{3,j},0,\omega'_{i})}_{\mathsf{path}_{3}} + 
					\underbrace{\frac{1}{\sqrt{3}}(q_{4,j},0,\omega'_{i})}_{\mathsf{path}_{4}}  \\
				\mathsf{path}_{2}: \delta'(q_{2,i},\sigma) & = & 
					\underbrace{\frac{1}{\sqrt{3}}(q_{2,j},c_{2},\omega_{i})}_{\mathsf{path}_{2}} +
					\underbrace{\frac{1}{\sqrt{3}}(q_{3,j},c_{2},\omega'_{i})}_{\mathsf{path}_{3}} - 
					\underbrace{\frac{1}{\sqrt{3}}(q_{4,j},c_{2},\omega'_{i})}_{\mathsf{path}_{4}} 
			\end{array},
		\end{equation}
		and if $ \delta(q_{i},\sigma,0,q_{j}) = 1 $, then
		\begin{equation}
			\begin{array}{lcl}
				\mathsf{path}_{3}: \delta'(q_{3,i},\sigma) & = & 
					\underbrace{(q_{3,j},0,\omega_{i})}_{\mathsf{path}_{3}} \\
				\mathsf{path}_{4}: \delta'(q_{4,i},\sigma) & = & 
					\underbrace{(q_{4,j},0,\omega_{i})}_{\mathsf{path}_{4}} \\
			\end{array},
		\end{equation}
		where $ c_{2}=1 $ only if $ D_{c}(q_{j})=-1 $.
		\\
		(iii) On symbol $ \dollar $: for each $ q_{i} \in Q $, if $ \delta(q_{i},\dollar,1,q_{j}) = 1 $
		and $ q_{j} \in Q_{2} $, then
		\begin{equation}
			\begin{array}{lcl}
				\mathsf{path}_{1}: \delta'(q_{1,i},\sigma) & = & 
					\underbrace{\frac{1}{\sqrt{3}}(q_{1,j},0,\omega_{i})}_{\mathsf{path}_{1}} +
					\underbrace{\frac{1}{\sqrt{3}}(q_{3,j},0,\omega_{i})}_{\mathsf{path}_{3}} + 
					\underbrace{\frac{1}{\sqrt{3}}(q_{4,j},0,\omega_{i})}_{\mathsf{path}_{4}}  \\
				\mathsf{path}_{2}: \delta'(q_{2,i},\sigma) & = & 
					\underbrace{\frac{1}{\sqrt{3}}(q_{2,j},c_{2},\omega_{i})}_{\mathsf{path}_{2}} +
					\underbrace{\frac{1}{\sqrt{3}}(q_{3,j},c_{2},\omega_{i})}_{\mathsf{path}_{3}} - 
					\underbrace{\frac{1}{\sqrt{3}}(q_{4,j},c_{2},\omega_{i})}_{\mathsf{path}_{4}} 
			\end{array},
		\end{equation}		
		and if $ \delta(q_{i},\dollar,0,q_{j}) = 1 $, then
		\begin{equation}
			\begin{array}{lcl}
				\mathsf{path}_{3}: \delta'(q_{3,i},\sigma) & = & 
					\underbrace{(q_{3,j},0,\omega_{i})}_{\mathsf{path}_{3}} \\
					\mathsf{path}_{4}: \delta'(q_{4,i},\sigma) & = & 
						\underbrace{(q_{4,j},0,\omega_{i})}_{\mathsf{path}_{4}} \\
				\end{array},
			\end{equation}
			where $ c_{2}=1 $ only if $ D_{c}(q_{j})=-1 $.
		\end{minipage}
		}
		\end{center}
	\caption{The details of the transition function of the RT-QFA-IOC presented in the proof of 
		Theorem \ref{wom:1-rev-RT-D1CA-by-RT-QFA-IOC} (II)}
		\vskip\baselineskip
		\label{wom:fig:transition-1-rev-RT-D1CA-by-RT-QFA-IOC-2}
	\end{figure}
	
	$ \mathcal{M} $ starts by branching to two paths, $ \mathsf{path}_{1} $ and  $ \mathsf{path}_{2} $, 
	with equal amplitude. These paths simulate 
	$ \mathcal{D} $ in parallel according to the specifications in Figure 
	\ref{wom:fig:transition-1-rev-RT-D1CA-by-RT-QFA-IOC-1} until $ \mathcal{D} $ decrements its counter 
	for the first time. From that step on, $ \mathsf{path}_{1} $ and  $ \mathsf{path}_{2} $ split further 
	to create new offshoots (called $ \mathsf{path}_{3} $ and  $ \mathsf{path}_{4} $,) on every symbol 
	until the end of the computation, as seen in Figure \ref{wom:fig:transition-1-rev-RT-D1CA-by-RT-QFA-IOC-2}.
	Throughout the computation, $ \mathsf{path}_{1} $ (resp., $ \mathsf{path}_{2} $) increments its counter 
	whenever $ \mathcal{D} $ is supposed to increment (resp., decrement) its counter. Since $ \mathcal{M} $'s 
	counter is write-only, it has no way of determining which transition $ \mathcal{D} $ makes depending on its 
	counter sign. This problem is solved by assigning different paths of $ \mathcal{M} $ to these branchings of 
	$ \mathcal{D} $:  $ \mathsf{path}_{1} $ and  $ \mathsf{path}_{2} $ (the ``pre-zero paths'') always assume that $ \mathcal{D} $'s counter 
	has not returned to zero yet by being decremented, whereas $ \mathsf{path}_{3} $s and  $ \mathsf{path}_{4} $s (the ``post-zero paths'') 
	carry out their simulations by assuming otherwise. Except for $ \mathsf{path}_{4} $s, all paths imitate $ 
	\mathcal{D} $'s decision at the end of the computation. $ \mathsf{path}_{4} $s, on the other hand, accept if 
	and only if their simulation of $ \mathcal{D} $ rejects the input.

	If $ \mathcal{D} $ never decrements its counter, 
	$ \mathcal{M} $ ends up with the same decision as $ \mathcal{D} $ with probability 1.
	We now focus on the other cases.
As seen in Figure \ref{wom:fig:transition-1-rev-RT-D1CA-by-RT-QFA-IOC-2}, the pre-zero paths lose some of their amplitude on each symbol in this stage by performing a QFT to a new pair of post-zero paths. The outcome of this transformation depends on the status of $ \mathcal{D} $'s counter at this point in the simulation by the pre-zero paths:
	\begin{itemize}
		\item If $ \mathcal{D} $'s counter has not yet returned to zero, then  $ \mathsf{path}_{2} $'s counter has a smaller value than $ \mathsf{path}_{1} $'s counter, and so they cannot interfere via the QFT.
			The newly created post-zero paths contribute equal amounts to the acceptance and rejection probabilities at the end of the computation.
		\item If $ \mathsf{path}_{1} $ and $ \mathsf{path}_{2} $ have the same counter value as a result of this transition, this indicates that $ \mathcal{D} $ has performed exactly as many decrements as its previous increments, and its counter is therefore zero. The paths interfere, the target $ \mathsf{path}_{4} $'s cancel each other, and $ \mathsf{path}_{3} $ survives after the QFT with a probability that is twice
			that of the total probability of the ongoing pre-zero paths. 			
\end{itemize}
	As a result, it is guaranteed that the path that is carrying out the correct simulation of $ \mathcal{D} $ dominates $ \mathcal{M} $'s decision at the end of the computation: If $ \mathcal{D} $'s counter ever returns to zero, the $ \mathsf{path}_{3} $ that is created at the moment of that last decrement has sufficient probability to tip the accept/reject balance. If $ \mathcal{D} $'s counter never returns to zero, then the common decision by the pre-zero paths on the right end-marker determines whether the overall acceptance or the rejection probability is greater than $ \frac{1}{2} $.
 \end{proof}

$ L_{NH} $ is recognizable by both 1-rev-RT-D1CAs and RT-N1BCAs\footnote{RT-N1BCAs can also recognize
$ L_{center} $,
and the languages presented in Figure \ref{uerr:nonstochastic-languages}, 
none of which can be recognized by RT-QFAs with unbounded error.} (and so RT-NQ1BCAs).
It is known  (see Section \ref{uerr:general}) that neither a RT-QFA nor a 
$o(\log(\log(n)))$-space PTM can recognize $ L_{NH} $ with unbounded error. 
We therefore have the following corollary.

\begin{corollary}
	\label{wom:cor:QTM-WOM-superior-PTM-WOM}
       QTM-WOMs are strictly superior to PTM-WOMs for any space bound $ o(\log(\log(n))) $
       in terms of language recognition with unbounded error.
\end{corollary}
\section{Realtime Quantum Finite Automata with Push-Only Stacks} \label{wom:POSs}

We conjecture that allowing more than one nonblank/nonempty symbol in
the WOM tape alphabet of a QFA increases its computational power. We
consider, in particular, the language $ L_{twin}=\{wcw \mid w \in \{a,b\}^{*} \} $:

\begin{theorem}
       \label{wom:thm:Ltwin-by-RT-QFA-POS}
       There exists a RT-QFA-POS that recognizes the language $ L_{twin} $
       with negative one-sided error bound $ \frac{1}{2} $.
\end{theorem}
\begin{proof}
       We construct a RT-QFA-POS $
\mathcal{M}=(Q,\Sigma,\Gamma,\Omega,\delta,q_{1},Q_{a}) $,
       where $ Q=\{q_{1},q_{2},q_{3}, $ $ p_{1},p_{2},p_{3}\} $, $ Q_{a} = \{q_{2}\} $,
       $ \Omega=\{ \omega_{1},\omega_{2} \} $, and $ \Gamma=\{\#,a,b,\varepsilon\} $.
       The transition details are shown in Figure \ref{wom:fiq:Ltwin}.
       
       \begin{figure}[h!]
       \begin{center}
               \fbox{
               \begin{minipage}{0.85\textwidth}
               \footnotesize
               On symbol $ \cent $:
               \begin{equation}
                       \delta(q_{1},\cent) = \underbrace{ \frac{1}{\sqrt{2}} 
                       (q_{1},\varepsilon,\omega_{1}) }_{\mathsf{path}_{1}} +
                               \underbrace{\frac{1}{\sqrt{2}} (p_{1},\varepsilon,\omega_{1})}_{\mathsf{path}_{2}}
               \end{equation}
               On symbols from $ \Sigma $:
               \begin{equation}
                       \mathsf{path}_{1}: \left\{
                               \begin{array}{lcl}
                                       \delta(q_{1},a) & = & (q_{1},a,\omega_{1}) \\
                                       \delta(q_{2},a) & = & (q_{2},\varepsilon,\omega_{1}) \\
                                       \delta(q_{1},b) & = & (q_{1},b,\omega_{1}) \\
                                       \delta(q_{2},b) & = & (q_{2},\varepsilon,\omega_{1}) \\
                                       \delta(q_{1},c) & = & (q_{2},\varepsilon,\omega_{1}) \\
                                       \delta(q_{2},c) & = & (q_{3},\varepsilon,\omega_{1}) \\
                                       \delta(q_{3},a) & = & (q_{3},\varepsilon,\omega_{2}) \\
                                       \delta(q_{3},b) & = & (q_{3},\varepsilon,\omega_{2}) \\
                                       \delta(q_{3},c) & = & (q_{3},\varepsilon,\omega_{2})
                               \end{array}
                       \right.
               \end{equation}
               \begin{equation}
                       \mathsf{path}_{2}: \left\{
                               \begin{array}{lcl}
                                       \delta(p_{1},a) & = & (p_{1},\varepsilon,\omega_{1}) \\
                                       \delta(p_{2},a) & = & (p_{2},a,\omega_{1}) \\
                                       \delta(p_{1},b) & = & (p_{1},\varepsilon,\omega_{1}) \\
                                       \delta(p_{2},b) & = & (p_{2},b,\omega_{1}) \\
                                       \delta(p_{1},c) & = & (p_{2},\varepsilon,\omega_{1}) \\
                                       \delta(p_{2},c) & = & (p_{3},\varepsilon,\omega_{1}) \\
                                       \delta(p_{3},a) & = & (p_{3},\varepsilon,\omega_{2}) \\
                                       \delta(p_{3},b) & = & (p_{3},\varepsilon,\omega_{2}) \\
                                       \delta(p_{3},c) & = & (p_{3},\varepsilon,\omega_{2})                                                                              
                               \end{array}
                       \right.
               \end{equation}
               On symbol $ \dollar $:
               \begin{equation}
                       \mathsf{path}_{1}: \left\{
                               \begin{array}{lcl}
                                       \delta(q_{1},\dollar) & = & (q_{1},\varepsilon,\omega_{1}) \\
                                       \delta(q_{2},\dollar) & = & \frac{1}{\sqrt{2}} (q_{2},\varepsilon,\omega_{1}) +
                                               \frac{1}{\sqrt{2}} (q_{3},\varepsilon,\omega_{2}) \\
                                       \delta(q_{3},\dollar) & = & (q_{3},\varepsilon,\omega_{1})
                               \end{array}
                       \right.
               \end{equation}
               \begin{equation}
                       \mathsf{path}_{2}: \left\{
                               \begin{array}{lcl}
                                       \delta(p_{1},\dollar) & = & (p_{1},\varepsilon,\omega_{1}) \\
                                       \delta(p_{2},\dollar) & = & \frac{1}{\sqrt{2}} (q_{2},\varepsilon,\omega_{1}) -
                                               \frac{1}{\sqrt{2}} (q_{3},\varepsilon,\omega_{2}) \\
                                        \delta(p_{3},\dollar) & = & (q_{3},\varepsilon,\omega_{1})
                               \end{array}
                       \right.
               \end{equation}
       \end{minipage}
       }
       \end{center}
               \caption{The transitions of the RT-QFA-POS of Theorem \ref{wom:thm:Ltwin-by-RT-QFA-POS}}
	\vskip\baselineskip
       \label{wom:fiq:Ltwin}
       \end{figure}

       \begin{enumerate}
               \item The computation splits into two paths, $ \mathsf{path}_{1} $
and $ \mathsf{path}_{2} $, with equal
                       amplitude at the beginning.
               \item $ \mathsf{path}_{1} $ (resp., $ \mathsf{path}_{2} $) scans the
input, and
                       copies $ w_{1} $ (resp., $ w_{2} $) to the POS if the input is of
the form $ w_{1} c w_{2} $,
                       where $ w_{1},w_{2} \in \{a,b\}^{*} $.
                       \begin{enumerate}
                               \item If the input is not of the form $ w_{1} c w_{2} $, both paths reject.
                               \item Otherwise, $ \mathsf{path}_{1} $
                                       and $ \mathsf{path}_{2} $ perform a QFT at the end of the
computation, where the distinguished range element is an accept state.
               \end{enumerate}
\end{enumerate}
The configurations at the ends of $ \mathsf{path}_{1} $ and $
\mathsf{path}_{2} $ interfere with each other,
i.e., the machine accepts with probability $ 1 $, if and only if the
input is of the form
$ w c w $, $ w \in \{a,b\}^{*} $.
Otherwise, each of $ \mathsf{path}_{1} $ and $ \mathsf{path}_{2} $
contributes at most $ \frac{1}{4} $ to the overall acceptance probability,
and the machine accepts with probability at most $ \frac{1}{2} $.
\end{proof}

\begin{lemma}
	\label{wom:lem:no-ptm-wom-L-twin}
        No PTM (or PTM-WOM) using $ o(\log(n)) $ space can recognize $
L_{twin} $ with bounded error.
\end{lemma}
\begin{proof}
       Any PTM using $ o(\log(n)) $ space to recognize $ L_{twin} $ with bounded error
       can be used to construct a PTM recognizing the palindrome language
       $ L_{pal} $ with bounded error using the same amount of space.
       (One would only need to modify the $ L_{twin} $ machine to treat the
right end-marker on the tape as
       the symbol $ c $, and switch its head direction when it attempts to
go past that symbol.)
       It is however known \cite{FK94} that no PTM using $ o(\log(n)) $
space can recognize $ L_{pal} $
       with bounded error.
\end{proof}

We are now able to state a stronger form of Corollary \ref{cor:QTM-WOMS-superior-PTM-WOMs},
which referred only to one-sided error:

\begin{corollary}
	\label{cor:QTM-WOMS-superior-PTM-WOMs-bounded-error}
       QTM-WOMs are strictly superior to PTM-WOMs for any space bound $ o(log(n)) $
       in terms of language recognition with bounded error.
\end{corollary}

\begin{openproblem}
       Can a one-way probabilistic pushdown automaton recognize $ L_{twin} $ with
bounded error?
\end{openproblem}

\section{Realtime Quantum Finite Automata with Write-Only Memories} \label{wom:WOMs}

In this section, we present a bounded-error RT-QFA-WOM that recognizes a language for which we currently do not know a RT-QFA-POS algorithm, namely,
\begin{equation}
	L_{rev} = \{wcw^{\rev}\ \mid w \in \{a,b\}^{*} \},
\end{equation}
where $w^{\rev}$ is the reverse of string $w$.
Note that this language can also be recognized by a deterministic pushdown automaton.

\begin{theorem}
	There exists a RT-QFA-WOM that recognizes $ L_{rev} $
	with negative one-sided error bound $ \frac{1}{2} $.
\end{theorem}
\begin{proof}
	(Sketch)
	We use almost the same technique presented in the proof of Theorem \ref{wom:thm:Ltwin-by-RT-QFA-POS}.
	The computation is split into two paths ($ \mathsf{path}_{1} $ and $ \mathsf{path}_{2} $) 
	with equal amplitude at the beginning of the computation.
	Each path checks whether the input string is of the form $ w_{1}cw_{2} $, where $ w_{1},w_{2} \in \{a,b\}^{*} $
	and rejects with probability 1 if it is not.
	We assume that the input string is of the form $ w_{1}cw_{2} $ in the rest of this proof.
	Until the $ c $ is read, $ \mathsf{path}_{1} $ copies $ w_{1} $ to the WOM tape, and $ \mathsf{path_{2}} $
	just moves the WOM tape head one square to the right at each step, without writing anything.
	After reading the $ c $, the direction of the WOM tape head is reversed in both paths.
	That is, $ \mathsf{path}_{1} $ moves the WOM tape head one square to the left at each step, without writing 
	anything, while $ \mathsf{path_{2}} $ writes $ w_{2} $ in the reverse direction (from the right to the left) 
	on the WOM tape.
	When the right end-marker is read,
	the paths make a QFT, as in the proof of Theorem \ref{wom:thm:Ltwin-by-RT-QFA-POS}.
	It is easy to see that the two paths interfere if and only if $ w_{1} = w_{2}^{\rev} $, and the input string is
	accepted with probability 1 if it is a member of $ L_{rev} $, and with probability $ \frac{1}{2} $ otherwise.
\end{proof}

By an argument similar to the one used in the proof of Lemma \ref{wom:lem:no-ptm-wom-L-twin}, 
$ L_{rev} $ can not be recognized with bounded error by any PTM using $ o(\log(n)) $ space,
since the existence of any such machine would lead to a PTM that recognizes the palindrome
language using the same amount of space.

\begin{openproblem}
	Can a RT-QFA-POS recognize $ L_{rev} $ with bounded error?
\end{openproblem}

It is easy to see that a WOM of constant size adds no power to a conventional machine. 
All the algorithms we considered until now used $ \Omega(n) $ squares of the WOM tape on worst-case inputs. 
What is the minimum amount of WOM that is required by a QFA-WOM recognizing a nonregular language? 
Somewhat less ambitiously, one can ask whether there is any nonregular language recognized 
by a RT-QFA-WOM with sublinear space. We answer this question positively for middle-space usage, 
that is, when we are only concerned with the space used by the machine when the input is a member of the language.

Let $ (i)^{\rev}_{2} $ be the reverse of the binary representation of $ i \in \mathbb{N} $. Consider the language
\begin{equation}
	L_{rev-bins}=\{ a (0)_{2}^{\rev} a (1)_{2}^{\rev} a \cdots a (k)_{2}^{\rev} a \mid k \in \mathbb{Z}^{+} \}.
\end{equation}
\begin{theorem}
	\label{wom:thm:L-rev-bins-RT-QFA-WOM}
	$ L_{rev-bins} $ can be recognized by a RT-QFA-WOM $ \mathcal{M} $ with negative one-sided error bound 
	$ \frac{3}{4} $, and the WOM usage of $ \mathcal{M} $ for the members of $ L_{rev-bins} $ is $ O(\log n) $,
	where $ n $ is the length of the input string.
\end{theorem}
\begin{proof}
	It is not hard to modify the RT-QFA-POS recognizing $ L_{twin} $ to obtain a new RT-QFA-POS, say $ \mathcal{M}' $, 
	in order to recognize language
	$ L_{twin'}=\{ (i)^{\rev}_{2} a (i+1)_{2}^{\rev} \mid i \geq 0 \} $ 
	with negative one-sided error bound $ \frac{1}{2} $.
	Our construction of $ \mathcal{M} $ is based on $ \mathcal{M'} $.
	The main idea is to use $ \mathcal{M}' $ in a loop in order to check the consecutive
	blocks of $ \{0,1\}^{+}a\{0,1\}^{+} $ between two $ a $'s. In each iteration,
	the WOM tape head reverses direction, and so the previously used space can be used again and again.
	Note that, whenever $ \mathcal{M}' $ executes a rejecting transition, $ \mathcal{M} $ enters a path which  rejects the input when it arrives at the right end-marker, and whenever $ \mathcal{M}' $ is supposed to execute an accepting transition 
	(except at the end of the computation), $ \mathcal{M} $ enters the next iteration.
	At the end of the input, the input is accepted by $ \mathcal{M} $
	if $ \mathcal{M}' $ accepts in its last iteration.
	
	Let $ w  $ be an input string.
	We assume that $ w $ is of the form
	\begin{equation}
		a \{0,1\}^{+} a \{0,1\}^{+} a \cdots a \{0,1\}^{+}a .
	\end{equation} 
	(Otherwise, it is rejected with probability 1.)
	At the beginning, the computation is split equiprobably into two branches,
	$ \mathsf{branch}_{1} $ and $ \mathsf{branch}_{2} $.
	(These never interfere with each other.)
	$ \mathsf{branch}_{1} $ (resp., $ \mathsf{branch}_{2} $) enters the block-checking loop after reading the
	first (resp., the second) $ a $. Thus, at the end of the computation, 
	one of the branches is in the middle of an iteration, and the other one has just finished its final iteration.
	The branch whose iteration is interrupted by reading the end-marker accepts with probability 1.
	
	If $ w \in L_{rev-bins} $, neither branch enters a reject state, and the input is accepted with probability 1. 
	On the other hand, if $ w \notin L_{rev-bins}  $, there must be at least one block $ \{0,1\}^{+}a\{0,1\}^{+} $
	that is not a member of $ L_{twin'} $, and so the input is rejected with probability $ \frac{1}{2} $
	in one branch. 
	Therefore, the overall accepting probability can be at most $ \frac{3}{4} $.
	
	It is easy to see that the WOM usage of this algorithm for members of $ L_{rev-bins} $ is $ O(\log n) $.
\end{proof}

\chapter{SUBLINEAR-SPACE REALTIME TURING MACHINES} \label{rtm}

In this chapter, we give space lower bounds of realtime classical Turing machines recognizing nonregular languages. 
In fact, we validate the lower bounds of one-way machines for realtime machines.
We refer the reader to \cite{Sz94} for a detailed background.

By combining previous results with ours, we can obtain Figure \ref{conc:fig:lower-bounds},
in which the lower bounds following from our results are presented in bold.
Also note that the slots containing symbol ``?'' in the figure are still open.

\begin{table}[h!]	
	\vskip\baselineskip
	\caption{Lower bounds of \{D,N,A\}TMs in order to recognize a nonregular language}
	\footnotesize
	\begin{center}
	\begin{tabular}{|l|l|l|l|l|l|l|}
	 	\hline
	 	& \multicolumn{3}{c|}{general case} & \multicolumn{3}{c|}{unary case} 
	 	\\ \hline  	
	 	& Strong & Middle & Weak & Strong & Middle & Weak 
	 	\\ \hline
	 	1DTM & $ \log n $ & $ \log n $ & $ \log n $ & $ \log n $ & $ \log n $ & $ \log n $
	 	\\ \hline
	 	1NTM & $ \log n $ & $ \log n $ & $ \log \log n $ & $ \log n $ & $ \log n $ & $ \log \log n $
	 	\\ \hline
	 	1ATM & $ \log n $ & $ \log \log n $ & $ \log \log n $ & $ \log n $ & $ \log n $ & $ \log \log n $
	 	\\ \hline
	 	RT-DTM & $ \mathbf{log} \mspace{2mu} n $ & $ \mathbf{log} \mspace{2mu} n $ &
	 		$ \mathbf{log} \mspace{2mu} n $ & $ \mathbf{log} \mspace{2mu} n $ &
	 		$ \mathbf{log} \mspace{2mu} n $ & $ \mathbf{log} \mspace{2mu} n $
	 	\\ \hline
	 	RT-NTM & $ \mathbf{log} \mspace{2mu} n $ & $ \mathbf{log} \mspace{2mu} n $ &
	 		$ \mathbf{log} \mspace{2mu} \mathbf{log} \mspace{2mu} n $ & $ \mathbf{log} \mspace{2mu} n $ &
	 		$ \mathbf{log} \mspace{2mu} n $ & ?
	 	\\ \hline
	 	RT-ATM & $ \mathbf{log} \mspace{2mu} n $ & $ \mathbf{log} \mspace{2mu} \mathbf{log} \mspace{2mu} n $ &
	 		$ \mathbf{log} \mspace{2mu} \mathbf{log} \mspace{2mu} n $ & $ \mathbf{log} \mspace{2mu} n $ &
	 		$ \mathbf{log} \mspace{2mu} n $ & ?
	 	\\ \hline 
	 \end{tabular}
	 \end{center}
	 \label{conc:fig:lower-bounds}
\end{table}

Let $ h_{\kappa} $ be a homomorphism such that 
\begin{itemize}
	\item $ h_{\kappa}(x)=x $ if $ x \neq \kappa $ and
	\item $ h_{\kappa}(\kappa)= \varepsilon $;
\end{itemize}
$ \mathcal{D} $ be the RT-DTM running in logarithmic space and
$ L_{upal\mbox{-}\kappa} $ satisfying $ h_{\kappa}(L_{upal\mbox{-}\kappa}) = L_{upal} $ 
be the language recognized by it;
$ \Sigma=\{a,b,\kappa\} $ and $ \Gamma=\{ \LofC ,0,1,\RofC,\# \} $.

For a given input $ w \in \Sigma^{*} $, the specifications of $ \mathcal{D} $ are as follows:
\begin{enumerate}
	\item During the computation, parallel to its main tasks, which is described in the following items,
		$ \mathcal{D} $ checks whether $ h_{\kappa}(w) $ is of the form $ a^{*}b^{*} $.
		If not, $ w $ is rejected.
	\item By reading three $ \kappa $'s, $ \mathcal{D} $ moves the work tape head three squares to the right.
		By reading three more $ \kappa $'s, $ \mathcal{D} $ consecutively writes symbols 
		$ \RofC $, $ 0 $, and $ \LofC $ on the work tape
		in the reverse direction. If any of the first six symbols is different than $ \kappa $,
		then the input is rejected.
		After reading six consecutive $ \kappa $'s,
		$ \mathcal{D} $ does nothing until reading a symbol different than $ \kappa $.
	\item Now, the content of the work tape is 
		``$ \LofC  0 \RofC \# \# \cdots  $'' and the work tape head is positioned on symbol $ \LofC $.
		Our aim is to keep a counter, used to compare the number of $ a $'s and $ b $'s, 
		in reverse direction between symbols $ \LofC $ and $ \RofC $.
	\item Let $ \mathsf{P}^{+} $ (resp., $ \mathsf{P}^{-} $) 
		be a deterministic procedure having a finite instruction set,
		that starts its movements when the work tape head on symbol $ \LofC $ 
		and then increases (resp., decreases) 
		the counter by one and finishes its movement by leaving the work tape head again on symbol $ \LofC $.
		The number of steps required by $ \mathsf{P}^{+} $ (resp., $ \mathsf{P}^{-} $) 
		depends on the current content of the counter.
		Let $ p^{+} $ (resp., $ p^{-} $) $: \{ \LofC (0 \cup 1)^{*} \RofC \} \rightarrow \mathbb{N} $ 
		be the function representing this number.
		Additionally, $ \mathsf{P}^{-} $ always checks whether the value of the counter becomes 0 or not.
	\item After reading an $ a $ (resp., a $ b $), 
		$ \mathcal{D} $ expects to read at least $ p^{+}(\LofC n \RofC ) $ 
		(resp., $ p^{-}(\LofC n \RofC ) $) $ \kappa $'s
		before reading a symbol different than $ \kappa $
		in order to complete the task of $ \mathsf{P}^{+} $ (resp., $ \mathsf{P}^{-} $) successfully,
		where $ n $ is the value of the counter before updating.
		That is, $ \mathcal{D} $ executes each step of $ \mathsf{P}^{+} $ (resp., $ \mathsf{P}^{-} $) 
		by reading a single $ \kappa $.
		If $ \mathcal{D} $ reads a symbol different than $ \kappa $ before finishing the task of 
		$ \mathsf{P}^{+} $ (resp., $ \mathsf{P}^{-} $), the input is rejected.
		After completing the task of $ \mathsf{P}^{+} $ (resp., $ \mathsf{P}^{-} $),
		$ \mathcal{D} $ does nothing until reading a symbol different than $ \kappa $.
	\item Whenever the counter becomes 0, $ \mathcal{D} $ expects to read $ \kappa^{*}\dollar $.
		If so, the input is accepted. Otherwise, the input is rejected.
\end{enumerate}

It can easily be verified that the space used by $ \mathcal{D} $ is $ O(log(|w|)) $. 
Assume that $ L_{upal\mbox{-}\kappa} $ is regular. Then, $ L_{upal} $ must be regular, too, since regular languages
are closed under homomorphism. We arrive at a contradiction. Therefore, $ L_{upal\mbox{-}\kappa} $ 
is a nonregular language.

Now, we present a unary nonregular language, say $ L_{power\mbox{-}\kappa} \subset \{\kappa\}^{*} $, 
recognized by a RT-DTM, say $ \mathcal{D} $, in logarithmic space.
The idea behind is similar to the previous algorithm.

For a given input $ w \in \Sigma^{*} $, the specifications of $ \mathcal{D} $ are as follows:
\begin{enumerate}
	\item By reading exactly six consecutive $ \kappa $'s, the counter, again kept in reverse direction, 
		is prepared as $ \LofC 0 \RofC $. If there are fewer leading $ \kappa $'s, 
		then the input is rejected.
		(If the seventh symbol is $ \dollar $, then the input is accepted.)
		The work tape head is placed on symbol $ \LofC $.
	\item By reading a block of $ \kappa $'s,
		\begin{itemize}
			\item the work tape head goes from symbol $ \LofC $ to symbol $ \RofC $ 
				while incrementing the counter by 1, and then,
			\item the work tape head goes from symbol $ \RofC $ to symbol $ \LofC $ and while checking 
				whether the counter value is a power of 2 or not.
		\end{itemize}
		During the travel of the work tape head, we assume that it cannot be stationary on any symbol.		
		Note that, if required, the place of symbol $ \RofC $ on the work tape is shifted 
		one square to the right.
		If $ \mathcal{D} $ reads $ \dollar $ before the work tape head is placed on symbol $ \LofC $,
		the input is rejected.
	\item When the work tape head is placed on symbol $ \LofC $ and 
		the current scanned symbol from the input tape is $ \dollar $,
		the input is accepted if the counter is a power of 2 and rejected otherwise.
\end{enumerate}

The members of $ L_{power\mbox{-}\kappa} $ can be enumerated 
as shorter strings come first.
Let $ a_{i} $ be the $ i^{th} $ element of $ L_{power\mbox{-}\kappa} $, where $ i > 0 $.
It can be verified easily that the value of $ | a_{i+1} | - | a_{i} | $ increases when $ i $ gets bigger.
Therefore, for any pumping length $ p $, we get a nonmember by pumping $ a_{i} $ once
when $ |a_{i}| > p $ and $ | a_{i+1} | - | a_{i} | > p $.
We can conclude that $ L_{power\mbox{-}\kappa} $ is a nonregular language.

We assume in the following parts that all 1TMs never move their work tape heads to the square indexed by 0.
Note that, any arbitrary 1\textbf{X}TM can be converted to a 1\textbf{X}TM having the restriction above
without lose of any space resource, where \textbf{X} $ \in \{ \mbox{D, N, A, P} \} $.
Additionally, we assume that $ \Sigma $ does not contain symbol $ \kappa $ and 
$ \Sigma_{\kappa} = \Sigma \cup \{ \kappa \} $.

For a given 1TM, a \textit{stationary-transition} is a transition in which the 1TM does not move 
its input head to the right and a \textit{to-right-transition} is a transition in which
the 1TM moves its input head to the right.

\begin{definition}	
	For a given 1DTM (resp., 1NTM or 1ATM) $ \mathcal{D} $,
	$ \mathcal{D}_{\kappa} $ is a RT-DTM (resp., RT-NTM or RT-ATM) implementing 
	the instructions of $ \mathcal{D} $ with the following specifications:
	\begin{itemize}
		\item on each non-$ \kappa $ symbol, i.e. a symbol different than $ \kappa $,
			$ \mathcal{D}_{\kappa} $ behaves as if $ \mathcal{D} $ scans this symbol;
		\item $ \mathcal{D}_{\kappa} $ requires to see a $ \kappa $ on the input tape 
			after implementing a stationary-transition of $ \mathcal{D} $; and,
			if this is not the case, called \textit{missing-$ \kappa $ case},
			$ \mathcal{D}_{\kappa} $ rejects the input;
		\item on each $ \kappa $, $ \mathcal{D}_{\kappa} $ behaves as if $ \mathcal{D} $ scans the symbol that
			is the last non-$ \kappa $ symbol scanned by $ \mathcal{D}_{\kappa} $; 
		\item after implementing a to-right-transition of $ \mathcal{D} $, $ \mathcal{D}_{\kappa} $
			waits to scan a non-$ \kappa $ symbol and so discards all $ \kappa $'s until seeing a non-$ \kappa $
			symbol.
	\end{itemize}
	At the end, $ \mathcal{D}_{\kappa} $ follows the decision of $ \mathcal{D} $ unless the input is rejected
	due to missing-$ \kappa $ case.
\end{definition}

\begin{definition}
	For a given 1DTM (resp., 1NTM or 1ATM) $ \mathcal{D} $,
	if $ L $ is the language recognized by $ \mathcal{D} $, then,
	$ L_{\mbox{-}\kappa} \subset \Sigma_{\kappa}^{*} $ is the language recognized by $ \mathcal{D}_{\kappa} $.
\end{definition}
\begin{remark}
	$ h_{\kappa} ( L_{\mbox{-}\kappa} ) = L $.
\end{remark}

\begin{remark}
	If $ L $ is not a member of class $ \mathcal{C} $ that is closed under homomorphism,
	then $ L_{\mbox{-}\kappa} $ is not a member of $ \mathcal{C} $, too.
\end{remark}

\begin{theorem}
	If $ L $ is recognized by a 1DTM (resp., a 1NTM or a 1ATM) $ \mathcal{D} $ in strong-space $ s(n) $
	and there exists a nondecreasing function $ t(n) $ such that $ s(n) \in O(t(n)) $,
	then $ L_{\mbox{-}\kappa} $ is recognized by $ \mathcal{D}_{\kappa} $ in strong-space $ O(t(n)) $.
\end{theorem}
\begin{proof}	
	For any $ w \in \Sigma^{*} $, the space used by $ \mathcal{D}_{\kappa} $ on any input which is a 
	member of $ \{ w_{\kappa} \in \Sigma_{\kappa}^{*} \mid h_{\kappa} ( w_{\kappa} ) = w \} $
	is at most equal to the space used by $ \mathcal{D} $ on $ w $.
\end{proof}

\begin{theorem}
	If $ L $ is recognized by a 1DTM (resp., a 1NTM or a 1ATM) $ \mathcal{D} $ in middle-space $ s(n) $
	and there exists a nondecreasing function $ t(n) $ such that $ s(n) \in O(t(n)) $,
	then $ L_{\mbox{-}\kappa} $ is recognized by $ \mathcal{D}_{\kappa} $ in middle-space $ O(t(n)) $.
\end{theorem}
\begin{proof}
	Similar to the previous one.
\end{proof}

Let $ L_{gcm} = \{ a^{n} b^{M} \mid M \mbox{ is a common multiple of all } i \leq n \} $.
Szepietowski \cite{Sz88} showed that $ L_{gcm} $ is in middle-one-way-ASPACE($ \log \log (n) $).
Let $ \mathcal{A} $ be the 1ATM presented by Szepietowski.

\begin{corollary}
	$ L_{gcm \mbox{-}\kappa} $ is a nonregular language and 
	recognized by RT-ATM $ \mathcal{A}_{\kappa} $ in middle-space $ O(\log\log(n)) $.
\end{corollary}

$ L_{m \neq n} = \{ a^{m}b^{n} \mid m \neq n \} $ is in weak-one-way-NSPACE($ \log \log (n) $) \cite{Sz94}.
Let $ \mathcal{N} $ be the 1NTM, presented on Page 23 in \cite{Sz94}, recognizing 
$ L_{m \neq n} $ in weak-space $ O(\log \log (n)) $.
We give a brief description of $ \mathcal{N} $ below.
At the beginning of the computation, $ \mathcal{N} $ nondeterministically selects a number $ l > 1 $.
The work tape of $ \mathcal{N} $ is divided into three tracks so that
the top track can store the value of $ l $, which is written at the beginning of the computation and
the middle track (resp., bottom track) can store the value of $ m \mod(l) $ (resp., $ n \mod(l) $), 
which is iteratively calculated whenever an $ a $ (resp., a $ b $) is read.
When symbol $ \dollar $ is read, the values of the middle and bottom tracks are compared.
If they are not equal, the input is accepted.
The nondeterministic path corresponding to $ l $, namely $ \mathsf{npath}_{l} $,
needs only $ \Theta(\log(l)) $ space.
We can conclude that $ L_{m \neq n} $ is recognized by 
$ \mathcal{N} $ in weak-space($ \log\log(n) $), by using the following fact.

\begin{fact}
	(On Page 22 of \cite{Sz94}) There exists a constant $ c $ such that for every pair of natural numbers $ m $
	and $ n $, $ m \neq n $, there exists a number $ l < c\log(m+n) $ such that $ m \not\equiv n \mod(l) $.
\end{fact}

In the above algorithm, after reading a symbol on the input tape,
$ \mathcal{N} $ implements some operations on the work tape by making stationary steps on the input tape.
The number of such stationary steps can be fixed to a value, say $ d_{l} $, 
depending on only the number written at the top track of the work tape ($ l > 1 $).
More specifically, we can use the following strategy:
\begin{itemize}
	\item the work tape head is always placed on the leftmost nonblank symbol before reading the next symbol
		on the input tape and
	\item the operations on the work tape can be completed by alternating the work tape head,
		which operates with the speed of one square per step,
		between the leftmost and rightmost nonblank symbols $ 2 d_{1} > 0 $ times.
\end{itemize}
By selecting $ d_{1} $ sufficiently bigger,
$ d_{l} $ can be set to $ 2d_{1} \lceil \log l \rceil + d_{2} $ for nondeterministic path $ \mathsf{npath}_{l} $,
where $ d_{2} \in \mathbb{Z} $ depends on how to store the numbers on the work tape.
Thus, we can guarantee that for any $ l' > l $,
the number of the stationary steps (on the input tape) 
of $ \mathsf{npath}_{l'} $ cannot be less than that of $ \mathsf{npath}_{l} $.
Let $ \mathcal{N}' $ be the 1NTM implementing the above algorithm
and $ L_{m \neq n\mbox{-}\kappa} $ be the language recognized by $ \mathcal{N}_{\kappa}' $. 
\begin{lemma}
	$ L_{m \neq n \mbox{-}\kappa} $ is a nonregular language and is 
	recognized by RT-NTM $ \mathcal{N}'_{\kappa} $ in weak-space $ O(\log\log(n)) $.
\end{lemma}
\begin{proof}
	$ L_{m \neq n \mbox{-}\kappa} $ is a nonregular language due to the fact that
	$ L_{m \neq n} $ is nonregular and REG is closed under homomorphism.
	
	Let $ w = a^{m}b^{n} $ be a member of $ L_{m \neq n} $ and $ l $ be the smallest
	number satisfying that $ m \not\equiv n \mod (l) $. 
	Then, for all $ w_{\kappa} \in L_{m \neq n \mbox{-}\kappa} $ satisfying $ h_{\kappa} (w_{\kappa}) = w $,
	the nondeterministic path of $ \mathcal{N}_{\kappa}' $ corresponding to $ l $ always accepts the computation.
\end{proof}

\chapter{RELATED WORK} \label{rwork}

In this chapter, we present many results not classified as a separate section.
In most cases, we give the sketch of the proofs.

\section{Counter and Pushdown Automata} \label{rwork:ca}

Let $ L_{ijk} = \{ a^{i}b^{j}c^{k} \mid i \neq j, i \neq k, j \neq k, 0 \leq i,j,k \} $ be
a language over alphabet $ \Sigma=\{a,b,c\} $.

\begin{theorem}
	$ L_{ijk} $ is in S$ ^{\neq}_{\mathbb{Q}} $ (NQAL).
\end{theorem}
\begin{proof}
	(Sketch)
	Let $ w $ be an input string of the form $ a^{*}b^{*}c^{*} $.
	We can design a GFA to calculate the value of $ (|w|_{a}-|w|_{b}) $.
	Thus, we can also construct a GFA to calculate the value of 
	\begin{equation}
		(|w|_{a}-|w|_{b})^{2} (|w|_{a}-|w|_{c})^{2} (|w|_{b}-|w|_{c})^{2}.
	\end{equation}
	This value is a positive integer if $ w $ is a member of $ L_{ijk} $ and
	it is zero if $ w $ is not a member of $ L_{ijk} $.
\end{proof}
Since $ L_{ijk} $ is also a member of RT-NQSTACK(1) and not in CFL, we can obtain the following corollary.

\begin{corollary} For any $ s \in O(n) $,
	\begin{equation}
		\mbox{one-way-NSTACK($ s $)} \subsetneq \mbox{one-way-NQSTACK($ s $)}
	\end{equation}
	and
	\begin{equation}
		\mbox{RT-NSTACK($ s $)} \subsetneq \mbox{RT-NQSTACK($ s $)}.
	\end{equation}
\end{corollary}

By using Theorem \ref{wom:thm:L-rev-bins-RT-QFA-WOM}, we have the following result.
\begin{corollary}
	$ L_{rev-bins} \in $ middle-RT-BQSTACK($ \log n $).
\end{corollary}
Remember form the previous section that 
$ (i)_{2} $ is the binary representation of $ i \in \mathbb{N} $ and $ (i)^{\rev}_{2} $ is 
the reverse of the binary representation of $ i $.
\begin{equation}
	 L_{twin-rev-bins} = 
	 \{ a(0)_{2}a (1)_{2}^{\rev}a(2)_{2}a (3)_{2}^{\rev}a \cdots a (2k)_{2} a (2k+1)_{2}^{\rev} \mid k>0 \}
\end{equation}
\begin{theorem}
	$ L_{twin-rev-bins} $ can be recognized by a RT-PPDA with negative one-sided error bound $ \frac{1}{2} $.
	The machine uses $ O(\log n) $ space on its stack for the members.
\end{theorem}
\begin{proof}
	(Sketch)
	The computation splits equiprobable into two branch. 
	One of the branch (resp., the other branch) 
	consecutively checks the binary numbers positioned in 1 and 2, 3 and 4, etc. (resp., 2 and 3, 4 and 5, etc.)
	If the given input string is a member of the language, then all checks succeed and the input is accepted with 
	probability 1. If it is not a member of the language, then at least one check fails and the input is 
	accepted with probability at most $ \frac{1}{2} $.
\end{proof}

\begin{theorem}
	$ L_{rev-bins} $ and $ L_{twin-rev-bins} $ 
	can be recognized by a RT-P1CA with negative one-sided unbounded error.
	The machine uses $ O(\log n) $ space on its counter for the members.
\end{theorem}
\begin{proof}
	(Sketch)
	The proof is similar to the above one.
	By using a counter, $ L_{twin} $ or $ L_{rev} $ can be recognized by a PFA with negative
	one-sided unbounded error.
	Adding or subtracting a constant from a binary number can be easily implemented if 
	it is accessed in reverse direction.
\end{proof}
\begin{corollary}
	$ L_{rev-bins} $ and $ L_{twin-rev-bins} $ can be recognized by a RT-A1CA.
	The machine uses $ O(\log n) $ space on its counter for the members.
\end{corollary}

$ L_{\mathbb{N}} = \{baba^{2}ba^{3}b \cdots b a^{k} \mid k \in \mathbb{N} \} $ is recognized by a RT-D2CA exactly
and a RT-P1CA with negative one-sided error bound $ \frac{1}{2} $. 
The RT-D2CA uses $ O(\sqrt{n}) $ space for all strings but
the RT-P1CA uses $ O(\sqrt{n}) $ space for the members.

$ L_{twin-rev} = \{ w_{1}cw_{2}cw_{1}^{\rev}cw_{2}^{\rev} \mid w_{1},w_{2} \in \{a,b\}^{*}, |w_{1}|=|w_{2}| \} $
cannot be recognized by a 1NPDA with an auxiliary tape on which $ o(n) $ space is used \cite{Br77}.

\begin{theorem}
	$ L_{twin-rev} $ is recognized by a RT-PPDA with negative one-sided error bound $ \frac{2}{3} $.
\end{theorem}
\begin{proof}
	(Sketch)
	The members of the language mainly have three deterministic checks, each of which can be implemented by
	a RT-DPDA.
	Therefore, each deterministic check is selected with probability $ \frac{1}{3} $.
\end{proof}
\begin{theorem}
	$ L_{twin-rev} $ is recognized by a RT-QPDA with negative one-sided error bound $ \frac{5}{8} $.
\end{theorem}
\begin{proof}
	(Sketch)
	The computation split equiprobably into two branches, say $ \mathsf{branch}_{1} $ and $ \mathsf{branch}_{2} $.
	Suppose the input is of the form $ w_{1} c w_{2} c w_{3} c w_{4} $.
	$ \mathsf{branch}_{1} $ (resp., $ \mathsf{branch}_{2} $) splits into two paths, say
	$ \mathsf{path}_{11} $ and $ \mathsf{path}_{12} $ (resp., $ \mathsf{path}_{21} $ and $ \mathsf{path}_{22} $)
	with equal amplitude.
	
	$ \mathsf{path}_{11} $ (resp., $ \mathsf{path}_{21} $) checks the length of $ w_{1} $ and $ w_{2} $
	(resp., $ w_{3} $ and $ w_{4} $). If fails, the input is rejected.
	
	$ \mathsf{path}_{12} $ (resp., $ \mathsf{path}_{22} $) checks the equality of $ w_{1} $ and $ w_{3}^{r} $
	(resp., $ w_{2} $ and $ w_{4}^{r} $). If fails, the input is rejected. 
	
	At the end of the computation, $ \mathsf{path}_{11} $ and $ \mathsf{path}_{12} $
	(resp., $ \mathsf{path}_{21} $ and $ \mathsf{path}_{22} $) makes a 2-way QFT
	with the distinguished target of an accepting case.
	
	In a branch, if both paths succeed, the input is accepted with probability 1.
	On the other hand, if a path fails, then the input is accepted with probability at most $ \frac{1}{4} $
	in the branch.
	Therefore, the overall accepting probability for nonmembers can be at most 
	$ \frac{5}{8} (=\frac{1}{2} + \frac{1}{8}) $.
\end{proof}

$ L_{ij-i \vee  j} = \{a^{i}b^{j}c^{k} \mid i \neq j, ( i=k \vee j=k ), 0 \leq i,j,k \} $
is not in CFL.

\begin{theorem}
	$ L_{ij-i \vee  j} $ can be recognized by a RT-Q1BCA with one-sided error bound $ \frac{3}{4} $.
\end{theorem}
\begin{proof}
	Suppose that the input is of the form $ a^{i}b^{j}c^{k} $ for some $ i,j,k > 0 $.
	The computation is split into two paths and one of them (resp., the other) write $ i $ (resp., $ j $)
	on the counter. Then, both makes a 2-way QFT with the distinguished target of an rejecting case.
	If $ i=j $, the input is rejected with probability 1. Otherwise, each path continuous the computation
	with probability $ \frac{1}{4} $. In the remaining part, both paths decrease their counter value by 1
	whenever reading a $ c $ and then enter an accepting case after reading $ \dollar $.
	Since $ i \neq j $, only the counter of a one path can be zero, which can only be occurred for
	the members of the language.
\end{proof}

\begin{theorem}
	Any language recognized by a RT-D$ k $BCA is in S$ ^{=}_{\mathbb{Q}} $.
\end{theorem}
\begin{proof}
	(Sketch)
	We can design a GFA for any given RT-D$ k $BCA. 
	The values of the counters are represented by a single number, say $ c $.
	Initially $ c $ is set to 1. Let $ p_{i} $ be the $ i^{th} $ smallest prime, where $ 1 \le i \le k $.
	If the value of the $ i^{th} $ counter is updated by $ 1 $ (resp., $ -1 $),
	then $ c $ is multiplied by $ p_{i} $ (resp., $ \frac{1}{p_{i}} $).
	It can be easily verified that $ c=1 $ if and only if the values of all counters are zeros.
\end{proof}

\begin{lemma}
	$ L_{NH} \notin $ uS can be recognized by a RT-N$ 1 $BCA.
\end{lemma}
\begin{corollary}
	The class of languages recognized by RT-D$ k $BCA is a proper subset of the class of the languages
	recognized by RT-N$ k $BCA for any $ k > 0 $.
\end{corollary}

\begin{lemma}
	$ L_{neq} $ can be recognized by a RT-N$ 1 $BCA.
\end{lemma}
\begin{proof}
	At the beginning of the computation, the machine assumes that the number of $ a $'s is bigger than
	the number of $ b $'s, (or vice versa).
	Then, any nondeterministic path test the equality of the number of $ a $'s and the number of $ b $'s
	without counting some $ a $'s (or $ b $'s).
\end{proof}

\begin{fact}
	\cite{Gr78}
	$ L_{upal}^{*} $ cannot be recognized by any 1N$ k $BCA, where $ k>0 $.
\end{fact}

\begin{lemma}
	$ L_{upal}^{*} $ is recognized by a RT-Q1BCA($ m $) with negative one-sided error bound
	$ \frac{1}{m} $, where $ m>1 $.
\end{lemma}

\begin{proof}
	(Sketch)
	The computation is split into $ m $ paths, i.e. $ \mathsf{path}_{j} $ $ (1 \le j \le m) $,
	at the beginning of each $ a^{+}b^{+} $ block.
	In $ \mathsf{path}_{j} $, the counter value is increased (resp., decreased) by $ j $
	when reading a $ a $ (resp., $ b $).
	After finishing the block, all paths make a $ m $-way QFT, i.e.,
	(i) the computation continuous on the distinguished target if symbol $ a $ is read 
	and the input is rejected on the other targets;
	(ii) if symbol $ \dollar $ is read, the input is accepted on the distinguished target.
\end{proof}

\section{Promise Problems} \label{rwork:promise}

In this section, we use the Hadamard transform (2-way QFT) with the following setup.
Let $ v $ be a $ 2 \times 1 $ column vector and $ H $ be the Hadamard transform i.e.
\begin{equation}
	\left( \begin{array}{ll}
		\frac{1}{\sqrt{2}} & \frac{1}{\sqrt{2}} \\
		\frac{1}{\sqrt{2}} & \frac{-1}{\sqrt{2}}	
	\end{array} \right).
\end{equation}
If $ v = \left( \begin{array}{r} \frac{1}{\sqrt{2}} \\ \frac{1}{\sqrt{2}} \end{array} \right) $
or $ v = \left( \begin{array}{r} \frac{-1}{\sqrt{2}} \\ \frac{-1}{\sqrt{2}} \end{array} \right) $, then
$ v' = Hv = \left( \begin{array}{r} 1 \\ 0 \end{array} \right) $
or $ v' = Hv = \left( \begin{array}{r} -1 \\ 0 \end{array} \right) $, respectively.
On the other hand,
If $ v = \left( \begin{array}{r} \frac{-1}{\sqrt{2}} \\ \frac{1}{\sqrt{2}} \end{array} \right) $
or $ v = \left( \begin{array}{r} \frac{1}{\sqrt{2}} \\ \frac{-1}{\sqrt{2}} \end{array} \right) $, then
$ v' = Hv = \left( \begin{array}{r} 0 \\ -1 \end{array} \right) $
or $ v' = Hv = \left( \begin{array}{r} 0 \\ 1 \end{array} \right) $, respectively.

Suppose that the computation is split into two paths with equal amplitude.
In each path, a deterministic check is done and the amplitude is multiplied with $ -1 $ if the check fails.
After that, two paths make the Hadamard transformation by assuming the targets surely interfere.
Thus, we can obtain the following logical relation:
\begin{equation}
	\begin{array}{|c|c|r|r|}
		\hline
		\mathsf{path}_{1} & \mathsf{path}_{2} & \mathsf{target}_{1} & \mathsf{target}_{2} 
		\\ \hline
		1 & 1 & 1 & 0 
		\\ \hline
		0 & 0 & -1 & 0
		\\ \hline
		1 & 0 & 0 & 1
		\\ \hline
		0 & 1 & 0 & -1
		\\ \hline
	\end{array}
\end{equation}
In other words, if both paths succeed or fail, then $ \mathsf{target}_{1} $ appears with probability 1,
and, in the other cases, $ \mathsf{target}_{2} $ appears with probability 1.

Consider the strings of the form $ w_{1} c w_{2} c w_{3} c w_{4} $, where
$ w_{1 \le i \le 4} \in \{a,b\}^{*} $ and $ |w_{1}| = |w_{2}| = |w_{3}| = |w_{4}| $:
\begin{itemize}
	\item If $ w_{1} = w_{3}^{\rev} \wedge w_{2} = w_{4}^{\rev} $ or 
		$ w_{1} \neq w_{3}^{\rev} \wedge w_{2} \neq w_{4}^{\rev} $, then it is a member of language $ A $;
	\item It is a member of $ B $, otherwise.
\end{itemize}

\begin{theorem}
	A 1-rev-RT-QPDA can separate languages $ A $ and $ B $ with zero error.
\end{theorem}
\begin{proof}
	(Sketch)
	Since the length of strings $ w_{i} $ $ (1 \le i \le 4) $ are equal, they can interfere after the Hadamard.
\end{proof}

\section{Multihead Finite Automata} \label{rwork:multihead}

The following language can be recognized by 1QFAs with 1-head and 1PFAs with 3-heads 
for any negative one-sided error bound $ \epsilon \in (0,\frac{1}{2}) $,
\begin{equation}
	\{ a^{n_{1}} b a^{n_{2}} b \cdots b a^{n_{k}} b a^{n_{k}} b \cdots b a^{n_{2}} b a^{n_{1}} \mid
	n_{i} > 0 \mbox{ for any } i \in \{1,\ldots,k\} \},
\end{equation}
where $ k \in \mathbb{Z}^{+} $.

\chapter{CONCLUSION} \label{conc}

In this thesis, we define a new kind of quantum Turing machine for space bounded computation and 
some of its restricted variants, i.e. quantum pushdown, counter, and finite automata.
Moreover, we present many results related to sublogarithmic and linear space 
both in classical and quantum computation.

As seen from Section \ref{qtm}, a common open problem for most of the space-bounded quantum computational models 
is whether their computational powers decrease or not when restricted to unidirectional ones.

In \cite{Wa98,Wa99,Wa03}, for any space-constructable $ s \in \Omega(\log n) $,
Watrous showed the equivalence of the classes PrSPACE$ _{\mathbb{X}} $($ s $) and
PrQSPACE$ _{\mathbb{X}} $($ s $), 
where $ \mathbb{X} \in \{ \mathbb{Q},\mathbb{A} \} $.
For sublogarithmic space bounds, we have obtained the results: 
\begin{equation}
	\mbox{one-way-PrSPACE($ s $)} \subsetneq \mbox{one-way-PrQSPACE($ s $)}, ~~
	\mbox{for any } s \in o(\log n)
\end{equation}
and
\begin{equation}
	\mbox{PrSPACE($ s $)} \subsetneq \mbox{PrQSPACE($ s $)}, ~~
	\mbox{for any } s \in o(\log\log n).
\end{equation}
For $ s \in \Omega(\log\log n)  \cap o(\log(n)) $,
it still remains as an open problem whether PrSPACE($ s $) is a proper subset of PrQSPACE($ s $).
It is also interesting to ask whether the results presented by Watrous \cite{Wa98,Wa99,Wa03} 
can be extended for the classes defined by computable or arbitrary real numbers amplitudes.

Some of our results for the constant space-bounded computation together with previously known ones
are shown in Figure \ref{conc:constant-space-bound},
in which the containment hierarchy is defined from bottom to top such that 
	dotted arrows indicate subset relationships and unbroken arrows represent the cases 
	where it is known that the inclusion is proper.
One natural question is which dotted arrows in the figure can be replaced by unbroken arrows.
Some other more specific questions are presented below.
\begin{openproblem}
	Is $ L_{pal} $ or $ L_{lt} $ in one-way-BQSPACE(1)?
\end{openproblem}
\begin{openproblem}
	Is there any nonstochastic language in BQSPACE(1)?
\end{openproblem}

\begin{figure}[h!]
	\begin{center}
	\fbox{
	\begin{minipage}{0.9\textwidth}
		\centering
		~~\\~~\\
		\ifx\JPicScale\undefined\def\JPicScale{1}\fi
		\unitlength \JPicScale mm
		\begin{picture}(95,110)(10,0)
		
		\scriptsize

		\put(50,8){\makebox(0,0)[l]{one-way-BPSPACE(1)}}
		\put(50,4){\makebox(0,0)[l]{= RT-BQSPACE(1)}}
		\put(50,0){\makebox(0,0)[l]{= REG}}
		
		\put(60,10){\vector(2,1){38}}
		\put(60,10){\vector(1,2){9}}
		\put(60,10){\vector(-1,1){18}}
		\put(60,10){\vector(-2,1){38}}
		
		\put(15,34){\makebox(0,0)[l]{NQAL}}
		\put(15,30){\makebox(0,0)[l]{= S$ ^{\neq} $}}
		\put(38,34){\makebox(0,0)[l]{LPostS}}
		\put(38,30){\makebox(0,0)[l]{= PostS}}
		\put(62,30){\makebox(0,0)[l]{S$ _{\mathbb{Q}}^{\neq} $ $ \cup $ S$ _{\mathbb{Q}}^{=} $}}
		\put(100,30){\makebox(0,0)[cc]{one-way-BQSPACE(1)}}
		
		\put(35,50){\makebox(0,0)[l]{BPSPACE(1)}}
		\multiput(42,36)(0,1){11}{$ . $}
		\put(42.1,47){\vector(0,1){1.5}}
		
		\put(42,36){\vector(2,1){24}}
		\put(62,50){\makebox(0,0)[l]{PostQAL}}
		\put(67,33){\vector(0,1){16}}
		
		\put(12,64){\makebox(0,0)[l]{QAL $ \cap $ coQAL}}
		\put(12,60){\makebox(0,0)[l]{= S $ \cap $ coS}}
		\put(12,56){\makebox(0,0)[l]{= UMM}}
		\put(18,36){\vector(0,1){18}}
		
		\put(62,63){\makebox(0,0)[l]{LPostQAL}}
		\put(68,53){\vector(0,1){8}}
		
		\put(38,82){\makebox(0,0)[l]{PrSPACE(1)}}
		\put(38,78){\makebox(0,0)[l]{=uQAL}}
		\multiput(21,66)(2,1){10}{$ . $}
		\put(39,75){\vector(2,1){1.5}}
		\put(44,52){\vector(0,1){24}}
		\multiput(67,66)(-2,1){10}{$ . $}
		\put(48.5,75.5){\vector(-2,1){1.5}}
		
		\put(87,78){\makebox(0,0)[l]{BQSPACE(1)}}
		\multiput(70,52)(1,1){24}{$ . $}
		\put(93.3,75){\vector(1,1){1.5}}
		\multiput(97,33)(0,1){42}{$ . $}
		\put(97.2,75){\vector(0,1){1.5}}
		
		\put(50,95){\makebox(0,0)[l]{one-way-PrQSPACE(1)}}
		\put(45,84){\vector(2,1){18}}
		
		\put(55,110){\makebox(0,0)[l]{PrQSPACE(1)}}
		\multiput(66,97)(0,1){10}{$ . $}
		\put(66.2,107){\vector(0,1){1.5}}

	\end{picture}
	~~\\~~\\
	\end{minipage}}
	\end{center}
	\caption{The relationships among classical and quantum constant space-bounded classes}
	\vskip\baselineskip
	\label{conc:constant-space-bound}
\end{figure}
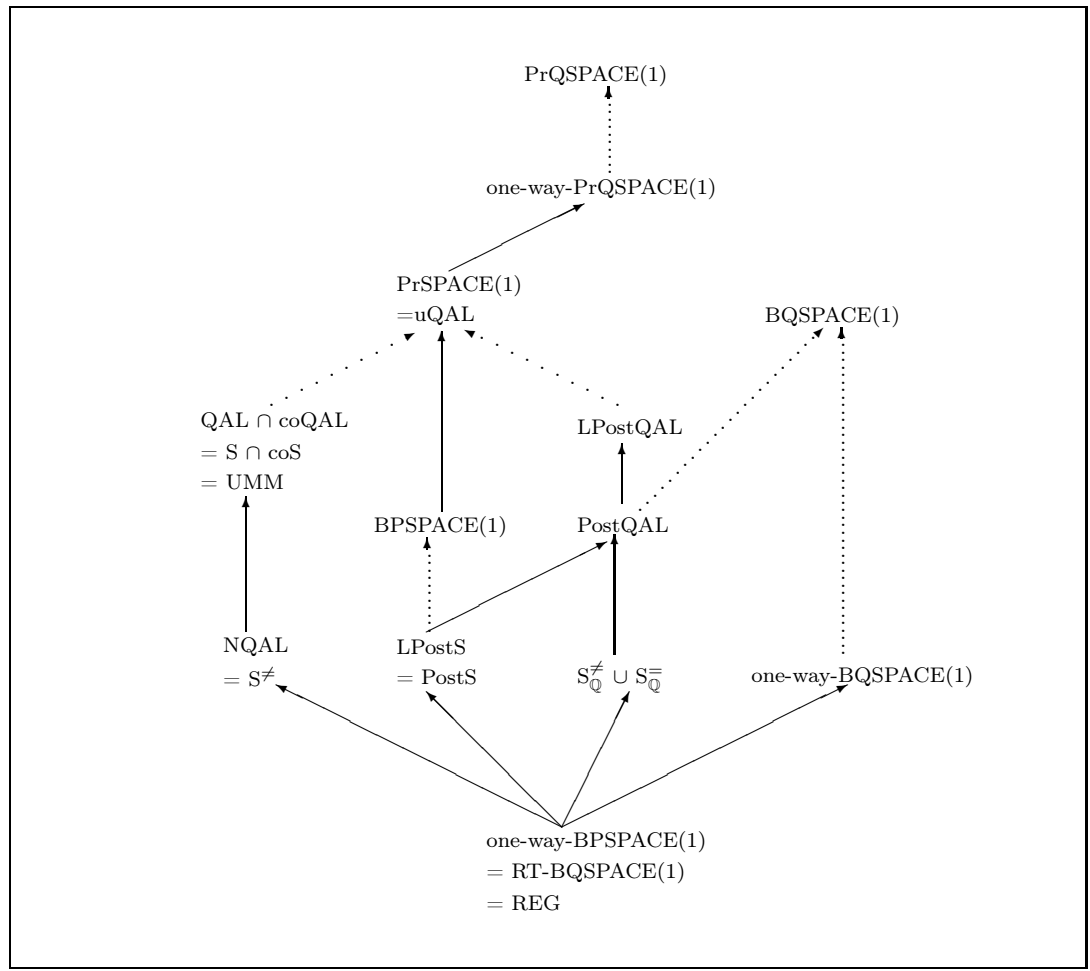 

In a more general way, the following question is also interesting.
\begin{openproblem}
	Is there any hierarchy theorem for space-bounded quantum classes?
\end{openproblem}

We also give a special emphasis to the quantum models with WOM (Chapter \ref{wom}) 
since they can be seen as good steps to understand the nature of quantum computation,
specifically the power of configuration interferences, which is impossible in classical computation.

Moreover, we also do not know of lower bounds for RT-QTMs and RT-PTMs in order to recognize a nonregular language.
The same question can also be extended for CAs and PDAs by considering the space usage
on their counters and stacks, respectively.

\appendix      

\chapter{LOCAL CONDITIONS FOR QUANTUM MACHINES} \label{well}

\section{Quantum Turing Machines} \label{well:QTMs}

We present the list of the local conditions for QTM wellformedness in below:
Let $ c_{j_{1}} $ and $ c_{j_{2}} $ be two configurations with corresponding columns 
$ v_{j_{1}} $ and $ v_{j_{2}} $ in $ \mathsf{E} $ (See Figure \ref{qtm:fiq:superoperators}).
The value of $ v_{j_{1}}[i]  $ is determined by $ \delta $ if the $
i^{th} $ entry of $ v_{j_{1}} $
corresponds to a configuration to which $ c_{j_{1}} $ can evolve in one step,
and it is zero otherwise.
Let $ x_{1} $ and $ x_{2} $ be the positions of the input tape head
and $ y_{1} $ and $ y_{2} $ be the positions of the work tape head
for the configurations $ c_{j_{1}} $ and $ c_{j_{2}} $, respectively.
In order to evolve to the same configuration in one step, the difference between
$ x_{1} $ and $ x_{2} $ and/or $ y_{1} $ and $ y_{2} $ must be at most 2.
Therefore, we obtain a total of 13 different cases, listed below, that completely
define the restrictions on the transition function.
Note that, by taking the conjugates of each summation,
we handle the symmetric cases that are shown in the parentheses.

For all $ q_{1},q_{2} \in Q $; 
$ \sigma,\sigma_{1},\sigma_{2} \in \Sigma $; and
$ \gamma_{1},\gamma_{2},\gamma_{3},\gamma_{4} \in \Gamma $
(the summations are taken over $ q' \in Q $; $ \gamma' \in \Gamma $;
$ d,d_{1},d_{2} \in <\mspace{-4mu}> $; and $ \omega \in \Omega $),
\\\\
\footnotesize
1. $ x_{1} = x_{2} $ and $ y_{1} = y_{2} $:
\begin{equation}			
	\sum\limits_{q',\gamma',d_{1},d_{2},\omega }
	\overline{
	\delta(q_{1},\sigma,\gamma_{1},q',d_{1},\gamma',d_{2},\omega)}
	\delta(q_{2},\sigma,\gamma_{2},q',d_{1},\gamma',d_{2},\omega)
	= 
	\left\lbrace
		\begin{array}{ll}
			1 & q_{1} = q_{2} , \gamma_{1} = \gamma_{2} \\
			0 & \mbox{otherwise}
		\end{array}
	\right.
\end{equation}
2. $ x_{1} = x_{2} $ and $ y_{1} = y_{2}-1 $ ($ x_{1} = x_{2} $ and  $ y_{1} = y_{2}+1 $):
\begin{equation}
	\begin{array}{ccl}
		\displaystyle \sum\limits_{q',d,\omega} &
		\multicolumn{2}{l}{
		\overline{
		\delta(q_{1},\sigma,\gamma_{1},q',d,\gamma_{2},\downarrow,\omega)
		}
		\delta(q_{2},\sigma,\gamma_{3},q',d,\gamma_{4},\leftarrow,\omega)}
		\\
		& + & 
		\overline{
		\delta(q_{1},\sigma,\gamma_{1},q',d,\gamma_{2},\rightarrow,\omega)
		}		
		\delta(q_{2},\sigma,\gamma_{3},q',d,\gamma_{4},\downarrow,\omega)
		= 0
	\end{array}
\end{equation}
3. $ x_{1} = x_{2} $ and $ y_{1} = y_{2}-2 $ ($ x_{1} = x_{2} $ and $ y_{1} = y_{2}+2 $):
\begin{equation}
	\begin{array}{ccl}
		\displaystyle \sum\limits_{q',d,\omega} & 
		\multicolumn{2}{l}{
		\overline{
		\delta(q_{1},\sigma,\gamma_{1},q',d,\gamma_{2},\rightarrow,\omega)
		}		
		\delta(q_{2},\sigma,\gamma_{3},q',d,\gamma_{4},\leftarrow,\omega)
		= 0 }
	\end{array}			
\end{equation}
4. $ x_{1} = x_{2}-1 $ and $ y_{1} = y_{2} $ ($ x_{1} = x_{2}+1 $ and $ y_{1} = y_{2} $):
\begin{equation}
	\begin{array}{ccl}
		\displaystyle \sum\limits_{q',\gamma',d,\omega} &
		\multicolumn{2}{l}{
		\overline{
		\delta(q_{1},\sigma_{1},\gamma_{1},q',\downarrow,\gamma',d,\omega)
		}		
		\delta(q_{2},\sigma_{2},\gamma_{2},q',\leftarrow,\gamma',d,\omega)}
		\\
		& + &
		\overline{
		\delta(q_{1},\sigma_{1},\gamma_{1},q',\rightarrow,\gamma',d,\omega)
		}		
		\delta(q_{2},\sigma_{2},\gamma_{2},q',\downarrow,\gamma',d,\omega)
		= 0
	\end{array}			
\end{equation}
5. $ x_{1} = x_{2}-1 $ and $ y_{1} = y_{2}-1 $ ($ x_{1} = x_{2}+1 $ and $ y_{1} = y_{2}+1 $):
\begin{equation}
	\begin{array}{ccl}
		{ \displaystyle \sum\limits_{q',\omega}} &
		\multicolumn{2}{l}{
		\overline{
		\delta(q_{1},\sigma_{1},\gamma_{1},q',\downarrow,\gamma_{2},\downarrow,\omega)
		}		
		\delta(q_{2},\sigma_{2},\gamma_{3},q',\leftarrow,\gamma_{4},\leftarrow,\omega)}
		\\
		& + &
		\overline{
		\delta(q_{1},\sigma_{1},\gamma_{1},q',\downarrow,\gamma_{2},\rightarrow,\omega)
		}		
		\delta(q_{2},\sigma_{2},\gamma_{3},q',\leftarrow,\gamma_{4},\downarrow,\omega)
		\\
		& + &
		\overline{
		\delta(q_{1},\sigma_{1},\gamma_{1},q',\rightarrow,\gamma_{2},\downarrow,\omega)
		}		
		\delta(q_{2},\sigma_{2},\gamma_{3},q',\downarrow,\gamma_{4},\leftarrow,\omega)
		\\
		& + &
		\overline{
		\delta(q_{1},\sigma_{1},\gamma_{1},q',\rightarrow,\gamma_{2},\rightarrow,\omega)
		}		
		\delta(q_{2},\sigma_{2},\gamma_{3},q',\downarrow,\gamma_{4},\downarrow,\omega)
		= 0
	\end{array}
\end{equation}
6. $ x_{1} = x_{2}-1 $ and $ y_{1} = y_{2}+1 $ ($ x_{1} = x_{2}+1 $ and $ y_{1} = y_{2}-1 $):
\begin{equation}
	\begin{array}{ccl}
		{  \displaystyle \sum\limits_{q',\omega}} &
		\multicolumn{2}{l}{
		\overline{
		\delta(q_{1},\sigma_{1},\gamma_{1},q',\downarrow,\gamma_{2},\downarrow,\omega)
		}		
		\delta(q_{2},\sigma_{2},\gamma_{3},q',\leftarrow,\gamma_{4},\rightarrow,\omega)}
		\\
		& + &
		\overline{
		\delta(q_{1},\sigma_{1},\gamma_{1},q',\downarrow,\gamma_{2},\leftarrow,\omega)
		}		
		\delta(q_{2},\sigma_{2},\gamma_{3},q',\leftarrow,\gamma_{4},\downarrow,\omega)
		\\
		& + &
		\overline{
		\delta(q_{1},\sigma_{1},\gamma_{1},q',\rightarrow,\gamma_{2},\downarrow,\omega)
		}		
		\delta(q_{2},\sigma_{2},\gamma_{3},q',\downarrow,\gamma_{4},\rightarrow,\omega)
		\\
		& + &
		\overline{
		\delta(q_{1},\sigma_{1},\gamma_{1},q',\rightarrow,\gamma_{2},\leftarrow,\omega)
		}		
		\delta(q_{2},\sigma_{2},\gamma_{3},q',\downarrow,\gamma_{4},\downarrow,\omega)
		= 0
	\end{array}
\end{equation}
7. $ x_{1} = x_{2}-1 $ and $ y_{1} = y_{2}-2 $ ($ x_{1} = x_{2}+1 $ and $ y_{1} = y_{2}+2 $):
\begin{equation}
	\begin{array}{ccl}
		{ \displaystyle \sum\limits_{q',\omega}} &
		\multicolumn{2}{l}{
		\overline{
		\delta(q_{1},\sigma_{1},\gamma_{1},q',\downarrow,\gamma_{2},\rightarrow,\omega)
		}		
		\delta(q_{2},\sigma_{2},\gamma_{3},q',\leftarrow,\gamma_{4},\leftarrow,\omega)}
		\\
		& + &
		\overline{
		\delta(q_{1},\sigma_{1},\gamma_{1},q',\rightarrow,\gamma_{2},\rightarrow,\omega)
		}		
		\delta(q_{2},\sigma_{2},\gamma_{3},q',\downarrow,\gamma_{4},\leftarrow,\omega)
		= 0
	\end{array}
\end{equation}
8. $ x_{1} = x_{2}-1 $ and $ y_{1} = y_{2}+2 $ ($ x_{1} = x_{2}+1 $ and $ y_{1} = y_{2}-2 $):
\begin{equation}
	\begin{array}{ccl}
		{ \displaystyle \sum\limits_{q',\omega}} &
		\multicolumn{2}{l}{
		\overline{
		\delta(q_{1},\sigma_{1},\gamma_{1},q',\downarrow,\gamma_{2},\leftarrow,\omega)
		}		
		\delta(q_{2},\sigma_{2},\gamma_{3},q',\leftarrow,\gamma_{4},\rightarrow,\omega)}
		\\
		& + &
		\overline{
		\delta(q_{1},\sigma_{1},\gamma_{1},q',\rightarrow,\gamma_{2},\leftarrow,\omega)
		}		
		\delta(q_{2},\sigma_{2},\gamma_{3},q',\downarrow,\gamma_{4},\rightarrow,\omega)
		= 0
	\end{array}
\end{equation}
9. $ x_{1} = x_{2}-2 $ and $ y_{1} = y_{2} $ ($ x_{1} = x_{2}+2 $ and $ y_{1} = y_{2} $):
\begin{equation}
	\begin{array}{ccl}
		{ \displaystyle \sum\limits_{q',\gamma',d,\omega}} &
		\multicolumn{2}{l}{
		\overline{
		\delta(q_{1},\sigma_{1},\gamma_{1},q',\rightarrow,\gamma',d,\omega)
		}		
		\delta(q_{2},\sigma_{2},\gamma_{2},q',\leftarrow,\gamma',d,\omega)}
		= 0
	\end{array}
\end{equation}
10. $ x_{1} = x_{2}-2 $ and $ y_{1} = y_{2}-1 $ ($ x_{1} = x_{2}+2 $ and $ y_{1} = y_{2}+1 $):
\begin{equation}
	\begin{array}{ccl}
		{ \displaystyle \sum\limits_{q',\omega}} &
		\multicolumn{2}{l}{
		\overline{
		\delta(q_{1},\sigma_{1},\gamma_{1},q',\rightarrow,\gamma_{2},\rightarrow,\omega)
		}		
		\delta(q_{2},\sigma_{2},\gamma_{3},q',\leftarrow,\gamma_{4},\downarrow,\omega)}
		\\
		& + &
		\overline{
		\delta(q_{1},\sigma_{1},\gamma_{1},q',\rightarrow,\gamma_{2},\downarrow,\omega)
		}		
		\delta(q_{2},\sigma_{2},\gamma_{3},q',\leftarrow,\gamma_{4},\leftarrow,\omega)
		= 0
	\end{array}
\end{equation}
11. $ x_{1} = x_{2}-2 $ and $ y_{1} = y_{2}+1 $ ($ x_{1} = x_{2}+2 $ and $ y_{1} = y_{2}-1 $):
\begin{equation}
	\begin{array}{ccl}
		{ \displaystyle \sum\limits_{q',\omega}} &
		\multicolumn{2}{l}{
		\overline{
		\delta(q_{1},\sigma_{1},\gamma_{1},q',\rightarrow,\gamma_{2},\leftarrow,\omega)
		}		
		\delta(q_{2},\sigma_{2},\gamma_{3},q',\leftarrow,\gamma_{4},\downarrow,\omega)}
		\\
		& + &
		\overline{
		\delta(q_{1},\sigma_{1},\gamma_{1},q',\rightarrow,\gamma_{2},\downarrow,\omega)
		}		
		\delta(q_{2},\sigma_{2},\gamma_{3},q',\leftarrow,\gamma_{4},\rightarrow,\omega)
		= 0
	\end{array}
\end{equation}
12. $ x_{1} = x_{2}-2 $ and $ y_{1} = y_{2}-2 $ ($ x_{1} = x_{2}+2 $ and $ y_{1} = y_{2}+2 $):
\begin{equation}
	\begin{array}{ccl}
		{ \displaystyle \sum\limits_{q',\omega}} &
		\multicolumn{2}{l}{
		\overline{
		\delta(q_{1},\sigma_{1},\gamma_{1},q',\rightarrow,\gamma_{2},\rightarrow,\omega)
		}		
		\delta(q_{2},\sigma_{2},\gamma_{3},q',\leftarrow,\gamma_{4},\leftarrow,\omega)}
		= 0
	\end{array}
\end{equation}
13. $ x_{1} = x_{2}-2 $ and $ y_{1} = y_{2}+2 $ ($ x_{1} = x_{2}+2 $ and $ y_{1} = y_{2}-2 $):
\begin{equation}
	\begin{array}{ccl}
		{ \displaystyle \sum\limits_{q',\omega}} &
		\multicolumn{2}{l}{
		\overline{
		\delta(q_{1},\sigma_{1},\gamma_{1},q',\rightarrow,\gamma_{2},\leftarrow,\omega)
		}		
		\delta(q_{2},\sigma_{2},\gamma_{3},q',\leftarrow,\gamma_{4},\rightarrow,\omega)}
		= 0
	\end{array}
\end{equation}
\normalsize

\section{Quantum Pushdown Automata} \label{well:QPDAs}

\noindent \underline{Local conditions for 2QPDA well-formedness}:

Let $ x_{i} $ and $ x_{j} $ be the positions of the input tape head
and $ y_{i} $ and $ y_{j} $ be the positions of the stack head
for the configurations $ c_{i} $ and $ c_{j} $, respectively.
In order to evolve to the same configuration in one step, the difference between
$ x_{i} $ and $ x_{j} $ and/or $ y_{i} $ and $ y_{j} $ must be at most 2.
Therefore, we obtain totally 13 different cases that completely define the restrictions on the transition function.
Note that, by taking conjugate of each summation, we obtain the symmetric cases that are shown in the parenthesis.
\newline \newline
For each choice of $ q_{1},q_{2} \in Q $; 
$ \sigma,\sigma_{1},\sigma_{2} \in \tilde{ \Sigma } $; and
$ \gamma_{1},\gamma_{2},\gamma_{3},\gamma_{4} \in \Gamma $
(the summations are taken over $ q' \in Q $; $ \gamma' \in \tilde{\Gamma} $;
$ d,d_{1},d_{2} \in <\mspace{-4mu}> $; and $ \omega \in \Omega $),
\\\\
1. $ x_{i} = x_{j} $ and $ y_{i} = y_{j} $:
\footnotesize
\begin{equation}			
	\sum\limits_{q',d_{1},d_{2},\gamma^{'},\omega }
	\overline{
	\delta(q_{1},\sigma,\gamma_{1},q',d_{1},d_{2},\gamma^{'},\omega)}
	\delta(q_{2},\sigma,\gamma_{2},q',d_{1},d_{2},\gamma^{'},\omega)
	= 
	\left\lbrace
		\begin{array}{ll}
			1 ~~& q_{1} = q_{2} , \gamma_{1} = \gamma_{2} \\
			0 & \mbox{otherwise}
		\end{array}
	\right.
\end{equation}
\normalsize
2. $ x_{i} = x_{j} $ and $ y_{i} = y_{j}-1 $ ($ x_{i} = x_{j} $ and  $ y_{i} = y_{j}+1 $):
\footnotesize
\begin{equation}
	\begin{array}{ccl}
		\displaystyle \sum\limits_{q',d,\omega} &
		\multicolumn{2}{l}{
		\overline{
		\delta(q_{1},\sigma,\gamma_{1},q',d,\downarrow,\gamma_{2},\omega)
		}
		\delta(q_{2},\sigma,\gamma_{3},q',d,\leftarrow,\gamma_{3},\omega)}
		\\
		& + & 
		\overline{
		\delta(q_{1},\sigma,\gamma_{1},q',d,\rightarrow,\gamma_{2},\omega)
		}		
		\delta(q_{2},\sigma,\gamma_{3},q',d,\downarrow,\gamma_{2},\omega)
		= 0
	\end{array}
\end{equation}
\normalsize
3. $ x_{i} = x_{j} $ and $ y_{i} = y_{j}-2 $ ($ x_{i} = x_{j} $ and $ y_{i} = y_{j}+2 $):
\footnotesize
\begin{equation}
	\begin{array}{ccl}
		\displaystyle \sum\limits_{q',d,\omega} & 
		\multicolumn{2}{l}{
		\overline{
		\delta(q_{1},\sigma,\gamma_{1},q',d,\rightarrow,\gamma_{2},\omega)
		}		
		\delta(q_{2},\sigma,\gamma_{3},q',d,\leftarrow,\gamma_{3},\omega)
		= 0 }
	\end{array}			
\end{equation}
\normalsize
4. $ x_{i} = x_{j}-1 $ and $ y_{i} = y_{j} $ ($ x_{i} = x_{j}+1 $ and $ y_{i} = y_{j} $):
\footnotesize
\begin{equation}
	\begin{array}{ccl}
		\displaystyle \sum\limits_{q',d,\gamma',\omega} &
		\multicolumn{2}{l}{
		\overline{
		\delta(q_{1},\sigma_{1},\gamma_{1},q',\downarrow,d,\gamma',\omega)
		}		
		\delta(q_{2},\sigma_{2},\gamma_{2},q',\leftarrow,d,\gamma',\omega)}
		\\
		& + &
		\overline{
		\delta(q_{1},\sigma_{1},\gamma_{1},q',\rightarrow,d,\gamma',\omega)
		}		
		\delta(q_{2},\sigma_{2},\gamma_{2},q',\downarrow,d,\gamma',\omega)
		= 0
	\end{array}			
\end{equation}
\normalsize
5. $ x_{i} = x_{j}-1 $ and $ y_{i} = y_{j}-1 $ ($ x_{i} = x_{j}+1 $ and $ y_{i} = y_{j}+1 $):
\footnotesize
\begin{equation}
	\begin{array}{ccl}
		{ \displaystyle \sum\limits_{q',\omega}} &
		\multicolumn{2}{l}{
		\overline{
		\delta(q_{1},\sigma_{1},\gamma_{1},q',\downarrow,\downarrow,\gamma_{2},\omega)
		}		
		\delta(q_{2},\sigma_{2},\gamma_{3},q',\leftarrow,\leftarrow,\gamma_{3},\omega)}
		\\
		& + &
		\overline{
		\delta(q_{1},\sigma_{1},\gamma_{1},q',\downarrow,\rightarrow,\gamma_{2},\omega)
		}		
		\delta(q_{2},\sigma_{2},\gamma_{3},q',\leftarrow,\downarrow,\gamma_{2},\omega)
		\\
		& + &
		\overline{
		\delta(q_{1},\sigma_{1},\gamma_{1},q',\rightarrow,\downarrow,\gamma_{2},\omega)
		}		
		\delta(q_{2},\sigma_{2},\gamma_{3},q',\downarrow,\leftarrow,\gamma_{3},\omega)
		\\
		& + &
		\overline{
		\delta(q_{1},\sigma_{1},\gamma_{1},q',\rightarrow,\rightarrow,\gamma_{2},\omega)
		}		
		\delta(q_{2},\sigma_{2},\gamma_{3},q',\downarrow,\downarrow,\gamma_{2},\omega)
		= 0
	\end{array}
\end{equation}
\normalsize
6. $ x_{i} = x_{j}-1 $ and $ y_{i} = y_{j}+1 $ ($ x_{i} = x_{j}+1 $ and $ y_{i} = y_{j}-1 $):
\footnotesize
\begin{equation}
	\begin{array}{ccl}
		{  \displaystyle \sum\limits_{q',\omega}} &
		\multicolumn{2}{l}{
		\overline{
		\delta(q_{1},\sigma_{1},\gamma_{1},q',\downarrow,\downarrow,\gamma_{2},\omega)
		}		
		\delta(q_{2},\sigma_{2},\gamma_{3},q',\leftarrow,\rightarrow,\gamma_{2},\omega)}
		\\
		& + &
		\overline{
		\delta(q_{1},\sigma_{1},\gamma_{1},q',\downarrow,\leftarrow,\gamma_{1},\omega)
		}		
		\delta(q_{2},\sigma_{2},\gamma_{3},q',\leftarrow,\downarrow,\gamma_{4},\omega)
		\\
		& + &
		\overline{
		\delta(q_{1},\sigma_{1},\gamma_{1},q',\rightarrow,\downarrow,\gamma_{2},\omega)
		}		
		\delta(q_{2},\sigma_{2},\gamma_{3},q',\downarrow,\rightarrow,\gamma_{2},\omega)
		\\
		& + &
		\overline{
		\delta(q_{1},\sigma_{1},\gamma_{1},q',\rightarrow,\leftarrow,\gamma_{1},\omega)
		}		
		\delta(q_{2},\sigma_{2},\gamma_{3},q',\downarrow,\downarrow,\gamma_{4},\omega)
		= 0
	\end{array}
\end{equation}
\normalsize
7. $ x_{i} = x_{j}-1 $ and $ y_{i} = y_{j}-2 $ ($ x_{i} = x_{j}+1 $ and $ y_{i} = y_{j}+2 $):
\footnotesize
\begin{equation}
	\begin{array}{ccl}
		{ \displaystyle \sum\limits_{q',\omega}} &
		\multicolumn{2}{l}{
		\overline{
		\delta(q_{1},\sigma_{1},\gamma_{1},q',\downarrow,\rightarrow,\gamma_{2},\omega)
		}		
		\delta(q_{2},\sigma_{2},\gamma_{3},q',\leftarrow,\leftarrow,\gamma_{3},\omega)}
		\\
		& + &
		\overline{
		\delta(q_{1},\sigma_{1},\gamma_{1},q',\rightarrow,\rightarrow,\gamma_{2},\omega)
		}		
		\delta(q_{2},\sigma_{2},\gamma_{3},q',\downarrow,\leftarrow,\gamma_{3},\omega)
		= 0
	\end{array}
\end{equation}
\normalsize
8. $ x_{i} = x_{j}-1 $ and $ y_{i} = y_{j}+2 $ ($ x_{i} = x_{j}+1 $ and $ y_{i} = y_{j}-2 $):
\footnotesize
\begin{equation}
	\begin{array}{ccl}
		{ \displaystyle \sum\limits_{q',\omega}} &
		\multicolumn{2}{l}{
		\overline{
		\delta(q_{1},\sigma_{1},\gamma_{1},q',\downarrow,\leftarrow,\gamma_{1},\omega)
		}		
		\delta(q_{2},\sigma_{2},\gamma_{3},q',\leftarrow,\rightarrow,\gamma_{4},\omega)}
		\\
		& + &
		\overline{
		\delta(q_{1},\sigma_{1},\gamma_{1},q',\rightarrow,\leftarrow,\gamma_{1},\omega)
		}		
		\delta(q_{2},\sigma_{2},\gamma_{3},q',\downarrow,\rightarrow,\gamma_{4},\omega)
		= 0
	\end{array}
\end{equation}
\normalsize
9. $ x_{i} = x_{j}-2 $ and $ y_{i} = y_{j} $ ($ x_{i} = x_{j}+2 $ and $ y_{i} = y_{j} $):
\footnotesize
\begin{equation}
	\begin{array}{ccl}
		{ \displaystyle \sum\limits_{q',d,\gamma',\omega}} &
		\multicolumn{2}{l}{
		\overline{
		\delta(q_{1},\sigma_{1},\gamma_{1},q',\rightarrow,d,\gamma',\omega)
		}		
		\delta(q_{2},\sigma_{2},\gamma_{2},q',\leftarrow,d,\gamma',\omega)}
		= 0
	\end{array}
\end{equation}
\normalsize
10. $ x_{i} = x_{j}-2 $ and $ y_{i} = y_{j}-1 $ ($ x_{i} = x_{j}+2 $ and $ y_{i} = y_{j}+1 $):
\footnotesize
\begin{equation}
	\begin{array}{ccl}
		{ \displaystyle \sum\limits_{q',\omega}} &
		\multicolumn{2}{l}{
		\overline{
		\delta(q_{1},\sigma_{1},\gamma_{1},q',\rightarrow,\rightarrow,\gamma_{2},\omega)
		}		
		\delta(q_{2},\sigma_{2},\gamma_{3},q',\leftarrow,\downarrow,\gamma_{2},\omega)}
		\\
		& + &
		\overline{
		\delta(q_{1},\sigma_{1},\gamma_{1},q',\rightarrow,\downarrow,\gamma_{2},\omega)
		}		
		\delta(q_{2},\sigma_{2},\gamma_{3},q',\leftarrow,\leftarrow,\gamma_{3},\omega)
		= 0
	\end{array}
\end{equation}
\normalsize
11. $ x_{i} = x_{j}-2 $ and $ y_{i} = y_{j}+1 $ ($ x_{i} = x_{j}+2 $ and $ y_{i} = y_{j}-1 $):
\footnotesize
\begin{equation}
	\begin{array}{ccl}
		{ \displaystyle \sum\limits_{q',\omega}} &
		\multicolumn{2}{l}{
		\overline{
		\delta(q_{1},\sigma_{1},\gamma_{1},q',\rightarrow,\leftarrow,\gamma_{1},\omega)
		}		
		\delta(q_{2},\sigma_{2},\gamma_{3},q',\leftarrow,\downarrow,\gamma_{4},\omega)}
		\\
		& + &
		\overline{
		\delta(q_{1},\sigma_{1},\gamma_{1},q',\rightarrow,\downarrow,\gamma_{2},\omega)
		}		
		\delta(q_{2},\sigma_{2},\gamma_{3},q',\leftarrow,\rightarrow,\gamma_{2},\omega)
		= 0
	\end{array}
\end{equation}
\normalsize
12. $ x_{i} = x_{j}-2 $ and $ y_{i} = y_{j}-2 $ ($ x_{i} = x_{j}+2 $ and $ y_{i} = y_{j}+2 $):
\footnotesize
\begin{equation}
	\begin{array}{ccl}
		{ \displaystyle \sum\limits_{q',\omega}} &
		\multicolumn{2}{l}{
		\overline{
		\delta(q_{1},\sigma_{1},\gamma_{1},q',\rightarrow,\rightarrow,\gamma_{2},\omega)
		}		
		\delta(q_{2},\sigma_{2},\gamma_{3},q',\leftarrow,\leftarrow,\gamma_{3},\omega)}
		= 0
	\end{array}
\end{equation}
\normalsize
13. $ x_{i} = x_{j}-2 $ and $ y_{i} = y_{j}+2 $ ($ x_{i} = x_{j}+2 $ and $ y_{i} = y_{j}-2 $):
\footnotesize
\begin{equation}
	\begin{array}{ccl}
		{ \displaystyle \sum\limits_{q',\omega}} &
		\multicolumn{2}{l}{
		\overline{
		\delta(q_{1},\sigma_{1},\gamma_{1},q',\rightarrow,\leftarrow,\gamma_{1},\omega)
		}		
		\delta(q_{2},\sigma_{2},\gamma_{3},q',\leftarrow,\rightarrow,\gamma_{4},\omega)}
		= 0
	\end{array}
\end{equation}
\normalsize

\noindent \underline{Local conditions for 1QPDA well-formedness}:

In order to evolve to the same configuration in one step, the difference between
$ x_{i} $ and $ x_{j} $ must be at most 1 in case of 1QPDA.
Therefore, we obtain totally 8 different cases that completely define the restrictions on the transition function.
Similar to 2QPDA, by taking conjugate of each summation, we obtain the symmetric cases that are shown in the parenthesis.
\newline \newline
For each choice of $ q_{1},q_{2} \in Q $; 
$ \sigma,\sigma_{1},\sigma_{2} \in \tilde{ \Sigma } $; and
$ \gamma_{1},\gamma_{2},\gamma_{3},\gamma_{4} \in \Gamma $
(the summations are taken over $ q' \in Q $; $ \gamma' \in \tilde{\Gamma} $;
$ d\in \rhd $,
$ e \in <\mspace{-4mu}> $; and $ \omega \in \Omega $),
\\\\
1. $ x_{i} = x_{j} $ and $ y_{i} = y_{j} $:
\footnotesize
\begin{equation}			
	\sum\limits_{q',d,e,\gamma^{'},\omega }
	\overline{
	\delta(q_{1},\sigma,\gamma_{1},q',d,e,\gamma^{'},\omega)}
	\delta(q_{2},\sigma,\gamma_{2},q',d,e,\gamma^{'},\omega)
	= 
	\left\lbrace
		\begin{array}{ll}
			1 ~~& q_{1} = q_{2} , \gamma_{1} = \gamma_{2} \\
			0 & \mbox{otherwise}
		\end{array}
	\right.
\end{equation}
\normalsize
2. $ x_{i} = x_{j} $ and $ y_{i} = y_{j}-1 $ ($ x_{i} = x_{j} $ and  $ y_{i} = y_{j}+1 $):
\footnotesize
\begin{equation}
	\begin{array}{ccl}
		\displaystyle \sum\limits_{q',d,\omega} &
		\multicolumn{2}{l}{
		\overline{
		\delta(q_{1},\sigma,\gamma_{1},q',d,\downarrow,\gamma_{2},\omega)
		}
		\delta(q_{2},\sigma,\gamma_{3},q',d,\leftarrow,\gamma_{3},\omega)}
		\\
		& + & 
		\overline{
		\delta(q_{1},\sigma,\gamma_{1},q',d,\rightarrow,\gamma_{2},\omega)
		}		
		\delta(q_{2},\sigma,\gamma_{3},q',d,\downarrow,\gamma_{2},\omega)
		= 0
	\end{array}
\end{equation}
\normalsize
3. $ x_{i} = x_{j} $ and $ y_{i} = y_{j}-2 $ ($ x_{i} = x_{j} $ and $ y_{i} = y_{j}+2 $):
\footnotesize
\begin{equation}
	\begin{array}{ccl}
		\displaystyle \sum\limits_{q',d,\omega} & 
		\multicolumn{2}{l}{
		\overline{
		\delta(q_{1},\sigma,\gamma_{1},q',d,\rightarrow,\gamma_{2},\omega)
		}		
		\delta(q_{2},\sigma,\gamma_{3},q',d,\leftarrow,\gamma_{3},\omega)
		= 0 }
	\end{array}			
\end{equation}
\normalsize
4. $ x_{i} = x_{j}-1 $ and $ y_{i} = y_{j} $ ($ x_{i} = x_{j}+1 $ and $ y_{i} = y_{j} $):
\footnotesize
\begin{equation}
	\begin{array}{ccl}
		\displaystyle \sum\limits_{q',e,\gamma',\omega} &
		\multicolumn{2}{l}{
		\overline{
		\delta(q_{1},\sigma_{1},\gamma_{1},q',\rightarrow,e,\gamma',\omega)
		}		
		\delta(q_{2},\sigma_{2},\gamma_{2},q',\downarrow,e,\gamma',\omega)
		= 0}
	\end{array}			
\end{equation}
\normalsize
5. $ x_{i} = x_{j}-1 $ and $ y_{i} = y_{j}-1 $ ($ x_{i} = x_{j}+1 $ and $ y_{i} = y_{j}+1 $):
\footnotesize
\begin{equation}
	\begin{array}{ccl}
		{ \displaystyle \sum\limits_{q',\omega}} &
		\multicolumn{2}{l}{
		\overline{
		\delta(q_{1},\sigma_{1},\gamma_{1},q',\rightarrow,\downarrow,\gamma_{2},\omega)
		}		
		\delta(q_{2},\sigma_{2},\gamma_{3},q',\downarrow,\leftarrow,\gamma_{3},\omega)
		}
		\\
		& + &
		\overline{
		\delta(q_{1},\sigma_{1},\gamma_{1},q',\rightarrow,\rightarrow,\gamma_{2},\omega)
		}		
		\delta(q_{2},\sigma_{2},\gamma_{3},q',\downarrow,\downarrow,\gamma_{2},\omega)
		= 0
	\end{array}
\end{equation}
\normalsize
6. $ x_{i} = x_{j}-1 $ and $ y_{i} = y_{j}+1 $ ($ x_{i} = x_{j}+1 $ and $ y_{i} = y_{j}-1 $):
\footnotesize
\begin{equation}
	\begin{array}{ccl}
		{  \displaystyle \sum\limits_{q',\omega}} &
		\multicolumn{2}{l}{
		\overline{
		\delta(q_{1},\sigma_{1},\gamma_{1},q',\rightarrow,\downarrow,\gamma_{2},\omega)
		}		
		\delta(q_{2},\sigma_{2},\gamma_{3},q',\downarrow,\rightarrow,\gamma_{2},\omega)
		}
		\\
		& + &
		\overline{
		\delta(q_{1},\sigma_{1},\gamma_{1},q',\rightarrow,\leftarrow,\gamma_{1},\omega)
		}		
		\delta(q_{2},\sigma_{2},\gamma_{3},q',\downarrow,\downarrow,\gamma_{4},\omega)
		= 0
	\end{array}
\end{equation}
\normalsize
7. $ x_{i} = x_{j}-1 $ and $ y_{i} = y_{j}-2 $ ($ x_{i} = x_{j}+1 $ and $ y_{i} = y_{j}+2 $):
\footnotesize
\begin{equation}
	\begin{array}{ccl}
		{ \displaystyle \sum\limits_{q',\omega}} &
		\multicolumn{2}{l}{
		\overline{
		\delta(q_{1},\sigma_{1},\gamma_{1},q',\rightarrow,\rightarrow,\gamma_{2},\omega)
		}		
		\delta(q_{2},\sigma_{2},\gamma_{3},q',\downarrow,\leftarrow,\gamma_{3},\omega)
		= 0
		}	
	\end{array}
\end{equation}
\normalsize
8. $ x_{i} = x_{j}-1 $ and $ y_{i} = y_{j}+2 $ ($ x_{i} = x_{j}+1 $ and $ y_{i} = y_{j}-2 $):
\footnotesize
\begin{equation}
	\begin{array}{ccl}
		{ \displaystyle \sum\limits_{q',\omega}} &
		\multicolumn{2}{l}{
		\overline{
		\delta(q_{1},\sigma_{1},\gamma_{1},q',\rightarrow,\leftarrow,\gamma_{1},\omega)
		}		
		\delta(q_{2},\sigma_{2},\gamma_{3},q',\downarrow,\rightarrow,\gamma_{4},\omega)
		= 0}
	\end{array}
\end{equation}
\normalsize

\noindent \underline{Local conditions for RT-QPDA well-formedness}:

For RT-QPDA, we only consider the stack tape head positions.
In order to evolve to the same configuration in one step, the difference between
$ y_{i} $ and $ y_{j} $ must be at most 2.
Therefore, we obtain totally 3 different cases that completely define the restrictions on the transition function.
Similar to the previous ones, 
by taking conjugate of each summation, we obtain the symmetric cases that are shown in the parenthesis.
\newline \newline
For each choice of $ q_{1},q_{2} \in Q $; 
$ \sigma,\sigma_{1},\sigma_{2} \in \tilde{ \Sigma } $; and
$ \gamma_{1},\gamma_{2},\gamma_{3},\gamma_{4} \in \Gamma $
(the summations are taken over $ q' \in Q $; $ \gamma' \in \tilde{\Gamma} $;
$ e \in <\mspace{-4mu}> $; and $ \omega \in \Omega $),
\\\\
1. $ x_{i} = x_{j} $ and $ y_{i} = y_{j} $:
\footnotesize
\begin{equation}			
	\sum\limits_{q',e,\gamma^{'},\omega }
	\overline{
	\delta(q_{1},\sigma,\gamma_{1},q',e,\gamma^{'},\omega)}
	\delta(q_{2},\sigma,\gamma_{2},q',e,\gamma^{'},\omega)
	= 
	\left\lbrace
		\begin{array}{ll}
			1 ~~& q_{1} = q_{2} , \gamma_{1} = \gamma_{2} \\
			0 & \mbox{otherwise}
		\end{array}
	\right.
\end{equation}
\normalsize
2. $ x_{i} = x_{j} $ and $ y_{i} = y_{j}-1 $ ($ x_{i} = x_{j} $ and  $ y_{i} = y_{j}+1 $):
\footnotesize
\begin{equation}
	\begin{array}{ccl}
		\displaystyle \sum\limits_{q',\omega} &
		\multicolumn{2}{l}{
		\overline{
		\delta(q_{1},\sigma,\gamma_{1},q',\downarrow,\gamma_{2},\omega)
		}
		\delta(q_{2},\sigma,\gamma_{3},q',\leftarrow,\gamma_{3},\omega)}
		\\
		& + & 
		\overline{
		\delta(q_{1},\sigma,\gamma_{1},q',\rightarrow,\gamma_{2},\omega)
		}		
		\delta(q_{2},\sigma,\gamma_{3},q',\downarrow,\gamma_{2},\omega)
		= 0
	\end{array}
\end{equation}
\normalsize
3. $ x_{i} = x_{j} $ and $ y_{i} = y_{j}-2 $ ($ x_{i} = x_{j} $ and $ y_{i} = y_{j}+2 $):
\footnotesize
\begin{equation}
	\begin{array}{ccl}
		\displaystyle \sum\limits_{q',\omega} & 
		\multicolumn{2}{l}{
		\overline{
		\delta(q_{1},\sigma,\gamma_{1},q',\rightarrow,\gamma_{2},\omega)
		}		
		\delta(q_{2},\sigma,\gamma_{3},q',\leftarrow,\gamma_{3},\omega)
		= 0 }
	\end{array}			
\end{equation}
\normalsize

\chapter{QUANTUM COMPUTATION} \label{qc}

We review some basic concepts related to quantum computation in this appendix.
We refer the reader to \cite{NC00} for a complete reference.
We focus on finite dimensional systems.

The \textit{state space} of a quantum system ($ A $) with $ n $ \textit{classical states}, 
$ Q = \{ q_{i} \mid 1 \le i \le n \} $, is an $ n $-dimensional 
Hilbert space\footnote{It is a complex vector space with inner product.}, 
denoted as $ \mathcal{H}_{A} $.
The set $ \mathcal{B}_{A} = \{ \ket{q_{i}} \mid 1 \le i \le n \} $ is an orthonormal basis for $ \mathcal{H}_{A} $,
where the $ i^{th} $ entry of $ \ket{q_{i}} $ is 1 and the remaining ones are zeros.
If isolated, $ A $ is completely described by its \textit{state vector}, say $ \ket{\psi} $,
that is a linear combination of \textit{basis states}\footnote{We fixed it as $ \mathcal{B}_{A} $. 
However, note that, one can also use different orthonormal basis.}
\begin{equation}
	\ket{\psi} = \alpha_{1} \ket{q_{1}}+ \ldots + \alpha_{n} \ket{q_{n}}, 
\end{equation}
where $ \alpha_{i} $ is the amplitude of $ \ket{q_{i}} $, whose modulus squared, 
$ |\alpha_{i}|^{2} $, is the probability of being in state $ q_{i} $,
and $ \sum_{i} |\alpha_{i}|^{2}=1 $  ($ 1 \le i \le n $).
When $ \ket{\psi} $ contains more than one basis state with nonzero amplitude,
then it is said that $ A $ is in a \textit{superposition} (of the corresponding basis states).

The evolution of an isolated quantum system, i.e. closed quantum system, is governed by unitary transformations,
which are norm preserving operators.
If a unitary transformation, say $ U $, is applied on state $ \ket{\psi} $,
then the new state is obtained by
\begin{equation}
	\ket{\psi'} = U \ket{\psi}.
\end{equation}

To retrieve the information from a quantum system ($ A $), it is measured by a quantum operator
(in physical sense, it is observed by a measurement apparatus).
Throughout the thesis, Von Neumann measurements\footnote{They are a special case of 
projective measurements.} are used. 
That is, the set of the classical states $ (Q) $ are divided into some pairwise-disjoint subsets,
$ Q=Q_{1}+ \cdots +Q_{k} $ and so $ \mathcal{H}_{A} $ becomes a composition of mutually orthogonal subspaces,
$ \mathcal{H}_{A_{1}}, \ldots, \mathcal{H}_{A_{k}} $, i.e.
\begin{equation}
	\mathcal{H} = \mathcal{H}_{A_{1}} \oplus \cdots \oplus \mathcal{H}_{A_{k}}
\end{equation}
and
\begin{equation}
	\mathcal{H}_{A_{i}} = \mbox{span} \{ \ket{q} \mid q \in Q_{i} \}, 
\end{equation}
where $ 1 \le i \le k $.
Thus, we define Von Neumann measurement $ P $ with $ k $ outcomes as follows:
$ P $ is composed of a list of measurement operators, or specifically orthogonal projectors, 
$ P_{1}, \ldots, P_{k} $, defined as
\begin{equation}
	\label{equation:P}
	P = \{ P_{i} \mid P_{i} = \sum_{q \in Q_{i}} \ket{q} \bra{q}, 1 \le i \le k \}.
\end{equation}
If $ A $ is measured by $ P $ when it is in state $ \ket{\psi} $,
the outcome $ i \in \{1,\ldots,k\} $ is obtained with probability 
\begin{equation}
	|| \mspace{1mu} \ket{\psi_{i}} \mspace{1mu} ||^{2} = 
\braket{\psi_{i}}{\psi_{i}} =   \bra{\psi} P_{i} \ket{\psi},
\end{equation}
where $ \ket{\psi_{i}} = P_{i} \ket{\psi} $,
and then the system collapses to state
\begin{equation}
	\frac{\ket{\psi_{i}}}{\sqrt{\braket{\psi_{i}}{\psi_{i}}}},
\end{equation}
which is normalized.

Let $ A $ and $ B $ be quantum systems with sets of the classical states 
$ Q=\{q_{1},\ldots,q_{n}\} $ and $ R=\{ r_{1},\ldots,r_{m} \} $, respectively.
$ A \otimes B $, or shortly $ AB $, denotes the joint system composed by $ A $ and $ B $.
The state space of $ AB $ is denoted by $ \mathcal{H}_{AB} $,
which is obtained by $ \mathcal{H}_{A} \otimes \mathcal{H}_{B} $.
Similarly, $ \mathcal{B}_{AB} = \mathcal{B}_{A} \otimes \mathcal{B}_{B} $,
i.e. $ \mathcal{B}_{AB} = \{ \ket{q_{i}r_{j}} = \ket{q_{i}} \ket{r_{j}} \mid 1 \le i \le n, 1 \le j \le m  \} $.
If $ A $ is in state $ \ket{\psi_{A}} $ and $ B $ is in state $ \ket{\psi_{B}} $, then 
$ AB $ is in state $ \ket{\psi_{AB}} = \ket{\psi_{A}} \ket{\psi_{B}} $.
However, it is not always possible for a joint system ($ AB $) to be in a state that is 
a composition of two separate states belonging to the participant subsystems ($ A $  and $ B $, respectively).
This situation is named as the \textit{entanglement} of two subsystems.

In some cases, a quantum system can be in more than one state 
with some probabilities, i.e. $ \{ (p_{l},\ket{\psi_{l}})  \mid 1 < l \le M < \infty \} $,
where $ p_{l} $ is the probability of being in state $ \ket{\psi_{l}} $,
$ \sum_{l} p_{l}=1 $ and $ \braket{\psi_{l}}{\psi_{l}} = 1 $.
If the system is in exactly one state, then 
the system is said to be in \textit{pure state}, but the cases in which the system is in more 
then one pure state with some probabilities, an ensemble of pure states, 
are said to be in \textit{mixed states}. 
A convenient representation tool describing the state of a quantum system in mixed states
is the density matrix. 
The \textit{density matrix}\footnote{The trace of a density matrix is 1 
and each density matrix is positive semidefinite.} representation of 
$ \{ (p_{l},\ket{\psi_{l}}) \mid 1 \le l \le M < \infty \} $ is as follows:
\begin{equation}
	\rho = \sum_{l} p_{l} \ket{\psi_{l}} \bra{\psi_{l}}.
\end{equation}
If the system is observed with measurement $ P=\{P_{i} \mid 1 \le i \le k \} $, then the outcome $ i $ is obtained 
with probability
\begin{equation}
	tr( \tilde{\rho_{i}} ),
\end{equation}
where 
\begin{equation}
	\tilde{\rho_{i}} = P_{i} \rho P_{i}^{\dagger},
\end{equation}
and so the corresponding density matrix becomes
\begin{equation}
	\rho_{i} = \frac{ \tilde{\rho_{i}} }{tr( \tilde{\rho_{i}} )}.
\end{equation}
In fact, the mixture is replaced by one of these sub-mixtures with the corresponding probability after the measurement.
If the measurement is \textit{non-selective}, i.e. there is no distinction between the outcomes,
the new state of the system can be seen as a mixture of sub-mixtures
\begin{equation}
	\label{equation:mixture-of-density-matrices}
	\{ (tr(\tilde{\rho_{i}}), \frac{ \tilde{\rho_{i}} }{tr( \tilde{\rho_{i}} )} ) \mid 1 \le i \le k \}
\end{equation}
that forms the density matrix 
\begin{equation}
	\label{equation:union-of-mixture-of-density-matrices}
	\rho' = \sum_{i} \tilde{\rho_{i}} = \sum_{i} P_{i} \rho P_{i}^{\dagger}.
\end{equation}
In the case of \textit{selective} measurements,
the new states (sub-mixtures) are kept separately. 
Here, each sub-mixture can be considered as the state of a sub-computation obtained by splitting the computation
due to the measurement.

Note that, as seen from the above, $ \tilde{\rho_{i}} $ alone represents $ (tr(\tilde{\rho_{i}}),\rho_{i}) $.
It may be helpful in some contexts to consider the unnormalized state representation\footnote{It is still
positive semidefinite but the trace may be less than $ 1 $.} ($ \tilde{\rho_{i}} $), 
in which the probability of being in state, the trace of the matrix ($ tr(\tilde{\rho_{i}}) $), 
exists unconditionally.

By a suitable setup, a quantum system, as a part of a bigger closed system which is governed by
a unitary transformation, can evolve with respect to the general quantum operators, 
namely \textit{superoperators}, that form a superset of both unitary and classical (stochastic) operators.
That is, we compose the system of interest, \textit{the principal system}, 
with a system, \textit{an environment}, which is prepared in a specified state.
The whole system evolves unitarily and then the environment is discarded, which is measured in some cases 
before discarding. 
If the environment is selectively measured, the operator applied on the principal system is called
\textit{selective quantum operator}. Otherwise, it is called \textit{admissible operator}.
The above setup is known as \textit{open quantum systems}.

Specifically, in the case of admissible operators, 
if the principal system is in $ \rho $ and the environment is prepared in $ \rho_{env} $,
then a unitary operator, say $ U $, is applied on $ \rho \otimes \rho_{env} $.
After that, the environment is discarded and new state becomes $ \rho' = \mathcal{E}(\rho) $,
where $ \mathcal{E} $ is the admissible operator, a part of $ U $, applied on the principal system.
There are many equivalent 
representations of them. We follow \textit{operator-sum representation}.
An admissible operator $ \mathcal{E} $ is a finite collection of matrices (operation elements) 
$ \{ E_{i} \mid 1 \le i \le m \} $ satisfying that
\begin{equation}
	\label{equation:superoperator-constraint}
	\sum_{i} E_{i}^{\dagger} E_{i} = I.
\end{equation} 
When admissible operator $ \mathcal{E} $ is applied on a density matrix $ \rho $, then
the new density matrix is obtained by
\begin{equation}
	\label{equation:new-rho-obtained-by-superoperator}
	\rho' = \mathcal{E} (\rho) = \sum_{i} E_{i} \rho E_{i}^{\dagger}.
\end{equation}
One of the nice properties of the admissible operator is that when two of them are composed, 
we obtain a new one\footnote{Note that, the composed operator is implemented probably with a bigger environment.}.
If $ \mathcal{E}=\{E_{i} \mid 1 \le i \le m \} $ and 
$ \mathcal{E}'=\{ E'_{i} \mid 1 \le j \le m' \} $, then
\begin{equation}
	\mathcal{E'} \circ \mathcal{E} = \{ E_{j}'E_{i} \mid 1 \le i \le m, 1 \le j \le m' \}.
\end{equation}

In case of selective quantum operators, the environment is measured before the discarding.
In fact, selective quantum operators\footnote{We also refer to the reader to \cite{Wa03}.} 
are the operators in order to handle the case of the selective measurements.
Let $ \Delta $ be a finite set of outcomes. A selective quantum operator $ \mathcal{E} $
is a collection of matrices $ \{ E_{ \tau,i } \mid \tau \in \Delta, 1 \le i \le m_{\tau}  \} $
satisfying that
\begin{equation}
	\label{equation:selective-operator-constraint}
	\sum_{\tau,i} E_{\tau,i}^{\dagger} E_{\tau,i} = I.
\end{equation} 
We can classify the matrices of $ \mathcal{E} $ with respect to the elements of the set of the outcomes:
\begin{equation}
	\mathcal{E}_{\tau} = \{ E_{\tau,i} \}.
\end{equation}
When selective quantum operator $ \mathcal{E} $ is applied on a density matrix $ \rho $, then
the unnormalized matrix representing the sub-mixture $ \tau \in \Delta $ is obtained by
\begin{equation}
	\label{equation:new-rho-obtained-by-selective-operator}
	\tilde{\rho_{\tau}} = \tilde{\mathcal{E}_{\tau}} (\rho) = \sum_{i} E_{\tau,i} \rho E_{\tau,i}^{\dagger},
\end{equation}
with probability $ tr(\tilde{\rho_{\tau}}) $.
The normalized version is obtained as 
\begin{equation}
	\rho_{\tau} = \mathcal{E}_{\tau}(\rho) = \frac{\tilde{\rho_{\tau}}}{tr(\tilde{\rho_{\tau}})}.
\end{equation}

The distinction between admissible and selective quantum
operators as an issue of this thesis is that
a space-bounded computation (which may go on forever), should be
checked regularly to see whether it has halted or not.
However, note that,
if the computation is time bounded, all sub-mixtures are kept alive
until the end of the computation,
and so selective quantum operators can be replaced with admissible
ones with the addition of some suitable measurements.

\begin{figure}[h!]
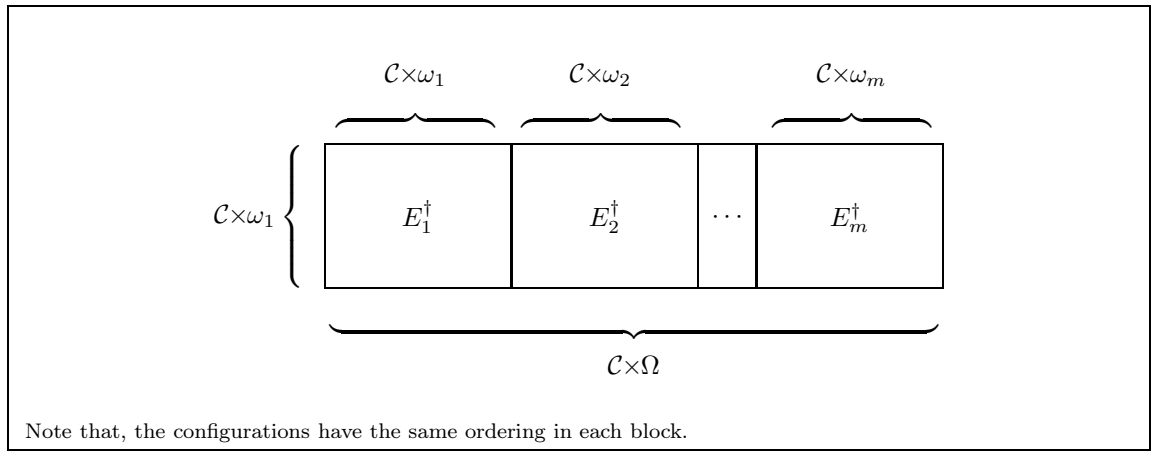

	\footnotesize
	\begin{center}
	\fbox{
	\begin{minipage}{0.95\textwidth}	
	\begin{equation*}
		\begin{array}{cc}
			& 
			\begin{array}{cccc}
				\mathcal{C} \mspace{-5mu} \times \mspace{-5mu} \omega_{1} & 
				\mathcal{C} \mspace{-5mu} \times \mspace{-5mu} \omega_{2} & & 
				\mathcal{C} \mspace{-5mu} \times \mspace{-5mu} \omega_{m} 
				\\ 
				\overbrace{\hspace*{60pt}} & \overbrace{\hspace*{60pt}} &
				\hspace*{15pt}
				& \overbrace{\hspace*{60pt}} \\
			\end{array}
			\\		
			\mathcal{C} \mspace{-5mu} \times \mspace{-5mu} \omega_{1}
			\left \lbrace \mspace{-25mu}
			\begin{array}{p{0.1pt}} \\  \\ \\  \end{array}
			\right. 
			&
			\begin{array}{cccc}
				\hline
				\multicolumn{1}{|c|}{
				\begin{array}{c}  \hspace*{50pt}  \\ E_{1}^{\dagger} \\  \hspace*{50pt}  \\  \end{array}
				}
				& 
				\multicolumn{1}{|c|}{
					\begin{array}{c}  \hspace*{50pt} \\ E_{2}^{\dagger}  \\ \hspace*{50pt}  \\  \end{array}
				}
				& \cdots & 
				\multicolumn{1}{|c|}{
					\begin{array}{c}  \hspace*{50pt}  \\ E_{m}^{\dagger}  \\ \hspace*{50pt}  \\  \end{array}
				} 			
				\\ \hline
			\end{array}
			\\
			&  \underbrace{ \hspace*{230pt}}
			\\
			& \mathcal{C} \mspace{-5mu} \times \mspace{-5mu} \Omega
		\end{array}
	\end{equation*}
	\scriptsize{Note that, the configurations have the same ordering in each block.}
	\end{minipage}	
	}
	\end{center}
	\caption{The matrix obtained by deleting all rows of $ U^{\dagger} $ except the ones in	
		$ \mathcal{C} \mspace{-5mu} \times \mspace{-5mu} \omega_{1} $.}
	\vskip\baselineskip
	\label{figure:U1}
\end{figure}

In the remaining part, we present the general setup implementing superoperators.
Let $ \mathcal{H}_{\mathcal{C}} $ and $ \mathcal{H}_{\Omega} $ be the principal and environment systems, respectively,
where $ \mathcal{C}=\{c_{1},\ldots,c_{n}\} $ is the finite dimensional configuration set
and $ \Omega=\{\omega_{1},\ldots,\omega_{m}\} $, having a special symbol 
$ \omega_{1} $, \textit{the initial symbol}, is the set of symbols for a finite register.
Suppose that $ \mathcal{H}_{\mathcal{C}} $ is in (mixed) state $ \rho $.
We prepare $ \mathcal{H}_{\Omega} $ with $ \ket{\omega_{1}} $ and apply unitary operator $ U $
on the joint system $ \mathcal{H}_{\mathcal{C}} \otimes \mathcal{H}_{\Omega} $, then we discard 
the environment $ \mathcal{H}_{\Omega} $. In Figure \ref{figure:U1}, the parts of $ U $ effecting
$ \mathcal{H}_{\mathcal{C}} $, which form an admissible operator $ \mathcal{E} $ with operation elements
$ \{E_{1},\ldots,E_{m}\} $, are seen. 
A straightforward calculation can validate that
$ \mathcal{E} $ satisfies Equation \ref{equation:superoperator-constraint} and 
the new density matrix of $ \mathcal{H}_{\mathcal{C}} $ is obtained as in Equation 
\ref{equation:new-rho-obtained-by-superoperator}.

In order to handle selective quantum operators, 
we add the following specifications to the \textit{general setup}:
\begin{enumerate}
	\item $ \Omega $ is partitioned into $ |\Delta| = k $ disjoint subsets, 
		$ \Omega_{\tau_{1}}, \ldots , \Omega_{\tau_{k}} $, i.e.
		$ |\Omega_{\tau_{i}}| = m_{i} $, $ 1 \le i \le k $, and 
	\item the environment is measured by $ P=\{ P_{\tau \in \Delta } \mid 
		P_{\tau} = \sum_{\omega \in \Omega_{\tau}} \ket{\omega} \bra{\omega} \} $ before being discarded.
\end{enumerate}
Hence, the operation elements that are explicitly shown in Figure \ref{figure:U1} are
grouped with respect to the partitioning of $ \Omega $ (see Figure \ref{figure:U1-re-partioned}), 
i.e., $ \mathcal{E} = \{ \mathcal{E}_{\tau} \} $ and 
$ \mathcal{E}_{\tau} = \{ E_{i} \mid \omega_{i} \in \Omega_{\tau} \} 
= \{ E_{\tau,j} \mid 1 \le j \le |\Omega_{\tau}| \} $.
Note that, the cardinality of $ \Omega_{\tau} $'s may not be the same.

\begin{figure}[h!]
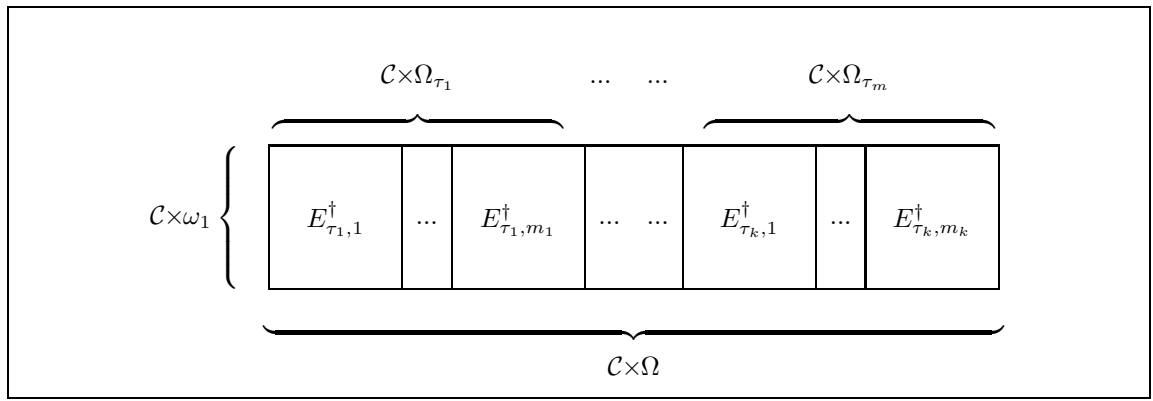

	\footnotesize
	\begin{center}
	\fbox{
	\begin{minipage}{0.95\textwidth}	
	\begin{equation*}
		\begin{array}{cc}
			& 
			\begin{array}{cccc}
				\mathcal{C} \mspace{-5mu} \times \mspace{-5mu} \Omega_{\tau_{1}} 
				& ... & ... & 
				\mathcal{C} \mspace{-5mu} \times \mspace{-5mu} \Omega_{\tau_{m}} 
				\\ 
				\overbrace{\hspace*{110pt}} &  & 
				\hspace*{15pt}
				& \overbrace{\hspace*{110pt}} \\
			\end{array}
			\\		
			\mathcal{C} \mspace{-5mu} \times \mspace{-5mu} \omega_{1}
			\left \lbrace \mspace{-25mu}
			\begin{array}{p{0.1pt}} \\ \\ \\  \end{array}
			\right. 
			&
			\begin{array}{cccccccc}
				\hline
				\multicolumn{1}{|c|}{
				\begin{array}{c}  \hspace*{30pt} \\ E_{\tau_{1},1}^{\dagger} \\ \hspace*{30pt}  \\  \end{array}
				}
				& 
				...
				&
				\multicolumn{1}{|c|}{
					\begin{array}{c}  \hspace*{30pt} \\ E_{\tau_{1},m_{1}}^{\dagger} \\ \hspace*{30pt}  \\  \end{array}
				}
				& ... & ... &
				\multicolumn{1}{|c|}{
					\begin{array}{c}  \hspace*{30pt} \\ E_{\tau_{k},1}^{\dagger} \\ \hspace*{30pt}  \\  \end{array}
				} 
				& ... &
				\multicolumn{1}{|c|}{
					\begin{array}{c}  \hspace*{30pt} \\ E_{\tau_{k},m_{k}}^{\dagger} \\ \hspace*{30pt}  \\  \end{array}
				} 			
				\\ \hline
			\end{array}
			\\
			&  \underbrace{ \hspace*{280pt}}
			\\
			& \mathcal{C} \mspace{-5mu} \times \mspace{-5mu} \Omega
		\end{array}
	\end{equation*}
	\end{minipage}	
	}
	\end{center}
	\caption{Re-partitioning of the matrix in Figure \ref{figure:U1}}
	\vskip\baselineskip
	\label{figure:U1-re-partioned}
\end{figure}

\bibliographystyle{fbe_tez_v12}
\bibliography{YakaryilmazSay}
%
%
\end{document}